\documentclass[a4paper, 11pt,twosided]{article}

\usepackage{amsmath}
\usepackage{amssymb}
\usepackage{amsfonts}
\usepackage{amsthm}
\usepackage{graphicx}
\usepackage[utf8]{inputenc}

\usepackage{enumitem}
\usepackage{comment}
\usepackage{stmaryrd}
\usepackage{array,multirow}
\usepackage{xparse} 

\pdfoutput=1

\usepackage[linesnumbered,algoruled,vlined]{algorithm2e}

\usepackage{listings}
\usepackage{cite}

\usepackage[textwidth=16cm,height=24cm]{geometry}

\usepackage[bf,small]{caption} 
\usepackage{url}
\usepackage{thmtools,thm-restate}
\usepackage{stmaryrd}

\title{    
\Large{
Masterarbeit im Studiengang Computer Science\\
\vspace{1cm}
}
\LARGE{Virtual Network Embedding via Decomposable LP Formulations:}\\
\Large{Orientations of Small Extraction Width and Beyond}   
}
\author{
\vspace{2cm}\\
\Large{Elias Döhne}\\
\vspace{0.5cm}\\
Technische Universität Berlin\\
\vspace{0.1cm}\\
\includegraphics[height=1.5cm]{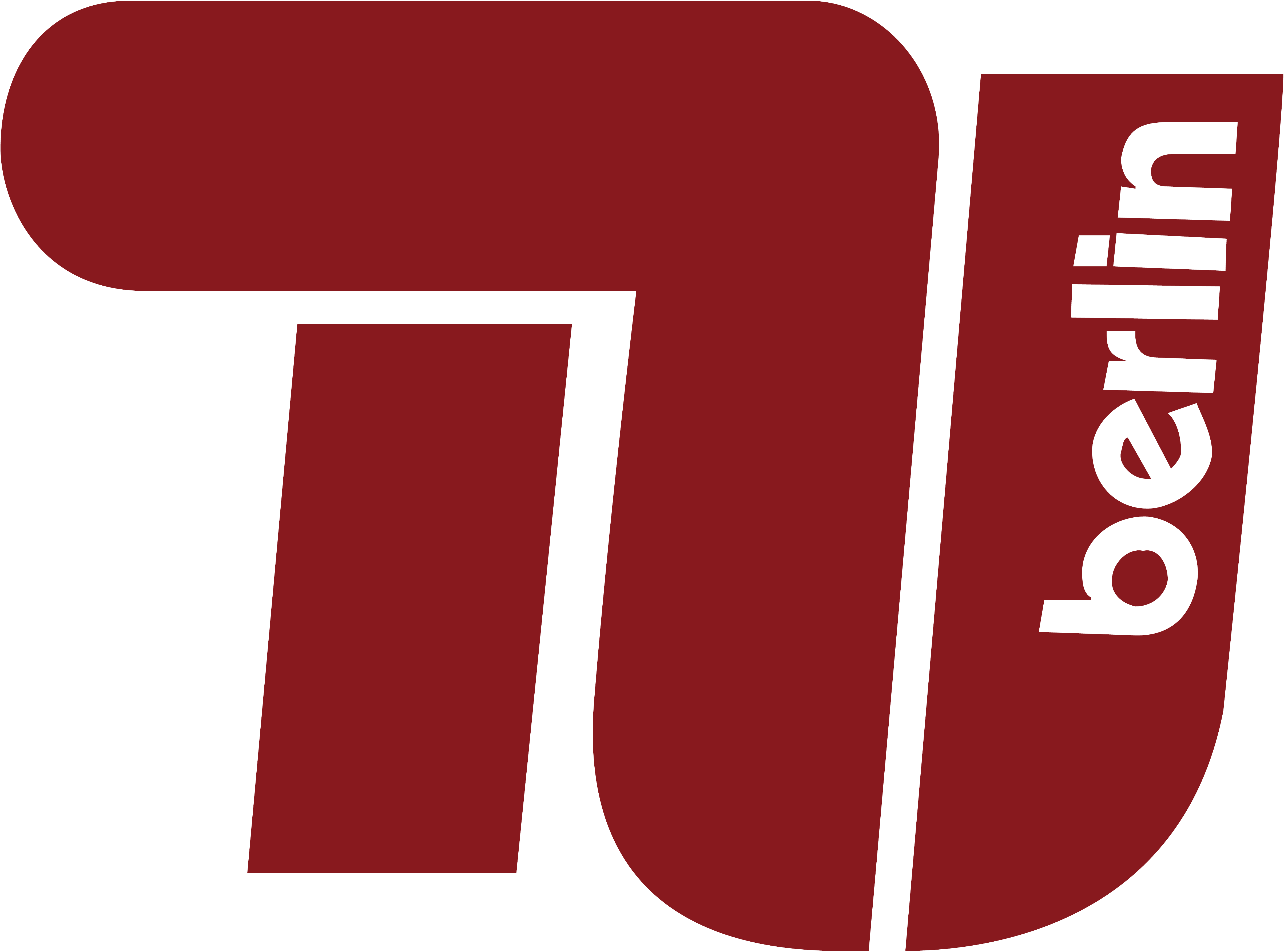}
\vspace{2cm} \\
Gutachter:\\
Prof. Anja Feldmann, Ph.D. \\
Univ.-Prof. Dr.sc. Stefan Schmid \\
\vspace{0.5cm}\\
Betreuer:\\
Matthias Rost, M.Sc.
\vspace{1cm}
}
\date{12. Juni 2018}

\newcommand{\substrateTopology}{\ensuremath{G_S}}
\newcommand{\substrateNodes}{\ensuremath{V_S}}
\newcommand{\substrateEdges}{\ensuremath{E_S}}
\newcommand{\substrateCapacity}{\ensuremath{c_{S}}}
\newcommand{\SRV}{\ensuremath{R^V_{S}}}   
\newcommand{\SR}{\ensuremath{R_{S}}}  
\newcommand{\substrateResources}{\ensuremath{\SR}}  
\newcommand{\substrateTopologyDef}{\ensuremath{\substrateTopology = (\substrateNodes, \substrateEdges)}}
\newcommand{\pathSet}{\ensuremath{\mathcal{P}}}
\newcommand{\substratePathSet}{\ensuremath{\pathSet_S}}
\newcommand{\PowerSet}{\ensuremath{\mathbb{P}}}

\newcommand{\conv}{\ensuremath{\text{conv}}}

\newcommand{\complexityNP}{\ensuremath{\mathcal{NP}}}
\newcommand{\bigO}{\ensuremath{\mathcal{O}}}

\newcommand{\VNEPInstance}{\ensuremath{(\reqTopology, \substrateTopology)}}
\newcommand{\lpvars}{\ensuremath{(\vec{y}, \vec{z}, \vec{\gamma}, \vec{a})}}
\newcommand{\nodeTypeSet}{\ensuremath{\mathcal{T}}}
\newcommand{\nodeType}{\ensuremath{\tau}}
\newcommand{\reqNodeType}{\ensuremath{\nodeType_R}}
\newcommand{\subNodeType}{\ensuremath{\nodeType_S}}
\newcommand{\allocationFunction}{\ensuremath{A}}
\DeclareDocumentCommand{\substrateNodesByType}{O{\nodeType}}{\ensuremath{\substrateNodes^{#1}}}

\newcommand{\hypergraph}{\ensuremath{\mathcal{H}}}
\newcommand{\hypergraphNodes}{\ensuremath{V}}
\newcommand{\hypergraphEdges}{\ensuremath{\mathcal{E}}}
\newcommand{\hypergraphDef}{\ensuremath{\hypergraph = (\hypergraphNodes, \hypergraphEdges)}}

\newcommand{\reqTopology}{\ensuremath{G_R}}
\newcommand{\reqNodes}{\ensuremath{V_R}}
\newcommand{\reqEdges}{\ensuremath{E_R}}
\newcommand{\reqTopologyDef}{\ensuremath{\reqTopology = (\reqNodes, \reqEdges)}}
\newcommand{\reqDemand}{\ensuremath{c_{R}}}

\DeclareDocumentCommand{\undirected}{O{G}}{\ensuremath{\overline{#1}}}
\newcommand{\undirectedGraph}{\ensuremath{\undirected[G]}}
\newcommand{\undirectedEdges}{\ensuremath{\undirected[E]}}
\newcommand{\reqUndirectedTopology}{\ensuremath{\undirectedGraph_R}}

\newcommand{\inEdgesNoArg}{\ensuremath{\delta^- }}
\newcommand{\outEdgesNoArg}{\ensuremath{\delta^+ }}
\newcommand{\inEdges}[1]{\ensuremath{\inEdgesNoArg(#1) }}
\newcommand{\outEdges}[1]{\ensuremath{\outEdgesNoArg(#1) }}

\newcommand{\extractionOrderCharacter}{\ensuremath{\mathcal{X}}}
\newcommand{\reqExtractionOrder}{\ensuremath{ {\reqTopology^\extractionOrderCharacter} }}
\newcommand{\reqExtractionOrderEdges}{\ensuremath{ {\reqEdges^\extractionOrderCharacter} }}
\newcommand{\reqExtractionOrderRoot}{\ensuremath{{r^\extractionOrderCharacter}}}
\newcommand{\reqExtractionOrderDef}{\ensuremath{\reqExtractionOrder = (\reqNodes, \reqExtractionOrderEdges, \reqExtractionOrderRoot)}}

\DeclareDocumentCommand{\reqEOLabelSubgraph}{O{k}}{\ensuremath{\reqExtractionOrder(#1)}}
\DeclareDocumentCommand{\reqEOLabelSubgraphNodes}{O{k}}{\ensuremath{\reqNodes^\extractionOrderCharacter(#1)}}
\DeclareDocumentCommand{\reqEOLabelSubgraphEdges}{O{k}}{\ensuremath{\reqExtractionOrderEdges (#1)}}
\DeclareDocumentCommand{\reqEOLabelSubgraphRoot}{O{k}}{\ensuremath{r^{#1}}}
\DeclareDocumentCommand{\reqEOLabelSubgraphDef}{O{k}}{\ensuremath{\reqEOLabelSubgraph[#1] = \left(\reqEOLabelSubgraphNodes[#1], \reqEOLabelSubgraphEdges[#1] \right)}}

\newcommand{\genExtractionOrderChar}{\ensuremath{{\mathcal{A}}}}
\newcommand{\reqDAGOrientation}{\ensuremath{\reqTopology^\genExtractionOrderChar}}
\newcommand{\reqDAGOrientationEdges}{\ensuremath{\reqEdges^\genExtractionOrderChar}}
\newcommand{\reqDAGOrientationDef}{\ensuremath{\reqDAGOrientation = (\reqNodes, \reqDAGOrientationEdges)}}

\DeclareDocumentCommand{\reqInducedEO}{O{\reqDAGOrientation}}{\ensuremath{\reqTopology^\extractionOrderCharacter (#1)}}

\DeclareDocumentCommand{\reqDAGAncestorSet}{O{V}}{\ensuremath{V_A(#1)}}
\DeclareDocumentCommand{\reqDAGAncestorNodes}{O{V}}{\reqDAGAncestorSet[#1]}
\DeclareDocumentCommand{\reqDAGAncestorGraph}{O{V}}{\ensuremath{G_A^{\extractionOrderCharacter}(#1)}}
\DeclareDocumentCommand{\reqDAGAncestorEdges}{O{V}}{\ensuremath{E_A^{\extractionOrderCharacter}(#1)}}
\DeclareDocumentCommand{\reqDAGAncestorGraphDef}{O{V}}{\ensuremath{\reqDAGAncestorGraph[#1] = (\reqDAGAncestorSet[#1], \reqDAGAncestorEdges[#1])}}

\newcommand{\superroot}{\ensuremath{{s^\extractionOrderCharacter}}}
\newcommand{\rootRegionSet}{\ensuremath{\mathcal{R}^\genExtractionOrderChar}}
\newcommand{\rootSet}{\ensuremath{R^\genExtractionOrderChar}}

\DeclareDocumentCommand{\labelsetEdgeExtended}{O{e}}{\ensuremath{\singleLabelset^\genExtractionOrderChar_{#1}}}
\DeclareDocumentCommand{\labelsetIncomingExtended}{O{i}}{\ensuremath{\labelsetIndexed{#1}^-}}
\DeclareDocumentCommand{\rootRegionSubgraphExtended}{O{r}}{\ensuremath{{G^{\extractionOrderCharacter, \text{ext}}_{ #1}}}}
\DeclareDocumentCommand{\rootRegionExtended}{O{r}}{\ensuremath{E^{\genExtractionOrderChar, \text{ext}}_{ #1}}}

\DeclareDocumentCommand{\rootRegionSubgraph}{O{r}}{\ensuremath{{G^\extractionOrderCharacter_{ #1}}}}
\DeclareDocumentCommand{\reqTopologyRootRegionSubgraph}{O{r}}{\ensuremath{{\reqTopology}( #1)}}
\DeclareDocumentCommand{\rootRegionNodes}{O{r}}{\ensuremath{V_{ #1}}}
\DeclareDocumentCommand{\rootRegion}{O{r}}{\ensuremath{{E^\genExtractionOrderChar_{ #1}}}}
\DeclareDocumentCommand{\rootRegionSubgraphDef}{O{r}}{\ensuremath{\rootRegionSubgraph[#1] = (\rootRegionNodes[#1], \rootRegion[#1], #1)}}

\DeclareDocumentCommand{\rootRegionSubgraphReq}{O{r}}{\ensuremath{{\reqTopology(#1)}}}
\DeclareDocumentCommand{\rootRegionSubgraphEdgesReq}{O{r}}{\ensuremath{{\reqEdges^{#1}}}}
\DeclareDocumentCommand{\rootRegionSubgraphReqDef}{O{r}}{\ensuremath{\rootRegionSubgraphReq = (\rootRegionNodes[#1], \rootRegionSubgraphEdgesReq[#1])}}

\DeclareDocumentCommand{\rootRegionBoundary}{O{r}}{\ensuremath{B_{#1}}}
\DeclareDocumentCommand{\rootRegionBoundaryIncoming}{O{r}}{\ensuremath{B^-_{#1}}}
\DeclareDocumentCommand{\rootRegionBoundaryPair}{O{r_1}O{r_2}}{\ensuremath{B_{#1, #2}}}
\DeclareDocumentCommand{\rootRegionBoundaryPairOrdering}{O{r_1}O{r_2}}{\ensuremath{O^B_{#1, #2}}}

\DeclareDocumentCommand{\boundaryAncestorSet}{O{r_1}O{r_2}O{n}}{\ensuremath{A_{#1, #2}^{#3}}}
\DeclareDocumentCommand{\incomingRRBoundary}{O{r}}{\rootRegionBoundary[#1]^{-}}
\DeclareDocumentCommand{\outgoingRRBoundary}{O{r}}{\rootRegionBoundary[#1]^{+}}

\DeclareDocumentCommand{\topdownMappingNodes}{O{i}}{\ensuremath{m_{V, #1}^\downarrow}}
\DeclareDocumentCommand{\topdownMappingEdges}{O{i}}{\ensuremath{m_{E, #1}^\downarrow}}
\DeclareDocumentCommand{\topdownMappingDef}{O{i}}{\ensuremath{m_{#1}^\downarrow = \topdownMappingEdges[#1], \topdownMappingEdges[#1]}}

\newcommand{\RRExtractionOrder}{\ensuremath{G^\mathcal{G}}}
\newcommand{\RRExtractionOrderEdges}{\ensuremath{E^{\mathcal{G}}_\genExtractionOrderChar}}
\newcommand{\RRExtractionOrderRoot}{\ensuremath{r^{\mathcal{G}}}}
\newcommand{\RRExtractionOrderDef}{\ensuremath{\RRExtractionOrder = (\rootSet, \RRExtractionOrderEdges, \RRExtractionOrderRoot)}}

\newcommand{\inEdgesRR}[1]{\ensuremath{\inEdgesNoArg_\mathcal{G}(#1) }}
\newcommand{\outEdgesRR}[1]{\ensuremath{\outEdgesNoArg_\mathcal{G}(#1) }}
\newcommand{\handledRREdges}[1]{\ensuremath{ E^{RR} }}

\newcommand{\inEdgesDAGOrderNoArg}{\ensuremath{\inEdgesNoArg_\genExtractionOrderChar }}
\newcommand{\outEdgesDAGOrderNoArg}{\ensuremath{\outEdgesNoArg_\genExtractionOrderChar }}

\newcommand{\inEdgesDAGOrder}[1]{\ensuremath{\inEdgesDAGOrderNoArg(#1) }}
\newcommand{\outEdgesDAGOrder}[1]{\ensuremath{\outEdgesDAGOrderNoArg(#1) }}

\DeclareDocumentCommand{\VGe}{O{e}}{\ensuremath{G_{e}}}
\DeclareDocumentCommand{\VVe}{O{e}}{\ensuremath{V_{e}}}
\DeclareDocumentCommand{\VEe}{O{e}}{\ensuremath{E_{e}}}

\newcommand{\inEdgesExtractionOrderNoArg}{\ensuremath{\inEdgesNoArg_\extractionOrderCharacter }}
\newcommand{\outEdgesExtractionOrderNoArg}{\ensuremath{\outEdgesNoArg_\extractionOrderCharacter }}

\newcommand{\EOEdgeToOriginal}{\ensuremath{\vec{E}_\extractionOrderCharacter}}
\newcommand{\inEdgesExtractionOrder}[1]{\ensuremath{\inEdgesExtractionOrderNoArg(#1) }}
\newcommand{\outEdgesExtractionOrder}[1]{\ensuremath{\outEdgesExtractionOrderNoArg(#1) }}

\DeclareDocumentCommand{\mappingChar}{O{k}}{m}
\newcommand{\mappingRequest}{\ensuremath{\mappingChar_R}}
\newcommand{\mappingNodes}{\ensuremath{\mappingChar_V}}
\newcommand{\mappingEdges}{\ensuremath{\mappingChar_E}}
\newcommand{\mappingRequestDef}{\ensuremath{\mappingRequest = (\mappingNodes, \mappingEdges)}}

\newcommand{\mappingRequestTilde}{\ensuremath{\mappingChar_R'}}
\newcommand{\mappingNodesTilde}{\ensuremath{\mappingChar_V'}}
\newcommand{\mappingEdgesTilde}{\ensuremath{\mappingChar_E'}}

\DeclareDocumentCommand{\RRmappingRequestIteration}{O{r}}{\ensuremath{\mappingRequest^{#1}}}
\DeclareDocumentCommand{\RRmappingNodesIteration}{O{r}}{\ensuremath{\mappingNodes^{#1}}}
\DeclareDocumentCommand{\RRmappingEdgesIteration}{O{r}}{\ensuremath{\mappingEdges^{#1}}}

\DeclareDocumentCommand{\mappingRequestIteration}{O{k}}{\ensuremath{\mappingRequest^{#1}}}
\DeclareDocumentCommand{\mappingNodesIteration}{O{k}}{\ensuremath{\mappingNodes^{#1}}}
\DeclareDocumentCommand{\mappingEdgesIteration}{O{k}}{\ensuremath{\mappingEdges^{#1}}}
\DeclareDocumentCommand{\mappingRequestIterationDef}{O{k}}{\ensuremath{\mappingRequestIteration[#1] = (\mappingNodesIteration[#1], \mappingEdgesIteration[k])}}

\DeclareDocumentCommand{\probSubscript}{O{k}}{\ensuremath{f_{#1}}}
\DeclareDocumentCommand{\mappingIteration}{O{k}}{\ensuremath{m_{#1}}}
\DeclareDocumentCommand{\mappingNodesIterationNoR}{O{k}}{\ensuremath{\mappingNodes^{#1}}}
\DeclareDocumentCommand{\mappingEdgesIterationNoR}{O{k}}{\ensuremath{\mappingEdges^{#1}}}

\DeclareDocumentCommand{\mappingEdgeTilde}{O{k}}{\widetilde{m}_e}

\DeclareDocumentCommand{\VEbfsBags}{O{i} }{\ensuremath{\mathcal{B}^{+}_{#1}}}
\DeclareDocumentCommand{\VEbfsLabelsEdgeCapBag}{O{b} O{e}}{\ensuremath{\singleLabelset^{\extractionOrderCharacter}_{#1 \cap #2}}}
\newcommand{\BagIPlus}{\ensuremath{\mathcal{B}^+_i}}
\DeclareDocumentCommand{\outEdgeBag}{ O{i} O{b}}{\ensuremath{B_{#1, #2}}}

\DeclareDocumentCommand{\bagLabelSets}{ O{i}}{\ensuremath{\mathcal{L}^B_{#1}}}
\DeclareDocumentCommand{\bagLabelSet}{ O{i} O{b}}{\ensuremath{\labelsetIndexed{#1, #2}}}

\newcommand{\labelsets}{\ensuremath{\mathcal{L}}}
\newcommand{\labelsetOrder}{\ensuremath{\omega}}
\newcommand{\labelsetOrderSet}{\ensuremath{\Omega}}
\newcommand{\labelsetOrderSetSet}{\ensuremath{\mathcal{O}}}

\newcommand{\bagLabelsetOrderSet}{\ensuremath{\Omega^B}}

\newcommand{\labelsetOrderTilde}{\ensuremath{\widetilde{\omega}}}
\newcommand{\labelsetOrderSetTilde}{\ensuremath{\widetilde{\Omega}}}

\newcommand{\labelsetPredecessors}{\ensuremath{L^p}}
\newcommand{\labelsetPredecessorFunction}{\ensuremath{\pi}}

\DeclareDocumentCommand{\mapVedge}{O{e}}{\ensuremath{m_{#1}}}
\DeclareDocumentCommand{\mapVbag}{O{\alpha}}{\ensuremath{m_{#1}}}
\DeclareDocumentCommand{\mapVinter}{O{e} O{a}}{\ensuremath{m_{#1 \cap #2}}}

\newcommand{\mappingPredecessors}{\ensuremath{m^p}}

\DeclareDocumentCommand{\labelsetEdge}{O{e}}{\ensuremath{{\singleLabelset^\extractionOrderCharacter_{#1}}}}
\DeclareDocumentCommand{\labelsetIncoming}{O{i}}{\ensuremath{\labelsetIndexed{#1}^-}}

\newcommand{\labelsetsEdgesExtractionOrder}{\ensuremath{{\mathcal{L}^\extractionOrderCharacter_E}}}
\DeclareDocumentCommand{\labelsetsEdges}{ O{E} }{\ensuremath{\mathcal{L}_{#1}}}

\DeclareDocumentCommand{\labelsetsEdgesRootRegion}{O{r}}{\ensuremath{\mathcal{L}^\extractionOrderCharacter_{#1}}}
\DeclareDocumentCommand{\labelsetsEdgesTilde}{ O{E} }{\ensuremath{\widetilde{\mathcal{L}}_{#1}}}
\DeclareDocumentCommand{\labelsetEdgeTilde}{ O{e} }{\ensuremath{{\widetilde{\singleLabelset}^\extractionOrderCharacter_{#1}}}}

\newcommand{\singleLabelset}{\ensuremath{L}}
\newcommand{\labelSetIndex}{\ensuremath{a}}
\newcommand{\labelSetIndexTwo}{\ensuremath{b}}
\newcommand{\labelSetIndexThree}{\ensuremath{c}}
\newcommand{\labelsetIndexed}[1]{\ensuremath{\singleLabelset_{#1}}}
\DeclareDocumentCommand{\labelsetRepresentative}{O{e}}{\ensuremath{\singleLabelset^\text{rep}_{ #1}}}
\newcommand{\labelsetmappingIndex}{\ensuremath{\alpha}}
\newcommand{\labelsetmappingIndexTwo}{\ensuremath{\beta}}
\newcommand{\mappedRequestNodes}{\ensuremath{\reqNodes^{m}}}

\newcommand{\ytilde}{\ensuremath{\tilde{y}}}
\newcommand{\gammatilde}{\ensuremath{\tilde{\gamma}}}
\newcommand{\ztilde}{\ensuremath{\tilde{z}}}
\newcommand{\atilde}{\ensuremath{\tilde{a}}}
\newcommand{\lpvarstilde}{\ensuremath{(\vec{\ytilde}, \vec{\ztilde}, \vec{\gammatilde}, \vec{\atilde})}}

\DeclareDocumentCommand{\prob}{O{k}}{\ensuremath{f^{#1}}}

\DeclareDocumentCommand{\confluence}{O{i} O{j}}{\ensuremath{C^\extractionOrderCharacter_{#1, #2} }}
\newcommand{\confluenceSet}{\ensuremath{\mathcal{C}^\extractionOrderCharacter}}

\DeclareDocumentCommand{\path}{O{i} O{j}}{\ensuremath{P_{#1, #2}}}

\DeclareDocumentCommand{\MappingSpace}{d[]}{\ensuremath{\mathcal{M}(#1)}}
\DeclareDocumentCommand{\restrict}{d[] d[]}{\ensuremath{\langle #1 \vert #2 \rangle}}
\DeclareDocumentCommand{\subLP}{d[] d[]}{\ensuremath{#1 \llbracket #2 \rrbracket}}

\DeclareDocumentCommand{\labelSetGraph}{O{\labelsets}}{\ensuremath{G_\labelsets(#1)}}
\DeclareDocumentCommand{\labelSetGraphNodes}{O{\labelsets}}{\ensuremath{V_{#1}}}
\DeclareDocumentCommand{\labelSetGraphEdges}{O{\labelsets}}{\ensuremath{E_{#1}}}
\DeclareDocumentCommand{\labelSetGraphDef}{O{\labelsets}}{\ensuremath{\labelSetGraph[#1] = (\labelSetGraphNodes[#1], \labelSetGraphEdges[#1])}}

\DeclareDocumentCommand{\labelsetsIncident}{O{i}}{\ensuremath{\labelsets^+_{ #1 }}}

\newcommand{\genericGraph}{\ensuremath{G}}
\newcommand{\genericGraphNodes}{\ensuremath{V}}
\newcommand{\genericGraphEdges}{\ensuremath{E}}
\newcommand{\genericGraphDef}{\ensuremath{\genericGraph = (\genericGraphNodes, \genericGraphEdges)}}

\newcommand{\halfwheelGraph}{\ensuremath{G_w}}
\newcommand{\halfwheelGraphNodes}{\ensuremath{V_w}}
\newcommand{\halfwheelGraphEdges}{\ensuremath{E_w}}
\newcommand{\halfwheelGraphDef}{\ensuremath{\halfwheelGraph = (\halfwheelGraphNodes, \halfwheelGraphEdges)}}

\newcommand{\setFamily}{\ensuremath{\mathcal{F}}}
\newcommand{\setFamilySet}{\ensuremath{S}}
\newcommand{\setFamilyDef}{\ensuremath{\setFamily = \{\setFamilySet_1, \ldots \setFamilySet_n\}}}

\newcommand{\hgPrimalGraph}{\ensuremath{G_P}}
\newcommand{\hgPrimalGraphNodes}{\ensuremath{V}}
\newcommand{\hgPrimalGraphEdges}{\ensuremath{E_P}}
\newcommand{\hgPrimalGraphDef}{\ensuremath{\hgPrimalGraph = (\hgPrimalGraphNodes, \hgPrimalGraphEdges)}}

\newcommand{\hgIncidenceGraph}{\ensuremath{G_I}}
\newcommand{\hgIncidenceGraphNodes}{\ensuremath{V_I}}
\newcommand{\hgIncidenceGraphEdges}{\ensuremath{E_I}}
\newcommand{\hgIncidenceGraphDef}{\ensuremath{\hgIncidenceGraph = (\hgIncidenceGraphNodes, \hgIncidenceGraphEdges)}}

\newcommand{\treeDecompTree}{\ensuremath{G_T}}
\newcommand{\treeDecompTreeNodes}{\ensuremath{V_T}}
\newcommand{\treeDecompTreeEdges}{\ensuremath{E_T}}
\newcommand{\treeDecompTreeDef}{\ensuremath{\treeDecompTree = (\treeDecompTreeNodes, \treeDecompTreeEdges)}}

\newcommand{\treeDecomp}{\ensuremath{T}}
\newcommand{\treeDecompSets}{\ensuremath{\setFamily}}
\newcommand{\treeDecompSet}{\ensuremath{\setFamilySet}}
\newcommand{\treeDecompDef}{\ensuremath{\treeDecomp = (\treeDecompSets, \treeDecompTree)}} 

\newcommand{\setOfTreeDecomps}{\ensuremath{\mathcal{T}}}

\newcommand{\joinTree}{\ensuremath{T}}
\DeclareDocumentCommand{\joinTreeNodes}{O{\hypergraphEdges}}{\ensuremath{#1}}
\newcommand{\joinTreeEdges}{\ensuremath{E_J}}
\DeclareDocumentCommand{\joinTreeDef}{O{\hypergraphEdges}}{\ensuremath{\joinTree = (\joinTreeNodes[#1], \joinTreeEdges)}} 

\newcommand{\treewidth}{\ensuremath{\textnormal{tw}}}
\newcommand{\decompwidth}{\ensuremath{\treewidth_\treeDecomp}}
\newcommand{\labelWidth}{\ensuremath{\textnormal{lw}}}
\newcommand{\labelWidthEQ}{\ensuremath{\labelWidth_\extractionOrderCharacter}}
\newcommand{\extractionWidthReq}{\ensuremath{\textnormal{ew}}}
\newcommand{\extractionWidth}{\ensuremath{\extractionWidthReq_\extractionOrderCharacter}}

\newcommand{\intersectionGraph}{\ensuremath{G_I}}
\newcommand{\intersectionGraphNodes}{\ensuremath{\setFamily}}
\newcommand{\intersectionGraphEdges}{\ensuremath{E_I}}
\DeclareDocumentCommand{\intersectionGraphDef}{O{\intersectionGraphNodes}}{\ensuremath{\intersectionGraph = (#1, \intersectionGraphEdges)}}

\newcommand{\labelsetOrderSetMR}{\ensuremath{\labelsetOrderSet^{\genExtractionOrderChar}}}
\DeclareDocumentCommand{\labelsetOrderRegion}{O{r}}{\ensuremath{\labelsetOrder_{#1}}}
\DeclareDocumentCommand{\labelsetOrderSetRegion}{O{r}}{\ensuremath{\labelsetOrderSet_{#1}}}

\newcommand{\req}{r}

\newcommand{\type}{\ensuremath{\tau}}

\NewDocumentCommand{\VGP}{O{\req} O{i} O{j}}{\ensuremath{G^{#2,#3}_{#1}}}

\newcommand{\Vsource}[1][\req]{\ensuremath{o^+_{\req}}}
\newcommand{\Vsink}[1][\req]{\ensuremath{o^-_{\req}}}

\newcommand{\Vtype}[1][\req]{\ensuremath{\tau_{#1}}}

\newcommand{\SV}{\ensuremath{V_S}}
\newcommand{\SE}{\ensuremath{E_S}}

\newcommand{\Scap}{\ensuremath{d_{S}}}

\DeclareDocumentCommand{\NodeG}{O{\req} O{\pi}}{\ensuremath{G^N_{#1,#2}}}
\DeclareDocumentCommand{\NodeV}{O{\req} O{\pi}}{\ensuremath{V^N_{#1,#2}}}
\DeclareDocumentCommand{\NodeE}{O{\req} O{\pi}}{\ensuremath{E^N_{#1,#2}}}

\DeclareDocumentCommand{\EdgeG}{O{\req} O{i} O{j} O{u}}{\ensuremath{G^E_{#1,#2,#3,#4}}}
\DeclareDocumentCommand{\EdgeV}{O{\req} O{i} O{j} O{u}}{\ensuremath{V^E_{#1,#2,#3,#4}}}
\DeclareDocumentCommand{\EdgeE}{O{\req} O{i} O{j} O{u}}{\ensuremath{E^E_{#1,#2,#3,#4}}}

\DeclareDocumentCommand{\VESD}{O{\req}}{\ensuremath{\overrightarrow{E}_{#1}}}
\DeclareDocumentCommand{\VEOD}{O{\req}}{\ensuremath{\overleftarrow{E}_{#1}}}
\DeclareDocumentCommand{\VESigmaD}{O{\req} O{\sigma}}{\ensuremath{E_{#1,#2}}}

\DeclareDocumentCommand{\NodeVRange}{O{\req} O{i} O{j}}{\ensuremath{V^N_{#1,#2,#3}}}
\DeclareDocumentCommand{\NodeVRangeRange}{O{\req} O{i} O{j} O{u} O{v}}{\ensuremath{V^N_{#1,#2,#3,#4,#5}}}

\DeclareDocumentCommand{\cycle}{O{k }}{\ensuremath{C_{#1}}}

\DeclareDocumentCommand{\NodePaths}{O{\req}}{\ensuremath{\mathcal{P}_{#1}}}

\DeclareDocumentCommand{\loadV}{O{\req} O{u}}{\ensuremath{l_{#1,#2}}}
\DeclareDocumentCommand{\loadE}{O{\req} O{u} O{v}}{\ensuremath{l_{#1,#2,#3}}}
\DeclareDocumentCommand{\loadX}{O{\req} O{x}}{\ensuremath{l_{#1,#2}}}
\DeclareDocumentCommand{\decomp}{O{k}}{\ensuremath{D^{#1}}}
\DeclareDocumentCommand{\decompHat}{O{\req} O{k}}{\ensuremath{{\hat{D}}_{#1}^{#2}}}
\DeclareDocumentCommand{\load}{O{\req} O{k}}{\ensuremath{l_{#1}^{#2}}}

\DeclareDocumentCommand{\loadHat}{O{\req} O{k}}{\ensuremath{\hat{l}_{#1}^{#2}}}
\DeclareDocumentCommand{\probHat}{O{\req} O{k}}{\ensuremath{\hat{f}_{#1}^{#2}}}
\DeclareDocumentCommand{\mappingHat}{O{\req} O{k}}{\ensuremath{\hat{m}_{#1}^{#2}}}

\DeclareDocumentCommand{\loadHat}{O{\req} O{k}}{\ensuremath{\hat{l}_{#1}^{#2}}}
\DeclareDocumentCommand{\probHat}{O{\req} O{k}}{\ensuremath{\hat{f}_{#1}^{#2}}}
\DeclareDocumentCommand{\mappingHat}{O{\req} O{k}}{\ensuremath{\hat{m}_{#1}^{#2}}}

\DeclareDocumentCommand{\randVarX}{O{\req} O{k}}{\ensuremath{X_{#1}^{#2}}}
\DeclareDocumentCommand{\randVarY}{O{\req}}{\ensuremath{Y_{#1}}}
\DeclareDocumentCommand{\randVarZ}{O{\req}}{\ensuremath{Z_{#1}}}
\DeclareDocumentCommand{\randVarL}{O{x}}{\ensuremath{L_{x}}}
\DeclareDocumentCommand{\randVarLX}{O{\req} O{x} O{y}}{\ensuremath{L_{#1,#2,#3}}}
\DeclareDocumentCommand{\randVarLNode}{O{\req} O{\type} O{u}}{\ensuremath{L_{#1,#2,#3}}}
\DeclareDocumentCommand{\randVarLEdge}{O{\req} O{u} O{v}}{\ensuremath{L_{#1,#2,#3}}}
\DeclareDocumentCommand{\randVarM}{O{\req}}{\ensuremath{M_{#1}}}
\DeclareDocumentCommand{\randVarC}{O{\req}}{\ensuremath{C_{#1}}}

\DeclareDocumentCommand{\ProbVarX}{O{1}}{\ensuremath{\mathbb{P}(\randVarX = #1)}}
\DeclareDocumentCommand{\ProbVarY}{O{1}}{\ensuremath{\mathbb{P}(\randVarY = #1)}}
\DeclareDocumentCommand{\ProbVarZ}{O{1}}{\ensuremath{\mathbb{P}(\randVarZ = #1)}}
\DeclareDocumentCommand{\ProbVarL}{O{1}}{\ensuremath{\mathbb{P}(\randVarL = #1)}}
\DeclareDocumentCommand{\ProbVarM}{O{1}}{\ensuremath{\mathbb{P}(\randVarM = #1)}}
\DeclareDocumentCommand{\ProbVarC}{O{1}}{\ensuremath{\mathbb{P}(\randVarC = #1)}}

\DeclareDocumentCommand{\WAC}{O{\req}}{\ensuremath{\textnormal{WC}_{\req}}}

\DeclareDocumentCommand{\PotEmbeddings}{}{\ensuremath{\mathcal{D}}}
\DeclareDocumentCommand{\PotEmbeddingsHat}{O{\req}}{\ensuremath{\hat{\mathcal{D}}_{#1}}}

\DeclareDocumentCommand{\maxLoadX}{O{x}}{\ensuremath{\textnormal{max}^{L,\sum}_{#1}}}
\DeclareDocumentCommand{\maxLoadV}{O{\req} O{\type} O{u}}{\ensuremath{\textnormal{max}^{L}_{#1,#2,#3}}}
\DeclareDocumentCommand{\maxLoadE}{O{\req} O{u} O{v}}{\ensuremath{\textnormal{max}^{L}_{#1,#2,#3}}}
\DeclareDocumentCommand{\maxLoadVSum}{O{\req} O{\type} O{u}}{\ensuremath{\textnormal{max}^{L,\Sigma}_{#1,#2,#3}}}
\DeclareDocumentCommand{\maxLoadESum}{O{u} O{v}}{\ensuremath{\textnormal{max}^{L,\Sigma}_{#1,#2}}}

\DeclareDocumentCommand{\VVroot}{O{\req}}{\ensuremath{s_{#1}}}
\DeclareDocumentCommand{\VVpred}{O{\req}}{\ensuremath{\pi_{#1}}}

\DeclareDocumentCommand{\Cycles}{O{\req}}{\ensuremath{\mathcal{C}_{#1}}}
\DeclareDocumentCommand{\Paths}{O{\req}}{\ensuremath{\mathcal{P}_{#1}}}

\DeclareDocumentCommand{\VVcycleSource}{O{\req} O{k}}{\ensuremath{s^{C_{#2}}_{#1}}}
\DeclareDocumentCommand{\VVcycleTarget}{O{\req} O{k}}{\ensuremath{t^{C_{#2}}_{#1}}}
\DeclareDocumentCommand{\VVpathSource}{O{\req} O{k}}{\ensuremath{s^{P_{#2}}_{#1}}}
\DeclareDocumentCommand{\VVpathTarget}{O{\req} O{k}}{\ensuremath{t^{P_{#2}}_{#1}}}

\DeclareDocumentCommand{\VGcycle}{O{\req} O{k}}{\ensuremath{G^{\extractionOrderCharacter,C_{#2}}_{#1}}}
\DeclareDocumentCommand{\VVcycle}{O{\req} O{k}}{\ensuremath{V^{\extractionOrderCharacter,C_{#2}}_{#1}}}
\DeclareDocumentCommand{\VEcycle}{O{\req} O{k}}{\ensuremath{E^{\extractionOrderCharacter,C_{#2}}_{#1}}}
\DeclareDocumentCommand{\VEcycleSame}{O{\req} O{k}}{\ensuremath{\overrightarrow{E}^{C_{#2}}_{#1}}}
\DeclareDocumentCommand{\VEcycleDiff}{O{\req} O{k}}{\ensuremath{\overleftarrow{E}^{C_{#2}}_{#1}}}

\DeclareDocumentCommand{\VGcycleOrig}{O{\req} O{k}}{\ensuremath{G^{C_{#2}}_{#1}}}
\DeclareDocumentCommand{\VVcycleOrig}{O{\req} O{k}}{\ensuremath{V^{C_{#2}}_{#1}}}
\DeclareDocumentCommand{\VEcycleOrig}{O{\req} O{k}}{\ensuremath{E^{C_{#2}}_{#1}}}
\DeclareDocumentCommand{\VGpath}{O{\req} O{k}}{\ensuremath{G^{P_{#2}}_{#1}}}
\DeclareDocumentCommand{\VVpath}{O{\req} O{k}}{\ensuremath{V^{P_{#2}}_{#1}}}
\DeclareDocumentCommand{\VEpath}{O{\req} O{k}}{\ensuremath{E^{P_{#2}}_{#1}}}

\DeclareDocumentCommand{\VEpathSame}{O{\req} O{k}}{\ensuremath{\overrightarrow{E}^{P_{#2}}_{#1}}}
\DeclareDocumentCommand{\VEpathDiff}{O{\req} O{k}}{\ensuremath{\overleftarrow{E}^{P_{#2}}_{#1}}}

\DeclareDocumentCommand{\VEDiff}{O{\req}}{\ensuremath{\overleftarrow{E}^{\extractionOrderCharacter}_{#1}}}

\DeclareDocumentCommand{\VEcycles}{O{\req}}{\ensuremath{E^{\mathcal{C}}_{#1}}}
\DeclareDocumentCommand{\VEpaths}{O{\req}}{\ensuremath{E^{\mathcal{P}}_{#1}}}

\DeclareDocumentCommand{\VVcycleSourcesTargets}{O{\req}}{\ensuremath{V^{\mathcal{C},\pm}_{#1}}}
\DeclareDocumentCommand{\VVpathSourcesTargets}{O{\req}}{\ensuremath{V^{\mathcal{P},\pm}_{#1}}}
\DeclareDocumentCommand{\VVSourcesTargets}{O{\req}}{\ensuremath{V^{\pm}_{#1}}}

\DeclareDocumentCommand{\VVcycleSources}{O{\req}}{\ensuremath{V^{\mathcal{C},+}_{#1}}}
\DeclareDocumentCommand{\VVpathSources}{O{\req}}{\ensuremath{V^{\mathcal{P},+}_{#1}}}

\DeclareDocumentCommand{\VVcycleTargets}{O{\req}}{\ensuremath{V^{\mathcal{C},-}_{#1}}}
\DeclareDocumentCommand{\VVpathTargets}{O{\req}}{\ensuremath{V^{\mathcal{P},-}_{#1}}}

\DeclareDocumentCommand{\VGcycleBranchR}{O{\req} O{k}}{\ensuremath{G^{C_{#2}, B_1}_{#1}}}
\DeclareDocumentCommand{\VVcycleBranchR}{O{\req} O{k}}{\ensuremath{V^{C_{#2}, B_1}_{#1}}}
\DeclareDocumentCommand{\VEcycleBranchR}{O{\req} O{k}}{\ensuremath{E^{C_{#2}, B_1}_{#1}}}

\DeclareDocumentCommand{\VGcycleBranchL}{O{\req} O{k}}{\ensuremath{G^{C_{#2}, B_2}_{#1}}}
\DeclareDocumentCommand{\VVcycleBranchL}{O{\req} O{k}}{\ensuremath{V^{C_{#2}, B_2}_{#1}}}
\DeclareDocumentCommand{\VEcycleBranchL}{O{\req} O{k}}{\ensuremath{E^{C_{#2}, B_2}_{#1}}}

\DeclareDocumentCommand{\VGdecomp}{O{\req} O{k}}{\ensuremath{G^{\mathcal{D}}_{#1}}}
\DeclareDocumentCommand{\VVdecomp}{O{\req} O{k}}{\ensuremath{V^{\mathcal{D}}_{#1}}}
\DeclareDocumentCommand{\VEdecomp}{O{\req} O{k}}{\ensuremath{E^{\mathcal{D}}_{#1}}}

\DeclareDocumentCommand{\VVbranching}{O{\req} }{\ensuremath{\mathcal{B}_{#1}}}
\DeclareDocumentCommand{\VVbranchingcycle}{O{\req} O{k}}{\ensuremath{\mathcal{B}^{C_{#2}}_{#1}}}
\DeclareDocumentCommand{\VVbranchingpath}{O{\req} O{k}}{\ensuremath{\mathcal{B}^{P_{k}}_{#1}}}
\DeclareDocumentCommand{\VVjoin}{O{\req} }{\ensuremath{\mathcal{J}_{#1}}}
\DeclareDocumentCommand{\VVaggregation}{O{\req} }{\ensuremath{\mathcal{A}_{#1}}}

\DeclareDocumentCommand{\VGextcycle}{O{\req} O{k}}{\ensuremath{G^{C_{#2}}_{#1,\textnormal{ext}}}}

\DeclareDocumentCommand{\VVextcycle}{O{\req} O{k}}{\ensuremath{V^{C_{#2}}_{#1,\textnormal{ext}}}}
\DeclareDocumentCommand{\VVextcycleSources}{O{\req} O{k}}{\ensuremath{V^{C_{#2}}_{#1,+}}}
\DeclareDocumentCommand{\VVextcycleTargets}{O{\req} O{k}}{\ensuremath{V^{C_{#2}}_{#1,-}}}
\DeclareDocumentCommand{\VVextcycleSubstrate}{O{\req} O{k}}{\ensuremath{V^{C_{#2}}_{#1,S}}}

\DeclareDocumentCommand{\VEextcycle}{O{\req} O{k}}{\ensuremath{E^{C_{#2}}_{#1,\textnormal{ext}}}}
\DeclareDocumentCommand{\VEextcycleSources}{O{\req} O{k}}{\ensuremath{E^{C_{#2}}_{#1,+}}}
\DeclareDocumentCommand{\VEextcycleTargets}{O{\req} O{k}}{\ensuremath{E^{C_{#2}}_{#1,-}}}
\DeclareDocumentCommand{\VEextcycleSubstrate}{O{\req} O{k}}{\ensuremath{E^{C_{#2}}_{#1,S}}}
\DeclareDocumentCommand{\VEextcycleF}{O{\req} O{k}}{\ensuremath{E^{C_{#2}}_{#1,F}}}

\DeclareDocumentCommand{\VGextpath}{O{\req} O{k}}{\ensuremath{G^{{P_{#2}}}_{#1,\textnormal{ext}}}}

\DeclareDocumentCommand{\VVextpath}{O{\req} O{k}}{\ensuremath{V^{P_{#2}}_{#1,\textnormal{ext}}}}
\DeclareDocumentCommand{\VVextpathSources}{O{\req} O{k}}{\ensuremath{V^{P_{#2}}_{#1,+}}}
\DeclareDocumentCommand{\VVextpathTargets}{O{\req} O{k}}{\ensuremath{V^{P_{#2}}_{#1,-}}}
\DeclareDocumentCommand{\VVextpathSubstrate}{O{\req} O{k}}{\ensuremath{V^{P_{#2}}_{#1,S}}}

\DeclareDocumentCommand{\VEextpath}{O{\req} O{k}}{\ensuremath{E^{P_{#2}}_{#1,\textnormal{ext}}}}
\DeclareDocumentCommand{\VEextpathSources}{O{\req} O{k}}{\ensuremath{E^{P_{#2}}_{#1,+}}}
\DeclareDocumentCommand{\VEextpathTargets}{O{\req} O{k}}{\ensuremath{E^{P_{#2}}_{#1,-}}}
\DeclareDocumentCommand{\VEextpathSubstrate}{O{\req} O{k}}{\ensuremath{E^{P_{#2}}_{#1,S}}}
\DeclareDocumentCommand{\VEextpathF}{O{\req} O{k}}{\ensuremath{E^{P_{#2}}_{#1,F}}}

\DeclareDocumentCommand{\forest}{O{\req}}{\ensuremath{\mathcal{F}_{#1}}}

\DeclareDocumentCommand{\VGforest}{O{\req}}{\ensuremath{G^{\mathcal{A},\mathcal{F}}_{#1}}}
\DeclareDocumentCommand{\VVforest}{O{\req}}{\ensuremath{V^{\mathcal{A},\mathcal{F}}_{#1}}}
\DeclareDocumentCommand{\VEforest}{O{\req}}{\ensuremath{E^{\mathcal{A},\mathcal{F}}_{#1}}}

\DeclareDocumentCommand{\VGforestOrig}{O{\req}}{\ensuremath{G^{\mathcal{F}}_{#1}}}
\DeclareDocumentCommand{\VVforestOrig}{O{\req}}{\ensuremath{V^{\mathcal{F}}_{#1}}}
\DeclareDocumentCommand{\VEforestOrig}{O{\req}}{\ensuremath{E^{\mathcal{F}}_{#1}}}

\DeclareDocumentCommand{\varFlowInput}{O{\req} O{i} O{u}}{\ensuremath{f^+_{#1,#2,#3}}}
\DeclareDocumentCommand{\varFlowOutput}{O{\req} O{i} O{u}}{\ensuremath{f^+_{#1,#2,#3}}}

\DeclareDocumentCommand{\VEextcycleHorizontal}{O{\req} O{k} O{u} O{v}}{\ensuremath{E^{C_{#2}}_{#1,\textnormal{ext},#3,#4}}}
\DeclareDocumentCommand{\VEextpathHorizontal}{O{\req} O{k} O{u} O{v}}{\ensuremath{E^{P_{#2}}_{#1,\textnormal{ext},#3,#4}}}
\DeclareDocumentCommand{\VEextcycleVertical}{O{\req} O{k} O{\type} O{u}}{\ensuremath{E^{C_{#2}}_{#1,\textnormal{ext},#3,#4}}}
\DeclareDocumentCommand{\VEextpathVertical}{O{\req} O{k} O{\type} O{u}}{\ensuremath{E^{P_{#2}}_{#1,\textnormal{ext},#3,#4}}}

\DeclareDocumentCommand{\VEextCGHorizontal}{O{\req} O{u} O{v}}{\ensuremath{E^{\textnormal{ext,SCG}}_{#1,#2,#3}}}
\DeclareDocumentCommand{\VEextCGVertical}{O{\req} O{\type} O{u}}{\ensuremath{E^{\textnormal{ext,SCG}}_{#1,#2,#3}}}

\DeclareDocumentCommand{\VVextCGFlowNodes}{O{\req}}{\ensuremath{V^{\textnormal{ext,SCG}}_{#1,\textnormal{flow}}}}
\DeclareDocumentCommand{\VEextCGFlowEdges}{O{\req}}{\ensuremath{E^{\textnormal{ext,SCG}}_{#1,\textnormal{flow}}}}

\DeclareDocumentCommand{\Queue}{}{\ensuremath{\mathcal{Q}}}

\DeclareDocumentCommand{\Variables}{}{\ensuremath{\mathcal{V}}}

\DeclareDocumentCommand{\VGextcycleFlow}{O{\req} O{k}}{\ensuremath{G^{C_{#2}}_{#1,\textnormal{ext},f}}}

\DeclareDocumentCommand{\VGextcycleFlowBranchR}{O{\req} O{k}}{\ensuremath{G^{C_{#2},{B}_1}_{#1,\textnormal{ext},f}}}
\DeclareDocumentCommand{\VVextcycleFlowBranchR}{O{\req} O{k}}{\ensuremath{V^{C_{#2},{B}_1}_{#1,\textnormal{ext},f}}}
\DeclareDocumentCommand{\VEextcycleFlowBranchR}{O{\req} O{k}}{\ensuremath{E^{C_{#2},{B}_1}_{#1,\textnormal{ext},f}}}

\DeclareDocumentCommand{\VGextcycleFlowBranchL}{O{\req} O{k}}{\ensuremath{G^{C_{#2},{B}_2}_{#1,\textnormal{ext},f}}}
\DeclareDocumentCommand{\VVextcycleFlowBranchL}{O{\req} O{k}}{\ensuremath{V^{C_{#2},{B}_2}_{#1,\textnormal{ext},f}}}
\DeclareDocumentCommand{\VEextcycleFlowBranchL}{O{\req} O{k}}{\ensuremath{E^{C_{#2},{B}_2}_{#1,\textnormal{ext},f}}}

\DeclareDocumentCommand{\VGextcycleBranchR}{O{\req} O{k}}{\ensuremath{G^{C_{#2},{B}_1}_{#1,\textnormal{ext}}}}
\DeclareDocumentCommand{\VVextcycleBranchR}{O{\req} O{k}}{\ensuremath{V^{C_{#2},{B}_1}_{#1,\textnormal{ext}}}}
\DeclareDocumentCommand{\VEextcycleBranchR}{O{\req} O{k}}{\ensuremath{E^{C_{#2},{B}_1}_{#1,\textnormal{ext}}}}

\DeclareDocumentCommand{\VGextcycleBranchL}{O{\req} O{k}}{\ensuremath{G^{C_{#2},{B}_2}_{#1,\textnormal{ext}}}}
\DeclareDocumentCommand{\VVextcycleBranchL}{O{\req} O{k}}{\ensuremath{V^{C_{#2},{B}_2}_{#1,\textnormal{ext}}}}
\DeclareDocumentCommand{\VEextcycleBranchL}{O{\req} O{k}}{\ensuremath{E^{C_{#2},{B}_2}_{#1,\textnormal{ext}}}}

\DeclareDocumentCommand{\VGextpathFlow}{O{\req} O{k}}{\ensuremath{G^{P_{#2}}_{#1,\textnormal{ext},f}}}
\DeclareDocumentCommand{\VVextpathFlow}{O{\req} O{k}}{\ensuremath{V^{P_{#2}}_{#1,\textnormal{ext},f}}}
\DeclareDocumentCommand{\VEextpathFlow}{O{\req} O{k}}{\ensuremath{E^{P_{#2}}_{#1,\textnormal{ext},f}}}

\DeclareDocumentCommand{\VVKSource}{O{\req} O{K}}{\ensuremath{s^{K}_{#1}}}
\DeclareDocumentCommand{\VVKTarget}{O{\req} O{K}}{\ensuremath{t^{K}_{#1}}}
\DeclareDocumentCommand{\VVKSourcesTargets}{O{\req}}{\ensuremath{V^{K,\pm}_{#1}}}

\DeclareDocumentCommand{\VEbfsPre}{O{j} O{\req}}{\ensuremath{E^{\extractionOrderCharacter,\mathrm{pre}}_{#2,#1}}}
\DeclareDocumentCommand{\VEbfsSuc}{O{i} O{\req}}{\ensuremath{E^{\extractionOrderCharacter,\mathrm{suc}}_{#2,#1}}}

\DeclareDocumentCommand{\VEbfsInter}{O{i} O{j} O{\req}}{\ensuremath{E^{\extractionOrderCharacter}_{#3,#1\leadsto #2}}}

\DeclareDocumentCommand{\VEbfsLabels}{O{e} O{\req} }{\ensuremath{\mathcal{L}^{\extractionOrderCharacter}_{#2,#1}}}

\DeclareDocumentCommand{\VEbfsLabelsOrig}{O{e} O{\req} }{\ensuremath{\mathcal{L}_{#2,#1}}}

\DeclareDocumentCommand{\VEbfsAC}{O{i} O{j} O{\req}}{\ensuremath{C^{\extractionOrderCharacter}_{#1,#2}}}

\DeclareDocumentCommand{\VEbfsACL}{O{i} O{j} O{\req} }{\ensuremath{P^{1}_{#1,#2}}}
\DeclareDocumentCommand{\VEbfsACR}{O{i} O{j} O{\req} }{\ensuremath{P^{2}_{#1,#2}}}

\DeclareDocumentCommand{\VEbfsBagIterator}{}{\ensuremath{b}}
\DeclareDocumentCommand{\VEbfsBag}{O{\VEbfsBagIterator}}{\ensuremath{B^{\extractionOrderCharacter}_{#1}}}

\DeclareDocumentCommand{\deltaMinusA}{O{i}}{\ensuremath{\delta^-_{\extractionOrderCharacter}(#1)}}
\DeclareDocumentCommand{\deltaPlusA}{O{i}}{\ensuremath{\delta^+_{\extractionOrderCharacter}(#1)}}

\DeclareDocumentCommand{\VEbfsLabelsBag}{O{b}}{\ensuremath{\mathcal{L}^{\extractionOrderCharacter}_{#1}}}

\DeclareDocumentCommand{\mapVedge}{O{e}}{\ensuremath{m^{\mathcal{L}}_{#1}}}
\DeclareDocumentCommand{\mapVbag}{O{a}}{\ensuremath{m^{\mathcal{L}}_{#1}}}

\DeclareDocumentCommand{\origEdge}{O{e}}{\ensuremath{\overrightarrow{E}_{\hspace{-1pt}\req}\hspace{-1pt} (#1)}}

\DeclareDocumentCommand{\extractionEdge}{O{e}}{\ensuremath{\overrightarrow{E}^{\hspace{-1pt}\extractionOrderCharacter}_{\hspace{-1pt}\req}\hspace{-2pt}(#1)}}

\theoremstyle{plain}
\newtheorem{theorem}{Theorem}
\newtheorem{lemma}[theorem]{Lemma}
\newtheorem{corollary}[theorem]{Corollary}
\theoremstyle{definition}
\newtheorem{problem}{Problem Statement}
\newtheorem{definition}[theorem]{Definition}

\newtheorem{remark}[theorem]{Remark}

\declaretheorem[style=definition,qed=$\square$,sibling=theorem]{definition}

\newcommand{\blankpage}{\newpage~ \thispagestyle{empty} \newpage}

\newcolumntype{F}{>{$\displaystyle}r<{$}@{\hspace{0.0em}}}
\newcolumntype{C}{>{$\displaystyle\,}c<{$}@{\hspace{0.0em}}}
\newcolumntype{B}{>{$\displaystyle\,}r<{$}@{\hspace{0.0em}}}
\newcolumntype{R}{>{$\displaystyle}r<{$}@{\hspace{0.2em}}}
\newcolumntype{S}{>{$\displaystyle}r<{$}@{\hspace{0.2em}}}
\newcolumntype{L}{>{$\displaystyle}l<{$}@{\hspace{0.2em}}}
\newcolumntype{Q}{>{$\displaystyle}l<{$}@{\hspace{0.3em}}}

\newcounter{ipCounter}
\NewDocumentEnvironment{IPFormulation}{m}{%
\refstepcounter{ipCounter}
\begin{algorithm}[#1]%

}{%
\end{algorithm}
\addtocounter{algocf}{-1}
}

\NewDocumentEnvironment{IPFormulationStar}{m}{%
\refstepcounter{ipCounter}
\begin{algorithm*}[#1]%

}{%
\end{algorithm*}
\addtocounter{algocf}{-1}
}

\makeatletter
\newcommand{\removelatexerror}{\let\@latex@error\@gobble}
\makeatother

\makeatletter
\newcommand{\nosemic}{\renewcommand{\@endalgocfline}{\relax}}
\newcommand{\dosemic}{\renewcommand{\@endalgocfline}{\algocf@endline}}
\newcommand{\pushline}{\Indp}
\newcommand{\popline}{\Indm\dosemic}
\let\oldnl\nl
\newcommand{\nonl}{\renewcommand{\nl}{\let\nl\oldnl}}
\makeatother

\newcommand{\defWhitespace}{\vspace{-4ex}}
\newcommand{\tagIt}[1]{\refstepcounter{equation}\textnormal{({\theequation})} \label{#1}}
\makeatletter

\newcommand{\eqnum}{\leavevmode\hfill\refstepcounter{equation}\textup{\tagform@{\theequation}}}
\makeatother

\begin{document}
\thispagestyle{empty}
\maketitle

\thispagestyle{empty}

\setcounter{page}{1}
\pagenumbering{roman} 
\thispagestyle{empty}
\noindent
This work is a minor revision of the author's master's thesis at Technische Universität Berlin.

%
%
%
%
%
%
%
%

\thispagestyle{empty}
\begin{abstract}
The Virtual Network Embedding Problem (VNEP) considers the efficient allocation of resources distributed in a substrate network to a set of request networks. Many existing works discuss either heuristics or exact algorithms for the VNEP, resulting in a choice between quick runtimes and quality guarantees for the solutions. 

Recently, the first fixed-parameter tractable approximation algorithm for the VNEP with arbitrary request and substrate topologies has been published by Rost and Schmid. This algorithm is fixed-parameter tractable with respect to the \emph{extraction width}, a graph parameter which is newly introduced by Rost and Schmid. It therefore combines positive traits of heuristics and exact solutions: The runtime is polynomial for instances with bounded extraction width, and the algorithm finds solutions with probabilistic quality guarantees. 

In this thesis, we propose two extensions of this \emph{base algorithm} in order to optimize the algorithm's runtime. The execution of the base algorithm is based on a traversal of an \emph{extraction order}, which is a rooted, directed acyclic graph representing the request topology. The base algorithm assigns to each edge a set of nodes, the \emph{label set}. Each request node's out-edges are partitioned into a set of \emph{edge bags}, such that the label sets associated with the edge bags are disjoint. The size of the largest edge bag label set determines the extraction width parameter. The complexity of the algorithm depends exponentially on the extraction width.

The first extension of the base algorithm allows partitioning the edges in such a way that overlap between the label sets is permissible. This results in smaller label sets, since fewer edges are collected in a single bag.  We represent this alternative partition through the introduction of \emph{extraction label set orderings}. The extraction label set orderings uses the well-known \emph{running intersection property}, which is closely related to the notion of \emph{tree decompositions}. 

We show that the resulting algorithm is fixed-parameter tractable with respect to a new graph parameter, the \emph{extraction label width}, which is shown to be bounded above by the base algorithm's extraction width parameter. We further generalize several findings by Rost and Schmid related to the extraction width parameter. For the class of \emph{half-wheel} graphs with central root placement, which exhibit a linear scaling of the extraction width parameter in terms of graph size, we show that the extraction label width is a constant. Rost and Schmid show that adding parallel paths to existing edges increases the extraction width at most by the maximal degree. We demonstrate that the extraction label width increases by at most one. Lastly, we discuss how the extraction label width relates to the very well-known treewidth parameter, and based on that connection show that finding optimal extraction label set orderings minimizing the runtime of the extended algorithm is $\complexityNP$-hard.

As a second extension of the base algorithm, we generalize the notion of an extraction order by removing the requirement that the extraction order must be rooted. This multi-root algorithm relies on a partitioning of the generalized extraction order into a set of rooted subgraphs. Each subgraph is processed independently, and the resulting solutions are combined to an overall solution of the VNEP instance. We introduce a method of ensuring that this step of combining embeddings for rooted subgraphs succeeds. Lastly, we consider the complexity of the multi-root algorithm and show that it is also fixed-parameter tractable with respect to the extraction label width parameter. We further investigate the impact of placing additional root nodes and show that the increase of the extraction label width is bounded by the size of the boundary region of the resulting rooted subgraph.

\end{abstract}
\pagebreak

\renewcommand{\abstractname}{Zusammenfassung}

\newcommand{\germanEO}{Extraktionsordnung}
\newcommand{\germanEW}{Extraktionsweite}
\newcommand{\germanELW}{Extraktionsmengenweite}
\newcommand{\germanELSO}{Extraktionsmengenordnung}
\newcommand{\germanELSOs}{Extraktionsmengenordnungen}
\newcommand{\germanLabelset}{Label-Menge}
\newcommand{\germanLabelsets}{Label-Mengen}
\newcommand{\germanMultirootAlg}{Multirootalgorithmus}
\newcommand{\germanEdgeBags}{Kantentaschen}
\newcommand{\germanEdgeBag}{Kantentasche}
\newcommand{\germanHW}{Halbrad}

\begin{abstract}
Das Virtual Network Embedding Problem (VNEP) befasst sich mit der effizienten Allokation von Ressourcen, die über eine Netzwerkinfrastruktur verteilt sind. In der Fachliteratur wurden bisher heuristische Ansätze und exakte Algorithmen zur Lösung des VNEP untersucht, welche entweder schnelle Laufzeiten, oder Lösungen mit Qualitätsgarantien bieten.

Vor kurzem wurde von Rost und Schmid der erste FPT-Approximations\-algorithmus für das VNEP veröffentlicht. Dieser Algorithmus ist parametrisiert bezüglich des Parameters \emph{\germanEW{}}, der von Rost und Schmid eingeführt wird. Somit vereint der Algorithmus positivie Eigenschaften von Heuristiken und von exakten Algorithmen: Für Instanzen mit beschränkter \germanEW{}~ist die Laufzeit polynomial, und zugleich bietet der Algorithmus Lösungen mit probabilistischen Qualitätsgarantien.

In dieser Masterarbeit werden zwei Erweiterungen dieses Basisalgorithmus entwickelt um bessere Laufzeiten zu erzielen. Die Ausführung des Algorithmus basiert auf einer Durchquerung  der \emph{\germanEO{}}, welche ein gewurzelter gerichteter azyklischer Graph ist, der die Anfragetopologie repräsentiert. Der Basisalgorithmus weist jeder Kante eine Knotenmenge zu, die \emph{\germanLabelsets{}}. Der Algorithmus partitioniert für jeden Knoten die ausgehenden Kanten in sogenannte \emph{\germanEdgeBags{}}, sodass die mit den \germanEdgeBags{} assoziierten \germanLabelsets{} disjunkt sind. Die Größe der größten solchen \germanLabelset{} bestimmt die \germanEW{}, und die Komplexität des Basisalgorithmus hängt exponentiell von der \germanEW{} ab.

Die erste Erweiterung des Basisalgorithmus ermöglicht es, statt der disjunkten Partitionierung der ausgehenden Kanten eine Partitionierung zu wählen, bei der ein Überlapp zwischen den \germanLabelsets{} erlaubt ist. Diese Partitionierung wird durch sogenannte \germanELSOs{} repräsentiert. Dadurch müssen weniger Kanten in eine \germanEdgeBag{} zusammengefasst werden, was zu kleineren \germanLabelsets{} führt. Die \germanELSO{} basiert auf der bekannten \emph{running intersection property}, welche mit dem Konzept einer \emph{Baumzerlegung} zusammenhängt.

Wir zeigen, dass der erweiterte Algorithmus parametrisiert bezüglich eines neuen Graphparameters ist, der \emph{\germanELW{}}, welche nach oben durch die \germanEW{} beschränkt ist. Wir verallgemeinern mehrere Resultate von Rost und Schmid: Für die Graphklasse der \germanHW{}-Graphen mit zentral platzierter Wurzel skaliert die \germanEW{} linear mit der Größe des Graphen. Wir zeigen, dass die \germanELW{} hingegen konstant ist. Rost und Schmid zeigen außerdem, dass das Hinzufügen von Pfaden parallel zu einer existierenden Kante die \germanELW{} höchstens um den maximalen Knotengrad erhöht. Wir zeigen, dass die \germanELW{} hierbei höchstens um eins steigt. Schließlich untersuchen wir den Bezug der \germanELW{} zum Graphparameter Baumweite und zeigen, dass die optimale Wahl einer \germanELSO{} ein $\complexityNP$-schweres Problem ist.

Als zweite Erweiterung des Basisalgorithmus  verallgemeinern wir das Konzept einer \germanEO{}, indem die Voraussetzung, dass die \germanEO{} gewurzelt sein muss, entfernt wird. Dieser  \emph{\germanMultirootAlg{}} basiert auf einer Partitionierung der verallgemeinerten \germanEO{} in mehrere gewurzelte Teilgraphen. Jeder Teilgraph wird separat bearbeitet, und die Teillösungen werden in einem zweiten Schritt zu einer vollständigen Lösung kombiniert. Wir entwickeln eine Methode, mit der diese Kombination von Teillösungen möglich ist. Schließlich betrachten wir die Komplexität des \germanMultirootAlg{} und zeigen, dass der Algorithmus ebenfalls mit der \germanELW{} parametrisiert ist. Wir untersuchen außerdem den Effekt, den die Platzierung zusätzlicher Wurzelknoten hat und zeigen, dass die Erhöhung der \germanELW{} durch Einführung eines zusätzlichen Wurzelknoten durch die Größe der Grenzregion des resultierenden gewurzelten Teilgraph beschränkt ist.

\end{abstract}
\thispagestyle{empty}

\begin{minipage}{0.96\textwidth}
\setcounter{tocdepth}{2}
\tableofcontents
\end{minipage}
\thispagestyle{empty}
\blankpage

\setcounter{page}{1}
\pagenumbering{arabic} 
\pagestyle{plain}

\section{Introduction} \label{sec:introduction}

With the rising popularity of virtualization technologies in cloud computing, and the ever-increasing complexity of distributed systems, the concept of network virtualization has been of great interest. Virtualization is appealing to clients as it removes the overhead of building and maintaining physical infrastructure. Today, the supply for virtualized computational resources is abundant, with multiple providers enabling deployment of virtual servers of varying capabilities at a moment's notice.

However, allocating virtual servers, each placed in a single location in the provider's network topology, is insufficient for some workloads. A client may instead require a set of \emph{networked} resources, allowing for direct communication between the localized computational resources. In particular, bandwidth guarantees between multiple computing nodes are necessary for certain workloads. This requirement increases the complexity of the allocation problem: Firstly, the resources now have to be placed in the physical network in a coordinated way, such that a fixed amount of bandwidth can be reserved along the routes between the nodes. Secondly, the limited bandwidth of the physical infrastructure becomes a contested resource between clients and has to be allocated, much like the computational resources themselves.

The Virtual Network Embedding Problem (VNEP) formalizes this question of how to best allocate resources that are distributed across a physical infrastructure, such as an ISP network or a datacenter. This \emph{substrate network} is specified as a network topology consisting of a set of nodes and edges, which represent the resources of some physical substrate network. Clients request resources by providing a request network topology, represented as a directed graph. An overview and classification of many existing VNEP algorithms is given by Fischer, Botero, Beck, de Meer and Hesselbach in the survey \cite{fischer2013virtual}. These algorithms can be grouped in the following two categories: 
\begin{itemize}
\item Heuristic approaches, obtaining solutions quickly with no optimality guarantees
\item Exact approaches, obtaining optimal solutions at the cost of exponential runtimes
\end{itemize}
Recently, the first VNEP approximation algorithms were published by Rost and Schmid in \cite{rost:vne-approx-leveraging-rand-round-2018, rostSchmidFPTApproximations}, bridging the gap between these two classes of approaches: Unlike heuristics, the approximation algorithm gives some guarantees regarding the quality of the obtained solutions, and has better runtimes than exact approaches. 

Rost and Schmid consider an offline setting, where a batch of requests are processed simultaneously. In \cite{rostSchmidFPTApproximations}, two optimization objectives are considered. Firstly, a profit maximization objective is considered. A positive profit value is assigned to each request and only a subset of requests from the batch are embedded, while other requests may be rejected. Secondly, a cost minimization objective is considered, assigning a cost to each allocation of a substrate resource. In this setting, each submitted request topology must be embedded.

In another recent paper \cite{rost:charting-complexity-of-vne-2018}, Rost and Schmid show that the VNEP is $\complexityNP$-complete. This result is shown for many VNEP variants, including the one discussed in this thesis. Additionally, it was shown that even deriving embeddings which may violate resource capacity constraints by a certain factor, is $\complexityNP$-complete. This result holds for restrictive classes of request topologies, such as planar or degree-bounded graphs. 

However, the approximation algorithm allows the derivation of embeddings in polynomial time for certain classes of request topologies. A new graph parameter, the \emph{extraction width}, was introduced in \cite{rostSchmidFPTApproximations} to parameterize the runtime of the approximation algorithm. Rost and Schmid show that the VNEP approximation algorithm runs in polynomial time when restricted to request topologies with bounded extraction width, i.e. the algorithm is \emph{fixed-parameter tractable} with respect to the extraction width parameter. In \cite{rost:vne-approx-leveraging-rand-round-2018}, Rost and Schmid give a detailed discussion of their algorithm applied to \emph{cactus graphs}. Cactus graphs are defined as the class of graphs, where no edge is simultaneously contained in two different cycles, and are shown to have a bounded extraction width in \cite{rostSchmidFPTApproximations}.

\subsection{Contributions}

In this thesis, we pursue two approaches to extend the VNEP approximation algorithm presented in \cite{rostSchmidFPTApproximations}, which we therefore call the \emph{base algorithm}. Both extensions aim to reduce the runtime of the algorithm. The base algorithm uses a representation of the request topology in the form of a rooted, acyclic graph, which is called the \emph{extraction order}.

Firstly, a modification of the base algorithm relying on \emph{extraction label set orderings} is introduced in Section~\ref{sec:hierarchical-bags}. The base algorithm relies on a partitioning of each node's out-edges in the extraction order. The runtime of the base algorithm and the size of the associated linear program depend exponentially on the size of the sets in this partitioning. We propose an extension of the base algorithm exploiting the well-known running intersection property, which allows a partitioning of the edges into significantly smaller sets. Due to the exponential impact of the size of these sets, the extension is a significant improvement in terms of algorithm runtime.

In particular, we introduce the new graph parameter \emph{extraction label width}. We show that the extended algorithm is fixed-parameter tractable with respect to the extraction label width of the request topology. We further show that the extraction label width is bounded from above by the extraction width parameter of the base algorithm, thereby proving that the extended algorithm is an improvement over base algorithm in terms of runtime. 

We further generalize several results which were derived by Rost and Schmid for the extraction width to the new extraction label width parameter. In particular, Rost and Schmid introduce the class of half-wheel graphs and show that when placing the extraction order's root suboptimally, the extraction width scales linearly with the size of the half-wheel graph, and is therefore unbounded for arbitrarily large half-wheel graphs. We show that the extraction label width for half-wheel graphs is a constant even for the suboptimal root placement. 

Another result by Rost and Schmid states that the extraction width of a graph can increase by at most the maximal degree, when a path is added parallel to an existing edge. We show an analogous result for the extraction label width with a much tighter bound, namely that the extraction label width increases by at most one.

Lastly, we discuss how the extraction label width is related to the well-established graph parameter treewidth. We show how the $\complexityNP$-complete treewidth decision problem for arbitrary graphs can be represented in terms of calculating the width parameter of a particular extraction order. We thus show that finding an optimal extraction label set ordering for a given extraction order is $\complexityNP$-hard.

In the second extension of the base algorithm in Section~\ref{sec:multiroot}, we consider the possibility of relaxing the requirement that the extraction order must be a rooted graph, generalizing to more arbitrary acyclic orientations of the request topology. The main result of Section~\ref{sec:multiroot} is an algorithm for generalized extraction orders, which are not required to be rooted. This algorithm partitions the extraction order in several rooted subgraph, the \emph{root regions}. Each such root region is processed independently, and we show that any solution for one root region can incrementally be extended to a solution for the overall request.

Lastly, we consider the complexity of the multi-root algorithm and show that it is also fixed-parameter tractable with respect to the extraction label width parameter. We further investigate the impact of placing additional root nodes and show that the increase of the extraction label width is bounded by the size of the boundary region.

\subsection{The Structure of this Thesis}

In the remainder of Section~\ref{sec:introduction}, we give a brief introduction to some prerequisite algorithmic concepts, such as linear and integer programming and randomized rounding.
In Section~\ref{sec:the-vnep}, we define the VNEP formally and introduce some related terminology. 
The algorithms presented in this work are extensions of the VNEP approximation algorithm introduced by Rost and Schmid in \cite{rost:vne-approx-leveraging-rand-round-2018}. This \emph{base algorithm} is discussed in Section~\ref{sec:solving-the-vnep}. Section~\ref{sec:background} gives a brief introduction to the closely interrelated concepts of tree decompositions, running intersection orderings, and hypergraph acyclicity, as they are relevant to the first extension of the base algorithm.

In Section~\ref{sec:hierarchical-bags}, we extend the base algorithm through the introduction of \emph{extraction label set orderings}, which allows for significant improvements over the base algorithm in terms of the complexity. In Section~\ref{sec:multiroot}, we introduce an algorithm which allows for extraction orders which are not rooted.
Finally, Section~\ref{sec:conclusion} provides a conclusion of this thesis' findings.

\subsection{Preliminaries: Linear Programming}\label{sec:01:linear-programming}
The VNEP approximation algorithms discussed in this thesis are based on linear programming. A linear program (LP) is defined as a set of constraints in the form of linear inequalities over a set of variables. The linear program is usually accompanied by some linear objective function. The objective of linear programming is to find a variable assignment that optimizes the objective function, while satisfying each of the constraints. 

Linear programs are commonly expressed in the canonical form
\begin{align}
\text{maximize } && c^T x \label{def:lp-objective}\\
\text{subject to } && \mathbf{A} x \leq b \label{def:lp-constraints}\\
\text{and } && x \geq 0 \label{def:var-def}
\end{align}
where $b$ and $c$ are vectors of coefficients specifying the bounds of the constraints and the objective function, respectively. $\mathbf{A}$ is a matrix of coefficients specifying the constraints and $x$ is a vector containing the variables.

LP formulations are solvable in polynomial time with respect to their size, a result first proven by the ellipsoid algorithm proposed by Khachiyan in \cite{khachiyan1980polynomial}. Several algorithms are used for solving linear programs, such as the well-known Simplex algorithm by Dantzig \cite{dantzig1951maximization}, which has no polynomial bound on its runtime but performs very well for many linear programs. While Khachiyan's ellipsoid algorithm performs poorly in practice, other so-called interior point methods with polynomial runtimes have since been developed, including Karmarkar's algorithm \cite{karmarkar1984new}, and barrier methods \cite{gill1986projected}. Today, the performance of such interior point methods is considered on par with the Simplex algorithm.

Closely related to linear programming is the notion of mixed-integer programming. A mixed-integer program (MIP) has a similar structure to a linear program, but it additionally requires some of the variables to be restricted to integer values. Unlike linear programming, the problem of solving arbitrary MIP formulations is $\complexityNP$-complete. Given some MIP formulation, we refer to the linear program obtained by removing the integrality constraint on the variables as its \emph{LP relaxation}. We also refer to solutions of a MIP formulation's LP relaxation as \emph{fractional solutions} of the MIP.

\subsection{Preliminaries: Randomized Rounding}\label{sec:01:rand-rounding}

The approach of \emph{randomized rounding}, introduced by Raghavan and Thompson in \cite{raghavan1987randomized}, is a widely used technique to design approximation algorithms for $\complexityNP$-hard problems. Randomized rounding seeks to exploit the fact that solving the LP relaxation is much faster than solving the corresponding MIP formulation. The general approach of randomized rounding is as follows:
\begin{enumerate}
\item State the problem as a MIP formulation.
\item Solve the LP relaxation of this MIP formulation.
\item Probabilistically round variables to integer values.
\end{enumerate}

In the context of the VNEP, the randomized rounding framework is used by Rost and Schmid in the offline setting, where multiple requests are considered simultaneously in batches. As we will discuss in more detail in Section~\ref{sec:solving-the-vnep}, fractional solutions to the LP formulation underlying the base algorithm can be  converted to a \emph{convex combination} of integer solutions, i.e. a set of integer solutions, each of which is assigned some value. The algorithm used for this conversion process is referred to as the \emph{decomposition algorithm}. 

The decomposition algorithm yields for each request in the batch a separate convex combination of integral sub-solutions, in the sense that each of these sub-solutions encodes an embedding of a single request. Due to the immense number of possible combinations, Rost and Schmid use randomized rounding to combine embeddings for individual requests to a single, overall solution of the problem instance. The value associated with each integral solution is interpreted as the probability that this solution should be selected in the overall solution for the entire batch. 

Since the runtime of the rounding is relatively small compared to solving the LP and processing the LP solution, many rounding iterations are performed, and the best overall solution is selected according to some heuristic. In \cite{rost:vne-approx-leveraging-rand-round-2018}, Rost and Schmid discuss three such rounding heuristics for cactus graphs in the profit maximization variant of the VNEP. 
The first heuristic selects the rounded solution with the smallest resource augmentations, using the achieved profit as a tie-breaker. The second heuristic instead selects the rounded solution with the maximal profit, using the solution's maximal resource augmentation to resolve ties. Finally, the third heuristic explicitly forbids resource augmentations by modifying the rounding procedure: In each rounding iteration, any randomly selected mappings that violate some resource capacity are rejected outright. From all rounding iterations, the solution with the greatest profit is selected. The order in which requests are considered is shuffled between iterations.

We can therefore summarize the heuristic approach presented by Rost and Schmid as follows:
\begin{enumerate}
\item Solve the decomposable LP formulation. 
\item Decomposition of the LP solution, yielding a convex combination of valid mappings.
\item  Perform a (large) number of randomized rounding iterations.
\item Select the best overall solution obtained in Step 3 according to some heuristic.
\end{enumerate}
The VNEP approximation algorithm performs similarly, however, instead of performing a fixed number of iterations and using a heuristic selection in the last step, the first suitable solution is returned.

The extensions proposed in this thesis are focused on the first two steps and are independent of the randomized rounding framework. Therefore, a simpler variant of the VNEP is discussed in this thesis, which only considers a single request. In this setting, randomized rounding is not necessary, as the best solution according to some metric can be selected directly by evaluating all options. However, the algorithms presented in this thesis generalize to the offline setting in a straightforward way. The single-request LP formulation can be used for each individual request in the input, and the capacity constraints can be defined in a shared way by summing over the set of requests.

\cleardoublepage
\section{The Virtual Network Embedding Problem} \label{sec:the-vnep}

In this section, we first introduce some prerequisite definitions and then give a formal statement of the VNEP. Since the thesis is based on the VNEP approximation algorithm by Rost and Schmid, we closely follow the notation used in  \cite{rost:vne-approx-leveraging-rand-round-2018,rostSchmidFPTApproximations} with some notational adaptations.

\subsection{Request and Substrate Topologies}

Request and substrate topologies are both given as directed graphs. For any directed graph $G = (V, E)$, we write $\undirectedGraph = (V, \undirectedEdges)$ to refer to its undirected version, i.e. $\undirectedEdges := \{\{x, y\}  |  (x, y) \in E\}$.

\begin{definition} (Request Topology) \label{def:introvnep:reqtop}\\
The request topology $\reqTopologyDef$ is a directed graph. Nodes of the request topology are usually named $i, j, k,$ etc. We assume that $\reqTopology$ contains no loop edges and antiparallel edges, i.e. if $(i, j) \in \reqEdges$, then $(j, i) \not\in \reqEdges$. We usually also assume that the request topology is a connected graph.
\end{definition}

\begin{definition} (Substrate Topology) \\
The substrate topology $\substrateTopologyDef$ is also a directed graph. Substrate nodes are usually named $u, v, w,$ etc. We again assume that the substrate contains no loop edges, although antiparallel edges may be included.
We write $\pathSet$ to denote the set of all simple paths in the substrate topology.
\end{definition}

\subsection{Request Mappings}

The allocation of resources from the substrate to some request network is represented as a mapping of each element of the request network to some subset of the substrate network. Since substrate nodes represent localized resources, each request node should be mapped to a single location in the substrate network. An example of a valid mapping is shown in Figure~\ref{fig:valid_mapping_example}.

\begin{definition}(Valid Mapping)\label{def:valid-mapping}\\
A valid request mapping $\mappingRequest := (\mappingNodes, \mappingEdges)$ is a tuple of a node mapping function $\mappingNodes: \reqNodes \rightarrow \substrateNodes$ and an edge mapping function $\mappingEdges: \reqEdges \rightarrow \pathSet(\substrateEdges)$.
The node mapping function maps each of the request's nodes to a single substrate node. The edge mapping function maps each request edge to a path in the substrate. Further, given an edge $e=(i,j) \in E_R$, the path $m_E(e)$ must start at the mapped location of its tail node $m_V(i)$ and end in the mapped location of its head node $m_V(j)$.
\end{definition}
\begin{definition}(Mapping Notation)\\
Given a node mapping $\mappingNodes$, a request node $i \in \reqNodes$ and a substrate node $u \in \substrateNodes$, we introduce the notation $i \mapsto u$ to mean $\mappingNodes(i) = u$. Analogously for edge mappings $\mappingEdges$, a request edge $e \in \reqEdges$ and a substrate path $p \in \pathSet$ we write $e \mapsto p$ to indicate that $\mappingEdges(e) = p$.
\end{definition}

\begin{figure}[tbph]
	\centering
	\includegraphics[width=0.5\textwidth]{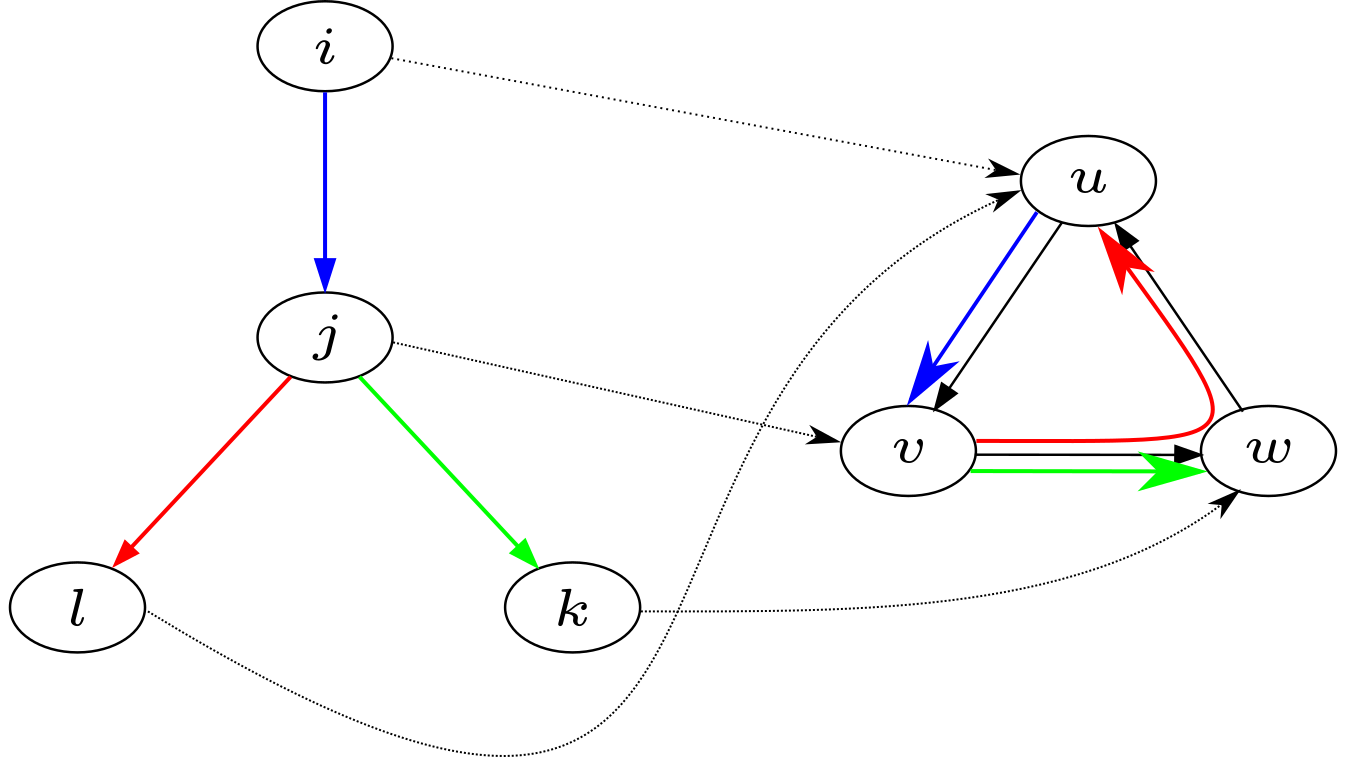}
	\caption{Example of a valid mapping of the request (nodes $i$, $j$, $k$ on the left) onto a substrate network (nodes $u$, $v$, $w$ on the right). The node mapping, is represented by the dotted arrows, and the edge mapping by the colored arrows in the substrate network. Explicitly stated, the mapping shown is given by $\mappingRequestDef$ with $\mappingNodes = \{i \mapsto u, j \mapsto v, k \mapsto w, l \mapsto u\}$ and \newline $\mappingEdges = \{(i,j) \mapsto ((u, v)), (j, k) \mapsto ((v, w)), (j, l) \mapsto ((v, w), (w, u)) \}$}
	\label{fig:valid_mapping_example}
\end{figure}

\subsection{Substrate Resources and Request Demands}

We distinguish between \emph{node resources} and \emph{edge resources}. Node resources are located at a single node in the substrate network. In practically relevant applications such as network function virtualization or service chaining, many different kinds of node resources are considered, such as x86 servers, hardware firewalls or other network appliances. This variety is represented through the introduction of \emph{node resource types} \cite{mehraghdam2014specifying,rost:vne-approx-leveraging-rand-round-2018}. Since the substrate network may provide multiple services at a single location, a single substrate node may concurrently and independently support an arbitrary number of resource types.

\begin{definition}(Substrate Resource Types)\\
We define $\nodeTypeSet$ as the set of all node types present in the substrate. The type assignment function $\subNodeType: \substrateNodes \rightarrow \mathcal{P}(\nodeTypeSet)$ maps each substrate node $u \in \substrateNodes$ to a set of types it supports. $\substrateNodesByType[\nodeType] \subset V_S$ denotes the set containing all substrate nodes supporting type $\nodeType$.
\end{definition}

Since a substrate node may support multiple resource types, a node resource is defined as a tuple of a substrate node and node type. 
Edge resources are considered as the bandwidth between pairs of node resources and otherwise have no associated resource type. They can therefore directly be identified with the substrate edges.
\begin{definition}(Substrate Resources)\\ 
The substrate node resources are given by $\SRV := \{ (u, \nodeType) ~|~ u \in \substrateNodes, \nodeType \in \subNodeType(u)\}$. The substrate edge resources are directly identified with the substrate edges.
The set of substrate resources is then defined as $\substrateResources := \SRV \cup \substrateEdges$.
\end{definition}

Node and edge resources in the substrate have a limited capacity. To encode this limitation, we introduce the \emph{substrate capacity function}, which assigns a non-zero capacity value to each substrate resource.
\begin{definition}(Substrate Capacity Function)\\ 
The capacity of each substrate resource is given by the substrate capacity function \mbox{$\substrateCapacity: \substrateResources \rightarrow \mathbb{R}^+$}.
\end{definition}

Each request node is assigned a single node type by the \emph{request node type function}.
\begin{definition}(Request Node Type Function)\\
Each request node is assigned a node type via the request node type function \mbox{$\reqNodeType: \reqNodes \rightarrow \nodeTypeSet$}.
\end{definition}

We quantify the amount of resources required by the request by assigning a demand value to each request node and edge. To this end, we define the \emph{request demand function}.
\begin{definition}(Request Demand Function)\\
We assign a demand value to each request node and edge with the request demand function \linebreak$\reqDemand: (\reqNodes \cup \reqEdges) \rightarrow \mathbb{R}_{\geq 0}$. This demand value gives the amount that must be allocated to the request node by any substrate resource that it is embedded in.
\end{definition}

Lastly, the \emph{mapping allocation function} computes the resources allocated to a single substrate resource by a given mapping.
\begin{definition}(Mapping Allocation Function) \label{def:mapping-allocation-function}\\
Given a mapping $\mappingRequestDef$, the mapping allocation function for node resources $(u, \nodeType) \in \SRV$ is defined as 
\begin{align}
\allocationFunction(\mappingRequest, u, \nodeType) := \sum_{\begin{subarray}{c}
i \in \reqNodes, \reqNodeType(i) = \nodeType,\\
\mappingNodes(i) = u
\end{subarray}}
\reqDemand(i, \nodeType),
\end{align}
and analogously for edge resources $(u, v) \in \substrateEdges$ as 
\begin{align}
\allocationFunction(\mappingRequest, u, v) := \sum_{\begin{subarray}{c}
(i, j) \in \reqNodes, \\
(u, v) \in \mappingEdges(i, j)
\end{subarray}}
\reqDemand(i, j).
\end{align}
\end{definition}

\subsection{The VNEP Problem Statement}
We now define the Virtual Network Embedding Problem.

\begin{problem} (Virtual Network Embedding Problem, VNEP) \label{def:VNEP}\\
\textbf{Input:} An instance of the VNEP $\VNEPInstance$ consists of a substrate and a request:
\begin{itemize}
\item A substrate, specified by a topology $\substrateTopology$, a set of substrate resources $\SR$ and a substrate capacity function $\substrateCapacity$.
\item A request, specified by a topology $\reqTopology$, a node type assignment $\reqNodeType$ and a request demand function $\reqDemand$.
\end{itemize}

\noindent \textbf{Output:}
A valid mapping of the request onto the substrate which does not exceed the capacities of the substrate resources, while optimizing some linear objective function.
\end{problem}

For the work presented in this thesis, the choice of objective function is arbitrary. We will assume that it is some linear function assigning a cost or benefit to each possible mapping decision.

The VNEP as defined above deviates from the VNEP as discussed in \cite{rostSchmidFPTApproximations} in two key aspects.
Firstly, unlike the offline setting discussed in \cite{rostSchmidFPTApproximations,rost:vne-approx-leveraging-rand-round-2018}, where batches of requests are considered simultaneously, we restrict the VNEP to a single request, which simplifies the notation. As discussed in Section~\ref{sec:01:rand-rounding}, randomized rounding is not necessary as a framework for the approximation algorithm under this restriction.
Secondly, we restrict the VNEP to only enforce substrate capacity constraints. Additional constraints, such as edge latency requirements, or node placement or routing restrictions, are not taken into account. Following the taxonomy introduced in \cite{rost:charting-complexity-of-vne-2018}, we can therefore classify the VNEP as defined above as $\langle VE \vert - \rangle$.

Many other variants of the VNEP are considered in the existing literature. A VNEP taxonomy is introduced in the survey by Fischer et al. \cite{fischer2013virtual}. The VNEP considered in this work can accordingly be classified as Centralized, Static and Concise. This classification means that decisions about the embedding of a request are made by a single entity, that embeddings are not subject to change, and that redundancies for fault-tolerance are not taken into account.

\cleardoublepage
\section{The Base VNEP Approximation Algorithm} \label{sec:solving-the-vnep}

In this section, we discuss two existing algorithms for solving the Virtual Network Embedding Problem, including the first VNEP approximation algorithm for general request topologies, which was introduced by Rost and Schmid in \cite{rostSchmidFPTApproximations}.

We first discuss the multi-commodity flow algorithm (MCF), which forms the basis of the base algorithm. We will discuss the algorithm in detail, and introduce an approximation algorithm based on LP relaxations of the MCF algorithm's underlying integer program. We discuss the limitations of this LP relaxation, which the base algorithm intends to resolve. 

\subsection{The Multi-Commodity Flow Algorithm} \label{sec:02:the-mcf-alg}
In this subsection, we discuss the multi-commodity flow (MCF) algorithm, which models the embedding decision for each request edge as a flow through the substrate network. The MCF algorithm is is based on a mixed-integer program (MIP), which is given in Mixed-Integer~Program~\ref{mip:matthias:MCF-formulation}. Variations of the MCF algorithm have been discussed in many publications, including \cite{chowdhury2009virtual, mehraghdam2014specifying, rost2014:its_about_time, rost:vne-approx-leveraging-rand-round-2018, rostSchmidFPTApproximations}. 

In Section~\ref{sec:mcf-algorithm-exact}, we will first introduce the MCF algorithm as an exact algorithm based on integer solutions of Mixed-Integer Program~\ref{mip:matthias:MCF-formulation}. In Sections~\ref{sec:mcf-algorithm-relaxation} and~\ref{sec:mcf-decomposition}, we discuss an approximation algorithm based on LP relaxations of the Mixed-Integer Program~\ref{mip:matthias:MCF-formulation}. Finally, in Section~\ref{sec:02:limitations-of-mcf}, the limitations of this approximation algorithm are discussed.


\begin{figure}[htbp]
\LinesNotNumbered
\renewcommand{\arraystretch}{1.1}
\removelatexerror

\begin{IPFormulation}{H}
\popline
\SetAlgorithmName{Mixed-Integer Program}{}{{}}

\begin{tabular}{FRLQB}
    \multicolumn{5}{r}{\parbox{0.975\textwidth}{~}} \\[-11pt]
    \left|\begin{array}{c} \max \\ \min \end{array}\right| ~ \text{obj}(\vec{y}, \vec{z}, \vec{l}) & & & & \tagIt{eq:classic_mcf:arbitrary_obj}\\
    \sum_{u \in \substrateNodesByType[\reqNodeType(i)]} y_i^u & = & 1  & \forall i \in \reqNodes &\tagIt{eq:classic_mcf:constr_force_emb}\\
    \sum_{(u, v) \in \outEdges{u}} z_{i,j}^{u,v} - \sum_{(v, u) \in \inEdges{u}} z_{i,j}^{v, u} &=& y_i^u - y_j^u ~~~~~~& \forall (i, j) \in \reqEdges; (u, v) \in \substrateEdges & \tagIt{eq:classic_mcf:constr_flow} \\
    y_i^u &=& 0 & \forall i \in \reqNodes, u \in \substrateNodesByType[\reqNodeType(i)]: \reqDemand(i) > \substrateCapacity(u, \reqNodeType(i)) &\tagIt{eq:classic_mcf:constr_forbid_small_nodes}\\
    z_{i,j}^{u,v} &=& 0 & \forall (i,j) \in \reqNodes, (u, v) \in \substrateEdges: 
\reqDemand(i, j) > \substrateCapacity(u, v) &\tagIt{eq:classic_mcf:constr_forbid_small_edges}\\
    \sum_{i \in V_R,~ \tau_R(i)=\tau} \reqDemand(i) \cdot y_i^u &=& a^{(u, \tau)} & \forall (u, \tau) \in \SR &\tagIt{eq:classic_mcf:constr_node_load}\\
    \sum_{(i, j) \in \reqEdges} \reqDemand(i, j) \cdot z_{i,j}^{u,v} &=& a^{(u, v)} & \forall (u, v) \in \substrateEdges &\tagIt{eq:classic_mcf:constr_edge_load}\\
    a^{x} \in [0, \substrateCapacity(x)] & & & \forall x \in R_S \\
    y_i^u \in \{0, 1\} & & & \forall i \in V_R, u \in V_S \\
    z_{i, j}^{u, v} \in \{0, 1\}  & & & \forall (i, j) \in \reqEdges, \quad \forall(u, v) \in \substrateEdges
\end{tabular}
\caption{MCF Formulation for VNEP, reproduced with notational adaptations from \cite{rost:vne-approx-leveraging-rand-round-2018}}
\label{mip:matthias:MCF-formulation}
\end{IPFormulation}
\end{figure}


\subsubsection{The Mixed-Integer Program of the MCF Algorithm} \label{sec:mcf-algorithm-exact}

We will now discuss the variables and constraints of the MCF algorithm's MIP formulation.

\paragraph{Variables}
Given a VNEP instance $\VNEPInstance$, the MCF formulation defines the following classes of variables:
\begin{itemize}
\item Node mapping variables $y$: For each possible node mapping decision of a request node $i \in \reqNodes$ onto some substrate node $u \in \substrateNodes$, a binary node mapping variable $y_i^u$ is defined. Setting $y_i^u = 1$ is interpreted as mapping $i$ onto $u$.
\item Edge mapping variables $z$: For each possible edge mapping decision, a binary  edge mapping variable $z_{i, j}^{u, v}$ is defined. Setting $z_{i, j}^{u, v} = 1$ is interpreted as mapping $(i, j)$ onto $(u, v)$. Note that for any request edge $(i, j) \in \reqEdges$, multiple $z$-variables may be set to one, as each request edge may be mapped onto a path of substrate edges of arbitrary length. In the case that $\mappingNodes(i) = \mappingNodes(j)$, all variables associated with $(i, j)$ may be set to zero.
\item Allocation variables $a$: For each substrate resource $x \in \SR$, a continuous variable $a^x$ is introduced. $a^x$ is bounded above by the associated substrate resource's capacity. The value of $a^x$ represents the amount of resource $x$ that is allocated to the VNEP solution corresponding to the variable assignment.
\end{itemize}

\paragraph{Constraints} 
The MCF formulation consists of the following constraints:

Constraint (\ref{eq:classic_mcf:constr_force_emb}) enforces that each request node is mapped onto a single substrate node. Due to the integrality of the node mapping variables, this constraint can only be satisfied by mapping the request node to exactly one substrate node.

Constraint (\ref{eq:classic_mcf:constr_flow}) enforces flow conservation at each substrate node and for each request edge $(i, j)$. The flow is induced by the variables $y^u_i$ and $y^u_j$ on the right side of the equation. The flow therefore originates in the substrate node to which $i$ is mapped and terminates in the substrate node to which $j$ is mapped. Since the edge mapping variables are binary, this flow is not splittable. 

Constraints (\ref{eq:classic_mcf:constr_forbid_small_nodes}) and (\ref{eq:classic_mcf:constr_forbid_small_edges}) explicitly forbid mapping any request nodes or edges onto substrate resources with insufficient capacity. While these constraints are redundant in the context of the exact algorithm based on integer solutions of Mixed-Integer Program~\ref{mip:matthias:MCF-formulation}, they are important when considering the LP relaxation for the MCF approximation algorithm.

Constraints (\ref{eq:classic_mcf:constr_node_load}) and (\ref{eq:classic_mcf:constr_edge_load}) track the resource allocations incurred by the node and edge mapping decisions, respectively. Since the allocation variables $\vec{a}$ are bounded by the corresponding substrate capacities, the constraints also implicitly enforce that the substrate capacities are not exceeded.

Therefore, constraints (\ref{eq:classic_mcf:constr_force_emb}) and (\ref{eq:classic_mcf:constr_flow}) ensure validity mapping, while constraints (\ref{eq:classic_mcf:constr_node_load}) and (\ref{eq:classic_mcf:constr_edge_load}), combined with the bounded allocation variables enforce the substrate capacity restrictions.

\begin{remark} (Existence of Circular Flows) \label{remark:eddy-flows}\\
Note that circular ``eddy currents'' in the edge mapping variables can exist which do not correspond to a mapping decision. That is, variable sets of the form $\{z^{u,v}_{i,j} ~|~ (i,j) \in \reqEdges, (u, v) \in C \}$ where $C$ is some cycle in the substrate topology, can be set to one. However, these circular currents are avoided by performing a graph search to extract the edge mapping.
\end{remark}

\subsubsection{Processing Solutions of the Mixed-Integer Program}

Any solution of Integer Program~\ref{mip:matthias:MCF-formulation} directly corresponds to a solution of the corresponding VNEP. In this section, we discuss how a solution for a VNEP instance is obtained from a solution of Integer Program~\ref{mip:matthias:MCF-formulation}.

We first introduce the \emph{local connectivity property} proposed by Rost and Schmid in \cite{rost:vne-approx-leveraging-rand-round-2018,rostSchmidFPTApproximations}. This property states the following: Let $y^u_i$ be some non-zero node variable representing the decision to map $i \mapsto u$. Let $j$ be some other request node, such that an edge between $i$ and $j$ exists. Then, there exists a set of non-zero edge mapping variables and a single non-zero node mapping variable $y^v_j$, such that the substrate edges associated with the edge mapping variables form a path connecting $u$ and $v$. If $i$ and $j$ are mapped to the same substrate node, no non-zero edge variables are required.

Importantly for later discussion, the local connectivity property is defined in terms of a \emph{fractional} solution, i.e. the LP relaxation of Mixed-Integer Program~\ref{mip:matthias:MCF-formulation}.

\begin{lemma}(Local Connectivity Property, cf. \cite{rost:vne-approx-leveraging-rand-round-2018, rostSchmidFPTApproximations}) \label{lem:local-connectivity-property} \\
Given a fractional solution $(\vec{y}, \vec{z}, \vec{a})$ to the MCF formulation, let $i \in \reqNodes$, $u \in \substrateNodes$ be request and substrate nodes such that $y^u_i > 0$. Then for any incoming edge $(k, i) \in \inEdges{i}$ and any outgoing edge $(i, j) \in \outEdges{i}$, there exist paths $P^{v, u}_{k, i}$ and $P^{u, w}_{i, j}$ such that 
\begin{enumerate}
    \item $P^{v, u}_{k, i}$ is a path from $v$ to $u$, such that $y^v_k > 0$ and $z^e_{k, i} >0$ for all $e \in P^{v, u}_{k, i}$ 
    \item $P^{u, w}_{i, j}$ is a path from $v$ to $u$, such that $y^w_j > 0$ and $z^e_{k, i} >0$ for all $e \in P^{u, w}_{i, j}$
\end{enumerate}
These paths can be found in $\mathcal{O}(|\substrateEdges|)$ time with a graph search.
\end{lemma}
\begin{proof}
See Lemma 4 in \cite{rostSchmidFPTApproximations}.
\end{proof}

The VNEP solution corresponding to a valid solution of Mixed-Integer Program~\ref{mip:matthias:MCF-formulation} can therefore be extracted as follows:
\begin{enumerate}
\item For each $i \in \reqNodes$ find a substrate node $u \in \substrateNodes$ such that the $y^u_i = 1$. Map $i$ onto $u$. $u$ is uniquely defined due to constraint~\ref{eq:classic_mcf:constr_force_emb} and integrality of the $\vec{y}$-variables.
\item For each $(i, j) \in E_R$, find a path $P^{u, v}_{i, j} \subseteq \substrateEdges$, such that $\mappingNodes(i) = u$ and $\mappingNodes(j) = v$, and $z^{u, v}_{i, j} = 1$ for all $(u, v) \in P^{u, v}_{i, j}$. Map $(i, j)$ onto $P^{u, v}_{i, j}$. Existence of such a path is guaranteed by Lemma~\ref{lem:local-connectivity-property}.
\end{enumerate}
Therefore, any solution of  Mixed-Integer Program~\ref{mip:matthias:MCF-formulation} corresponds directly to a valid mapping for the request.

A visualization of a variable assignment for a request consisting of a single edge is shown in Figure~\ref{fig:integral_flow_single_edge}, illustration how the notion of a flow in the variables $z^{\cdot, \cdot}_{i,j}$ corresponds to an edge mapping decision. Figure~\ref{fig:mcf-construction-tree} shows the construction of the MCF formulation for a request topology containing multiple edges. Unlike Figure~\ref{fig:integral_flow_single_edge}, no variable assignment is shown in the diagram. Rather, the diagram is intended to illustrate how a flow diagram representing the formulation can be constructed for requests containing multiple edges.

\begin{figure}[tbh]
	\centering
	\includegraphics[width=0.5\textwidth]{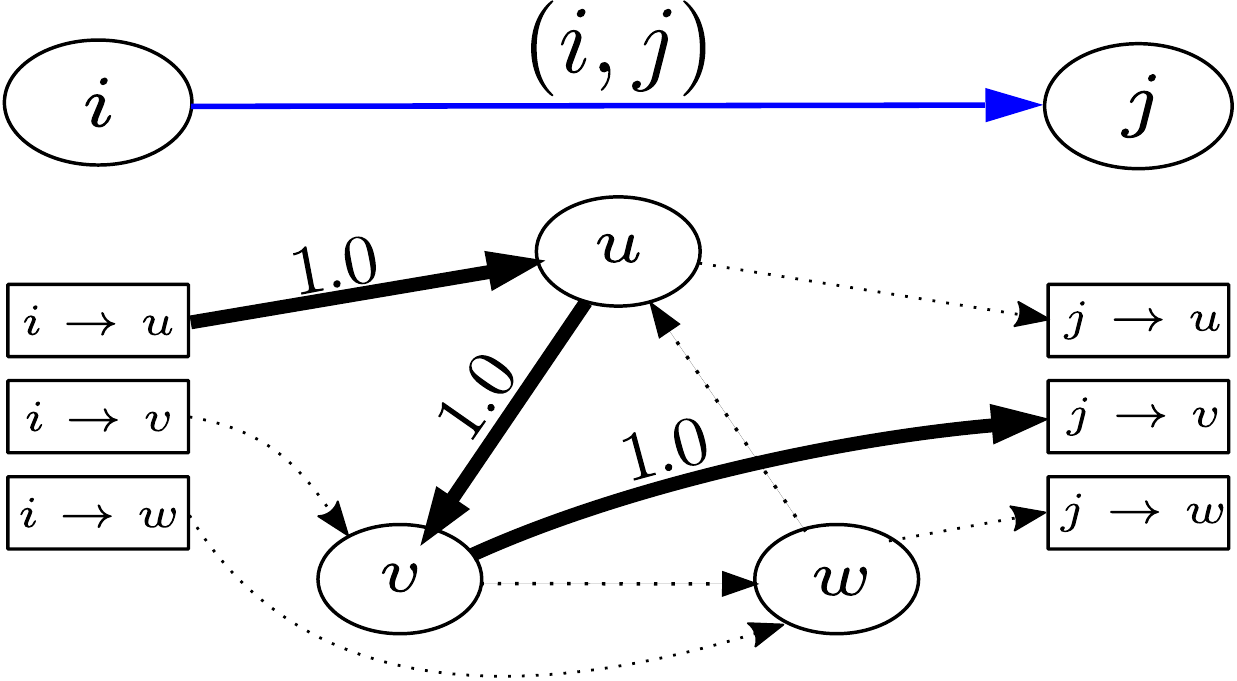}
	\caption{
		A valid variable assignment for Mixed-Integer Program~\ref{mip:matthias:MCF-formulation} for a request consisting of a single edge $(i,j)$, and the substrate introduced in Figure \ref{fig:valid_mapping_example}. The rectangular boxes on the left and right side represent the node mapping decision for $i$ and $j$, respectively. 
		The arrows between these boxes and the substrate nodes in the center represent the node mapping variables $\vec{y}$, which induce a flow represented by the arrows pointing to the substrate nodes corresponding to the node mapping decision. The substrate in the center represents the edge mapping decisions.        
		The flow induction, flow conservation and integrality constraints enforce that there is a path through the network carrying an integer flow.
		The mapping corresponding to this example flow is $\mappingNodes = \{i \mapsto u$, $j \mapsto v\}$ and $\mappingEdges = \{(i, j) \mapsto ((u, v))\}$, i.e. $(i, j)$ is mapped to the path containing the single edge $(u, v)$.
	}
	\label{fig:integral_flow_single_edge}
\end{figure}

\begin{figure}[tbph]
    \centering
    \includegraphics[width=0.75\textwidth]{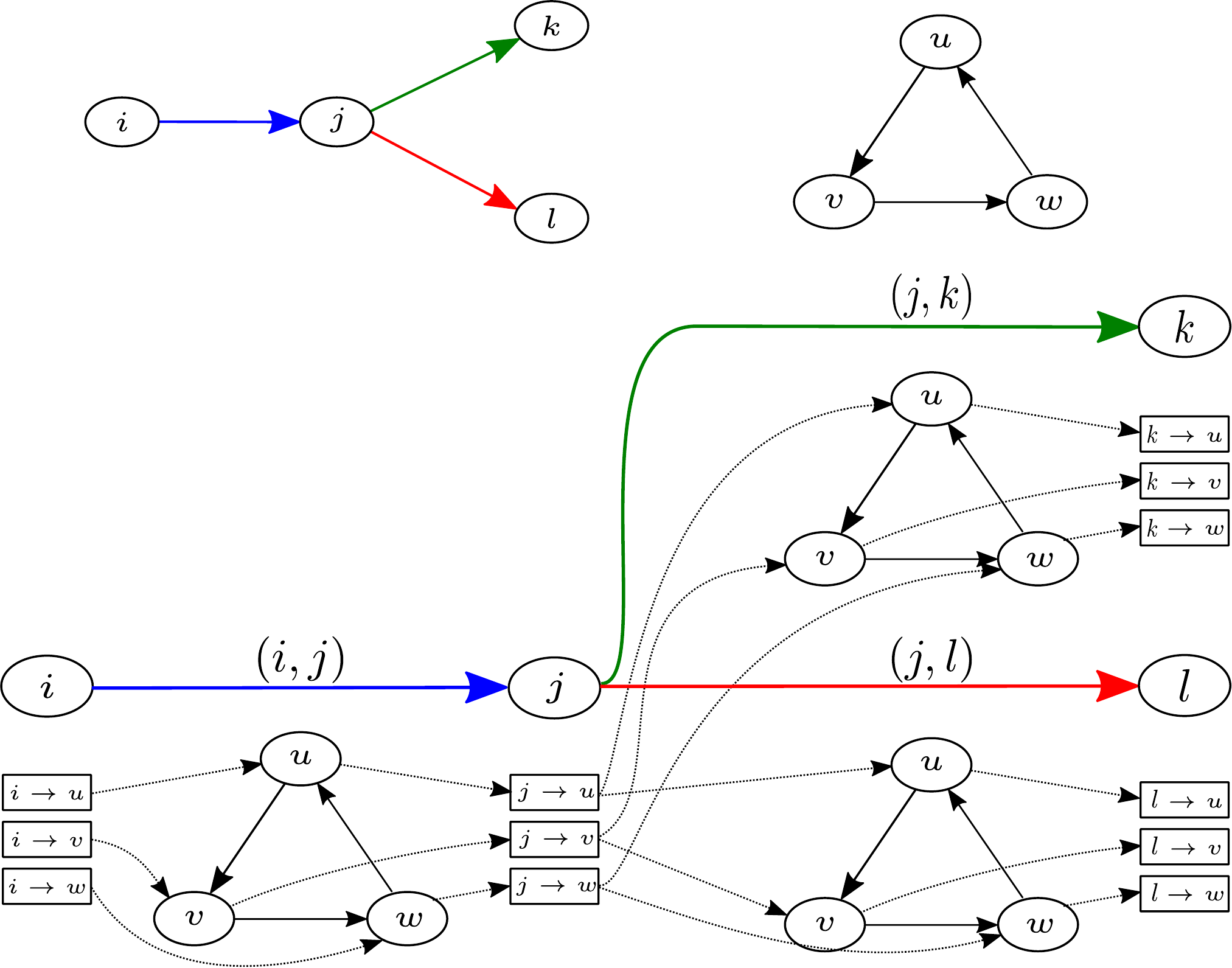}
    \caption{
An illustration of how the MCF formulation is constructed for a request with multiple edges.\newline
Top: Left: The request topology. Right: The substrate topology.\newline
Bottom: The MCF formulation, visualized in a style similar to Figure~\ref{fig:integral_flow_single_edge}. No variable assignment is shown in this diagram. The  flow diagram constructions for each edge are drawn under the edge's representation. Note that the node mapping variables for $j$ induce a flow in both edges $(j, k)$ and $(j, l)$.
    }
    \label{fig:mcf-construction-tree}
\end{figure}
\clearpage

\subsection{LP Relaxations of the MCF Formulation}\label{sec:mcf-algorithm-relaxation}

We now consider the Linear Program obtained by relaxing the integrality constraints on the node and edge mapping variables in the Multi-Commodity Flow formulation. As the example shown in Figure~\ref{fig:fractional_flow_single_edge} illustrates, due to the non-integrality, the variable assignment no longer corresponds to a single valid mapping of the request. However, as we will now discuss, the solutions for the LP relaxation may still contain valuable information about potential integer solutions, which in turn correspond to valid mappings of the request (see  Definition~\ref{def:valid-mapping}). 

\begin{figure}[htbp]
    \centering
    \includegraphics[width=0.5\textwidth]{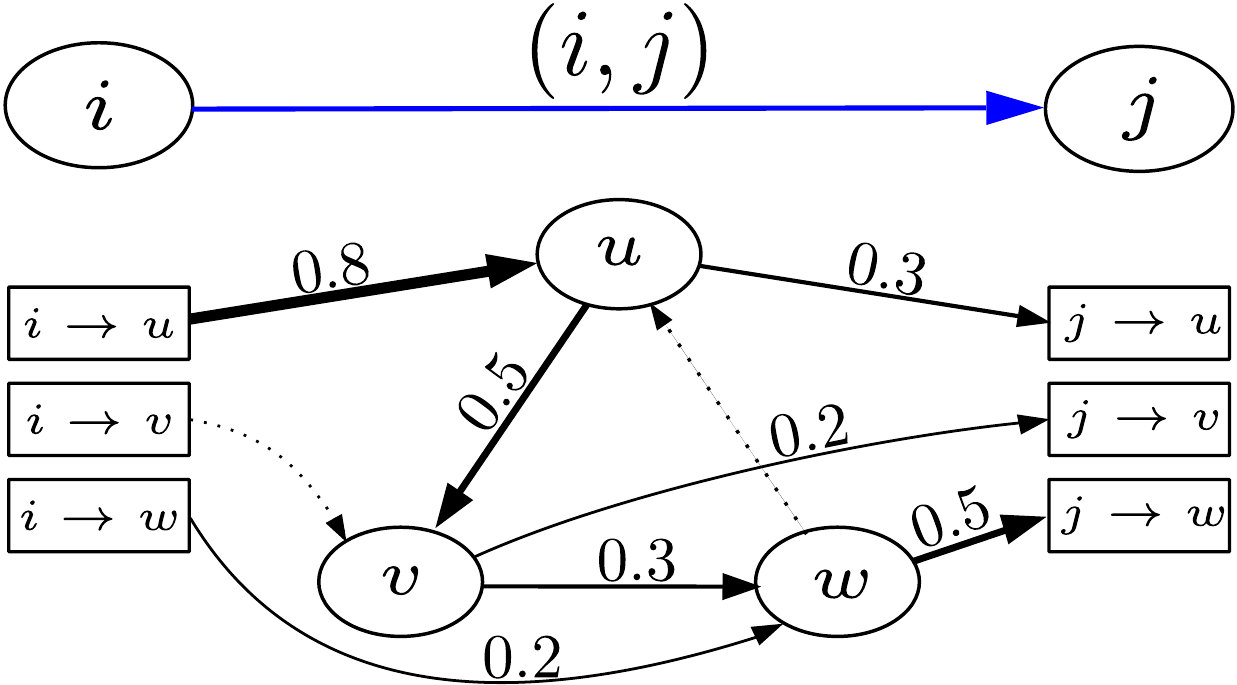}
    \caption{
        A valid variable assignment for the LP-relaxation of  Mixed-Integer Program~\ref{mip:matthias:MCF-formulation} for a request consisting of a single edge. As in the integral case, the flow induction and flow conservation constraints are satisfied, but the flow does not represent a definitive mapping decision.
    }
    \label{fig:fractional_flow_single_edge}
\end{figure}

Rost and Schmid show in \cite{rost:vne-approx-leveraging-rand-round-2018,rostSchmidFPTApproximations} that for tree-like requests, the relaxed solution of Linear Program~\ref{mip:matthias:MCF-formulation} consists of a \emph{superposition} of multiple valid mappings. More precisely, the LP solution can be converted to a \emph{convex combination} of a set of integer variable assignments, each of which corresponds to a valid mapping of the request. We now define the notion of a convex combination.
\begin{definition}(Convex Hull, Convex Combination) \label{def:conv-combi}  \\
Given a set of vectors $S = \{v_1, \ldots v_n \} \in \mathbb{R}^n$, the convex hull of $S$ is defined as  
\begin{align}
\conv(S) := \left\{\sum_{i=1}^{n} a_i \cdot v_i ~\middle|~ \vec{a} \in \mathbb{R}^n: \sum_{i=1}^{n} a_i = 1, \forall i \in \{1, ... n\}: a_i \geq 0 \right\}.
\end{align}
Any element of the convex hull is called a convex combination of $S$.
\end{definition}

In the context of the LP relaxation, each element of the set $S$ in Definition~\ref{def:conv-combi} can be considered as the vector obtained by concatenating the binary node and edge mapping variables corresponding to a valid mapping with potential capacity violations, i.e. \mbox{$(\vec{y}, \vec{z})^T \in \{0, 1\}^{|\vec{y}| + |\vec{z}|}$}. Each of these binary vectors then describes a valid integer solution, albeit with potential capacity violations, and therefore a valid mapping. 

Solutions of the LP-relaxation then correspond to a convex combination of these binary vectors. Once such a convex combination of valid mappings is obtained, some mechanism can be used to select a specific mapping. In the offline setting considered by Rost and Schmid, randomized rounding can be used (see Section~\ref{sec:01:rand-rounding}).

This property, by which solutions to the LP relaxation are convex combinations of integer solutions, does not hold for arbitrary integer programs. We call such formulations \emph{decomposable}, meaning that solutions to the linear program can be transformed to valid solutions for the original problem. We call algorithms which perform this transformation for some LP formulation \emph{decomposition algorithms}. Given a solution to the LP formulation, we refer to the corresponding convex combination of valid mappings as the \emph{decomposition} of the LP solution.

Figure~\ref{fig:fractional_flow_single_edge_decomposed} shows a decomposition for the LP solution shown in Figure~\ref{fig:fractional_flow_single_edge} into four valid mappings of varying strength. Note that the value of each variable in the LP solution for each variable is the sum of the variables in the individual mappings.

\begin{figure}[tbph]
    \centering
    \includegraphics[width=\textwidth]{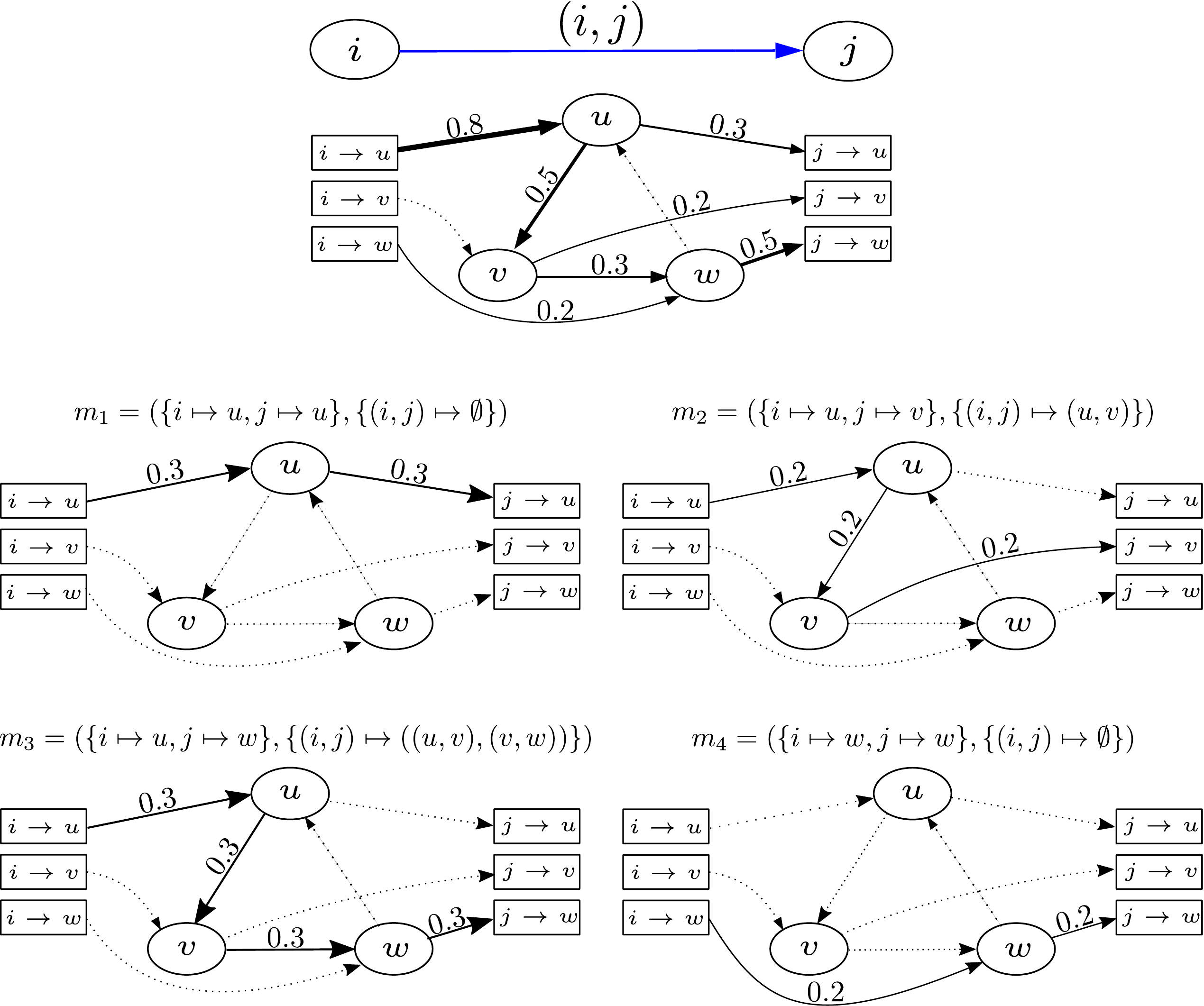}
    \caption{
A decomposition of the LP variable assignment shown in Figure~\ref{fig:fractional_flow_single_edge} to four valid mappings, illustrating that the non-integral flow of the LP solution corresponds to a superposition of several valid mappings. \newline
Top: The fractional solution.\newline
Center and Bottom: Four valid mappings $m_1$, $m_2$, $m_3$ and $m_4$ are extracted, with mapping values 0.3, 0.2, 0.3 and 0.2, respectively. Note that the sum of the four variable assignments yields the variable assignment for the LP solution.
    }
    \label{fig:fractional_flow_single_edge_decomposed}
\end{figure}

\subsubsection{Decomposition of the MCF Formulation for Tree Requests} \label{sec:mcf-decomposition}

In this section, we discuss the decomposition for LP solutions of the MCF formulation in more detail. The decomposition algorithms presented by Rost and Schmid in \cite{rost:vne-approx-leveraging-rand-round-2018, rostSchmidFPTApproximations} all rely on a top-down traversal of a rooted directed acyclic graph. In the case of tree-like request topologies, such a graph is also known as an arborescence. Since we allow for more general request topologies, which may not be acyclic or rooted, we introduce the notion of an \emph{extraction order} as an abstraction of the original request's orientation.
\begin{definition}(Extraction Order)\label{def:extraction-order}\\
Given a request topology $\reqTopologyDef$, an extraction order $\reqExtractionOrderDef$ is a directed acyclic graph, such that:
\begin{enumerate}
    \item $\reqExtractionOrder$ is a rooted in $\reqExtractionOrderRoot$, i.e. every node is reachable from $\reqExtractionOrderRoot$. 
    \item For every edge $(i, j) \in \reqEdges$, either $(i, j) \in \reqExtractionOrderEdges$ or $(j, i) \in \reqExtractionOrderEdges$, and vice versa. That is, $\undirected[\reqTopology] = \undirected[\reqExtractionOrder]$. Note that (anti-)parallel edges are not allowed in the request topology by Definition~\ref{def:introvnep:reqtop}.
\end{enumerate}
We further introduce the function $\EOEdgeToOriginal: \reqExtractionOrderEdges \rightarrow \reqEdges$, which maps each edge of the extraction order edge set to its counterpart in the original request edge set:
\begin{align}
    \EOEdgeToOriginal(i, j) = \begin{cases}
     (i, j) & \text{ if } (i, j) \in \reqEdges \\
     (j, i) & \text{ otherwise}
    \end{cases}
\end{align}
\end{definition}

The decomposition algorithm for the MCF formulation with tree-like request topologies is given as pseudocode in Algorithm~\ref{alg:matthias:decomposition:Algorithm-MCF-Tree}. We will now briefly describe how the algorithm operates and defer the reader to \cite{rostSchmidFPTApproximations, rost:vne-approx-leveraging-rand-round-2018} for a formal proof of the algorithm's correctness. In particular, Rost and Schmid prove that the algorithm returns a complete convex combination of valid mappings, and additionally, that the amount of substrate resources used for this convex combination of valid mappings is bounded by the values of the allocation variables, i.e. given the convex combination  $\PotEmbeddings$ returned by Algorithm~\ref{alg:matthias:decomposition:Algorithm-MCF-Tree},  $\sum_{(\prob, \mappingRequestIteration) \in \PotEmbeddings} \prob \cdot \allocationFunction(\mappingRequestIteration, x, y) \leq a^{x,y}$ holds for each resource $x,y \in \substrateResources$.

In each iteration of the loop starting in Line~\ref{algline:mcf-decomp:while-x-start}, the algorithm extracts a single valid mapping from the LP solution. First, a mapping for the root node is chosen and added to the mapping in Lines~\ref{algline:mcf-decomp:choose-root-mapping-node} and~\ref{algline:mcf-decomp:map-root}. The only requirement for this root node mapping is a non-zero value of the corresponding node mapping variable $y^u_\reqExtractionOrderRoot$.

The algorithm then iteratively extends the mapping along a traversal of the extraction order: For each node $i$ that is added to the mapping, the out-edges $\outEdgesExtractionOrder{i}$ are processed. Edges are mapped by performing the graph search mentioned in the local connectivity property, Lemma~\ref{lem:local-connectivity-property}. Note that if an edge is reversed in the extraction order with respect to the original request, the path $P^{v, u}_{j, i}$ extracted in Line~\ref{algline:mcf-decomp:choose-edge-path-reversed} is reversed to $P^{u, v}_{j, i}$ when assigning the edge mapping in Line~\ref{algline:mcf-decomp:assign-edge-mapping-reversed}.

Note that any node visited in this traversal is already mapped to some substrate node, as it can only be added to the queue in Line~\ref{algline:tree-decomp-add-node-to-q} once its in-edge has been mapped. The minimal value of all LP-variables that were used to decide node and edge mappings is calculated and subtracted from each used variable in Lines~\ref{alg:decomposition:compute-Vk} to \ref{algline:tree-decomp-adapt-variables}.

The algorithm terminates in polynomial time with respect to the number of LP variables, since in each iteration, at least one LP variable is decreased to zero in Line~\ref{algline:tree-decomp-adapt-variables}.

\begin{figure}[h!]

\begin{minipage}{1.01\textwidth}

\removelatexerror

\begin{algorithm*}[H]

\SetKwInOut{Input}{Input}
\SetKwInOut{Output}{Output}
\SetKwFunction{ProcessPath}{ProcessPath}{}{}
\SetKwFunction{reverse}{reverse}{}{}
\SetKwFunction{LP}{LP}
\SetKwFunction{LP}{LP}

\newcommand{\SET}{\textbf{set~}}
\newcommand{\ADD}{\textbf{add~}}
\newcommand{\DEFINE}{\textbf{define~}}
\newcommand{\AND}{\textbf{and~}}
\newcommand{\LET}{\textbf{let~}}
\newcommand{\WITH}{\textbf{with~}}
\newcommand{\COMPUTE}{\textbf{compute~}}
\newcommand{\FIND}{\textbf{find~}}
\newcommand{\CHOOSE}{\textbf{choose~}}
\newcommand{\DECOMPOSE}{\textbf{decompose~}}
\newcommand{\FORALL}{\textbf{for all~}}
\newcommand{\OBTAIN}{\textbf{obtain~}}
\newcommand{\WITHPROBABILITY}{\textbf{with probability~}}

\caption{Decomposition algorithm of fractional solutions of MIP \ref{mip:matthias:MCF-formulation} for Tree Requests, reproduced with notational adaptations from \cite{rost:vne-approx-leveraging-rand-round-2018}}
\label{alg:matthias:decomposition:Algorithm-MCF-Tree}

\Input{VNEP instance $\VNEPInstance$ such that $\reqUndirectedTopology$ is a tree,
    Extraction order $\reqExtractionOrderDef$, solution to the relaxation of IP~\ref{mip:matthias:MCF-formulation}~$(\vec{y},\vec{z},\vec{l})$ }
\Output{Convex combination~$\PotEmbeddings = \{\decomp = (\prob,\mappingRequestIteration)\}_k$ of valid mappings}
\BlankLine


  \SET $\PotEmbeddings  \gets \emptyset$ \AND $k \gets 1$ \AND $x \gets 1$\\
  \While{$x > 0$ \label{algline:mcf-decomp:while-x-start}}
  {
    
    \SET $\mappingRequestIteration \gets (\mappingNodesIteration,\mappingEdgesIteration)~\gets (\emptyset,\emptyset)$ \label{alg:sc-decomposition:init-maps}\\
    
    \SET $\Queue \gets \{\reqExtractionOrderRoot \}$\\
    
    \CHOOSE $u \in \substrateNodes$ \WITH $y^u_{\reqExtractionOrderRoot} > 0$ \label{algline:mcf-decomp:choose-root-mapping-node}\\ 
     \SET $\mappingNodesIteration(\reqExtractionOrderRoot)~\gets u$ \label{algline:mcf-decomp:map-root}\\
    
    \While{$|\Queue| > 0$}{  \label{alg:decomposition:begin-while-q}
      \CHOOSE $i \in \Queue$ \AND \SET$\Queue \gets \Queue \setminus \{i\}$\\
      \ForEach{$(i,j) \in \outEdgesExtractionOrder{\reqExtractionOrderEdges}(i)$}{
        \eIf{$(i,j) = \EOEdgeToOriginal(i,j)$}{
          \COMPUTE path ${P}^{u,v}_{i,j}$ from $\mappingNodesIteration(i)=u$ to  $v \in \substrateNodes$ according to Lemma~\ref{lem:local-connectivity-property}\\
        \pushline\pushline\pushline \nonl such that $y^v_{j} > 0$ and 
          $z^{u',v'}_{i,j} > 0$ hold for $(u',v') \in {P}^{u,v}_{i,j}$ \label{algline:mcf-decomp:choose-edge-path}\\
                    \vspace{2pt}
          \popline\popline\popline
        \SET $\mappingNodesIteration(j) \gets v$ \AND $\mappingEdgesIteration(i,j) \gets {P}^{u,v}_{i,j}$ \label{algline:mcf-decomp:assign-edge-mapping}\\
        }{
          \COMPUTE path ${P}^{v,u}_{j,i}$ from $v \in \substrateNodes$ to $\mappingNodesIteration(i)=u$ according to Lemma~\ref{lem:local-connectivity-property}\\  \pushline\pushline\pushline \nonl such that $y^v_{j} > 0$ and 
                    $z^{u',v'}_{j,i} > 0$ hold for $(u',v') \in {P}^{v,u}_{j,i}$ \label{algline:mcf-decomp:choose-edge-path-reversed}\\
          \vspace{2pt}
          \popline\popline\popline
          \SET $\mappingNodesIteration(j) \gets v$ \AND  $\mappingEdgesIteration(\EOEdgeToOriginal(i,j)) \gets {P}^{u,v}_{j,i} $ \label{algline:mcf-decomp:assign-edge-mapping-reversed}\\

        }
        \SET $\mathcal{Q} \gets \mathcal{Q} \cup \{j\}$ \label{algline:tree-decomp-add-node-to-q}\\
      }
    }
    \SET $\Variables_k \gets \{x\} \cup \{y^{\mappingNodesIteration(i)}_{i} | i \in \reqNodes\} \cup \{z^{u,v}_{i,j} | (i,j) \in \reqEdges, (u,v) \in \mappingEdgesIteration(i,j)\}$ \label{alg:decomposition:compute-Vk}\\
    \SET $\prob \gets \min \Variables_k$ \label{alg:decomposition:computing-prob} \\
    \SET $v \gets v - \prob$ \FORALL $v \in \Variables_k$ \AND \SET $a^{x,y} \gets a^{x,y} - \prob \cdot \allocationFunction(\mappingRequestIteration,x,y)$ \FORALL $(x,y) \in \SR$ \label{algline:tree-decomp-adapt-variables}\\  
    \ADD $\decomp = (\prob,m^k)$ to $\PotEmbeddings$ \AND \SET $k \gets k + 1$\\
  }

\KwRet{$\PotEmbeddings$}
\end{algorithm*}
\end{minipage}

\end{figure}

\clearpage

\subsubsection{Limitations of the MCF Formulation} \label{sec:02:limitations-of-mcf}

Rost and Schmid show in \cite{rost:vne-approx-leveraging-rand-round-2018,rostSchmidFPTApproximations} that the LP relaxation of Integer Program~\ref{mip:matthias:MCF-formulation} is not decomposable for request topologies containing ``undirected cycles'', i.e. requests where the undirected topology contains a cycle. 

Formally, this limitation is stated as the following theorem.
\begin{theorem} (Non-Decomposability of the MCF Formulation for Cyclical Requests, cf. \cite{rostSchmidFPTApproximations}) \\
Solutions to the LP relaxation of Integer Program~\ref{mip:matthias:MCF-formulation} can in general not be decomposed into convex combinations of valid mappings, if the request graph contains cycles. Accordingly, the integrality gap of the LP relaxation is unbounded for such request graphs.
\end{theorem}
\begin{proof}
See Theorem 7 in \cite{rostSchmidFPTApproximations}.
\end{proof}

In the following, we discuss this result using an example of a variable assignment satisfying all constraints of the MCF formulation for a cyclical request. The flow shown in Figure~\ref{fig:inconsistent_mapping_example} satisfies all constraints of the MCF formulation, i.e. flow induction (\ref{eq:classic_mcf:constr_force_emb}), flow conservation (\ref{eq:classic_mcf:constr_flow}) and substrate capacity (\ref{eq:classic_mcf:constr_node_load}), (\ref{eq:classic_mcf:constr_edge_load}) constraints. 

By following the differently colored flows in Figure~\ref{fig:inconsistent_mapping_example}, each of which originates at a specific mapping decision of the root node $i$, no consistent mapping of node $k$ in both branches can be reached. Since the flows do not intersect at any point, the decomposition is unambiguous. 

This counterexample shows that fractional solutions to the MCF MIP formulation can in general not be decomposed to valid mappings when considering request topologies containing cycles.

\begin{figure}[tbph]
    \centering
    \includegraphics[width=0.7\textwidth]{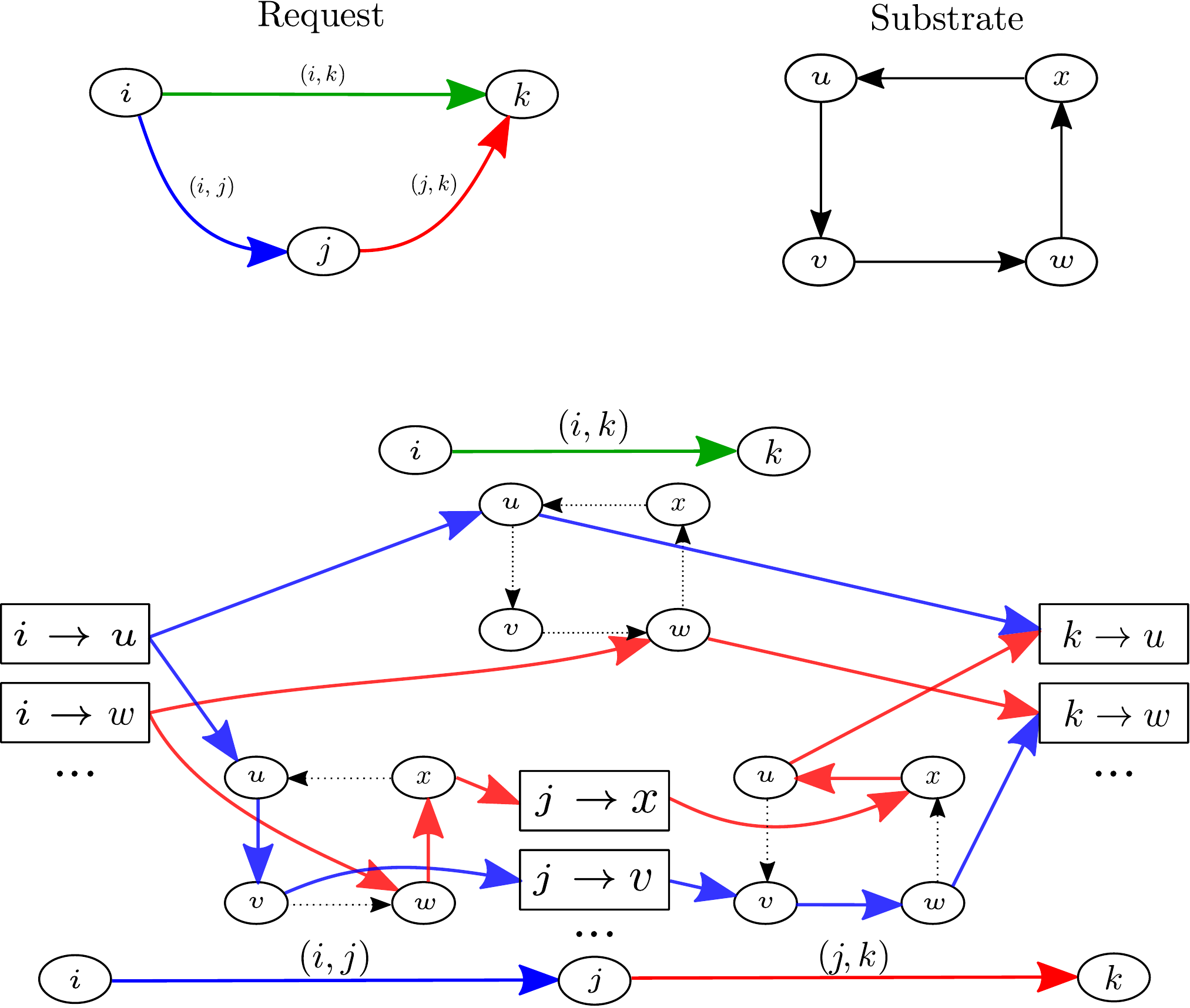}
    \caption{
A valid variable assignment for the relaxation of Integer Program~\ref{mip:matthias:MCF-formulation}, from which no valid mapping can be extracted. \newline
Top: The request and substrate graphs. \newline
Bottom: A valid, but non-decomposable variable assignment of the LP-relaxation of IP~\ref{mip:matthias:MCF-formulation}. Every non-zero flow variable, indicated by a red or blue arrow, has value 0.5. The flow corresponding to the mapping $i \mapsto u$ is shown in blue, and the flow for $i \mapsto w$ in red. Note that the two flows do not intersect, which implies that the edge mapping is unambiguous. The flows starting from either mapping decision for $i$ always lead to different mapping decisions for $k$. Therefore, no consistent mapping for the paths $(i, k)$ and $(i, j, k))$ can be extracted.
    }
    \label{fig:inconsistent_mapping_example}
\end{figure}

\subsection{The Base VNEP Approximation Algorithm} \label{sec:02:the-base-vnep-alg}

In this section, we will discuss the VNEP approximation algorithm for general request topologies, which was presented by Rost and Schmid in \cite{rostSchmidFPTApproximations}. Since this approximation algorithm forms the basis for the extensions presented in this thesis, we refer to it as the \emph{base algorithm}.

The base algorithm is structured very similarly to the algorithm presented in Section~\ref{sec:02:the-mcf-alg} for tree-like requests. 
Given a VNEP instance $\VNEPInstance$, where $\reqTopology$ is tree-like, the base algorithm proceeds as follows:
\begin{enumerate}
\item Define an extraction order $\reqExtractionOrder$, e.g. by performing a breadth-first search of the undirected topology $\undirected[\reqTopology]$.
\item Solve the decomposable linear programming formulation. 
\item Decomposition of the LP solution, yielding convex combination of valid mappings.
\item Select a mapping.
\end{enumerate}
The extensions proposed in this thesis are focused on the second and third step. 
We call the LP formulation associated with the base algorithm the \emph{base LP formulation}, and the decomposition algorithm the \emph{base decomposition algorithm}.

\subsubsection{The Approach of the Base Algorithm}

We first formalize the notion of an undirected cycle referenced in Section~\ref{sec:02:limitations-of-mcf} by introducing the concept of a \emph{confluence}.
\begin{definition}(Confluence) \label{def:confluence} \\
Given an extraction order $\reqExtractionOrderDef$, a confluence $\confluence = \path^1 \cup \path^2$ is a pair of paths $\path^1, \path^2 \subset \reqExtractionOrderEdges$ that both start in node $i \in \reqNodes$ and end in node $j\in \reqNodes$, and are otherwise node-disjoint. We call $i$ the start node and $j$ the target or end node of $\confluence$. 
We further denote by $\confluenceSet$ the set containing all confluences in $\reqExtractionOrder$.
\end{definition}

When an extraction order contains a confluence, a decomposition algorithm cannot rely on the local connectivity enforced by flow conservation, since the flows on both paths of the confluence must be reconciled when they merge at the confluence end node. Therefore, an additional mechanism is necessary, by which consistency can be enforced with respect to the confluence start and end node mappings.

The base algorithm introduces such a mechanism by introducing multiple \emph{LP subformulations} for each edge in the confluence. The subformulations consist of the variables and constraints of the MCF formulation for the subgraph defined by a single edge. 

Each subformulation is associated with a specific mapping of the confluence's end node. If an edge $e$ ends in a confluence end node, special constraints are added for each subformulation associated with $e$, which explicitly forbid any node mapping that is inconsistent with the mapping associated with the subformulation. 

Figure~\ref{fig:base_LP_cycle_example} illustrates the concept of LP subformulations. Note that any flow entering a subformulation associated with some end node mapping may only leave the subformulation through a single node mapping variable. Therefore, consistency of the confluence start node's mapping is enforced by flow induction in both branches of the confluence, while consistency of the confluence start node's mapping is enforced through the choice of the LP subformulation.

\begin{figure}[tbph]
    \centering
    \includegraphics[width=1.0\textwidth]{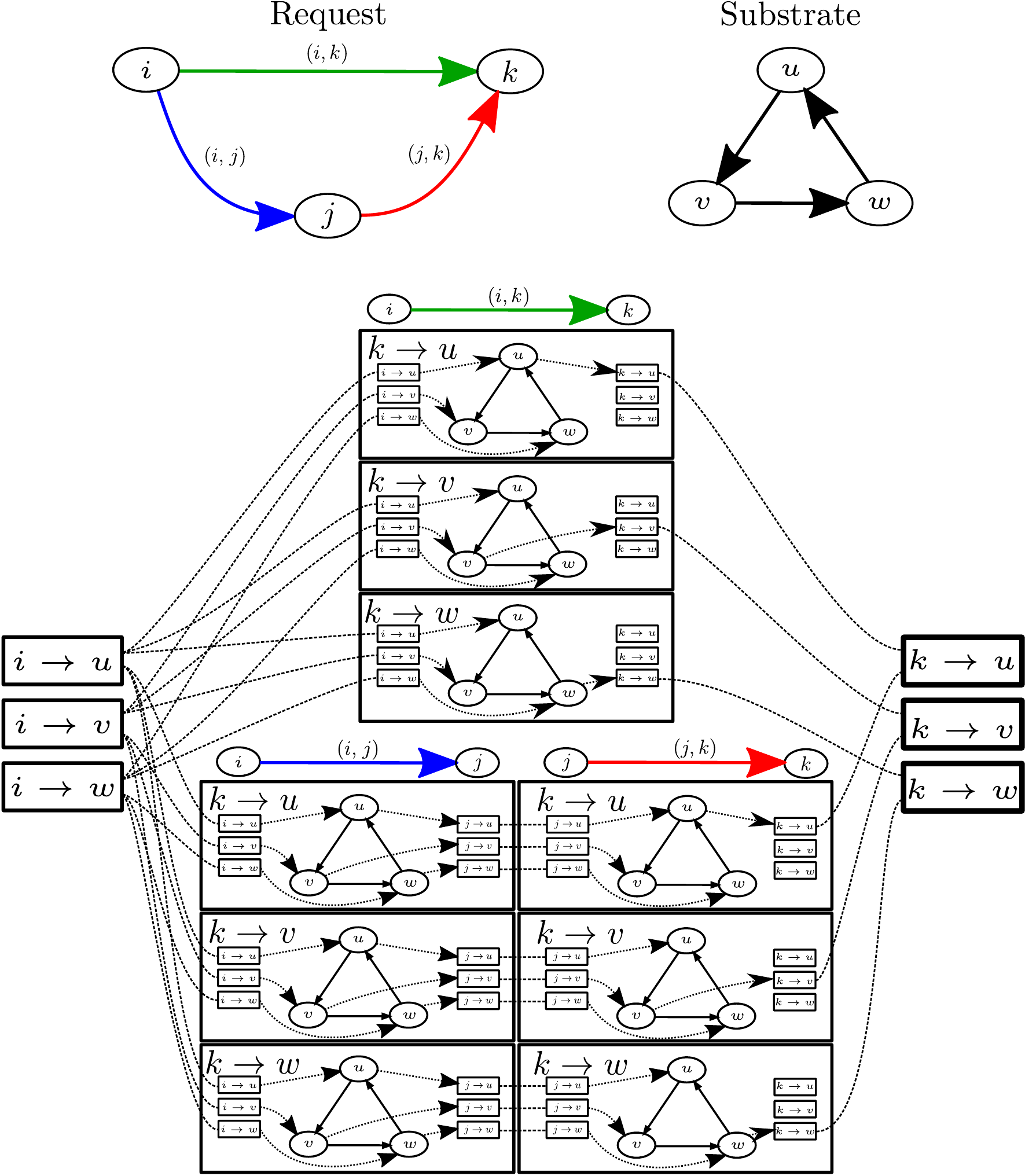}
    \caption{
An illustration of how the use of subformulations in Linear Program~\ref{IP:novel-AC} forces consistent mappings of confluence end nodes. \newline
Top: Left: The request topology, a confluence ending in $k$. Right: The substrate topology.\newline
Bottom: The construction of Linear Program~\ref{IP:novel-AC}, shown in a style similar to Figure~\ref{fig:fractional_flow_single_edge}. Each subformulation is enclosed in a rectangle annotated with the associated mapping of confluence end node $k$. 
On the left side, each global node mapping variable for $i$ is connected to its counterpart in each subformulation. On the right side, only one node mapping variable for $k$ is connected in each subformulation. All other mappings for $k$ are explicitly forbidden in the subformulations by Constraint (\ref{alg:lp:novel:forbidding-nodes-in-sub-lps}).
The bag variables $\vec{\gamma}$ and the global node variables for node $j$ are omitted in this diagram for clarity. In this example, by transitivity, $\subLP[y^x_j][(i, j), \{k \mapsto y\} ] = \subLP[y^x_j][(j, k), \{k \mapsto y\} ]$ holds for all substrate nodes $x, y \in \substrateNodes$ as indicated.
    }
    \label{fig:base_LP_cycle_example}
\end{figure}

A request edge may lie on multiple overlapping confluences. In this case, the base algorithm must ensure that the edge mapping is reconcilable with respect to the mapping of \emph{all} of the confluence end nodes. Therefore, LP-subformulations are introduced for every possible combination of target node mappings. Figure~\ref{fig:LP_cycle_nested_example} shows an example of such a case with two overlapping confluences. Only subformulations for the edge $(i,j)$ are shown. 

As Figure~\ref{fig:LP_cycle_nested_example} shows, the confluences to which any given edge belongs are indicated by labeling the edge with the set of the corresponding confluences' end nodes.

\begin{figure}[p]
    \centering
    \includegraphics[width=1.0\textwidth]{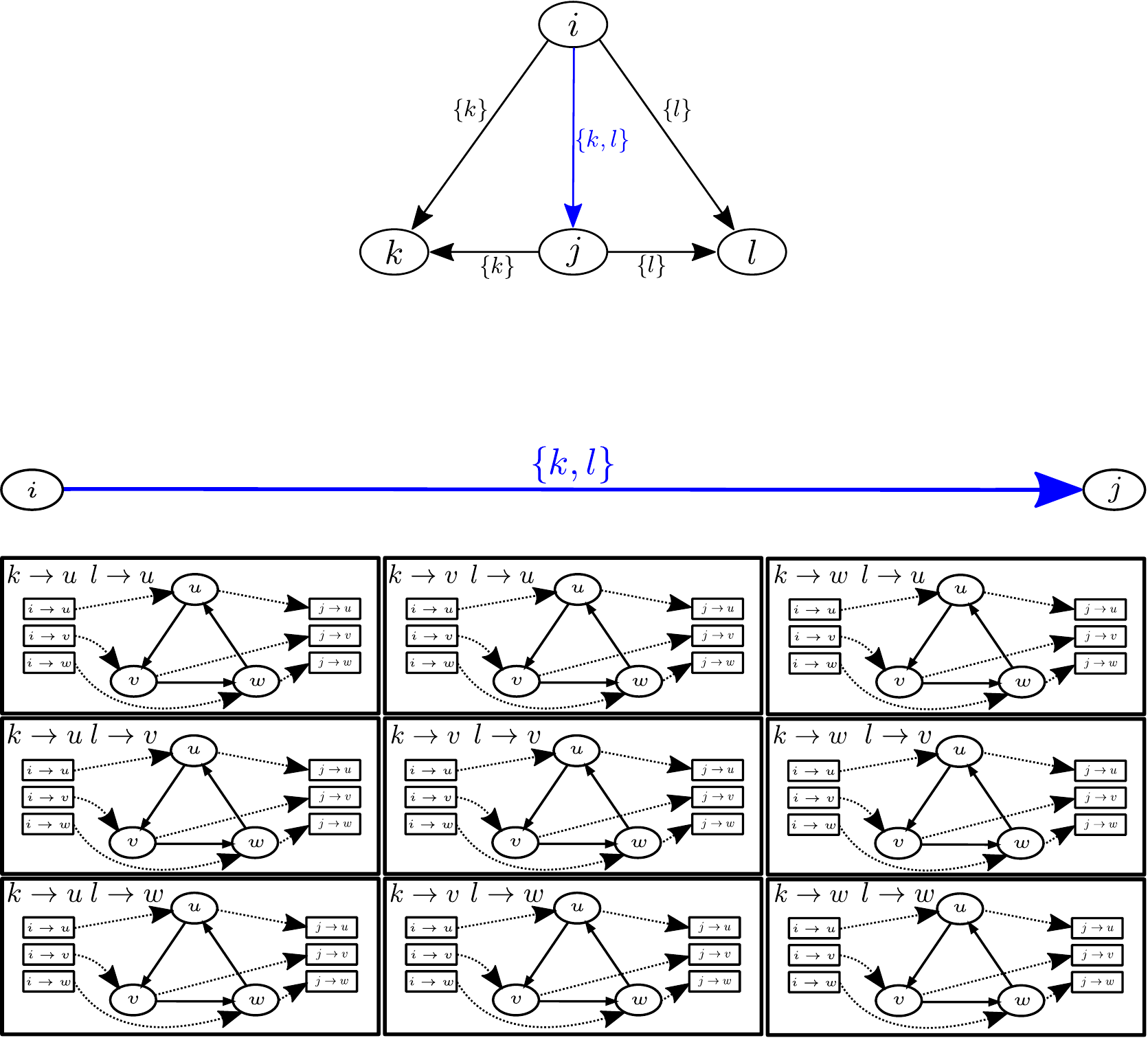}
    \caption{
Construction of LP subformulations for a request with overlapping confluences. The edge $(i, j)$ lies on two confluences, one ending in $k$ and one ending in $l$. \newline
Top: The request topology. Edges are labeled with the end nodes of confluences they belong to. The substrate topology is the same as in Figure~\ref{fig:base_LP_cycle_example}. \newline
Bottom: The LP subformulations for the edge $(i, j)$. Each of the nine possible combinations of mapping decisions for $k$ and $l$ is represented by a subformulation. Note that the diagram arranges the subformulations in a grid due to space limitations. Each of the subformulations must be connected to the global node mapping variables, similar to Figure~\ref{fig:base_LP_cycle_example}.
    }
    \label{fig:LP_cycle_nested_example}
\end{figure}

\subsubsection{Definitions}
In this subsection, we introduce the prerequisite concepts for the base algorithm more formally, following the terminology used in \cite{rostSchmidFPTApproximations}.

We first describe the \emph{confluence edge label assignment}, which assigns to each edge the set of end nodes of confluences containing the edge.
\begin{definition}(Confluence Edge Label Assignment) \label{def:confluence-edge-labels} \\
Given an extraction order $\reqExtractionOrderDef$, we define for each edge $e \in \reqExtractionOrderEdges$ a set of edge labels $\labelsetEdge \subset \reqNodes$, such that for each confluence $\confluence[i][j] \in \confluenceSet$ containing $e$, the end node is contained in the label assignment for $e$, i.e. $j \in \labelsetEdge$. 
Combining the edge label sets for each edge in the extraction order, we write $\labelsetsEdgesExtractionOrder$ for the confluence edge label assignment as a whole.
\end{definition}

The confluence edge label assignment is referred to as \emph{extraction edge labels} in \cite{rostSchmidFPTApproximations}. Since the extension of the base algorithm proposed in Section~\ref{sec:multiroot} requires different edge label assignments, we use the more specific term.

We now briefly discuss how to practically identify confluences and calculate edge labels given an extraction order $\reqExtractionOrderDef$. By Definition~\ref{def:confluence}, any confluence end node has at least two incoming edges. Conversely, any node $j$ with multiple in-edges must be a confluence end node. The request topology does not admit parallel edges by Definition~\ref{def:introvnep:reqtop} and since the extraction order is a rooted graph, any of $j$'s in-neighbors must be reachable from the root $\reqExtractionOrderRoot$, which implies the existence of at least one common ancestor. The confluence start node is the common ancestor of $j$'s parent nodes which is farthest from the root. To find it, we traverse the extraction order's edges in reverse order towards the root. Any edge that lies on a path from the common ancestor to $j$ is labeled with $j$.

The following property of the edge label assignment is important for proving the correctness of the base decomposition algorithm.
\begin{lemma}(Equality of Incoming Label Sets, cf. \cite{rostSchmidFPTApproximations}) \label{lemma:incomingLabelsUnique}\\
Given an extraction order $\reqExtractionOrderDef$, a node $i \in \reqNodes$ and any pair $e, e' \in \inEdgesExtractionOrder{i}$,   $\labelsetEdge[e] = \labelsetEdge[e']$ holds, where $\labelsetEdge[e]$ is the label assignment for $e$ in the confluence edge label assignment $\labelsetsEdges$.
\end{lemma}
\begin{proof}
See Lemma 15 in \cite{rostSchmidFPTApproximations}.
\end{proof}

Since a node's out-edges may have different, but overlapping label sets, the base algorithm must ensure that the flows induced in each edge's subformulations are consistent. To this end, another layer of variables is introduced, aggregating edges with overlapping label assignments. Rost and Schmid introduce the notion of \emph{edge bags}, which partition a request node's out-edges according to their edge labels, such that any of a node's edges whose label sets overlap are included in the same edge bag.
\begin{definition}(Edge Bag)\label{def:edge-bags}\\
Given an extraction order $\reqExtractionOrderDef$ and some node $i \in \reqNodes$, the set of edge bags $\BagIPlus := \{\outEdgeBag | \outEdgeBag \subseteq \outEdgesExtractionOrder{i}\}$ is defined as a disjoint partition of $i$'s outgoing edges, i.e. the following properties hold:
\begin{enumerate}
\item $\bigcup \BagIPlus = \outEdgesExtractionOrder{i}$
\item For all $e \in \outEdgeBag$ and $e' \in \outEdgeBag[i][c]$ with $b \neq c$, $\labelsetEdge \cap \labelsetEdge[e'] = \emptyset$.
\item $\BagIPlus$ is the partition minimizing the size of the associated label set, i.e. minimizing \\
$\max_{\outEdgeBag \in \BagIPlus} \left|\bigcup_{e \in \outEdgeBag} \labelsetEdge\right|$
\end{enumerate}
\end{definition}
While the first two properties of the definition of edge bags are sufficient to ensure the correctness of the base algorithm, the third property is added to uniquely identify the partition.

For each edge bag, the associated \emph{bag label set} is defined as the union of the contained edges' label sets.
\begin{definition}(Bag Label Set)\label{def:bag-label-set}\\
Given an extraction order $\reqExtractionOrderDef$, a request node $i \in \reqNodes$ and an edge bag $\outEdgeBag \in \BagIPlus$, we define the bag label set $\bagLabelSet$ as $\bagLabelSet := \bigcup_{e \in B_b} \labelsetEdge$.
We further denote by $\bagLabelSets := \{\bagLabelSet | \outEdgeBag \in \BagIPlus \}$ the set containing a node's bag label sets.
\end{definition}

We will now define \emph{mapping spaces}. Given a set of nodes, the corresponding mapping space is the set containing all possible combinations of mappings of the individual nodes.
\begin{definition}(Mapping Space) \\
Given an extraction order $\reqExtractionOrderDef$ and some set of request nodes $X \subseteq \reqNodes$, the mapping space $\MappingSpace[X]$ is defined as the set containing each combination of possible node mappings for the nodes contained in $X$. Formally, we define the mapping space as the set of all functions from $X$ to $\substrateNodes$.
\end{definition}
For example, for the VNEP instance from Figure~\ref{fig:base_LP_cycle_example}, the mapping space for the node set $\{i, j\}$ is given by
\begin{align}
\MappingSpace[\{i, j\}] =
\left\{
\begin{array}{ccc}
\{ i \mapsto u, j\mapsto u \}, & \{ i \mapsto v, j\mapsto u \}, & \{ i \mapsto w, j\mapsto u \}, \\ 
\{ i \mapsto u, j\mapsto v \}, & \{ i \mapsto v, j\mapsto v \}, & \{ i \mapsto w, j\mapsto v \}, \\ 
\{ i \mapsto u, j\mapsto w \}, & \{ i \mapsto v, j\mapsto w \}, & \{ i \mapsto w, j\mapsto w \}
\end{array} 
\right\}.
\end{align}

The next definition of a \emph{mapping projection} gives a method to restrict request mappings to subsets of their domain.
\begin{definition}(Mapping Projection) \\
Let  $V, V' \subseteq \reqNodes$ be sets of request nodes and $\mappingNodes: V \rightarrow \substrateNodes$ be a node mapping defined over $V$. The mapping projection of $\mappingNodes$ onto $V'$ is a function $\restrict[\mappingNodes][V']: V \cap V' \rightarrow \substrateNodes$, with
\begin{align}
\restrict[\mappingNodes][V'] (i) := \mappingNodes(i) \qquad \forall i \in V \cap V'.
\end{align}
We analogously define the projection for edge mappings. Given $E, E' \subseteq \reqEdges$, and an edge mapping $m_E$, we define $\restrict[\mappingEdges][E']: E \cap E' \rightarrow \substratePathSet$ as
\begin{align}
\restrict[\mappingEdges][E'] (e) := \mappingEdges(e) \qquad \forall e \in E \cap E'.
\end{align}
Finally, we extend the notation to request mappings: Let $G, G' := (V', E')$ be two subgraphs of the request. Let $\mappingRequestDef$ be a mapping for $G$. We define 
\begin{align}
\restrict[\mappingRequest][G'] := (\restrict[\mappingNodes][V'], \restrict[\mappingEdges][E']).
\end{align}
\end{definition}
Mapping projections are very useful to compare different mappings: Consider two node sets $V, V' \subseteq \reqNodes$ and two node mappings $m$ and $m'$ defined over $V$ and $V'$, respectively. We can now use $\restrict[m][V'] = \restrict[m'][V]$ to express that the mappings agree on the intersection of $V$ and $V'$. However, we will usually use the more explicit notation $\restrict[m][V' \cap V]$ to emphasize that the mapping projection is defined for the intersection of the domain of $m$ with the set $V'$ to which the mapping is projected.

We introduce \emph{Edge Subgraphs}, which are used to define the LP subformulations. An edge subgraph is a subgraph of the request induced by the tail and head node of a single edge.
\begin{definition}(Edge Subgraphs) \\
Given a request topology $\reqTopologyDef$ and an edge $e \in \reqEdges$, the edge subgraph $\VGe = (\VVe, \VEe)$ is the subgraph induced by the head and tail node of $e$, i.e. for $(i, j) \in \reqEdges$, we define $\VVe[(i, j)] := \{i, j\}$ and $\VEe[(i, j)] := \{(i, j)\}$.
\end{definition}
Note that the edge subgraph is uniquely defined for each edge, since we explicitly forbid parallel edges in the request topology (see Definition~\ref{def:introvnep:reqtop}).

\subsection{The Base LP Formulation}

The base algorithm's LP formulation is given in Linear Program~\ref{IP:novel-AC}. We will now discuss the variables and constraints that are introduced by the base LP formulation.


\begin{figure}[bh!]
  
  {
    \LinesNotNumbered
    \renewcommand{\arraystretch}{1.4}
    
    \removelatexerror

    \begin{IPFormulation}{H}
      
      \SetAlgorithmName{Linear Program}{}{{}}
      
      \caption{Decomposable Base Formulation for VNEP, adapted from \cite{rostSchmidFPTApproximations}}
      \label{IP:novel-AC}
      \popline

      \newcommand{\spaceIt}{\qquad\quad\quad}
      \newcommand{\miniSpace}{\hspace{1.5pt}}
      
\scalebox{0.88}{
\begin{minipage}{1.08\columnwidth}
      
      \begin{tabular}{FRLQB}
        \multicolumn{5}{r}{\parbox{0.975\textwidth}{~}} \\[-16pt]        
        
        \multicolumn{3}{L}{\textnormal{  \hspace{-10pt}(\ref{eq:classic_mcf:constr_flow}) -  (\ref{eq:classic_mcf:constr_edge_load}) for $\VGe$ on variables } \subLP[(\vec{y},\vec{z},\vec{a})][e,\mappingChar_e]} ~~& \forall e \in \reqEdges, \mappingChar_e \in \MappingSpace[\labelsetEdge]   & \tagIt{alg:lp:novel:subformulations}
        \\[6pt]
        
        \sum \limits_{u \in \substrateNodesByType[\reqNodeType(i)] } y^u_{i}  & ~=~ & 1 &  \forall i \in \reqNodes &  \tagIt{alg:lp:novel:node-embedding} \\
        
        y^u_{i} & ~=~ & 	\sum_{\mappingChar_e \in \MappingSpace[\labelsetEdge]} \hspace{-12pt} \subLP[y^u_{i}][e,\mappingChar_e] & \forall i \in \reqNodes,  u \in \substrateNodesByType[\reqNodeType(i)], e \in \reqEdges: i \in e &  \tagIt{alg:lp:novel:node-to-sub-node-mapping} \\
        
        \subLP[y^u_{i}][\EOEdgeToOriginal(e),\mappingChar_e]  & ~=~ & \sum_{\begin{subarray}{c}
            \mappingChar_\labelsetmappingIndexTwo \in \MappingSpace[\labelsetIndexed{\labelSetIndexTwo}]: \\
            \restrict[\mappingChar_\labelsetmappingIndexTwo][\bagLabelSet[i][b] \cap \labelsetEdge] = \mappingChar_e
        \end{subarray}} \gamma^u_{i,b,\labelsetmappingIndexTwo}  & \hspace{-5pt}\begin{array}{l}
          \forall i \in \reqNodes, u \in \substrateNodesByType[\reqNodeType(i)], \VEbfsBag \in \BagIPlus, \\ 
          \quad e \in \VEbfsBag, \mappingChar_e \in \MappingSpace[\labelsetEdge]
        \end{array} & \tagIt{alg:lp:novel:gamma-to-outgoing-edges}\\
        
        \sum_{\begin{subarray}{c}
            \mappingChar_e \in \MappingSpace[\labelsetEdge]: \\
            \restrict[\mappingChar_e][\bagLabelSet[i][b] \cap \labelsetEdge] =  \mapVinter[b][e]
        \end{subarray}} \hspace{-16pt} \subLP[y^u_{i}][\EOEdgeToOriginal(e),\mappingChar_e] & ~=~ & \sum_{\begin{subarray}{c}
            \mappingChar_\labelsetmappingIndexTwo \in \MappingSpace[\labelsetIndexed{\labelSetIndexTwo}]: \\
            \restrict[\mappingChar_\labelsetmappingIndexTwo][\bagLabelSet[i][b] \cap \labelsetEdge] =  \mapVinter[b][e]
        \end{subarray}} \hspace{-16pt} \gamma^u_{i,b,\labelsetmappingIndexTwo}  & \hspace{-5pt}\begin{array}{l}
          \forall  i \in \reqNodes, u \in \substrateNodesByType[\reqNodeType(i)], e \in \inEdgesExtractionOrder{i}, \\ 
          \quad \VEbfsBag \in \BagIPlus, \mapVinter[b][e] \in \MappingSpace[\bagLabelSet[i][b] \cap \labelsetEdge]
        \end{array}  & \tagIt{alg:lp:novel:incoming-edges-to-gamma-variables}\\

        \subLP[y^u_{i}][\EOEdgeToOriginal(e),\mappingChar_e] & ~=~ & 0  & \hspace{-5pt}\begin{array}{l}
          \forall i \in \reqNodes, e \in \inEdgesExtractionOrder{i}: i \in \labelsetEdge,  \\ 
          \quad \mappingChar_e \in \MappingSpace[\labelsetEdge], u \in  \substrateNodesByType[\reqNodeType(i)] \setminus \{\mappingChar_e(i)\}
        \end{array}  & \tagIt{alg:lp:novel:forbidding-nodes-in-sub-lps}\\
        
        a^{\type,u}  & ~=~  & 	 	\sum \limits_{i \in \reqNodes, \reqNodeType(i) = \type}  \hspace{-12pt}  \reqDemand(i) \cdot y^u_{i} & \forall (\type,u) \in  \SRV &  \tagIt{alg:lp:novel:node-load}\\
        
        a^{u,v} & ~=~ & \sum_{\begin{subarray}{c}
            e \in \reqEdges \\ \mappingChar_e \in \MappingSpace[\VEbfsLabelsOrig]
        \end{subarray}} \hspace{-12pt} \subLP[a^{u,v}][e,\mappingChar_e]  \quad & \forall (u,v) \in  \SE&  \tagIt{alg:lp:novel:edge-load}\\
        

        \multicolumn{4}{C}{
          \hspace{-10pt}
          \begin{array}{c}
            y^u_{i} \in [0,1],~\forall i \in \reqNodes, u \in  \substrateNodesByType[\reqNodeType(i)]; \hspace{18pt} a^{x,y} \in [0, \Scap(x,y)],~\forall (x,y) \in \SR \\
            \gamma^u_{i,b,\labelsetmappingIndexTwo} \in [0,1],~\forall  i \in \reqNodes, u \in  \substrateNodesByType[\reqNodeType(i)], \VEbfsBag \in \BagIPlus, \mappingChar_\labelsetmappingIndexTwo \in \MappingSpace[\labelsetIndexed{\labelSetIndexTwo}]
          \end{array}

        } & \tagIt{alg:lp:novel:variables}
      \end{tabular}
    
\end{minipage}}
    \end{IPFormulation}
  }

\end{figure}

\subsubsection{Base LP Variables} 
We first describe the different classes of variables in Linear Program~\ref{IP:novel-AC} and their meaning within the LP formulation.

The subformulation variables $\subLP[\vec{y}][e,\mappingChar_e]$, $\subLP[\vec{z}][e,\mappingChar_e]$ and $\subLP[\vec{a}][e,\mappingChar_e]$ have the same functions as their counterparts $\vec{y}$, $\vec{z}$ and $\vec{a}$ in Integer Program~\ref{mip:matthias:MCF-formulation}. The subformulation variables are each associated with a request edge $e \in \reqEdges$, through the edge subgraph of $e$. Subformulations are introduced for each mapping in the mapping space $\MappingSpace[\labelsetEdge[e]]$. The edge and the label set mapping, with which a subformulation variable is associated, are indicated in the brackets $\subLP[][\cdot]$.

The global node variables $\vec{y}$ aggregate the values of node variables $\subLP[\vec{y}][e, \mappingChar_e]$ in different subformulations. Similarly, the allocation aggregation variables $\vec{a}$ aggregate the resource footprints of all subformulations.

\pagebreak
The node bag variables $\vec{\gamma}$ are required to enforce the consistency of a mapping with respect to the label sets of a node's out-edges. A variable $\gamma^{u}_{i, \labelSetIndexTwo, \labelsetmappingIndexTwo}$ is defined for each possible combination of the following, as indicated by the indices:
\begin{description}[align=right,labelwidth=5cm]
\item[$i \in \reqNodes$:] a request node 
\item[$u \in \substrateNodes^{\reqNodeType(i)}$:] a substrate node, onto which $i$ can be mapped
\item[$\outEdgeBag \in \BagIPlus$:] an edge bag
\item[$\mappingChar_\labelsetmappingIndexTwo \in \mathcal{M}(L_{i,b})$:] a mapping of the edge bag's label set 
\end{description}

\subsubsection{Base LP Constraints}
We now discuss the constraints of the base LP formulation, beginning with the constraints related to the LP subformulations.
\begin{itemize}
\item Constraint (\ref{alg:lp:novel:subformulations}) applies constraints of Integer Program~\ref{mip:matthias:MCF-formulation} to each subformulation, with the exception of Constraint (\ref{eq:classic_mcf:constr_force_emb}). The omission of Constraint~(\ref{eq:classic_mcf:constr_force_emb}) allows any given subformulation to only partially embed each node / edge.
\end{itemize}
Next, we describe the constraints which connect the subformulations' variables to the global variables of the LP formulation.
\begin{itemize}
\item Constraint (\ref{alg:lp:novel:node-embedding}) enforces that each request node $i$ is completely mapped.
\item Constraint (\ref{alg:lp:novel:node-to-sub-node-mapping}) distributes the node mapping over the subformulations, ensuring that each node is mapped in some set of subformulations.
\item Constraint (\ref{alg:lp:novel:forbidding-nodes-in-sub-lps}) applies to any LP subformulation $\subLP[][e, \mappingChar_e]$, whose associated edge $e$ ends in a confluence end node $k$. This implies that $k \in \labelsetEdge[e]$. The constraint enforces that any node mapping for $k$ in the subformulation is consistent with $\mappingChar_e$. Since the mapping decision for $k$ is fixed in $\mappingChar_e$, consistency is achieved by explicitly forbidding any other mapping of the end node $k$.
\item Constraints (\ref{alg:lp:novel:node-load}) and (\ref{alg:lp:novel:edge-load}) track the resource allocations of the node and edge mapping, respectively. For the node mappings, the allocated resources can directly be determined through the global node mapping variables. For computing the edge allocations, the edge mapping variables of the LP subformulations are used.
\end{itemize}
Finally, Constraints (\ref{alg:lp:novel:gamma-to-outgoing-edges}) and (\ref{alg:lp:novel:incoming-edges-to-gamma-variables}) ensure that subformulations of different edge bags are consistent:
\begin{itemize}
\item Constraint (\ref{alg:lp:novel:gamma-to-outgoing-edges}) induces a flow in each of a node's out-edges' LP subformulations, that is consistent with the aggregated flow of node bag variables $\gamma$.
\item Constraint (\ref{alg:lp:novel:incoming-edges-to-gamma-variables}) enforces that the node mapping variable assignment of a node's incoming edges' subformulations is appropriately propagated to each of the edge bags.
\end{itemize}

\subsection{The Base Decomposition Algorithm} \label{sec:base_decomp}

The base decomposition algorithm is given as pseudocode in Algorithm~\ref{alg:decomposition:Algorithm-Novel-AC}. Similarly to Algorithm~\ref{alg:matthias:decomposition:Algorithm-MCF-Tree}, the base decomposition algorithm initially chooses a mapping for the extraction order root node $\reqExtractionOrderRoot$ according to some non-zero node mapping variable $y^{u}_{\reqExtractionOrderRoot}$, and subsequently traverses the extraction order. 

This traversal is realized with the queue $\mathcal{Q}$, to which a request node is only added in Line~\ref{algline:base-decomposition:add-j-to-q} after all in-edges are validly mapped. Therefore, in each iteration of the while-loop in Lines~\ref{algline:base-decomposition:begin-while-q} to \ref{algline:base-decomposition:add-j-to-q}, a mapping for request node $i$ and all of its ancestor nodes is already set according to the partial mapping $\mappingNodesIteration$.

The algorithm ensures that confluences are consistently mapped: Whenever node $i$ is a start node of some confluence, the confluence end node $t$ must occur as a label in one of the bag label sets $\bagLabelSet[i][b] \in \bagLabelSets[i]$, and is therefore mapped in Line~\ref{algline:base-decomposition:map-bag-label-nodes}.

Since Algorithm~\ref{alg:decomposition:Algorithm-Novel-AC} is besides slight notational differences the same as the decomposition algorithm by Rost and Schmid, we refer to \cite{rostSchmidFPTApproximations} for a correctness proof for Algorithm~\ref{alg:decomposition:Algorithm-Novel-AC}. In Section~\ref{sec:hierarchical-bags}, we will introduce an adapted LP formulation and decomposition algorithm of which the base LP formulation and decomposition algorithm can be considered a special case. For this extended algorithm, a correctness proof is provided.

Analogously to the decomposition algorithm for tree-like requests, Rost and Schmid show in \cite{rostSchmidFPTApproximations} that the resource allocations of a convex combination $\PotEmbeddings$ returned by Algorithm~\ref{alg:decomposition:Algorithm-Novel-AC} are bounded by the corresponding allocation variables, i.e. $\sum_{(\prob, \mappingRequestIteration) \in \PotEmbeddings} \prob \cdot \allocationFunction(\mappingRequestIteration, x, y) \leq a^{x, y}$ holds for each resource $(x, y) \in \substrateResources$.

\subsection{Complexity of the Base Algorithm} \label{sec:performance-base-approx}

As described in Section~\ref{sec:01:linear-programming}, Linear Program~\ref{IP:novel-AC} can be solved in polynomial time with respect to its size. Rost and Schmid show in \cite{rostSchmidFPTApproximations} that Algorithm~\ref{alg:decomposition:Algorithm-Novel-AC} terminates in polynomial time with respect to the size of the LP: In each iteration of the outer while-loop, at least one LP variable is set to zero in Line~\ref{algline:base-decomposition:adapt-variables-one}. 

Therefore, the size of the LP formulation determines the overall runtime of the base algorithm. We now discuss the contribution of the different variables to the LP size.

First, consider the subformulation variables. The number of variables required by each subformulations scales linearly with the size of the substrate topology. However, the number of subformulations required for a single request edge increases exponentially with the size of the edge's label set.

Similarly to the subformulations, the number of $\gamma$-variables associated with an edge bag $\outEdgeBag$ grows exponentially with the size of the bag label set $\bagLabelSet$. According to Definition~\ref{def:bag-label-set}, an edge bag's label set is the union of the contained edges' label sets. The largest bag label set in any extraction order is therefore at least as large as the size of the largest edge label set, and may in general be significantly larger than any individual edge label set. 

It follows that the number of $\gamma$-variables are the dominating factor in the LP size.


\begin{figure}[tbh!]
\scalebox{0.88}{
\begin{minipage}{1.08\columnwidth}

\removelatexerror
\begin{algorithm*}[H]

\SetKwInOut{Input}{Input}\SetKwInOut{Output}{Output}
\SetKwFunction{ProcessPath}{ProcessPath}{}{}
\SetKwFunction{reverse}{reverse}{}{}
\SetKwFunction{LP}{LP}
\SetKwFunction{LP}{LP}

\newcommand{\SET}{\textbf{set~}}
\newcommand{\ADD}{\textbf{add~}}
\newcommand{\EACH}{\textbf{each~}}
\newcommand{\DEFINE}{\textbf{define~}}
\newcommand{\AND}{\textbf{and~}}
\newcommand{\LET}{\textbf{let~}}
\newcommand{\WITH}{\textbf{with~}}
\newcommand{\COMPUTE}{\textbf{compute~}}
\newcommand{\FIND}{\textbf{find~}}
\newcommand{\CHOOSE}{\textbf{choose~}}
\newcommand{\DECOMPOSE}{\textbf{decompose~}}
\newcommand{\FORALL}{\textbf{for all~}}
\newcommand{\OBTAIN}{\textbf{obtain~}}
\newcommand{\WITHPROBABILITY}{\textbf{with probability~}}

\caption{Decomposition algorithm for solutions to the base formulation's LP \ref{IP:novel-AC}, reproduced with notational adaptations from \cite{rostSchmidFPTApproximations}}
\label{alg:decomposition:Algorithm-Novel-AC}

\Input{VNEP-instance (\substrateTopology, \reqTopology), request extraction order $\reqExtractionOrderDef$, solution to the relaxation of IP~\ref{IP:novel-AC}}
\Output{Convex combination~$\PotEmbeddings = \{\decomp = (\prob,\mappingRequestIteration)\}_k$ of valid mappings}
  \SET $x \gets 1$ \label{algline:base-decomposition:setRemainingFlowVar}\\
  \SET $\PotEmbeddings  \gets \emptyset$ \AND $k \gets 1$\\
  \While{$x > 0$ }
  {
    
    \SET $\mappingRequestIteration = (\mappingNodesIteration,\mappingEdgesIteration)~\gets (\emptyset,\emptyset)$ \label{algline:base-decomposition:init-maps}\\
    
    \SET $\Queue \gets \{\reqExtractionOrderRoot \}$ \\
    
    \CHOOSE $u \in \substrateNodes$ \WITH $y^u_{\reqExtractionOrderRoot} > 0$ \AND \SET $\mappingNodesIteration(\reqExtractionOrderRoot)~\gets u$\\
    
    \While{$|\Queue| > 0$}{  \label{algline:base-decomposition:begin-while-q}
      \CHOOSE $i \in \Queue$ \AND \SET$\Queue \gets \Queue \setminus \{i\}$\label{algline:base-decomposition:get-i-from-q}\\
      \ForEach{$\VEbfsBag \in \VEbfsBags$}{
        \LET $\mappedRequestNodes = (\mappingNodesIteration)^{-1}(\SV)$ denote the already mapped nodes\\
        \CHOOSE $m_\labelsetmappingIndexTwo \in \MappingSpace[\bagLabelSet[i][\labelSetIndexTwo]]$, s.t. $\gamma^{\mappingNodesIteration(i)}_{i,\labelSetIndexTwo,\labelsetmappingIndexTwo} > 0$ \AND $\restrict[m_\labelsetmappingIndexTwo][\bagLabelSet[i][\labelSetIndexTwo] \cap \mappedRequestNodes] = \restrict[\mappingNodesIteration][\bagLabelSet[i][\labelSetIndexTwo] \cap \mappedRequestNodes]$ \\
        \SET $\mappingNodesIteration(j) \gets m_\labelsetmappingIndexTwo(j)$ \FORALL $j \in \bagLabelSet[i][\labelSetIndexTwo] \setminus \mappedRequestNodes$ \label{algline:base-decomposition:map-bag-label-nodes}\\
        \ForEach{ $e=(i,j) \in \VEbfsBag$ \label{algline:base-decomposition:foreach-bag-label-set}}{
          \eIf{$(i,j) = \EOEdgeToOriginal(i, j)$}{
            \COMPUTE path ${P}^{u,v}_{i,j}$ from $\mappingNodesIteration(i)=u$ to  $v \in \substrateNodes$ according to Lemma~\ref{lem:local-connectivity-property}\\
          \pushline\pushline  \nonl s.t. 
          $\begin{array}{rl}
          \subLP[y^v_{j}][(i,j),\restrict[\mappingNodesIteration][\labelsetIndexed{e}]] > 0 & \textnormal{\AND} \\[2pt]
          \subLP[z^{u',v'}_{i,j}][(i,j),\restrict[\mappingNodesIteration][\labelsetIndexed{e}]] > 0 & \textnormal{\FORALL } (u',v') \in {P}^{u,v}_{i,j}
          \end{array}$\\
                      \vspace{2pt}
            \popline\popline
            
            \SET $\mappingEdgesIteration(i,j) \gets {P}^{u,v}_{i,j}$ \AND \textbf{if} $\mappingNodesIteration(j) = \emptyset$ \textbf{then} $\mappingNodesIteration(j) \gets v$  \label{algline:map-edge-if-not-reversed}\\
          }{
            \COMPUTE path ${P}^{v,u}_{j,i}$ from $v \in \substrateNodes$ to $\mappingNodesIteration(i)=u$ according to Lemma~\ref{lem:local-connectivity-property}\\
            \pushline \pushline \nonl s.t.
            $\begin{array}{rl}
              \subLP[y^v_{j}][\EOEdgeToOriginal(i, j),\restrict[\mappingNodesIteration][\labelsetIndexed{e}]] > 0 & \textnormal{\AND} \\[2pt]
              \subLP[z^{u',v'}_{j,i}][\EOEdgeToOriginal(i, j),\restrict[\mappingNodesIteration][\labelsetIndexed{e}]] > 0 & \textnormal{\FORALL } (u',v') \in {P}^{u,v}_{j,i}
            \end{array}$\\
            \vspace{2pt}
            \popline\popline

            \SET $\mappingEdgesIteration(\EOEdgeToOriginal(i,j)) \gets {P}^{u,v}_{j,i}$ \AND \textbf{if} $\mappingNodesIteration(j) = \emptyset$ \textbf{then} $\mappingNodesIteration(j) \gets v$ \label{algline:map-edge-if-reversed} \\
          
          }
          \If{$\mappingEdgesIteration(\EOEdgeToOriginal(e)) \neq \emptyset$ \textnormal{\FORALL}$e \in \inEdgesExtractionOrder{j}$}{
            \SET $\mathcal{Q} \gets \mathcal{Q} \cup \{j\}$ \label{algline:base-decomposition:add-j-to-q}\\
          }
        }
      }
    }
    \SET $\Variables_k \gets 
    \left( \begin{array}{ll}
            & \{x\}  \cup  \{y^u_{i} ~|~ i \in \reqNodes, u=\mappingNodesIteration(i)\} \\
        \cup   & \{~~\,\subLP[x][e,\restrict[\mappingNodesIteration][\labelsetIndexed{e}]] ~|~ e \in \reqEdges \} \\
        \cup   & \{~\,\subLP[y^u_{i}][e,\restrict[\mappingNodesIteration][\labelsetIndexed{e}]] ~|~ e \in \reqEdges, i \in e, u=\mappingNodesIteration(i) \}\\
        \cup    & \{\subLP[z^{u,v}_{i,j}][e,\restrict[\mappingNodesIteration][\labelsetIndexed{e}]]~|~ e=(i,j) \in \reqEdges, (u,v) \in \mappingEdgesIteration(i,j)\} \\
        \cup   & \{\gamma^{u}_{i,\labelSetIndexTwo,\labelsetmappingIndexTwo}~|~ i \in \reqNodes, u = \mappingNodesIteration(i), \VEbfsBag \in \VEbfsBags, m_\labelsetmappingIndexTwo=\restrict[\mappingNodesIteration][\bagLabelSet[i][\labelSetIndexTwo]]\} 
        \end{array}  \right)$ \label{algline:base-decomposition:compute-Vk}\\
    \SET $\prob \gets \min \Variables_k$ \label{algline:base-decomposition:computing-prob} \\
    \SET $v \gets v - \prob$ \FORALL $v \in \Variables$ \label{algline:base-decomposition:adapt-variables-one}\\
    \SET $a^{x,y} \gets a^{x,y} - \prob \cdot \allocationFunction(\mappingRequestIteration,x,y)$ \FORALL $(x,y) \in \SR$ \label{algline:base-decomposition:adapt-load-one}\\ 
    \ForEach{$(i,j) \in \reqEdges$ \textnormal{\AND \EACH} $(x,y) \in \{(\Vtype(i),i), (\Vtype(j),j), (i,j)\}$}{
      \SET $\subLP[a^{x,y}][(i,j),\restrict[\mappingNodes][\labelsetIndexed{e}]] \gets \subLP[a^{x,y}][\restrict[\mappingNodes][\labelsetIndexed{e}]] - \prob \cdot \allocationFunction(\mappingRequestIteration,x,y)$ \label{algline:base-decomposition:adapt-load}\\
    }
    \ADD $\decomp = (\prob,\mappingRequestIteration)$ to $\PotEmbeddings$ \AND \SET $k \gets k + 1$\\
  }

\KwRet{$\PotEmbeddings$}
\end{algorithm*}
\end{minipage}}

\end{figure}
\clearpage

\subsubsection{The Extraction Width Parameter} \label{sec:solving-the-vnep:extraction-width}

The size of the base LP formulation is determined by the choice of extraction order. The extraction width parameter was introduced in \cite{rostSchmidFPTApproximations} to parameterize the size of the LP formulation in terms of the selected extraction order.

We now give the definition of extraction width, both in terms of a specific extraction order, and as a graph parameter of the request topology.
\begin{definition}(Extraction Width, cf. \cite{rostSchmidFPTApproximations}) \label{def:extraction-width}\\
Let $\reqTopology$ be a request graph, and let $\reqExtractionOrderDef$ be an extraction order.

The width of $\reqExtractionOrder$ is then defined as
\begin{align}
\extractionWidth(\reqExtractionOrder) := 1 + \max_{i \in \reqNodes} \max_{\bagLabelSet \in \bagLabelSets} |\bagLabelSet|
\end{align}

We define the graph parameter extraction width by identifying with each graph the minimal width of all possible extraction orders: Let $\extractionOrderCharacter (\reqExtractionOrder) $ denote the set of all possible extraction orders of $\reqTopology$. We then define the extraction width of $\reqTopology$ as 
\begin{align}
\extractionWidthReq(\reqTopology) := \min_{\reqExtractionOrder \in \extractionOrderCharacter(\reqTopology)} \extractionWidth (\reqExtractionOrder).
\end{align}
\end{definition}

The following theorem, derived by Rost and Schmid in \cite{rostSchmidFPTApproximations}, gives the complexity of the base algorithm, in terms of LP size and decomposition runtime.
\begin{theorem} (Base Algorithm Complexity, cf. \cite{rostSchmidFPTApproximations}) \label{thm:base-alg-complexity}\\
Let $\VNEPInstance$ be an instance of VNEP. Given an extraction order $\reqExtractionOrderDef$ for the request $\reqTopology$, the size of Linear Program $\ref{IP:novel-AC}$ is bounded by $\mathcal{O}\left(|\substrateTopology|^{\extractionWidth(\reqExtractionOrder)} \cdot |\reqTopology| \right)$. Algorithm~\ref{alg:decomposition:Algorithm-Novel-AC} executes in time $\mathcal{O}\left(|\substrateTopology|^{2 \cdot \extractionWidth(\reqExtractionOrder) + 1} \cdot |\reqTopology|^2 \right)$.
\end{theorem}
\begin{proof}
See Theorem 17 and Theorem 18 in \cite{rostSchmidFPTApproximations}.
\end{proof}
Theorem~\ref{thm:base-alg-complexity} states that the base algorithm parameterized by extraction width is fixed-parameter tractable, i.e. both the size of the base LP formulation, and the runtime of the base decomposition algorithm are polynomially bounded for extraction orders with bounded width.

\subsubsection{Topologies with Small Extraction Width} \label{sec:solving-the-vnep:small-extraction-width-graphs}

Rost and Schmid derive several results related to the extraction width of specific graph classes. First, it is obvious by definition that tree-like requests have extraction width $1$: An extraction order $\reqExtractionOrderDef$ for a tree-like request cannot contain a confluence, and it follows for all edges $e \in \reqExtractionOrderEdges$ that $\labelsetEdge = \emptyset$ under the confluence edge label assignment. Then, all out-edges of a node can be partitioned in a single edge bag, whose bag label set is also empty.

A second class of topologies with bounded extraction width is given by the class of \emph{cactus graphs}.
\begin{definition}(Cactus Graphs) \label{def:cactus-graphs} \\
An undirected graph $\genericGraphDef$ is called a cactus graph, if any cycles contained in $\genericGraph$ intersect in at most a single node.
\end{definition}
We also refer to (directed) request topologies $\reqTopology$ as cactus graphs, if their undirected versions $\undirected[\reqTopology]$ meet this definition. Rost and Schmid show that cactus graphs have extraction width $2$. A detailed study of the base algorithm restricted to cactus graphs can be found in \cite{rost:vne-approx-leveraging-rand-round-2018}.

Rost and Schmid further introduce the class of half wheel graphs as an example of a topology, where suboptimal root placement has a drastic impact on the extraction width.
\begin{definition}(Half Wheel Graph, cf. \cite{rostSchmidFPTApproximations}) \label{def:half-wheel-graphs} \\
A half wheel graph $\halfwheelGraphDef$ is an undirected graph with node and edge sets defined as
\begin{align}
\halfwheelGraphNodes &:= \{ w_c\} \cup \{ w_1, w_2, \ldots w_n \} \\
\halfwheelGraphEdges &:= \{ \{ w_c, w_1 \}, \{ w_c, w_2 \}, \ldots, \{ w_c, w_n \} \} \cup \{\{w_1, w_2\}, \ldots, \{w_{n-1}, w_n \}\}
\end{align}
\end{definition}
Similar to our definition of cactus graphs, we also refer to directed request topologies as half wheel graphs, if their undirected version is a half wheel graph. 

The following lemma gives the width of extraction orders for a half-wheel graph in terms of the chosen root placement.
\begin{lemma} (Extraction Width of Half Wheel Graphs, cf. \cite{rostSchmidFPTApproximations}) \label{lem:extraction-width-half-wheel} \\
Consider a half wheel graph $\reqTopology$ with $n$ outer nodes. Considering any extraction order $\reqExtractionOrderDef$ for which the central node $w_c$ is chosen to be the root, $\extractionWidth (\reqExtractionOrder) \geq \lfloor n/2 \rfloor + 1$ holds. If an outer node is selected as the root node, an extraction order with $\extractionWidth (\reqExtractionOrder) = 2$ exists.
\end{lemma}
\begin{proof}
See Lemma 36 and Corollary 37 in \cite{rostSchmidFPTApproximations} for the case where $\reqExtractionOrderRoot = w_c$. 
When the root is placed at an outer node, any extraction order $\reqExtractionOrder$, with $\outEdgesExtractionOrder{w_c} = \emptyset$, i.e. all edges incident in the central node $w_c$ point towards $w_c$, has width $\extractionWidth(\reqExtractionOrder) = 2$ (cf. Figure 6 in \cite{rostSchmidFPTApproximations}).
\end{proof}
An example of two extraction orders of a half-wheel graph with $n=9$, demonstrating both root placements, is shown in Figure~\ref{fig:02:half_wheel_bad_root}.

\begin{figure}[h] 
	\centering
	\includegraphics[width=0.9\textwidth]{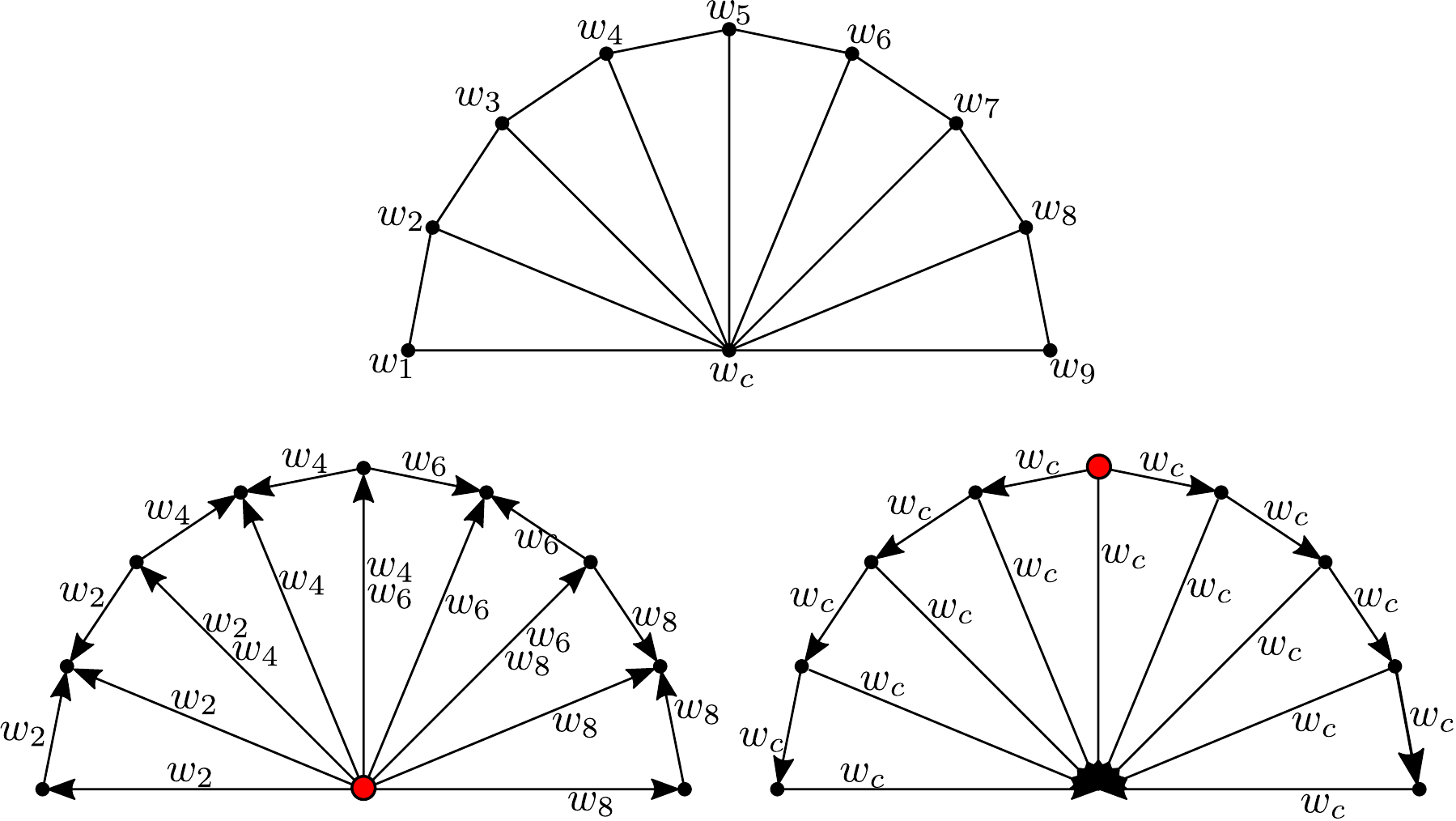}
	\caption{
Two extraction orders for a \emph{half-wheel graph} with 9 outer nodes. \newline
Top: The half-wheel graph as an undirected graph. \newline
Bottom: Two extraction orders, showing the optimal and non-optimal root placement. \newline 
Left: Extraction order obtained by designating the central node $w_c$ as the root node. While any given edge has at most two label nodes, all of the root node's out-edges are in the same edge bag with the bag label set $\{w_2, w_4, w_6, w_8\}$. By Lemma~\ref{lem:extraction-width-half-wheel}, the width scales linearly with the number of outer nodes. \newline
Right: Extraction order obtained by designating the outer node $w_5$ as the root node. Note that the width of this extraction order is $2$, as all confluences share the same end node.
}
	\label{fig:02:half_wheel_bad_root}
\end{figure}

The following result from \cite{rostSchmidFPTApproximations} gives a bound on the impact of adding parallel paths to an existing request topology.
\begin{lemma}(Bounded Increase of Extraction Width with Additional Parallel Paths, cf. \cite{rostSchmidFPTApproximations}) \label{thm:extraction-width-adding-parallel-paths} \\
Given an arbitrary graph $\reqTopologyDef$, adding any number of paths parallel to an existing edge increases the extraction width $\extractionWidthReq (\reqTopology)$ by at most the maximum degree of $\reqTopology$.
\end{lemma}
\begin{proof}
See Lemma 34 in \cite{rostSchmidFPTApproximations}.
\end{proof}

\begin{remark} (Clarification Regarding Parallel Edges) \\
Lemma~\ref{thm:extraction-width-adding-parallel-paths} is originally phrased in \cite{rostSchmidFPTApproximations} in terms of adding parallel \emph{edges} to the request topology. Since the definition of a request topology in this thesis, Definition~\ref{def:introvnep:reqtop}, explicitly forbids parallel edges, we only consider inserting parallel paths for an existing edge.
\end{remark}

The problem of finding an \emph{optimal} extraction order, i.e. one with minimal width, for a given request topology was shown to be NP-hard in \cite{rostSchmidFPTApproximations}.
\begin{theorem} (Hardness of Computing Optimal Extraction Orders, cf. \cite{rostSchmidFPTApproximations}) \\
Computing an extraction order of minimum extraction width is $\complexityNP$-hard.
\end{theorem}
\begin{proof}
See Theorem 40 in \cite{rostSchmidFPTApproximations}.
\end{proof}

\cleardoublepage
\section{Theoretical Background} \label{sec:background}

In this section, we give an overview of some theoretical concepts that are used for the improvement to the base algorithm presented in Section~\ref{sec:hierarchical-bags} of this thesis.

We will first discuss tree decompositions and the related graph parameter treewidth. Intuitively, a tree decomposition is some representation of an arbitrary graph in a tree-like structure. The treewidth parameter can then be considered to quantify the similarity of a graph to a tree. Treewidth has many applications in parameterizing the complexity of $\complexityNP$-hard problems. A survey of such results is given by Arnborg in \cite{arnborg1985efficient}. 

In particular, tree decompositions have led to very important results in the context of graph problems formulated in second-order monadic logic, where an important result was shown by Courcelle in \cite{courcelle1990monadic}, stating that such problems can be decided in linear time for graphs with bounded treewidth. Tree decompositions have applications in statistics and machine learning, where they are used in algorithms related to Bayesian networks. In this context, they are commonly referred to as junction trees \cite{lauritzen1988local,jensen1996introduction,nielsen2009bayesian,cowell2006probabilistic}. 

We then discuss the running intersection property, which is a property of set families closely related to tree decompositions. The running intersection property will be a defining characteristic of the extension of the base algorithm.

Finally, we give a brief overview of the theory of hypergraphs. Hypergraphs are a generalization of an undirected graph which allows for edges connecting more than two nodes. Any finite family of sets can be considered the edge set of such a hypergraph. This makes hypergraphs a useful way to represent the set of label sets assigned to a node's outgoing edges. We will also discuss how the concepts of tree decompositions and the running intersection property are related to acyclicity in hypergraphs. 

The theory of hypergraphs has applications in database theory: A database schema can be interpreted as a hypergraph, and Beeri, Fagin, Maier and Yannakakis show in \cite{beeri1983desirability} that the class of database schemas corresponding to \emph{acyclic} hypergraphs in particular have useful properties.

\subsection{Tree Decompositions and Treewidth} \label{sec:background:tree-decomp-width}

In this section, we introduce the concept of tree decompositions and the closely related graph parameter treewidth. 

The notion of treewidth was first introduced in 1986 by Robertson and Seymour in a series of works on graph minors \cite{robertson1986treewidth}. It has since gained significant attention, because many NP-hard problems are solvable in polynomial time for graphs with bounded tree width.

We first define tree decomposition, following the definition given by Bodlaender in \cite{bodlaender1996linear}.
\begin{definition}(Tree Decomposition, cf. \cite{bodlaender1996linear}) \label{def:tree-decomposition} \\
A tree decomposition of a graph $\genericGraphDef$ is a pair $\treeDecompDef$, where $\treeDecompTreeDef$ is a tree and $\treeDecompSets = \{\treeDecompSet_i ~|~ i \in \treeDecompTreeNodes \}$ is a family of subsets of $\genericGraphNodes$, containing one subset for each node of $\treeDecompTree$, such that: 
\begin{enumerate}
\item $\bigcup_{i \in \treeDecompTreeNodes} \treeDecompTreeNodes  = \genericGraphNodes$, i.e. each node in $\genericGraphNodes$ is included in at least one of the sets in $\treeDecompSets$. \label{def:tree-decomp:every-node-is-included}
\item For all edges $(v, w) \in \genericGraphEdges$, there exists an $i \in \treeDecompTreeNodes$ with $v \in \treeDecompSet_i$ and $w \in \treeDecompSet_i$. \label{def:tree-decomp:every-edge-is-covered}
\item For all $i, j, k \in \treeDecompTreeNodes$, if $j$ is on the path from $i$ to $k$ in $\treeDecompTree$, then $\treeDecompSet_i \cap \treeDecompSet_k \subseteq \treeDecompSet_j$. \label{def:tree-decomp:sets-on-path-contain-nodes}
\end{enumerate}
The sets in $\treeDecompSets$ are also referred to as the \emph{bags} of the tree decomposition.
\end{definition}
This definition of tree decompositions separates the family of node sets $\treeDecompSets$ from the node set of the tree $\treeDecompTree$. For simplicity, we directly identify the nodes of $\treeDecompTree$ with the corresponding node sets in $\treeDecompSets$, e.g. we may write ``$\treeDecompSet_i \in \treeDecompTreeNodes$'' instead of ``$\treeDecompSet_i \in \treeDecompSets$, where $i \in \treeDecompTreeNodes$''.

Examples of tree decompositions are shown in Figure~\ref{fig:tree-decomp-examples}.
\begin{figure}[htbp]
	\centering
	\includegraphics[width=0.8\textwidth]{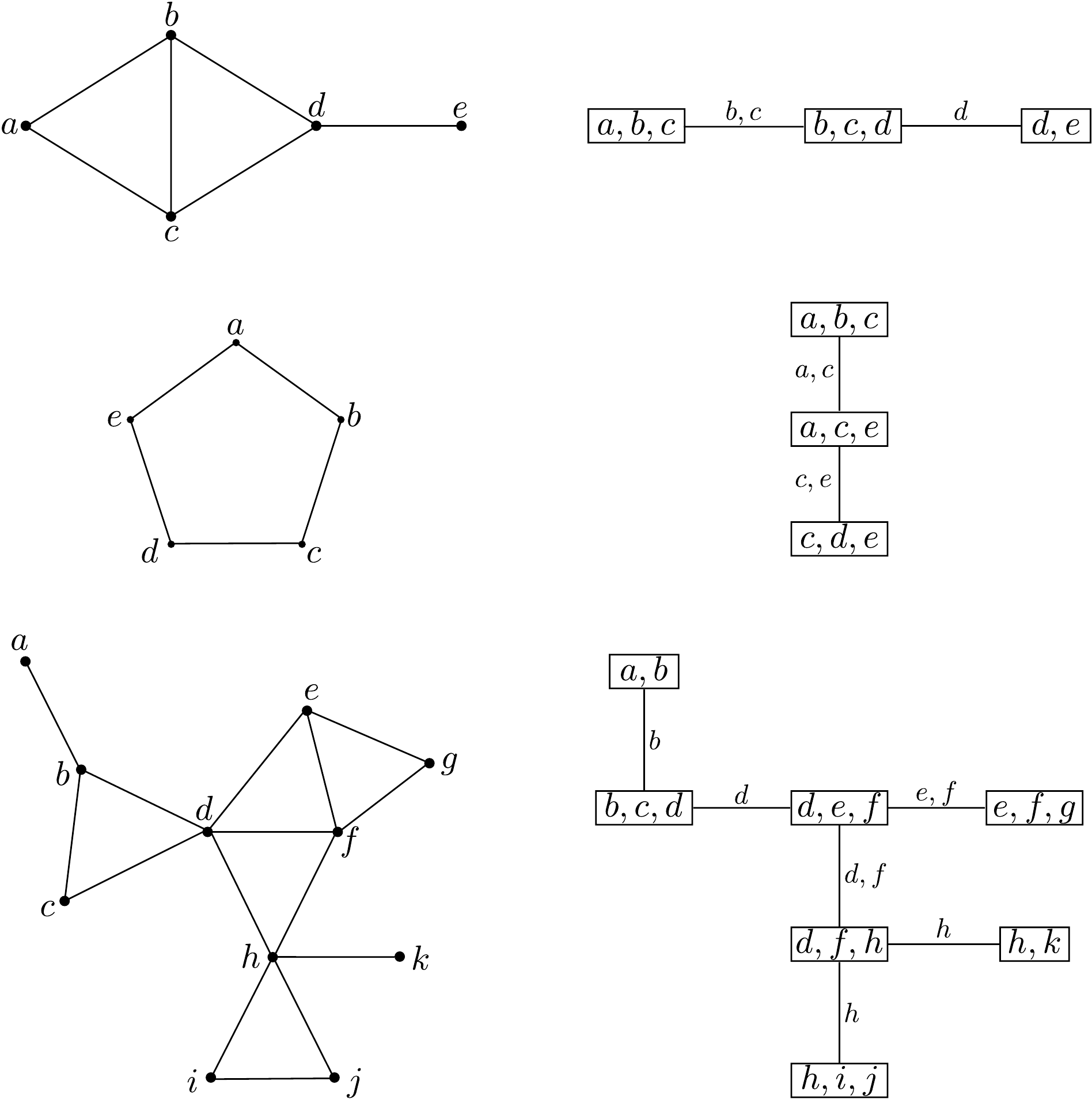}
	\caption{
Three examples of tree decompositions. The left column shows the original graphs, and the right column visualizes a tree decomposition for each graph. Note that the nodes of the tree decomposition are directly identified with the corresponding node sets. The edges of the tree decomposition's tree are labeled with the intersection of the incident nodes, to show the connection of the tree decompositions to \emph{join trees}, which are introduced later in Definition~\ref{def:join-tree}.
	}
	\label{fig:tree-decomp-examples}
\end{figure}

Property~\ref{def:tree-decomp:sets-on-path-contain-nodes} of the definition of tree decompositions can equivalently be stated as follows:
Given any node $v \in \genericGraphNodes$ of the undirected graph, the subgraph of the tree $\treeDecompTree$ induced by the sets containing $v$ is a tree.

A tree decomposition can be found for any undirected graph, as gathering all nodes in a single set trivially defines a tree decomposition. However, tree decompositions where only a small number of nodes are gathered in each single node set are of interest. The width of a tree decomposition is therefore defined as the size of its largest node set.
\begin{definition} (Width of a Tree Decomposition) \\
Given a tree decomposition $\treeDecompDef$ for some undirected graph $\genericGraphDef$, the width of $\treeDecomp$ is the size of the largest node set in $\treeDecompSets$:
\begin{align}
\decompwidth(\treeDecomp) := \max_{\treeDecompSet \in \treeDecompSets} |X| - 1
\end{align}
\end{definition}

The treewidth can then be defined as a graph parameter of the original undirected graph $\genericGraph$. It is the minimal width of all tree decompositions of $\genericGraph$.
\begin{definition} (Treewidth) \\
Let $\genericGraphDef$ be an undirected graph. The treewidth of $\genericGraph$ is defined as
\begin{align}
\treewidth(\genericGraph) := \min_{\treeDecomp \in \setOfTreeDecomps(\genericGraph)} \decompwidth(\treeDecomp)
\end{align}
where $\setOfTreeDecomps(\genericGraph)$ is the set of all tree decompositions of $\genericGraph$.
\end{definition}

Intuitively, treewidth can be thought of as a measure of how structurally similar a graph is to a tree. When a tree decomposition of small width exists, only few nodes are combined in any given set in $\treeDecompSets$, and the decomposition's tree $\treeDecompTree$ resembles the structure of original graph.

The following theorem, stating the hardness deciding the treewidth of a graph, was originally shown by Arnborg, Corneil and Proskurowski in \cite{arnborg1987complexity}, in the context of identifying subgraphs of so-called $k$-trees, which is equivalent to the problem of computing tree decompositions.
\begin{theorem}(Hardness of Treewidth Decision Problem, cf. \cite{arnborg1987complexity}) \label{thm:treewidth-hardness}\\
Given an arbitrary graph $\genericGraph$ and an integer $k$, the decision problem ``Is $\treewidth(G) \leq k$?'' is $\complexityNP$-complete.
\end{theorem}

However, for any fixed value of $k$, the treewidth decision problem is solvable in polynomial time, i.e. it is fixed-parameter tractable with the parameter $k$. For graphs with treewidth $k \leq 3$, linear-time algorithms for finding a tree decomposition have been found by Matoušek and Thomas in \cite{matouvsek1991algorithms}. For graphs with bounded treewidth, algorithms for computing the treewidth with polynomial runtime have been found, such as the algorithm by Reed with $\mathcal{O}(n \log(n))$ complexity in \cite{reed1992computing-treewidth}. Further, a treewidth approximation algorithm with polynomial runtime is given by Bodlaender, Gilbert, Hafsteinsson and Kloks in  \cite{bodlaender1995approximating}, yielding a tree decomposition of width within $\bigO(\log n)$ of the treewidth.

We now introduce some classes of graphs with bounded treewidth and state some theorems related to treewidth. A detailed discussion of related results, and in particular an overview of graph classes with bounded treewidth is given by Bodlaender in the survey \cite{bodlaender1996arboretum}. 

\begin{lemma}(Treewidth of Forests, cf. \cite{bodlaender1996arboretum}) \label{lemma:treewidth-forests}\\
A graph has treewidth $1$ if and only if it is a forest.
\end{lemma}
\begin{lemma}(Treewidth of Cactus Graphs, cf. \cite{bodlaender1996arboretum,bodlaender1986classes}) \label{lemma:treewidth-cactus-graphs}\\
Any cactus graph containing a cycle has treewidth $2$.
\end{lemma}

\begin{lemma}(Treewidth of Series-Parallel Graphs, cf. \cite{bodlaender1996arboretum}) \label{lemma:treewidth-series-parallel}\\
Any series-parallel graph has treewidth $\leq 2$.
\end{lemma}

The following lemma states that in a tree decomposition of a graph $\genericGraph$, any complete subgraph of $\genericGraph$ must be contained in a single node set. 
\begin{lemma}(Cliques in Tree Decomposition, cf. \cite{bodlaender1993pathwidth})\label{lemma:cliques-in-tree-decomp}\\
Let $\genericGraphDef$ be an undirected graph and let $\treeDecompDef$ be a tree decomposition for $\genericGraph$. For any clique $H = (V_H, E_H)$ in $\genericGraph$, there is a set $\treeDecompSet \in \treeDecompSets$ with $V_H \subseteq \treeDecompSet$.
\end{lemma}

As a direct consequence of the previous lemma, we state the treewidth of a complete graph.
\begin{corollary}(Tree Decomposition of Complete Graphs) \label{lemma:treewidth-clique}\\
The complete graph $K_n$ with $n$ nodes has treewidth $n-1$.
\end{corollary}

\subsection{The Running Intersection Property} \label{sec:background:rip}

We now discuss the \emph{running intersection property}, which is a property of families of sets. The running intersection property is closely related to tree decompositions. Indeed, as we will discuss in Section~\ref{sec:background:hypergraphs}, Property~\ref{def:tree-decomp:sets-on-path-contain-nodes} in the definition of tree decompositions is equivalent to enforcing the running intersection property for $\treeDecompSets$, the bags of a tree decomposition.

In Section~\ref{sec:hierarchical-bags}, we will discuss an extension of the base algorithm which removes the requirement that the bag label sets of edge bags must be disjoint. Instead, the algorithm will be extended such that some overlap between the label sets can be tolerated. The running intersection property will define the extent to which an overlap between label sets is permissible.

The running intersection property states that the sets in a family of sets can be ordered in such a way that each set's overlap with its predecessors in the ordering is entirely contained in one of the preceding sets.
\begin{definition}(Running Intersection Property, cf. \cite{beeri1983desirability})\label{def:runningIntersectionProperty}\\
Let $\setFamilyDef$ be a family of sets. An ordering $\labelsetOrder = (\setFamilySet_{1}, ... \setFamilySet_{n})$ of $\setFamily$ satisfies the running intersection property, if the intersection of each set $\setFamilySet_\labelSetIndex \in \labelsetOrder$ with its preceding sets according to the ordering is contained in one of the preceding sets $\setFamilySet_{\labelSetIndexTwo}$, i.e. for each index $\labelSetIndex \in \{2, ..., n\}$, there is an index $\labelSetIndexTwo < \labelSetIndex$, such that 
\begin{align}
\setFamilySet_{\labelSetIndex} \cap \left(\bigcup_{\labelSetIndexThree < \labelSetIndex} \setFamilySet_{\labelSetIndexThree} \right) \subseteq \setFamilySet_{\labelSetIndexTwo}.
\end{align}
Such an ordering is called a running intersection ordering.
\end{definition}

We next introduce notation to refer to the predecessor of some label set in a running intersection ordering. We define the predecessor label set as the overlap of the label set in question with all preceding label sets.
\begin{definition}(Predecessor Label Set)\label{def:ripPredecessorLabelSet}\\
Let $\setFamilyDef$ be a family of sets. Given an ordering $\labelsetOrder = (\setFamilySet_{1}, ... \setFamilySet_{n})$, and a set $\setFamilySet_{\labelSetIndex} \in \labelsetOrder$, we define the predecessor label set $\labelsetPredecessors_\labelSetIndex$ as the intersection of $\setFamilySet_{\labelSetIndex}$ with the union of its predecessors in the ordering, i.e. 
\begin{align}
\labelsetPredecessors_\labelSetIndex := \setFamilySet_{\labelSetIndex} \cap \left(\bigcup_{\labelSetIndexTwo < \labelSetIndex} \setFamilySet_{\labelSetIndexTwo} \right).
\end{align}
\end{definition}

We define the \emph{predecessor index function}, which maps the index of a label set to the index of the first label set in the running intersection order which contains the predecessor label set defined above.
\begin{definition}(Predecessor Index Function)\label{def:ripPredecessorIndexFunction}\\
Let $\setFamilyDef$ be a set family satisfying the running intersection property and let $\labelsetOrder = (\setFamilySet_{1}, ... \setFamilySet_{n})$ be a running intersection ordering for $\setFamily$. The predecessor index function \mbox{$\labelsetPredecessorFunction: \setFamily \rightarrow \{1, \ldots n\}$} maps each set to the index of the first set in the running intersection ordering which contains its predecessor label set.
\begin{align}
\labelsetPredecessorFunction (\setFamilySet_{\labelSetIndex}) := \begin{cases}
\labelSetIndex & \text{ if } \labelsetPredecessors_\labelSetIndex = \emptyset \\
\min_\labelSetIndexTwo \{\labelSetIndexTwo \in  \{1, \ldots, \labelSetIndex\}: \labelsetPredecessors_\labelSetIndex \subseteq \setFamilySet_{\labelSetIndexTwo} \} & \text{ otherwise}
\end{cases}
\end{align}
We also use $\labelsetPredecessorFunction(\labelSetIndex)$ as a shorthand notation directly using the label set's index in the ordering to refer to $\labelsetPredecessorFunction(\setFamilySet_{\labelSetIndex})$.
\end{definition}

\subsection{Hypergraphs}\label{sec:background:hypergraphs}

In this section, we introduce the theory of hypergraphs, which generalize the notion of an undirected graph by allowing for (hyper-)edges containing more than two nodes. We introduce this concept as the running intersection property and the concept of tree decompositions are closely related to acyclicity of the associated hypergraph.

Hypergraphs are useful as a direct representation of a set family, which in the context of this thesis will usually be a set of label sets. In particular, given some extraction order $\reqExtractionOrderDef$ containing the node $i \in \reqNodes$, we will be interested in the set containing the edge label assignment for the edges incident in $i$, i.e. in the set \mbox{$\{\labelsetEdge[e] ~|~ e \in \left( \inEdgesExtractionOrder{i} \cup \outEdgesExtractionOrder{i} \right) \}$}.

We begin by defining the notion of a hypergraph.
\begin{definition}(Hypergraph)\label{def:hypergraph} \\
A hypergraph $\hypergraphDef$ is a pair, where $V$ defines a set of nodes, and $\mathcal{E} \subseteq \PowerSet(\mathcal{V})$, where $\PowerSet(\cdot)$ denotes the power set, defines a set of hyperedges defined over the node set. The number of nodes contained in a single edge is unbounded.
\end{definition}
Note that this reduces to the usual definition of an undirected graph, if all edges of a hypergraph contain exactly two nodes.

\begin{remark} \label{remark:label-set-hypergraph}
Hypergraphs are useful as a direct representation of a set of label sets, where we consider each label set to be a hyperedge. When interpreting a set of label sets $\labelsets = \{\labelsetIndexed{1}, ..., \labelsetIndexed{n}\}$ as the edges of a hypergraph, we will assume that $V = \bigcup \labelsets$ and directly identify $\labelsets$ with the hypergraph defined by $(\bigcup \labelsets, \labelsets)$.
\end{remark}

We define the \emph{reduction} of a hypergraph, by removing any edge that is a subset of another edge.
\begin{definition}(Hypergraph Reduction)\label{def:hypergraph-reduction}\\
Given a hypergraph $\hypergraphDef$, we define its reduction as $\hypergraph = (V, \hypergraphEdges')$, where $\hypergraphEdges'$ are the inclusion-wise maximal edges, i.e. \mbox{$\hypergraphEdges' = \{e ~|~ e \in \hypergraphEdges:~ \nexists f \in \hypergraphEdges \setminus \{e\} : ~e \subseteq f \}$}.
\end{definition}

We now define \emph{join trees}, which are closely related to tree decompositions and running intersection orderings. The following definition is adapted from \cite{beeri1983desirability}.
\begin{definition}(Join Tree) \label{def:join-tree}\\
Given a set family $\setFamily$, a join tree $\joinTreeDef[\setFamily]$ is a tree, whose nodes are the sets in $\setFamily$, and whose edges can be labeled such that 
\begin{enumerate}
	\item Each join tree edge $\{X_1, X_2\}$ is labeled with $X_1 \cap X_2$. \label{def:join-tree:edge-labeled-with-intersection}
	\item For each pair of nodes in the join tree $X_1, X_2 \in \setFamily$, each edge along the path connecting $X_1$ and $X_2$ in $\joinTree$ is labeled with $X_1 \cap X_2$. \label{def:join-tree:paths-are-labeled}
\end{enumerate}
We may also identify a join tree with the hypergraph $(\bigcup \setFamily, \setFamily)$ instead of the set family $\setFamily$ (see Remark~\ref{remark:label-set-hypergraph}).
\end{definition}

The following lemma shows how join trees are related to tree decompositions. 
\begin{lemma} (Tree Decompositions and Join Trees) \label{lemma:tree-decomp-and-join-trees} \\
Let $\genericGraphDef$ be a graph with tree decomposition $\treeDecompDef$. Then, the tree $\treeDecompTreeDef$ of the tree decomposition is a join tree of the hypergraph $\hypergraphDef$, where $\hypergraphNodes = \bigcup \treeDecompSets$ and $\hypergraphEdges = \treeDecompSets$.
\end{lemma}
\begin{proof}
First, note that the tree $\treeDecompTreeDef$ is a tree whose nodes are the hyperedges $\hypergraphEdges$. We now have to show that the edges of the tree can be labeled according to Definition~\ref{def:join-tree}.

By Property~\ref{def:join-tree:edge-labeled-with-intersection} of the join tree's edge labeling, each edge $\{i,j\} \in \joinTreeEdges$ is assigned the label set $\treeDecompSet_i \cap \treeDecompSet_j$.

We now show that Property~\ref{def:join-tree:paths-are-labeled} of the join tree's edge labeling then follows from Property~\ref{def:tree-decomp:sets-on-path-contain-nodes} of a tree decomposition (see Definition~\ref{def:tree-decomposition}). Consider  a pair of nodes in the tree decomposition, $i, k \in \treeDecompTreeNodes$. Property~\ref{def:tree-decomp:sets-on-path-contain-nodes} of the tree decomposition states that given any node $j$ on the (uniquely defined) path between $i$ and $k$, $\treeDecompSet_i\cap \treeDecompSet_k \subseteq \treeDecompSet_j$ holds. Consider any edge $e = (u, v)$ on this path. By Property~\ref{def:join-tree:edge-labeled-with-intersection} of a join tree, $e$ is labeled with $\treeDecompSet_u \cap \treeDecompSet_v$. However, by Property~\ref{def:tree-decomp:sets-on-path-contain-nodes} of the tree decomposition, both $\treeDecompSet_i\cap \treeDecompSet_k \subseteq \treeDecompSet_u$ and $\treeDecompSet_i\cap \treeDecompSet_k \subseteq \treeDecompSet_v$ hold, implying that the label set of $e$ contains $\treeDecompSet_i\cap \treeDecompSet_k$.

Therefore, a tree decomposition defines a join tree for the hypergraph $(\bigcup \treeDecompSets, \treeDecompSets)$.
\end{proof}
Some examples of join trees can be found in Figure~\ref{fig:tree-decomp-examples}: Every tree in the right column is a join tree for the underlying set family.

We will further clarify the relation between join trees and tree decompositions towards the end of this section, in Lemma~\ref{lemma:join-trees-and-tree-decomp-primal-graphs}.

In the context of hypergraph theory, the existence of join trees and of running intersection orderings of the edge sets in hypergraphs is tied to \emph{$\alpha$-acyclicity}, one of several definitions of hypergraph acyclicity. Multiple distinct definitions of acyclicity for hypergraphs exist. In increasing order of restrictiveness, $\alpha$-, $\beta$-, $\gamma$- and Berge-acyclicity have been defined, all of which reduce to the usual definition of acyclicity when restricted to regular graphs. A detailed discussion of the hierarchy of hypergraph acyclicity notions can be found in the paper by Fagin, \cite{fagin1983degrees}.

We now give a definition of $\alpha$-acyclicity, in terms of several equivalent conditions. A proof of the equivalence can be found in \cite{beeri1983desirability} or \cite[p. 460ff]{maier1983theory}.
\begin{definition}($\alpha$-Acyclicity, cf. \cite{beeri1983desirability}) \label{def:hypergraph-alpha-acyclicity} \\
A hypergraph $\hypergraphDef$ is $\alpha$-acyclic, if and only if any of the following hold:
\begin{enumerate}
\item The reduction of $\hypergraph$ (cf. Definition~\ref{def:hypergraph-reduction}) is $\alpha$-acyclic. 
\item A join tree for $\hypergraphEdges$ exists.
\item $\hypergraphEdges$ satisfies the running intersection property, i.e. a running intersection ordering of $\hypergraphEdges$ exists.
\end{enumerate}
\end{definition} 
Since $\alpha$-acyclicity is the only kind relevant for this thesis, we will in the following omit the prefix ``$\alpha$'' and simply refer to it as acyclicity or label set acyclicity, when referring to a hypergraph induced by a set of label sets, as described in Remark~\ref{remark:label-set-hypergraph}. 

The inclusion of the second condition draws the connection between join trees and the existence of a running intersection order. To obtain a running intersection order of a label set, our approach will involve the generation of a join tree based on the label sets present. In the following lemmas, we will see how a join tree may be constructed given a hypergraph, and how a running intersection ordering can be constructed from this join tree.

We now examine how a join tree can be constructed for a given acyclic hypergraph. We first introduce the notion of an intersection graph, which is a representation of a family of sets in the form of an undirected graph.
\begin{definition} (Intersection Graph) \label{def:intersection-graph}\\
Let $\setFamilyDef$ be a family of sets. The intersection graph $\intersectionGraphDef[\setFamily]$ is an undirected graph with node set $\intersectionGraphNodes$, where two nodes $\setFamilySet, \setFamilySet' \in \setFamily$ are connected if their sets have a non-empty intersection, i.e. we define the edge set as $\intersectionGraphEdges = \{ \{\setFamilySet, \setFamilySet' \}  ~|~  \setFamilySet, \setFamilySet' \in \intersectionGraphNodes: \setFamilySet \cap \setFamilySet' \neq \emptyset \}$.
\end{definition}

A join tree can be obtained as a maximum spanning tree of the intersection graph for the hypergraph's edge set.
\begin{lemma}(Constructing a Join Tree from an Intersection Graph, cf. \cite{maier1983theory})\label{lemma:join-tree-max-spanning-tree}\\
Given an acyclic hypergraph $\hypergraphDef$, any maximum spanning tree of the intersection graph $\intersectionGraph$ with weight function $w: \intersectionGraphEdges \rightarrow \mathbb{N}_0$, $w(S, S') = |S \cap S' |$ is a join tree.
\end{lemma}
\begin{proof}
See Theorem 13.1 in \cite[p. 459]{maier1983theory}.
\end{proof}

Note that Lemma~\ref{lemma:join-tree-max-spanning-tree} only holds for acyclic hypergraphs, i.e. it assumes the existence of a join tree.

The following lemmas show that given a join tree, a running intersection ordering of the join tree's nodes can be constructed, and vice versa.
\begin{lemma}(Deriving a Running Intersection Ordering from a Join Tree, cf. \cite{maier1983theory,beeri1983desirability}) \label{lemma:rip-order-from-join-tree}\\
Let $\setFamily$ be a set family admitting a join tree $\joinTreeDef[\setFamily]$. Then, $\setFamily$ has the running intersection property and a running intersection ordering can be computed through a preorder traversal of $\joinTree$, i.e. a graph traversal of the join tree visiting each node before its sub-trees.
\end{lemma}
\begin{proof}
For preorder traversals, see Lemma 13.7 in \cite[p. 471]{maier1983theory}. For level-order traversals, see the proof of Theorem 3.4, $(9) \Rightarrow (10)$ in \cite[p. 494]{beeri1983desirability}.
\end{proof}

\begin{lemma}(Deriving a Join Tree from a Running Intersection Ordering) \label{lemma:join-tree-from-rip-order}\\
Given a set family $\setFamilyDef$ that has a running intersection ordering $(S_1, ... S_n)$, a join tree for $\setFamily$ can be constructed in polynomial time.
\end{lemma}
\begin{proof}
Since $\setFamily$ has a running intersection ordering, the corresponding hypergraph (see Remark~\ref{remark:label-set-hypergraph}) is acyclic. By Lemma~\ref{lemma:join-tree-max-spanning-tree}, any maximal spanning tree of the intersection graph (which is computable in polynomial time) is a join tree for $\setFamily$.
\end{proof}

We next show that for each set in a set family satisfying the running intersection property, a running intersection ordering starting at that set exists.
\begin{corollary}(First Element in a Running Intersection Ordering)\label{corollary:rip-arbitrary-start-edge}\\
Let $\setFamilyDef$ be a family of sets that satisfies the running intersection property. Then, for any $S_i \in \setFamily$, there exists a running intersection ordering $\labelsetOrder$ with $\labelsetOrder_1 = S_i$.
\end{corollary}
\begin{proof}
By Lemma~\ref{lemma:join-tree-from-rip-order}, a join tree for $\setFamily$ can be constructed. By Lemma~\ref{lemma:rip-order-from-join-tree}, any pre- or level-order traversal of the join tree is a running intersection ordering. This traversal can be started at $S_i$.
\end{proof}

We will now introduce two common representations of hypergraphs in terms of undirected graphs. We first introduce the \emph{primal graph}.
\begin{definition}(Primal Graph) \label{def:hypergraph-primal-graph} \\
Given a hypergraph $\hypergraphDef$, the primal graph $\hgPrimalGraphDef$ is an undirected graph with the same node set as $\hypergraph$. The edge set is defined as $\hgPrimalGraphEdges := \{ \{i, j\} ~|~ i, j \in \hgPrimalGraphNodes: \exists e \in \hypergraphEdges: \{i, j\} \subseteq e \}$, i.e. two nodes are connected in the primal graph, if both occur in one of the hyperedges in $ \hypergraphEdges$.
\end{definition}

Using this definition, we can establish the connection between the join graph defined for a hypergraph's edge set, and the concept of a tree decomposition.
\begin{lemma} (Join Trees and Tree Decompositions of Primal Graphs) \label{lemma:join-trees-and-tree-decomp-primal-graphs} \\
Let $\hypergraphDef$ with $\hypergraphNodes = \bigcup_{e \in \hypergraphEdges} e$ be an acyclic hypergraph and let $\hgPrimalGraphDef$ denote its primal graph. Any join tree for $\hypergraphEdges$ defines a tree decomposition for $\hgPrimalGraph$.
\end{lemma}
\begin{proof}
We show that the join tree satisfies the three defining properties of a tree decomposition for the primal graph.
First, consider Property~\ref{def:tree-decomp:every-node-is-included} of Definition~\ref{def:tree-decomposition}. Due to the requirement that $\hypergraphNodes = \bigcup_{e \in \hypergraphEdges} e$, every node in $\hypergraphNodes$ occurs in some hyperedge, and is therefore  contained in one of the nodes of the join tree.

Property~\ref{def:tree-decomp:every-edge-is-covered} is satisfied, since any nodes that are connected in the primal graph by an edge both occur in some set, which is a node in the join tree.

Finally, Property~\ref{def:tree-decomp:sets-on-path-contain-nodes} holds: By definition of a join tree, any node that occurs in the label set of an edge must also be contained in the two sets to which the edge is incident. Given three nodes $i, j, k \in \hypergraphNodes$, Property~\ref{def:join-tree:paths-are-labeled} in Definition~\ref{def:join-tree} states that each edge lying on the path from $i$ to $k$ must contain the intersection $i \cap k$ in the edge label set. It then follows that $i\cap k$ must also be contained in the set $k$.
\end{proof}

Finally, we define the representation of a hypergraph as an \emph{incidence graph}, which is also commonly referred to as the dual graph.
\begin{definition}(Incidence Graph) \label{def:hypergraph-incidence-graph} \\
Given a hypergraph $\hypergraphDef$, the incidence graph $\hgIncidenceGraphDef$ is a bipartite graph, where the node set is $\hgIncidenceGraphNodes := \hypergraphNodes \dot{\cup} \hypergraphEdges$, and the edge set is $\hgIncidenceGraphEdges = \{ \{ v, e \} ~|~ v \in \hypergraphNodes, e \in \hypergraphEdges: v \in e  \}$. 

The incidence graph therefore contains a node for each node and edge of the hypergraph, and an edge connecting each hyperedge to each of its incident nodes.
\end{definition}
An example of a hypergraph with its primal and incidence graph representations is shown in Figure~\ref{fig:label-set-graph-example}.

\begin{figure}[htbp]
	\centering
	\includegraphics[width=0.6\textwidth]{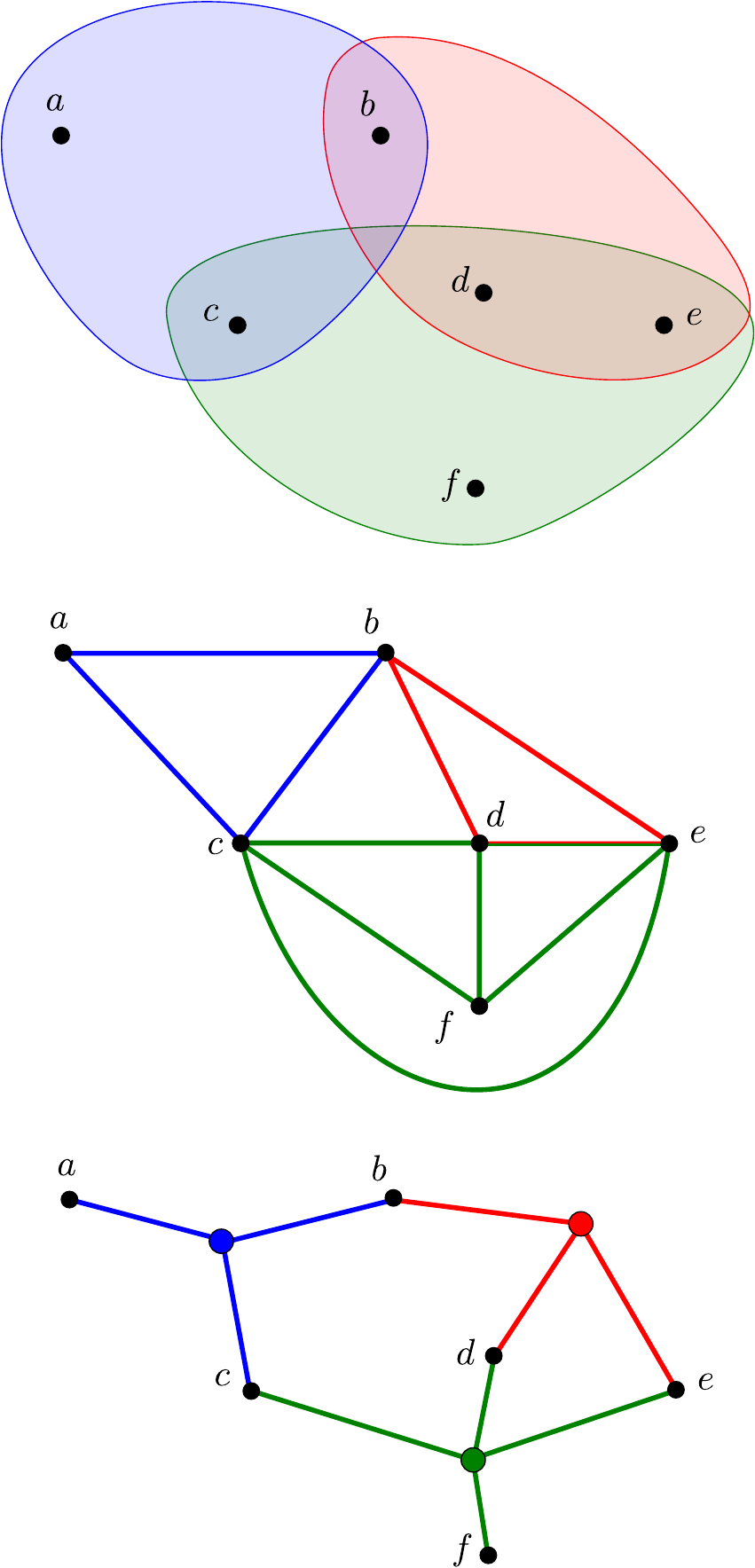}
	\caption{
The primal graph and incidence graph for the hypergraph \newline $\hypergraphDef = (\{a, b, c, d, e, f\} , \{ \{a, b, c\}, \{ b, d, e \}, \{ c, d, e, f\} \})$. \\
Top: The hypergraph $\hypergraph$. \\
Center: The primal graph $\hgPrimalGraphDef$ of $\hypergraph$. Two nodes are connected in $\hgPrimalGraphEdges$ if some hyperedge contains them both. \\
Bottom: The incidence graph $\hgIncidenceGraphDef$ of $\hypergraph$. A node is introduced for each node and hyperedge in $\hypergraphNodes$ and $\hypergraphEdges$, respectively. The three unlabeled nodes represent the hyperedges, and are colored accordingly in the diagram. Each such unlabeled node is connected to the nodes contained in its corresponding hyperedge.
	}
	\label{fig:label-set-graph-example}
\end{figure}

\cleardoublepage
\section{Extending the Base Algorithm with Label Set Acyclicity}\label{sec:hierarchical-bags}

The base algorithm assigns to  each edge in the extraction order a label set to enforce consistent mappings of confluence end nodes through the introduction of LP subformulations. In this labeling, the out-edges of a node may have different label assignments. To induce a consistent amount of flow in each LP subformulation, the base algorithm partitions each node's outgoing edges in edge bags. Edges with overlapping label sets are placed in the same edge bag, such that the resulting partition and the associated bag label sets are disjoint (see Definition \ref{def:edge-bags}). For each such edge bag, a set of $\gamma$-variables are defined, along with constraints enforcing a consistent induction of flow in the LP subformulations of edges contained in the edge bags.

As discussed in Section~\ref{sec:performance-base-approx}, both the number of LP subformulation variables $(\vec{y}, \vec{z}, \vec{a})$, and the number of bag variables $\vec{\gamma}$, increase exponentially with the size of the associated bag label sets. Since the size of the bag label sets is greater than the size of the individual edge label sets, the bag variables are the determining factor for the size of the LP formulation. Further, the example of a half-wheel graph (see Figure~\ref{fig:02:half_wheel_bad_root}) shows that the size of a bag label set may scale linearly with the number of nodes contained in the request for some extraction orders.

The extraction width can also be drastically impacted by relatively small changes to the request topology. Adding a single edge to an extraction order may require two edge bags to merge and may therefore drastically impact the overall size of the LP formulation. Figure~\ref{fig:impact-of-adding-one-edge} shows an example of such an extraction order, where the addition of a single edge causes roughly a doubling of the bag label set size. As discussed in Section~\ref{sec:performance-base-approx}, the LP size of the base formulation grows exponentially with the extraction width. Therefore, adding a single edge to a request can have significant runtime implications for the base algorithm.

In this section, we present a modification of the base algorithm, which relaxes the requirement of each label set in the partition to be disjoint. We will extend the base LP formulation and decomposition algorithm such that they allow for an overlap between the label sets defined by the partition. This means that the edges can be partitioned in groups with smaller associated label sets than in the case of disjoint edge bags. Therefore, the number of $\gamma$-variables, which depends exponentially on the size of the label sets, can be drastically reduced by using a partition with smaller label sets.

\begin{figure}[htbp]
    \centering
    \includegraphics[width=0.75\textwidth]{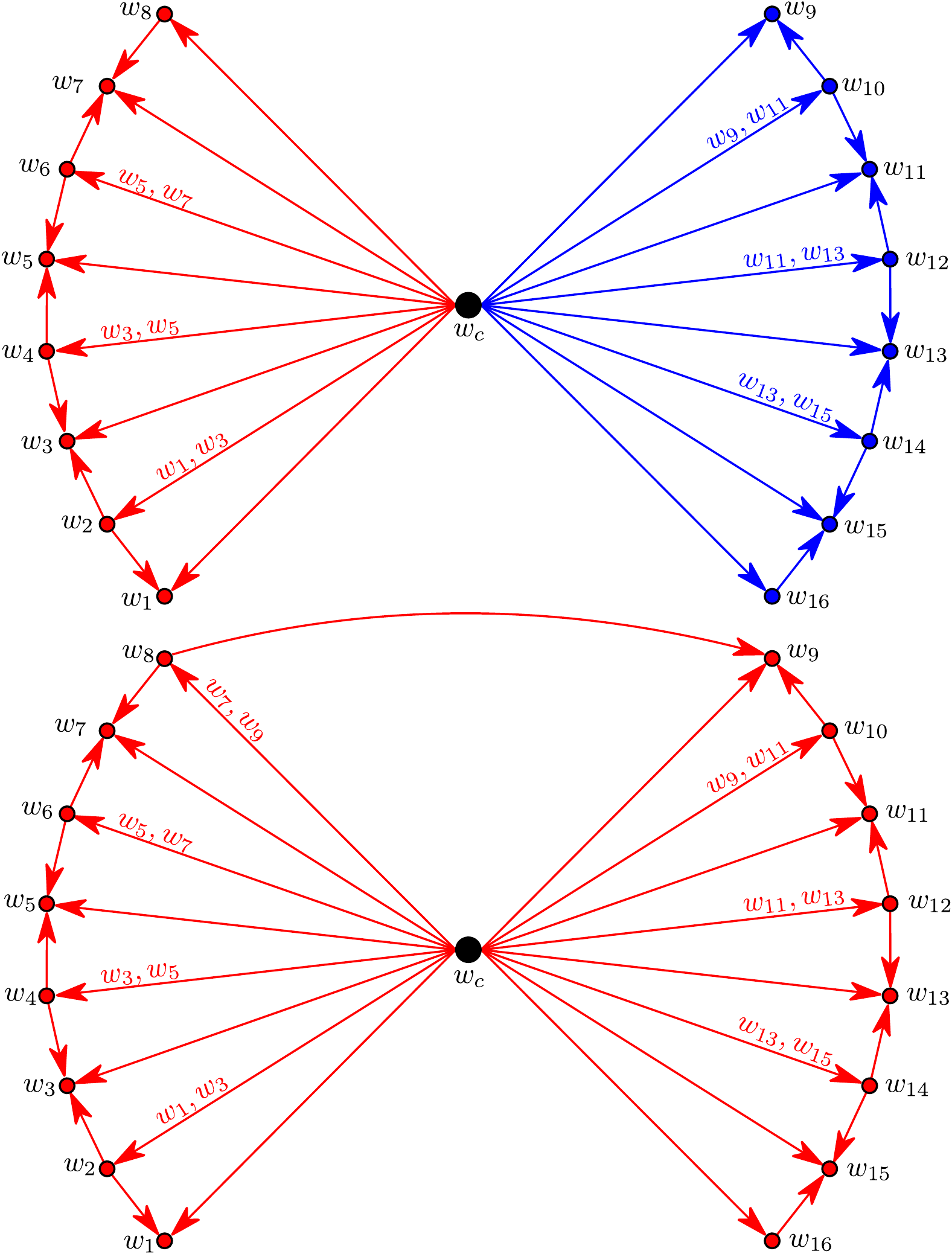}
    \caption{
This example shows the potential impact on the extraction width of adding a single edge to an extraction order. While only some edge label sets are shown, the label sets of all other edges are subsets of the shown label sets. \newline
Top: The original extraction order. The out-edges of the root node $w_c$ are partitioned into two edge bags shown in red and blue, with bag label sets $\bagLabelSets[w_c] = \{ \{w_1, w_3, w_5, w_7\}, \{w_9, w_{11}, w_{13}, w_{15}\} \}$, resulting in an extraction width of $5$. \newline
Bottom: The extraction order is modified by adding the edge $(w_8, w_9)$. Note that the two edge bags are now merged to a single edge bag with bag label set $\bagLabelSets[w_c] = \{ \{w_1, w_3, w_5, w_7, w_9, w_{11}, w_{13}, w_{15}\} \}$, resulting in an extraction width of $9$. This drastically increases the size of the corresponding base LP formulation, and the runtime of the base decomposition algorithm due to Theorem~\ref{thm:base-alg-complexity}.
    }
    \label{fig:impact-of-adding-one-edge}
\end{figure}

We begin by introducing the notion of an extraction label set ordering in Section~\ref{sec:hierarchical:extraction-label-set-ordering}. In Section~\ref{sec:hierarchical:adapted-lp}, we present the extension of the base LP formulation, discussing the variables and constraints. We also show constructively that the extended LP formulation has a solution, whenever a convex combination of mappings exists.
Section~\ref{sec:hierarchical:decomp-edge-label-assignments} introduces decomposable edge label assignments, which generalize the confluence edge label assignment from Definition~\ref{def:confluence-edge-labels}. This more general notion is required for the extension of the algorithm presented in Section~\ref{sec:multiroot}.
In Section~\ref{sec:hierarchical:decomposition-algorithm}, we discuss the decomposition algorithm for the extended LP formulation, including a correctness proof. In Section~\ref{sec:rip-reduction-to-base-algorithm}, we show how the approach using extraction label set orderings is related to the base algorithm, and show that the base algorithm emerges from the extended algorithm given a specific choice of extraction label set ordering.
In Section~\ref{sec:size-rip-lp-formulation}, we discuss the complexity of the extended algorithm, parameterized by the new parameter extraction label width. We further generalize some findings related to the extraction width to the new parameter. Lastly, in Section~\ref{sec:extraction-orders-tree-decomp}, we show the hardness of finding an optimal extraction label set ordering for a specific extraction order.

\subsection{Extraction Label Set Orderings}\label{sec:hierarchical:extraction-label-set-ordering}

The proposed adaptation relies on a partitioning of each node's out-edges in such a way, that the label sets  associated with each  set in the partition satisfy the running intersection property introduced in Section~\ref{sec:background:rip}.

Due to the equality of the label sets on all in-edges of a node (see Lemma~\ref{lemma:incomingLabelsUnique}), it is now convenient to identify this unique label set with the request node directly.
\begin{definition}(Incoming Label Set)\label{def:IncomingLabelset}\\
Given an extraction order $\reqExtractionOrderDef$ and a node $i \in \reqNodes$, the incoming label set of $i$ is given by 
\begin{align}
\labelsetIncoming := \begin{cases}
\emptyset & \text{ if } i = \reqExtractionOrderRoot\\
\labelsetEdge[e] \text{ for some } e \in \inEdgesExtractionOrder{i} & \text{ otherwise (cf. Lemma~\ref{lemma:incomingLabelsUnique})} 
\end{cases}
\end{align}
\end{definition}

We now define the \emph{extraction label set ordering}, which is essential for the adapted algorithm. 
\begin{definition} (Extraction Label Set Ordering) \label{def:extractionLabelSetOrdering}\\
Given an extraction order $\reqExtractionOrderDef$ with an edge label assignment $\labelsets_e = \{\labelsetEdge ~|~ e \in \reqExtractionOrderEdges\}$, an extraction label set ordering $\labelsetOrderSet$ defines for each node $i \in \reqNodes$ an ordering $\labelsetOrder_i$ of some set of label sets $\labelsets \subseteq \PowerSet(\bigcup_e \labelsetEdge)$ with the following properties:
\begin{enumerate}
    \item $\labelsetOrder_i$ is a running intersection ordering. \label{def:extractionLabelSetOrdering:item:running-intersection-ordering}
    \item The first element of $\labelsetOrder_i$ is $\labelsetIncoming$. \label{def:extractionLabelSetOrdering:item:first-nonlocal}
    \item For each $e \in \outEdgesExtractionOrder{i}$, there is exactly one representative label set $\labelsetRepresentative[e] \in \labelsetOrder_i$ with $\labelsetEdge[e] \subseteq \labelsetRepresentative[e]$. If multiple label sets $\labelsetIndexed{\labelSetIndex}  \in \labelsetOrder_i$ satisfy $\labelsetEdge[e] \subseteq \labelsetIndexed{\labelSetIndex}$, we define the representative label set as the first according to the ordering $\labelsetOrder_i$. \label{def:extractionLabelSetOrdering:item:representative-exists}
\end{enumerate}
\defWhitespace
\end{definition}

The extraction label set ordering replaces the notion of edge bags used in the base algorithm. Unlike the disjoint partition enforced by the edge bags, the extraction label set orderings allow for some overlap between different label sets. 

Allowing this overlap between the label sets allows the resulting label sets to be smaller than the corresponding edge bags in terms of the number of contained labels. Indeed, in cases where the edge label sets of each node's in- and out-edges can be arranged in a running intersection order, the size of the representative label set is exactly the size of the edge label set. In this case, the number of bag variables $\gamma$ is no longer the only dominating factor of the LP formulation's size.

Note that the uniqueness requirement in Property~\ref{def:extractionLabelSetOrdering:item:representative-exists} is not strictly required by the adapted algorithm. However, relaxing this requirement for a single representative will lead to the generation of some redundant constraints.

\begin{figure}[htbp]
    \centering
    \includegraphics[width=0.9\textwidth]{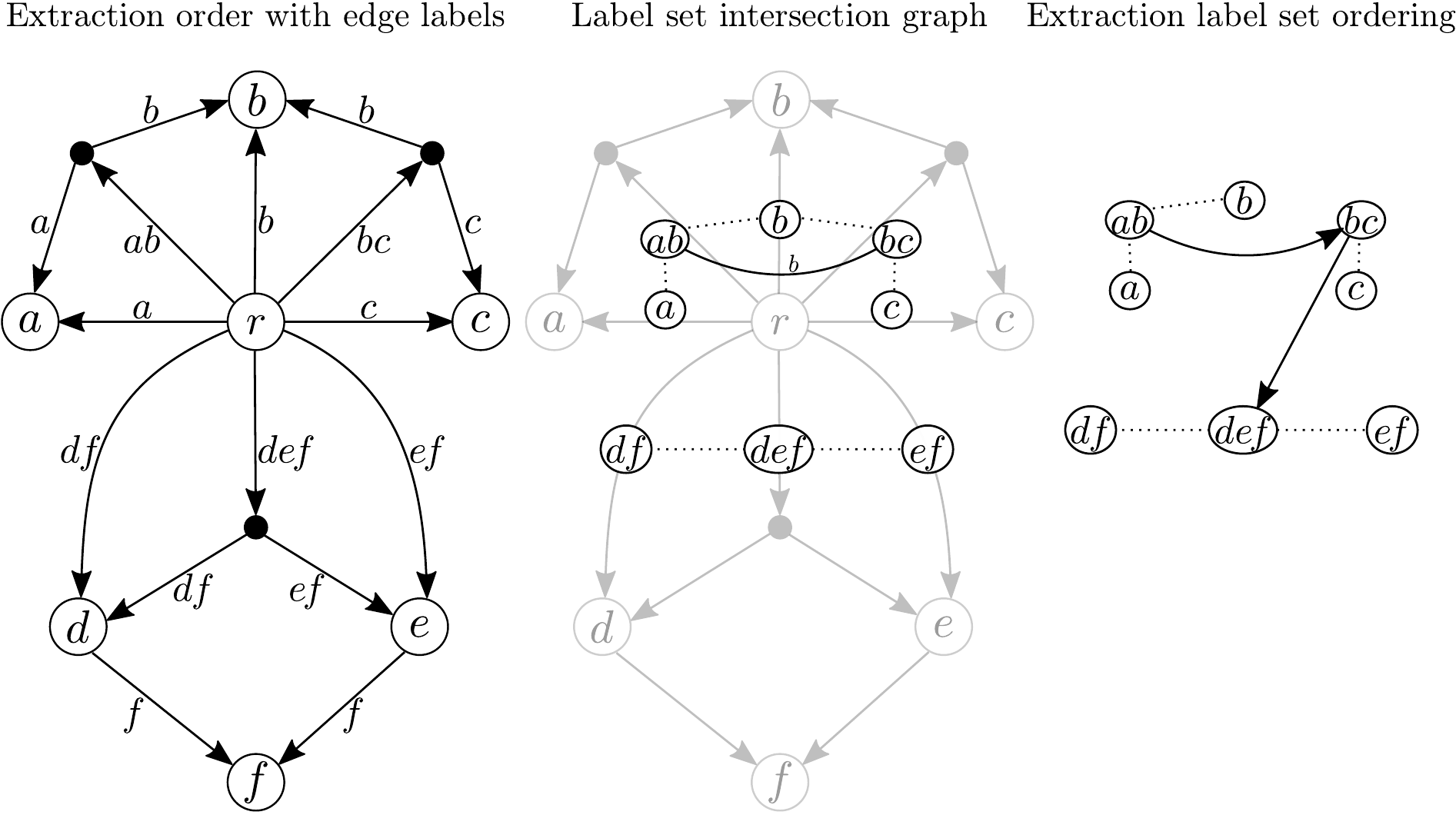}
    \caption{
An illustration of how an extraction label set ordering for an extraction label set ordering is derived for the root node $r$.  \newline
Left: The extraction order used for this example. Each edge is annotated with the confluence edge label assignment. \newline
Center: The label sets for the outgoing edges $\outEdgesExtractionOrder{r}$ are connected according to their intersecting labels. Dotted edges indicate that a label set is a subset of the other (e.g. $\{d, f\} \subset \{d, e, f\}$). \newline
Right: The resulting extraction label set ordering for the root node $r$ is  $\labelsetOrder_r = (\{a, b\}, \{b, c\}, \{d, e, f\})$.
    }
    \label{fig:rip:extraction_order_to_rip_example}
\end{figure}

An illustration of how an extraction label set ordering can be constructed given an extraction order with the confluence edge label assignment is shown in Figure~\ref{fig:rip:extraction_order_to_rip_example} for the root node $r$. In a first step, the intersection graph of the label sets incident in $r$ is constructed. Label sets which are subsets of other label sets are indicated with dotted lines. In this example, the resulting intersection graph already forms a join tree. A running intersection ordering is constructed through a traversal of this join trees.

In Section~\ref{sec:rip-reduction-to-base-algorithm}, we will show that the bag label sets associated with the base algorithm's edge bags satisfy the running intersection property. For now, we note that an extraction label set ordering can be defined for any configuration of edge labels, e.g. with the trivial ordering defined by collecting all label nodes in a single set.

\subsection{The Adapted LP Formulation}\label{sec:hierarchical:adapted-lp}


\begin{figure}[bh!]
{
  \LinesNotNumbered
  \renewcommand{\arraystretch}{1.4}
  
  \removelatexerror
  
  \begin{IPFormulation}{H}
    
    \popline
    
    \SetAlgorithmName{Linear Program}{}{{}}
    
    \newcommand{\spaceIt}{\qquad\quad\quad}
    \newcommand{\miniSpace}{\hspace{1.5pt}}

    \begin{tabular}{FRLQB}
      \multicolumn{5}{r}{\parbox{0.975\textwidth}{~}} \\[-16pt]
      \multicolumn{3}{L}{\textnormal{  \hspace{-10pt}(\ref{eq:classic_mcf:constr_flow}) -  (\ref{eq:classic_mcf:constr_edge_load}) for $\VGe$ on variables } \subLP[(\vec{y},\vec{z},\vec{a})][e,\mappingChar_e]} ~~& \forall e \in \reqEdges, \mappingChar_e \in \MappingSpace[\labelsetEdge]  & \tagIt{eq:lp:novel-rip:subformulations} \vspace{10pt} \\

      \multicolumn{3}{L}{
      \textnormal{(\ref{alg:lp:novel:node-embedding}) - (\ref{alg:lp:novel:node-to-sub-node-mapping}), (\ref{alg:lp:novel:forbidding-nodes-in-sub-lps}) - (\ref{alg:lp:novel:edge-load}) on variables $\lpvars$}}  
      &    & \tagIt{eq:lp:novel-rip:recycled} \vspace{10pt} \\

      \hspace{50pt}
      \subLP[y^u_{i}][\EOEdgeToOriginal(e),\mappingChar_{e}]  & ~=~ & \sum_{\begin{subarray}{c}
          \mappingChar_{\labelsetmappingIndex} \in \MappingSpace[\labelsetRepresentative[e]]: \\
          \restrict[\mappingChar_{\labelsetmappingIndex}][\labelsetEdge] = \mappingChar_{e}
      \end{subarray}}  \hspace{-10pt}\gamma^u_{i,\labelSetIndex,\labelsetmappingIndex}  
      & \hspace{-5pt}
      \begin{array}{r}
        \forall i \in \reqNodes, u \in  \substrateNodesByType[\reqNodeType(i)],  e \in \outEdgesExtractionOrder{i}, \\
        \labelsetRepresentative[e] \in \labelsetOrder_i, \mappingChar_{e} \in \MappingSpace[\labelsetEdge[e]]
      \end{array} & \tagIt{eq:lp:novel-rip:gamma-to-out-edges} \vspace{10pt}\\
      
      \hspace{50pt}
      \subLP[y^u_{i}][\EOEdgeToOriginal(e),\mappingChar_\labelsetmappingIndex] & ~=~ & 
      \gamma^u_{i,1,\labelsetmappingIndex}  
      & \hspace{-5pt}
      \begin{array}{r}
        \forall  i \in \reqNodes, e \in \inEdgesExtractionOrder{i}, u \in  \substrateNodesByType[\reqNodeType(i)], \\ \mappingChar_\labelsetmappingIndex \in \MappingSpace[\labelsetIncoming]
      \end{array}  & \tagIt{eq:lp:novel-rip:in-edges-to-gamma} \vspace{10pt}\\
      
      \hspace{50pt}
      \sum_{\begin{subarray}{c}
          \mappingChar_\labelsetmappingIndex \in \MappingSpace[\labelsetIndexed{\labelSetIndex}]: \\
          \restrict[\mappingChar_\labelsetmappingIndex][\labelsetPredecessors_\labelSetIndex] = \mappingPredecessors
      \end{subarray}} \hspace{-10pt} \gamma^u_{i, \labelSetIndex, \labelsetmappingIndex}\hspace{-5pt}  & ~=~ & 
            \sum_{\begin{subarray}{c}
                \mappingChar_\labelsetmappingIndexTwo \in \MappingSpace[\labelsetIndexed{\labelsetPredecessorFunction(\labelSetIndex)}]: \\
                \restrict[\mappingChar_\labelsetmappingIndexTwo][ \labelsetPredecessors_\labelSetIndex] = \mappingPredecessors
            \end{subarray}} \hspace{-10pt} \gamma^u_{i, \labelsetPredecessorFunction(\labelSetIndex), \labelsetmappingIndexTwo}  
      & \hspace{-5pt}
      \begin{array}{r}
         \forall i \in \reqNodes, u \in  \substrateNodesByType[\reqNodeType(i)], \labelsetOrder_i \in \labelsetOrderSet,\\ 
         \labelsetIndexed{\labelSetIndex} \in \labelsetOrder_i: \labelsetPredecessorFunction(\labelsetIndexed{\labelSetIndex}) \neq \labelSetIndex, \\
         \mappingPredecessors \in \MappingSpace[\labelsetPredecessors_\labelSetIndex]
      \end{array} & \tagIt{eq:lp:novel-rip:gamma-continuity} \vspace{10pt}\\
      
      \multicolumn{4}{C}{
        \begin{array}{c}
          y^u_{i} \in [0,1],~\forall i \in \reqNodes, u \in  \substrateNodesByType[\reqNodeType(i)]; \hspace{18pt} a^{x,y} \geq 0,~\forall (x,y) \in \SR \\
          \gamma^u_{i,\labelSetIndex,\labelsetmappingIndex} \in [0,1],~\forall  i \in \reqNodes, u \in  \substrateNodesByType[\reqNodeType(i)], \labelsetIndexed{\labelSetIndex} \in \labelsetOrder_i, m_\labelsetmappingIndex \in \MappingSpace[\labelsetIndexed{\labelSetIndex}]
        \end{array}
      } & \tagIt{eq:lp:novel-rip:variables}

    \end{tabular}
    \caption{Adapted Linear Program for Extraction Label Set Orderings}
    \label{LP:RunningIntersectionProperty}
  \end{IPFormulation}
}
\end{figure}

We now discuss the adapted LP formulation, which is given in Linear Program~\ref{LP:RunningIntersectionProperty}. We will first discuss the variables and constraints. We then show that the adapted formulation always has a solution, if a convex combination of mappings exists.

\subsubsection{Changes from the Base Formulation}
We first give an overview of the difference between the adapted formulation and the base formulation, Linear Program~\ref{IP:novel-AC}. The variables are defined analogously to the base formulation. Only the $\gamma$-variables have a slightly different meaning, since they are no longer associated with the bag label sets defined by the partitioning of edges into edge bags, but rather with the representative label sets of an extraction label set ordering.

Similar to the base formulation, LP subformulations are introduced in Constraint~(\ref{eq:lp:novel-rip:subformulations}). In Constraint~(\ref{eq:lp:novel-rip:recycled}), several constraints from the base formulation are used without modification:
\begin{itemize}
\item Constraint (\ref{alg:lp:novel:node-embedding}) forces the mapping of each node.
\item Constraint (\ref{alg:lp:novel:node-to-sub-node-mapping}) induces a flow in each LP subformulation.
\item Constraint (\ref{alg:lp:novel:forbidding-nodes-in-sub-lps}) forces the consistent mapping of confluence end nodes in subformulations.
\item Constraints (\ref{alg:lp:novel:node-load}) - (\ref{alg:lp:novel:edge-load}) track the allocations and enforce substrate capacity constraints implicitly through the domain of the allocation variables.
\end{itemize}

Two constraints from the base formulation related to the edge bags are replaced with slightly adapted constraints. These changes reflect the use of label set orderings instead of edge bags, and therefore differ primarily in notation.

Constraint (\ref{alg:lp:novel:gamma-to-outgoing-edges}) is replaced by Constraint (\ref{eq:lp:novel-rip:gamma-to-out-edges}), which connects the node label variables $\gamma^u_{i, a, \labelsetmappingIndex}$ of the request node $i \in \reqNodes$ to the node mapping variables $\subLP[y^u_i][e, \restrict[m_\labelsetmappingIndex][\labelsetEdge]]$ in the subformulations associated with the out-edges of $i$. Instead of the bag label set $\bagLabelSet[i][b]$, the representative label set $\labelsetRepresentative[e]$ defined by the label set ordering $\labelsetOrder_i$ is used.

Similarly, Constraint (\ref{alg:lp:novel:incoming-edges-to-gamma-variables}) is replaced by Constraint (\ref{eq:lp:novel-rip:in-edges-to-gamma}). This constraint propagates the value of the node mapping variables from subformulations associated with the node's in-edges to the node's $\gamma$-variables. As the incoming label set is the first in the sequence $\labelsetOrder_{i}$, the corresponding $\gamma$-variables are directly accessed with the index $1$. 

The label set continuity Constraint (\ref{eq:lp:novel-rip:gamma-continuity}) is newly introduced and has no counterpart in the base formulation. For each label set $\labelsetIndexed{a}$, it enforces consistency between the $\gamma$-variables associated with $\labelsetIndexed{a}$ and the $\gamma$-variables of the predecessor label set. Intuitively, the label set continuity constraint encodes the running intersection ordering defined by the label set ordering $\labelsetOrderSet$ in the LP formulation.

As a technicality, note that for each in-edge $e \in \inEdgesExtractionOrder{i}$, there is a one-to-one correspondence between the variables $\subLP[y^u_i][\EOEdgeToOriginal(e), \restrict[m_\labelsetmappingIndex][\labelsetEdge]]$ and $\gamma^u_{i, 1, \labelsetmappingIndex}$, because their associated label sets are equal. That is, $\labelsetEdge[e] = \labelsetIncoming[i]$ and therefore, $\MappingSpace[\labelsetEdge[e]] = \MappingSpace[\labelsetIncoming[i]]$. In practice, the variables $\gamma^u_{i, 1, \labelsetmappingIndex}$ could therefore be eliminated by applying Constraint (\ref{eq:lp:novel-rip:in-edges-to-gamma}). However, including them simplifies the notation and discussion of the adapted algorithm.

\subsubsection{Existence of Solutions}

We will now verify that Linear Program~\ref{LP:RunningIntersectionProperty} always has a solution, given that a convex combination of mappings exists. The proof of the following lemma constructively shows how to obtain such a solution given a convex combination of mappings. The proof works regardless of how edge labels are assigned.

\begin{lemma} (Solvability of Linear Program~\ref{LP:RunningIntersectionProperty}) \label{lemma:rip-lp-solvability-with-arbitrary-edge-labels} \\
Let $\VNEPInstance$ be a VNEP instance and let $\PotEmbeddings := \{(\probSubscript[1],\mappingIteration[1]), ... (\probSubscript[n],\mappingIteration[n])\}$ be a convex combination of valid mappings for $\reqTopology$, such that $\sum_k \probSubscript[k] = 1$ holds. Let $\reqExtractionOrderDef$ be an extraction order for $\reqTopology$ and let $\labelsetsEdgesTilde$ be an arbitrary edge label assignment for $\reqExtractionOrder$. Then there exists a solution $\lpvarstilde$ for Linear Program~\ref{LP:RunningIntersectionProperty} using label assignment $\labelsetsEdgesTilde$.
\end{lemma}
\begin{proof}
We prove the lemma constructively, by deriving a variable assignment from $\PotEmbeddings$, and showing that all constraints of Linear Program~\ref{LP:RunningIntersectionProperty} are satisfied by this assignment. Consider the following variable assignment for the label set variables $\gammatilde$, and the subformulation variables $\ytilde$ and $\ztilde$:
\begin{align}
\gammatilde^u_{i, \labelSetIndex, \labelsetmappingIndex} & = \sum_{
\begin{subarray}{c}
    (\probSubscript, \mappingIteration[k]) \in \PotEmbeddings \\
    \restrict[\mappingIteration[k]][\labelsetIndexed{\labelSetIndex}] = \mappingChar_{\labelsetmappingIndex}
\end{subarray}} \probSubscript & 
\begin{array}{r}
\forall e \in \reqExtractionOrderEdges, i \in e, u \in \substrateNodes, \\
\labelsetIndexed{\labelSetIndex} \in \labelsetOrder_i, \mappingChar_{\labelsetmappingIndex} \in \MappingSpace[\labelsetIndexed{\labelSetIndex}] 
\end{array} \label{eq:arbitrary-edge-labels:set-gamma-tilde}\\
\subLP[\ytilde^u_i][\EOEdgeToOriginal(e),\mappingEdgeTilde] & = \sum_{
\begin{subarray}{c}
    (\probSubscript, \mappingIteration[k]) \in \PotEmbeddings \\
    \restrict[\mappingIteration[k]][\labelsetEdgeTilde] = \mappingEdgeTilde
\end{subarray}} \probSubscript & 
\forall e \in \reqExtractionOrderEdges, i \in e, u \in \substrateNodes, \mappingEdgeTilde \in \MappingSpace[\labelsetEdgeTilde[e]] \label{eq:arbitrary-edge-labels:set-y-tilde}\\
\subLP[\ztilde^{u,v}_{e}][\EOEdgeToOriginal(e),\mappingEdgeTilde] & = \sum_{
\begin{subarray}{c}
    (\probSubscript, \mappingIteration[k]) \in \PotEmbeddings \\
    \restrict[\mappingIteration[k]][\labelsetEdgeTilde] = \mappingEdgeTilde \\
    (u, v) \in \mappingIteration[k](e)
\end{subarray}} \probSubscript &
\forall e \in \reqExtractionOrderEdges, (u, v) \in \substrateEdges, \mappingEdgeTilde \in \MappingSpace[\labelsetEdgeTilde[e]] \label{eq:arbitrary-edge-labels:set-z-tilde}
\end{align}

From these variables, the remaining variables may be assigned as follows: The allocation variables can be assigned according to Constraint (\ref{eq:classic_mcf:constr_node_load}) and (\ref{eq:classic_mcf:constr_edge_load}) for the subformulation allocation variables, and (\ref{alg:lp:novel:node-load}) and (\ref{alg:lp:novel:edge-load}) for the aggregated allocation variables. The aggregated node mapping variables $\ytilde$ which are not directly tied to a LP subformulation can be assigned according to Constraint (\ref{alg:lp:novel:node-to-sub-node-mapping}). All of these constraints are therefore trivially satisfied by the resulting variable assignment. 

We next show that all other constraints of the LP formulation are also satisfied. The argument will usually rely on the assumption that $\PotEmbeddings$ is a convex combination of \emph{valid} mappings. By the definition of a valid mapping (see Definition~\ref{def:valid-mapping}), this implies that each mapping $\mappingIteration[k]$ in $\PotEmbeddings$ maps each request node to a single substrate node. Therefore, for any set of request nodes (and in particular any label set) $L \subseteq \reqNodes$, the mapping $\mappingIteration[k]$ agrees with exactly one element of the mapping space $m\in \MappingSpace[L]$. 

We begin with the flow conservation/induction Constraint (\ref{eq:classic_mcf:constr_flow}). According to (\ref{eq:arbitrary-edge-labels:set-y-tilde}), given any edge $(i, j) \in \reqExtractionOrderEdges$, the subformulation node variables belonging to the nodes $i$ and $j$, namely $\subLP[y^{\mappingIteration(i)}_i][\EOEdgeToOriginal(i,j), \restrict[\mappingIteration][\labelsetEdge[(i,j)]]]$ and $\subLP[y^{\mappingIteration(j)}_j][\EOEdgeToOriginal(i,j), \restrict[\mappingIteration][\labelsetEdge[(i,j)]]]$, are increased by $\probSubscript$. Additionally, by the validity of $\mappingIteration[k]$, all edge mapping variables $\ztilde$ which are incremented by a flow $\probSubscript$ in (\ref{eq:arbitrary-edge-labels:set-z-tilde}) lie on a path starting at $\mappingIteration[k](i)$ and ending in $\mappingIteration[k](j)$. Therefore, the flow preservation and induction constraints are satisfied.

We next consider Constraint (\ref{alg:lp:novel:node-embedding}), which forces the embedding of each request node. Since each mapping $\mappingIteration[k]$ with $(\probSubscript, \mappingIteration[k]) \in \PotEmbeddings$ is by assumption valid, it agrees with exactly one of the mappings $m_e \in \MappingSpace[\labelsetEdge[e]]$ in the mapping space of the edge label set. Therefore, for each request node $i$, each $\probSubscript$ is added to the node mapping subvariables corresponding to a single edge label mapping in  (\ref{eq:arbitrary-edge-labels:set-y-tilde}). Each subformulation variable $\subLP[\ytilde^\cdot_i][\EOEdgeToOriginal(e),\mappingEdgeTilde]$ contributes to exactly one term in the sum on the right side of Constraint (\ref{alg:lp:novel:node-embedding}). Since $\sum_k \probSubscript = 1$, Constraint (\ref{alg:lp:novel:node-embedding}) holds.

We next consider Constraint (\ref{alg:lp:novel:forbidding-nodes-in-sub-lps}). The validity of each mapping $\mappingIteration[k]$ requires each node to be mapped to a single substrate node, and each request edge to be mapped to a substrate path connecting the appropriate nodes. Therefore, given any label node $l$ and any in-edge $e \in \inEdgesExtractionOrder{l}$ of $l$, only subformulation variables $\subLP[\ytilde^{\mappingIteration[k](l)}_l][\EOEdgeToOriginal(e),\mappingEdgeTilde]$ with $\mappingEdgeTilde(l) = \mappingIteration[k](l)$ are incremented in (\ref{eq:arbitrary-edge-labels:set-y-tilde}) and all other subformulation node variables remain zero.

We next consider Constraint (\ref{eq:lp:novel-rip:gamma-to-out-edges}). Let $e = (i,j) \in \reqExtractionOrderEdges$, and $\labelsetIndexed{\labelSetIndex} \in \labelsetOrder_i$ with $\labelsetEdge[e] \subseteq \labelsetIndexed{\labelSetIndex}$, and let $m_\labelsetmappingIndex \in \MappingSpace[\labelsetIndexed{\labelSetIndex}]$. Consider any mapping $(\probSubscript, \mappingIteration[k]) \in \PotEmbeddings$ such that $\restrict[\mappingIteration[k]][\labelsetIndexed{\labelSetIndex}] = m_\labelsetmappingIndex$. Since $\labelsetEdge[e] \subseteq \labelsetIndexed{\labelSetIndex}$, there is exactly one $m_e \in \MappingSpace[\labelsetEdge[e]]$ such that $\restrict[ m_\labelsetmappingIndex][\labelsetEdge[e]] =  m_e$. Whenever $\probSubscript$ is added to $\gammatilde^u_{i, \labelSetIndex, \labelsetmappingIndex}$  in  (\ref{eq:arbitrary-edge-labels:set-gamma-tilde}), it is also added to $\subLP[\ytilde^u_i][\EOEdgeToOriginal(e),m_e]$ in (\ref{eq:arbitrary-edge-labels:set-y-tilde}). Therefore, Constraint (\ref{eq:lp:novel-rip:gamma-to-out-edges}) holds.

Constraint (\ref{eq:lp:novel-rip:in-edges-to-gamma}) holds because each mapping $\mappingIteration[k]$ contributes to exactly one mapping of each node's incoming label set. Since the incoming label set is equal to the edge label set indexing the subformulation node variable, the same amount of flow is added to the subformulation node mapping variable and the $\gamma$-variable of the incoming label set in (\ref{eq:arbitrary-edge-labels:set-gamma-tilde}) and (\ref{eq:arbitrary-edge-labels:set-y-tilde}).

Finally, we consider Constraint (\ref{eq:lp:novel-rip:gamma-continuity}). Let $\labelsetIndexed{\labelSetIndex}$ be some label set with predecessor label set $\labelsetIndexed{\labelsetPredecessorFunction(\labelSetIndex)}$. We only consider label sets where $\labelsetIndexed{\labelSetIndex} \neq \labelsetIndexed{\labelsetPredecessorFunction(\labelSetIndex)}$, since Constraint (\ref{eq:lp:novel-rip:gamma-continuity}) is only generated in this case. For each mapping $(\probSubscript, \mappingIteration) \in \PotEmbeddings$ it holds that $\restrict[\mappingIteration][\labelsetPredecessors_\labelSetIndex] \in \MappingSpace[\labelsetPredecessors_\labelSetIndex]$, i.e. each $\mappingIteration$ is compatible with one predecessor mapping $\mappingPredecessors$. Due to the validity of each $\mappingIteration[k]$, it also holds that each $\mappingIteration[k]$ is compatible with exactly one $m_\labelSetIndex \in \MappingSpace[\labelsetIndexed{\labelSetIndex}]$ and exactly one $m_{\labelsetPredecessorFunction(\labelSetIndex)} \in \MappingSpace[\labelsetIndexed{\labelsetPredecessorFunction(\labelSetIndex)}]$. Since $m_\labelSetIndex$ and $m_{\labelsetPredecessorFunction(\labelSetIndex)}$ occur once in the sums on the left and right side of Constraint (\ref{eq:lp:novel-rip:gamma-continuity}), respectively, the same mapping value $\probSubscript$ is added on both sides and the constraint holds.
\end{proof}

\subsection{Decomposable Edge Label Assignments}\label{sec:hierarchical:decomp-edge-label-assignments}

In Section~\ref{sec:multiroot}, we will require extensions of the confluence edge label assignment as a mechanism to enforce non-local consistency of the request mapping. In this section, the required properties of these extended edge label assignments are defined. 

The decomposition algorithm and its correctness proof will therefore be formulated in terms of this more general notion of a \emph{decomposable edge label assignment}. We will verify that the confluence edge label assignment introduced in Definition~\ref{def:confluence-edge-labels} has these properties. Thus, the edge label assignment can still be thought of as the confluence label assignment for the remainder of Section~\ref{sec:hierarchical-bags}.

We first define the \emph{label-induced subgraph} which, given some label node $k$, is the subgraph defined by the edges labeled with $k$.
\begin{definition} (Label-Induced Subgraph) \\
Given an extraction order $\reqExtractionOrderDef$ with some edge label assignment $\labelsetsEdges$. Let $k \in \bigcup_{e \in \reqExtractionOrderEdges} \labelsetEdge[e]$ be a node that occurs as a label in some label set. We define the label-induced subgraph for $k$, $\reqEOLabelSubgraph[k]$, as the smallest subgraph containing all edges that are labeled with $k$:
\begin{align}
\reqEOLabelSubgraphDef[k],
\end{align}
where $\reqEOLabelSubgraphEdges[k] := \{e \in \reqExtractionOrderEdges ~|~ k \in \labelsetEdge[e]\}$ and $\reqEOLabelSubgraphNodes[k] := \bigcup_{(i, j) \in \reqEOLabelSubgraphEdges[k]} \{i, j\}$. 
\end{definition}

We now define the properties of decomposable edge label assignments.
\begin{definition} (Decomposable Edge Label Assignment) \label{def:decomposable-edge-labels} \\
Let $\reqTopologyDef$ be some request with extraction order $\reqExtractionOrderDef$. Let $\labelsetsEdgesExtractionOrder$ be the confluence edge label assignment derived according to Definition~\ref{def:confluence-edge-labels}, and let $\labelsetsEdgesTilde$ be a different edge label assignment. 

We call $\labelsetsEdgesTilde$ a decomposable edge label assignment, if it satisfies the following properties, where  $k$ is a node which occurs in some label set $\labelsetEdgeTilde[e] \in \labelsetsEdgesTilde$:  
\begin{enumerate}
\item $\labelsetsEdgesTilde$ extends the confluence edge labels, i.e. for each edge $e \in \reqExtractionOrderEdges$, $\labelsetEdge[e] \subseteq \labelsetEdgeTilde[e]$
\label{def:decomposability-extended-edge-labels:confluence-label-extension}
\item The label-induced subgraph $\reqEOLabelSubgraph[k]$ for $k$ is connected and rooted in some node $\reqEOLabelSubgraphRoot[k] \in \reqNodes(k)$.\label{def:decomposability-extended-edge-labels:rooted-labelinducedgraph}
\item Each label node is contained in its own label-induced subgraph, i.e. $k \in \reqEOLabelSubgraphNodes[k]$. 
\label{def:decomposability-extended-edge-labels:labelinducedgraph-contains-label-node}
\item Lemma~\ref{lemma:incomingLabelsUnique} holds, i.e. each of a node's incoming edges have the same label set. 
\label{def:decomposability-extended-edge-labels:in-edges-same-labels}
\item For any $e \in \outEdgesExtractionOrder{k}$, $k \not \in \labelsetEdgeTilde[e]$. 
\label{def:decomposability-extended-edge-labels:labels-end-in-label-node}
\end{enumerate}
\defWhitespace
\end{definition}
Several examples of both decomposable and non-decomposable edge label assignments for a specific extraction order are shown in Figure~\ref{fig:examples-extended-label-assignment}.

\begin{figure}[htbp]
    \centering
    \includegraphics[width=0.6\textwidth]{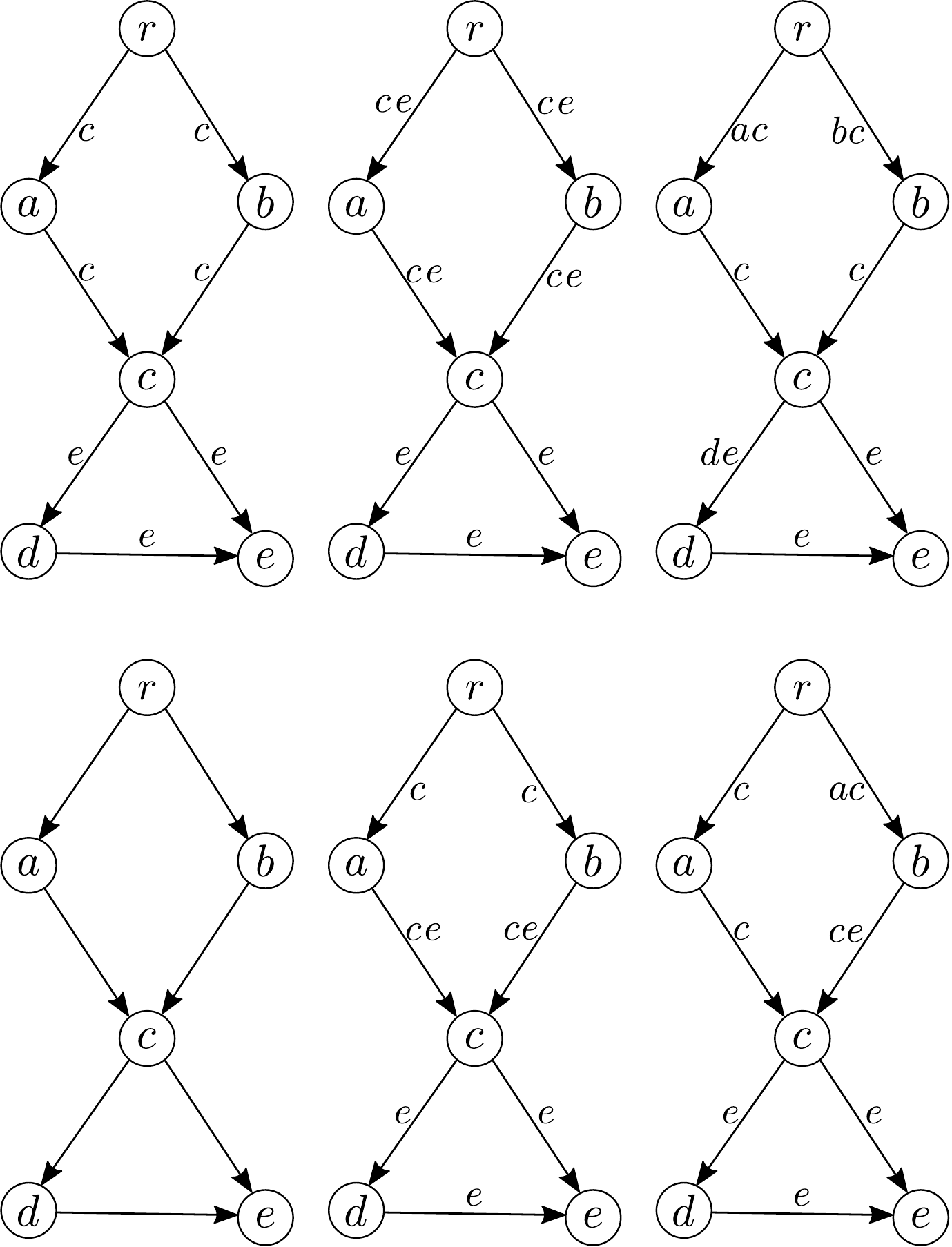}
    \caption{
        An extraction order $\reqExtractionOrder$, with examples of decomposable and non-decomposable edge label assignments according to Definition~\ref{def:decomposable-edge-labels}.
        \newline
        Top: Decomposable edge label assignments.
        \newline 
        Left: The confluence label assignment.
        Center: Edge label $e$ is extended towards the root node $r$. 
        Right: Edge labels for $a$, $b$ and $d$ are newly introduced.
        \newline
        Bottom: Non-decomposable edge label assignments.
        \newline 
        Left: Does not extend the confluence edge label assignment $\labelsetsEdgesExtractionOrder$.
        Center: The label-induced subgraph for $e$ is not rooted.
        Right: Violates Lemma~\ref{lemma:incomingLabelsUnique} for node $c$ and Property~\ref{def:decomposability-extended-edge-labels:labelinducedgraph-contains-label-node} for node $a$, since $a \not\in \reqEOLabelSubgraphNodes[a]$.
    }
    \label{fig:examples-extended-label-assignment}
\end{figure}

From the definition, we derive the following corollary, stating that if a node occurs as a label in some label set, all of its in-edges are labeled with it.
\begin{corollary} (In-Edges of Label Nodes) \label{corollary:decomp-labels:in-edges-label-nodes} \\
Let  $\reqExtractionOrder$ be an extraction order and let $\labelsetsEdges$ be a decomposable edge label assignment. Let $k \in \bigcup \labelsetsEdges$ be a node that occurs as a label in $\labelsetsEdges$. Then, each in-edge of $k$ is labeled with $k$.
\end{corollary}
\begin{proof}
By Property~\ref{def:decomposability-extended-edge-labels:labelinducedgraph-contains-label-node} of Definition~\ref{def:decomposable-edge-labels}, $k$ is contained in its label-induced subgraph and by Property~\ref{def:decomposability-extended-edge-labels:labels-end-in-label-node}, none of the out-edges of $k$ are labeled with $k$. It follows that $k$ must be contained in one of its in-edges' label sets, which implies by Property~\ref{def:decomposability-extended-edge-labels:in-edges-same-labels} that each in-edge of $k$ is labeled with $k$.
\end{proof}

In the next lemma, we show that once a label node has been introduced in a decomposable label assignment, any path leading to the label node is labeled with it. That is, once a label node is first encountered in some edge label set during a graph traversal, any edge leading towards the label node must be labeled with it.

\begin{lemma} (Label Path Continuity) \label{lemma:decomplabels:all-paths-are-labeled}\\
Given an extraction order $\reqExtractionOrder$ with some decomposable edge label assignment $\labelsetsEdges$, and some node $k \in \bigcup_e \labelsetEdge[e]$ that occurs as a label in $\labelsetsEdges$. Let $i \in \reqEOLabelSubgraphNodes[k]$ with $i \neq k$ be a node in the label-induced subgraph for $k$, i.e. some $e' \in \outEdgesExtractionOrder{i} \cup \inEdgesExtractionOrder{i}$ with $k \in \labelsetEdge[e']$ exists. Then, any path $\path[i][k] \subseteq \reqExtractionOrderEdges$ from $i$ to $k$ is labeled with $k$, i.e. for each $e \in \path[i][k]$ it holds that $k \in \labelsetEdge[e]$.
\end{lemma}
\begin{proof}
We assume that $k$ is reachable from $i$. Otherwise, the claim is trivially true, since no paths from $i$ to $k$ would exist. 
Let $\path[i][k]$ be some path from $i$ to $k$. Assume for the sake of contradiction, that some edge in $\path[i][k]$ exists that is not labeled with $k$, and let $(x, y)$ be the last such edge in the path. Then, by Property~\ref{def:decomposability-extended-edge-labels:in-edges-same-labels} of the decomposable edge label assignment, no in-edge of $y$ is labeled with $k$.

Let $(j,k) \in \inEdgesExtractionOrder{k} \cap \path[i][k]$ be the in-edge of $k$ which lies on the path $\path[i][k]$. By Corollary~\ref{corollary:decomp-labels:in-edges-label-nodes}, $k$ must be in the label set of $(j,k)$. Therefore, at least one edge in $\path[i][k]$ is labeled with $k$.
Since the node $y$ has no in-edge labeled with $k$, but at least one out-edge that is labeled with $k$, and since $\reqEOLabelSubgraph[k]$, the label-induced subgraph for $k$, is rooted by Property~\ref{def:decomposability-extended-edge-labels:rooted-labelinducedgraph}, $y$ must be this root node. However, $y$ can be reached with a path from $i$, and some edge containing $i$ is labeled with $k$. Due to the acyclicity of the extraction order, no path from $y$ to $i$ can exist.

It follows that the label-induced subgraph for $k$ cannot be rooted, violating Property~\ref{def:decomposability-extended-edge-labels:rooted-labelinducedgraph} of a decomposable label assignment. This contradicts the assumption that $\labelsetsEdges$ is a decomposable label assignment. We therefore conclude that each edge in the path $\path[i][k]$ is labeled with $k$.
\end{proof}

Finally, we verify that the confluence edge label assignment (see Definition~\ref{def:confluence-edge-labels}) is indeed a decomposable edge label assignment.
\begin{lemma} (Confluence Edge Labels are Decomposable) \label{lemma:confluence-labels-are-decomp}\\ 
Let $\reqExtractionOrderDef$ be an extraction order. The confluence edge label assignment $\labelsetsEdges$ (see Definition~\ref{def:confluence-edge-labels}) for $\reqExtractionOrder$ is a decomposable edge label assignment (see Definition~\ref{def:decomposable-edge-labels}).
\end{lemma}
\begin{proof}
We prove that each of the defining properties of decomposable label assignments are satisfied by the confluence edge label assignment.

Since $\labelsetsEdges$ is the confluence edge label assignment, Property~\ref{def:decomposability-extended-edge-labels:confluence-label-extension} is satisfied with equality.

Property~\ref{def:decomposability-extended-edge-labels:rooted-labelinducedgraph} holds, since by definition of the confluence label assignment, each edge that is labeled with $k$ lies on some confluence ending in $k$. Consider the start node of each such confluence. Assume for contradiction, that the label-induced subgraph for $k$ in the confluence label assignment is not rooted. Then there are (at least) two confluences ending in $k$ with different start nodes $i$ and $j$, such that no in-edge of $i$ and $j$ is labeled with $k$. However, since the extraction order is rooted, a path from $\reqExtractionOrderRoot$ to both $i$ and $j$ exists. Therefore, $i$ and $j$ share some common ancestor node, $a$. There are at least two node-disjoint paths from $a$ to $k$, namely one path via $i$ and one path via $j$. These two paths from $a$ to $k$ form a confluence. Since the in-edges of $i$ and $j$ are not labeled with $k$, this confluence is only partially labeled with $k$. This contradicts the definition of the confluence label assignment.

Property~\ref{def:decomposability-extended-edge-labels:labelinducedgraph-contains-label-node} holds, since the end node of any confluence must be reached by the edges that lie on the confluence. Since all these edges are labeled with the end node, it follows that any label node is contained in its label-induced subgraph.

The claim that Property~\ref{def:decomposability-extended-edge-labels:in-edges-same-labels} holds for confluence edge labels is the statement of Lemma~\ref{lemma:incomingLabelsUnique}. Accordingly, a proof is given by Rost and Schmid in \cite{rostSchmidFPTApproximations}.

Finally, Property~\ref{def:decomposability-extended-edge-labels:labels-end-in-label-node} holds: If any out-edge $e \in \outEdgesExtractionOrder{k}$ were labeled with $k$ by the confluence label assignment, it would lie on a confluence ending in $k$. Any branch of this confluence would form a cycle, contradicting the acyclicity of the extraction order $\reqExtractionOrder$.
\end{proof}


\begin{figure}[tbh!]

\removelatexerror

\begin{algorithm*}[H]

\SetKwInOut{Input}{Input}\SetKwInOut{Output}{Output}
\SetKwFunction{ProcessPath}{ProcessPath}{}{}
\SetKwFunction{reverse}{reverse}{}{}
\SetKwFunction{LP}{LP}
\SetKwFunction{LP}{LP}

\newcommand{\SET}{\textbf{set~}}
\newcommand{\ADD}{\textbf{add~}}
\newcommand{\EACH}{\textbf{each~}}
\newcommand{\DEFINE}{\textbf{define~}}
\newcommand{\AND}{\textbf{and~}}
\newcommand{\LET}{\textbf{let~}}
\newcommand{\WITH}{\textbf{with~}}
\newcommand{\COMPUTE}{\textbf{compute~}}
\newcommand{\FIND}{\textbf{find~}}
\newcommand{\CHOOSE}{\textbf{choose~}}
\newcommand{\DECOMPOSE}{\textbf{decompose~}}
\newcommand{\FORALL}{\textbf{for all~}}
\newcommand{\OBTAIN}{\textbf{obtain~}}
\newcommand{\WITHPROBABILITY}{\textbf{with probability~}}

\caption{Decomposition algorithm for solutions of Linear Program \ref{LP:RunningIntersectionProperty}}
\label{alg:decompositionAlgorithm-RIP}

\Input{VNEP-instance (\substrateTopology, \reqTopology), request extraction order $\reqExtractionOrderDef$, valid label set ordering $\labelsetOrderSet$, solution $\lpvars$ to Linear Program ~\ref{LP:RunningIntersectionProperty}}
\Output{Convex combination~$\PotEmbeddings = \{\decomp = (\prob,\mappingRequestIteration)\}_k$ of valid mappings}
  \SET $\PotEmbeddings  \gets \emptyset$ \AND $k \gets 1$\\
  \SET $x \gets 1$ \label{algline:rip-decomp:setRemainingFlowVar}\\
  \While{$x > 0$ \label{algline:rip-decomp:while-flow}}
  { \label{algline:rip-decomp:begin-while-flow}
    
    \SET $\mappingRequestIteration = (\mappingNodesIteration,\mappingEdgesIteration)~\gets (\emptyset,\emptyset)$ \label{algline:rip-decomp:init-maps}\\
    
    \SET $\Queue \gets \{\reqExtractionOrderRoot \}$ \\
    
    \CHOOSE $u \in \substrateNodes$ \WITH $y^u_{\reqExtractionOrderRoot} > 0$ \AND \SET $\mappingNodesIteration(\reqExtractionOrderRoot)~\gets u$ \label{algline:rip-decomp:nodeMapRoot}\\
    
    \While{$|\Queue| > 0$}{  \label{algline:rip-decomp:begin-while-q}
      \CHOOSE $i \in \Queue$ \AND \SET$\Queue \gets \Queue \setminus \{i\}$ \label{algline:rip-decomp:retrieve-i-from-q}\\
      \ForEach{$\labelsetIndexed{\labelSetIndex} \in \labelsetOrder_i$  \label{algline:rip-decomp:label-set-loop}}{
        \LET $\mappedRequestNodes = (\mappingNodesIteration)^{-1}(\substrateNodes)$ denote the already mapped nodes \label{algline:rip-decomp:iterationstart}\\
        \CHOOSE $m_\labelsetmappingIndex \in \MappingSpace[\labelsetIndexed{\labelSetIndex}]$, s.t. $\gamma^{\mappingNodesIteration(i)}_{i,\labelSetIndex,\labelsetmappingIndex} > 0$ \AND $\restrict[m_\labelsetmappingIndex][\labelsetIndexed{\labelSetIndex} \cap \mappedRequestNodes] = \restrict[\mappingNodesIteration][\labelsetIndexed{\labelSetIndex} \cap \mappedRequestNodes]$ 
        \label{algline:rip-decomp:choose-compatible-labelset-mapping} \\
        \SET $\mappingNodesIteration(j) \gets m_\labelsetmappingIndex(j)$ \FORALL $j \in \labelsetIndexed{\labelSetIndex} \setminus \mappedRequestNodes$\\
        \ForEach{$e=(i,j) \in \outEdgesExtractionOrder{i}: \labelsetRepresentative[e] = \labelsetIndexed{\labelSetIndex}$ \label{algline:rip-decomp:edge-in-label-set-loop}}{
          \eIf{$(i,j) = \EOEdgeToOriginal(i, j)$}{
            \COMPUTE path ${P}^{u,v}_{i,j}$ from $\mappingNodesIteration(i)=u$ to  $v \in \substrateNodes$ according to Lemma~\ref{lem:local-connectivity-property}\\
          \pushline\pushline  \nonl s.t. 
          $\begin{array}{rl}
          \subLP[y^v_{j}][(i,j),\restrict[\mappingNodesIteration][\labelsetEdge]] > 0 & \textnormal{\AND} \\[2pt]
          \subLP[z^{u',v'}_{i,j}][(i,j),\restrict[\mappingNodesIteration][\labelsetEdge]] > 0 & \textnormal{\FORALL } (u',v') \in {P}^{u,v}_{i,j}
          \end{array}$ \label{algline:rip-decomp:compute-path}\\
                      \vspace{2pt}
            \popline\popline
            
            \SET $\mappingEdgesIteration(i,j) \gets {P}^{u,v}_{i,j}$ \AND \textbf{if} $\mappingNodesIteration(j) = \emptyset$ \textbf{then} $\mappingNodesIteration(j) \gets v$
            \label{algline:rip-decomp:map-edge-head-node}\\
          }{
            \COMPUTE path ${P}^{v,u}_{j,i}$ from $v \in \substrateNodes$ to $\mappingNodesIteration(i)=u$ according to Lemma~\ref{lem:local-connectivity-property} \label{algline:rip-decomp:compute-path-rev-edge}\\
            \pushline \pushline \nonl s.t.
            $\begin{array}{rl}
              \subLP[y^v_{j}][\EOEdgeToOriginal(i, j),\restrict[\mappingNodesIteration][\labelsetEdge]] > 0 & \textnormal{\AND} \\[2pt]
              \subLP[z^{u',v'}_{j,i}][\EOEdgeToOriginal(i, j),\restrict[\mappingNodesIteration][\labelsetEdge]] > 0 & \textnormal{\FORALL } (u',v') \in {P}^{u,v}_{j,i}
            \end{array}$\\
            \vspace{2pt}
            \popline\popline

            \SET $\mappingEdgesIteration(\EOEdgeToOriginal(i,j)) \gets {P}^{u,v}_{j,i}$ \AND \textbf{if} $\mappingNodesIteration(j) = \emptyset$ \textbf{then} $\mappingNodesIteration(j) \gets v$
            \label{algline:rip-decomp:map-edge-head-node-reversed}\\
          
          }
          \If{$\mappingEdgesIteration(\EOEdgeToOriginal(e)) \neq \emptyset$ \textnormal{\FORALL}$e \in \inEdgesExtractionOrder{j}$}{
            \SET $\mathcal{Q} \gets \mathcal{Q} \cup \{j\}$ \label{algline:rip-decomp:add-j-to-q} \label{algline:rip-decomp:end-while-q}\\
          }
        }
      }
    }
    \SET $\Variables_k \gets 
    \left( \begin{array}{ll}
            & \{x\}  \cup  \{y^u_{i} ~|~ i \in \reqNodes, u=\mappingNodesIteration(i)\} \\
        \cup   & \{\,\subLP[y^u_{i}][e,\restrict[\mappingNodesIteration][\labelsetEdge]] ~|~ e \in \reqEdges, i \in e, u=\mappingNodesIteration(i) \}\\
        \cup    & \{\subLP[z^{u,v}_{i,j}][e,\restrict[\mappingNodesIteration][\labelsetEdge]]~|~ e=(i,j) \in \reqEdges, (u,v) \in \mappingEdgesIteration(i,j)\} \\
        \cup   & \{\gamma^{u}_{i,\labelSetIndex,\labelsetmappingIndex}~|~ i \in \reqNodes, u = \mappingNodesIteration(i), \labelsetIndexed{\labelSetIndex} \in \labelsetOrder_i, m_\labelsetmappingIndex=\restrict[\mappingNodesIteration][\labelsetIndexed{\labelSetIndex}]\} 
        \end{array}  \right)$ \label{algline:rip-decomp:compute-Vk}\\
    \SET $\prob \gets \min \Variables_k$ \label{algline:rip-decomp:computing-prob} \\
    \SET $v \gets v - \prob$ \FORALL $v \in \Variables_k$ \label{algline:rip-decomp:adapt-variables-one}\\
    \SET $a^{x,y} \gets a^{x,y} - \prob \cdot \allocationFunction(\mappingRequestIteration,x,y)$ \FORALL $(x,y) \in \SR$\label{algline:rip-decomp:adapt-load-one}\\ 
    \ForEach{$(i,j) \in \reqEdges$ \textnormal{\AND \EACH} $(x,y) \in \{(\Vtype(i),i), (\Vtype(j),j), (i,j)\}$}{
      \SET $\subLP[a^{x,y}][(i,j),\restrict[\mappingNodes][\labelsetEdge]] \gets \subLP[a^{x,y}][(i,j),\restrict[\mappingNodes][\labelsetEdge]] - \prob \cdot \allocationFunction(\mappingRequestIteration,x,y)$ \label{algline:rip-decomp:adapt-load}\\
    }
    \ADD $\decomp = (\prob,\mappingRequestIteration)$ to $\PotEmbeddings$ \\
    \SET $k \gets k + 1$ \\
  }

\KwRet{$\PotEmbeddings$}
\end{algorithm*}

\end{figure}

\clearpage

\subsection{Decomposition Algorithm for Extraction Label Set Orderings} \label{sec:hierarchical:decomposition-algorithm}

We now discuss the decomposition algorithm for Linear Program~\ref{LP:RunningIntersectionProperty}, which is given in pseudocode in Algorithm~\ref{alg:decompositionAlgorithm-RIP}. The algorithm proceeds similarly to the base decomposition algorithm: In each iteration of the outer loop starting in Line~\ref{algline:rip-decomp:while-flow}, a single mapping for the request is extracted. Within the body of the loop, a mapping for the extraction order root node is chosen, and incrementally extended through a graph traversal of the extraction order. 

The main difference to the base decomposition algorithm is the iteration over the label sets in Lines~\ref{algline:rip-decomp:label-set-loop} to~\ref{algline:rip-decomp:edge-in-label-set-loop}, where the base algorithm's iteration over edge bags is replaced by iteration over the extraction label set ordering $\labelsetOrderSet$. Unlike the iteration over the edge bags, the order in which the label sets are processed is now crucial.
\begin{remark} (Iteration over Extraction Label Set Ordering) \\
In Line~\ref{algline:rip-decomp:edge-in-label-set-loop} of Algorithm~\ref{alg:decompositionAlgorithm-RIP}, the iteration over the label sets of the extraction label set ordering $\labelsetOrder_i$ must be performed according to the order defined by $\labelsetOrder_i$.
\end{remark}

We will now show the correctness of Algorithm~\ref{alg:decompositionAlgorithm-RIP}.
\begin{theorem}(Decomposability of LP~\ref{LP:RunningIntersectionProperty} with Extraction Label Set Orderings)\\
Linear Program~\ref{LP:RunningIntersectionProperty} is decomposable using Algorithm~\ref{alg:decompositionAlgorithm-RIP}, if an extraction label set ordering $\Omega$ and a decomposable edge label assignment $\labelsetsEdges$ are used. Additionally, the resource allocations of the resulting convex combination of mappings $\PotEmbeddings$ are bounded by the allocation variables, such that $\sum_{(\prob, \mappingRequestIteration) \in \PotEmbeddings} \prob \cdot \allocationFunction(\mappingRequestIteration, x, y) \leq a^{x, y}$ holds for each resource $(x, y) \in \substrateResources$.
\label{thm:decomp_of_rip_orderings}
\end{theorem}
\begin{proof}
For the sake of contradiction, assume that the LP formulation is non-decomposable, i.e. a valid variable assignment for $\lpvars$ exists, such that the decomposition algorithm cannot extract a valid mapping.

The LP subformulations $(\vec{y}, \vec{z}, \vec{a})$ are unchanged from Linear Program~\ref{IP:novel-AC}, and therefore the local connectivity property from Lemma~\ref{lem:local-connectivity-property} holds within each subformulation. Consider any edge $e = (i, j) \in \reqExtractionOrderEdges$ of the extraction order. If any node mapping variable $\subLP[y^u_i][\EOEdgeToOriginal(e), m_e]$ with $m_e \in \MappingSpace[\labelsetEdge]$ for the tail node $i$ and some substrate node $u\in \substrateNodes$ is non-zero, local connectivity guarantees the existence of some non-zero variable $\subLP[y^v_j][\EOEdgeToOriginal(e), m_e]$ mapping the head node $j$ to some substrate node $v \in \substrateNodes$. Additionally, it guarantees that a path $P^{u,v}_{i,j}$ through the substrate exists with $z^{x, y}_{i, j} > 0$ for all path edges $(x, y) \in P^{u,v}_{i,j}$.

The decomposition algorithm initially determines a mapping for the root node $\reqExtractionOrderRoot$ in Line~\ref{algline:rip-decomp:nodeMapRoot}. Given a request node $i$, let 
$s_i := \sum_{u \in \substrateNodes}y^u_i$ denote the sum of all variables related to the mapping of $i$.
The following invariant holds in each iteration: For any $i \in \reqNodes$, $s_i$ is equal to the variable $x$, which tracks the total remaining mapping probability. Initially, $x=1$ is set in Line~\ref{algline:rip-decomp:setRemainingFlowVar}, while $s_i=1$ holds due to Constraint (\ref{alg:lp:novel:node-embedding}). In each iteration of the algorithm $s_i$ and $x$ are reduced by the same amount $\prob$ in Line~\ref{algline:rip-decomp:adapt-variables-one}.

Once $x=0$, it follows that $s_{\reqExtractionOrderRoot} = 0$, such that the while-condition in Line~\ref{algline:rip-decomp:begin-while-flow} fails and the algorithm terminates before reaching Line~\ref{algline:rip-decomp:nodeMapRoot} again. Therefore, some non-zero node mapping variable for the root node $\reqExtractionOrderRoot$ exists in each iteration of the while-loop, and the mapping of $\reqExtractionOrderRoot$ in Line~\ref{algline:rip-decomp:nodeMapRoot} always succeeds.

Next, assume that the algorithm fails to produce a valid mapping in some iteration $k$. This might occur in the following ways:
\begin{enumerate}
\item The algorithm maps an edge $e = (i, j) \in \reqEdges$ to a path $P^{u, v}_{i, j}$ in Line~\ref{algline:rip-decomp:map-edge-head-node} or~\ref{algline:rip-decomp:map-edge-head-node-reversed}, when either $i$ or $j$ has previously been mapped to some other node earlier in the algorithm's execution, violating the validity of the mapping. \label{thm:decomp_of_rip_orderings:contradiction:edge-fail}
\item For some $i \in \reqNodes$ with $u = \mappingNodesIteration(i)$, and for some label set $\labelsetIndexed{\labelSetIndex}$, the algorithm fails to choose a mapping $\mappingChar_\labelsetmappingIndex$  in Line~\ref{algline:rip-decomp:choose-compatible-labelset-mapping}, such that $\gamma^u_{i, a, \alpha} > 0$ and such that $\mappingChar_\labelsetmappingIndex$ is compatible with the mapped nodes $V^m_R$. \label{thm:decomp_of_rip_orderings:contradiction:node-fail}
\end{enumerate}

Case~\ref{thm:decomp_of_rip_orderings:contradiction:edge-fail} also occurs in the correctness proof of Algorithm~\ref{alg:decomposition:Algorithm-Novel-AC} and is proven contradictory by Rost and Schmid in \cite{rostSchmidFPTApproximations}. We reproduce the argument here, using the more general terminology of decomposable edge label assignments.

Assume that the edge mapping operation produces an invalid edge mapping for some edge $e = (i, j) \in \reqExtractionOrderEdges$. W.l.o.g., assume that $\EOEdgeToOriginal(e) = e$. If $\mappingNodesIteration(j)=\emptyset$, i.e. no previous node mapping for $j$ exists, $j$ is immediately and consistently mapped with the edge mapping in Line~\ref{algline:rip-decomp:map-edge-head-node}. 

To obtain an invalid mapping, $\mappingNodesIteration(j)$ must therefore have been set at an earlier time by the decomposition algorithm. This can only occur, if $j$ is contained in some edge label set. Then, since $j$ is by definition of a decomposable label assignment (see Definition~\ref{def:decomposable-edge-labels}) included in its label-induced subgraph, it must also hold that $j \in \labelsetIncoming[j]$. This implies that the algorithm in Line~\ref{algline:rip-decomp:compute-path} selected a path ${P}^{u,v}_{i,j}$ ending in some node $v \neq \mappingNodesIteration(j)$, such that $\subLP[y^v_{j}][(i,j),\restrict[\mappingNodesIteration][\labelsetEdge]] > 0$. However, Constraint (\ref{alg:lp:novel:forbidding-nodes-in-sub-lps}) explicitly sets $\subLP[y^v_{j}][(i,j),\restrict[\mappingNodesIteration][\labelsetEdge]] = 0$, contradicting the assumption that the variable assignment is a valid solution of Linear Program~\ref{LP:RunningIntersectionProperty}. 

Further, due to the definition of a decomposable label assignment, there must be a single node, namely the root of the label-induced subgraph for $j$, where $j$ was first introduced as a label. Due to the label path continuity shown in Lemma~\ref{lemma:decomplabels:all-paths-are-labeled}, all edges between this root node and $j$ are labeled with $j$. Since the previous edges on this path were successfully mapped, it follows that some non-zero flow is induced in the LP subformulation, such that a valid edge mapping can be extracted through a graph search, as guaranteed by the local connectivity property.

We now consider Case~\ref{thm:decomp_of_rip_orderings:contradiction:node-fail}. Note that every request node $i \in \reqNodes$ that is removed from $\mathcal{Q}$ in Line~\ref{algline:rip-decomp:retrieve-i-from-q} is already mapped by $\mappingNodesIteration$. We denote this set of already mapped nodes by $\mappedRequestNodes$. Assume for contradiction that for some $i \in \reqNodes$ with $u = \mappingNodesIteration(i)$ and some label set $\labelsetIndexed{\labelSetIndex} \in \labelsetOrder_i$, no mapping $\mappingChar_\labelsetmappingIndex \in \MappingSpace[\labelsetIndexed{\labelSetIndex}]$ exists with  $\restrict[\mappingChar_\labelsetmappingIndex][\labelsetIndexed{\labelSetIndex} \cap \mappedRequestNodes] = \restrict[\mappingNodesIteration][\labelsetIndexed{\labelSetIndex} \cap \mappedRequestNodes]$ and $\gamma^{u}_{i, \labelSetIndex, \labelsetmappingIndex}> 0$.
Any \emph{non-locally} mapped label nodes, i.e. nodes in $\labelsetIndexed{\labelSetIndex}$ that were mapped in previous iterations of the while loop in Lines~\ref{algline:rip-decomp:begin-while-q}~to~\ref{algline:rip-decomp:add-j-to-q}, must be contained in $\labelsetIncoming$.
Since $\labelsetOrder_i$ is part of an extraction label set ordering, $\labelsetIncoming$ is the first element of $\labelsetOrder_i$. From the running intersection property, it then follows that $\labelsetIndexed{\labelSetIndex} \cap \mappedRequestNodes \subseteq \labelsetPredecessors_\labelSetIndex$, i.e. any node that has been mapped by the algorithm and that is a label in $\labelsetIndexed{\labelSetIndex}$ is also contained in the predecessor label set $\labelsetPredecessors_\labelSetIndex$. 

Let $\labelSetIndexTwo = \labelsetPredecessorFunction(\labelsetIndexed{\labelSetIndex})$ be the index of the associated predecessor in the label set ordering. If $\labelSetIndexTwo = \labelSetIndex$, the predecessor label set of $\labelsetIndexed{\labelSetIndex}$ is empty, i.e. the label set $\labelsetIndexed{\labelSetIndex}$ does not contain any nodes that are already mapped in $\mappingNodesIteration$. Since $i$ was retrieved from the queue, all in-edges of $i$ are already mapped in $\mappingRequestIteration$. At the time when these in-edges were mapped, the criterion $\subLP[y^v_{j}][(i,j),\restrict[\mappingNodesIteration][\labelsetEdge]] > 0$ in Line~\ref{algline:rip-decomp:compute-path}, or analogously in Line~\ref{algline:rip-decomp:compute-path-rev-edge} had to be satisfied. Constraint (\ref{alg:lp:novel:node-to-sub-node-mapping}) then requires that $y^u_i > 0$. From Constraints (\ref{alg:lp:novel:node-to-sub-node-mapping}) and (\ref{eq:lp:novel-rip:gamma-to-out-edges}), it follows that some $\mappingChar_\labelsetmappingIndex \in \MappingSpace[\labelsetIndexed{\labelSetIndex}]$ with $\gamma^{u}_{i, \labelSetIndex, \labelsetmappingIndex } > 0$ exists, in contradiction to the assumption.

Therefore, assume that $\labelSetIndexTwo < \labelSetIndex$  and let $\mappingChar_\labelsetmappingIndexTwo := \restrict[\mappingNodesIteration][\labelsetIndexed{\labelSetIndexTwo}]$. Since the algorithm has reached Line~\ref{algline:rip-decomp:choose-compatible-labelset-mapping} for label set $\labelsetIndexed{\labelSetIndex}$,  $\labelsetIndexed{\labelSetIndexTwo}$ has been successfully processed by the algorithm. Therefore, $\gamma^{u}_{i, b, \labelsetmappingIndexTwo } > 0$ holds.
The sum on the right side of Constraint (\ref{eq:lp:novel-rip:gamma-continuity}) includes $\gamma^{u}_{i, b, \labelsetmappingIndexTwo}$. As all $\gamma$-variables are non-negative, it follows that there is at least one non-zero term in the sum on the left side of Constraint (\ref{eq:lp:novel-rip:gamma-continuity}). From the running intersection property, it follows that $\labelsetIndexed{\labelSetIndex} \setminus \labelsetPredecessors_\labelSetIndex$ only contains labels that are newly introduced by $\labelsetIndexed{\labelSetIndex}$, i.e. $\restrict[\mappingNodesIteration][\labelsetIndexed{\labelSetIndex}\setminus \labelsetPredecessors_\labelSetIndex] = \emptyset$ holds. Therefore, mapping these new labels cannot introduce any inconsistencies with already mapped nodes. 
It then follows that there is at least one mapping $\mappingChar_\labelsetmappingIndex \in \MappingSpace[\labelsetIndexed{\labelSetIndex}]$ such that $\gamma^{u}_{i, \labelSetIndex, \labelsetmappingIndex } > 0$ holds, and the algorithm can therefore successfully perform the choose operation in Line~\ref{algline:rip-decomp:choose-compatible-labelset-mapping}. 
This contradicts the assumption that the algorithm fails to construct a valid mapping in any iteration of the while loop starting in Line~\ref{algline:rip-decomp:begin-while-flow}, and by iterating this extraction process as long as $x > 0$ holds.

Hence, a convex combination of mappings $\PotEmbeddings$ is returned by Algorithm~\ref{alg:decompositionAlgorithm-RIP}.

We next show that after the subtraction operations in Lines~\ref{algline:rip-decomp:adapt-variables-one} to~\ref{algline:base-decomposition:adapt-load}, all constraints of Linear Program~\ref{LP:RunningIntersectionProperty} other than Constraint (\ref{alg:lp:novel:node-embedding}), which forces the full embedding of nodes, still hold. Since the mapping $\mappingRequestIterationDef$ which was extracted in iteration $k$ is valid, it holds that each $i \in \reqNodes$ is mapped to a single substrate node, and that each $(i,j) \in \reqEdges$ is mapped to a substrate path from $\mappingNodesIteration(i)$ to $\mappingNodesIteration(j)$. From the validity of $\mappingRequestIteration$, it follows that Constraints (\ref{eq:classic_mcf:constr_flow}), (\ref{alg:lp:novel:node-to-sub-node-mapping}), (\ref{eq:lp:novel-rip:gamma-to-out-edges}), (\ref{eq:lp:novel-rip:in-edges-to-gamma}), and (\ref{eq:lp:novel-rip:gamma-continuity}) hold: In each of these constraints, for every variable from which $\prob$ is subtracted on one side of the equation, a corresponding term exists on the other side, balancing the constraint. Note that by the same argument, Constraint (\ref{alg:lp:novel:node-embedding}) also holds after each iteration, if the constant $1$ on the right side of Constraint (\ref{alg:lp:novel:node-embedding}) is replaced with the variable $x$. A similar argument applies to the Constraints (\ref{alg:lp:novel:node-load}) and (\ref{alg:lp:novel:edge-load}) tracking the global allocations and Constraints (\ref{eq:classic_mcf:constr_node_load}) and (\ref{eq:classic_mcf:constr_edge_load}) tracking the allocations in each subformulation, by taking into account the definition of the mapping resource allocation function $\allocationFunction(\mappingRequestIteration, x, y)$ (see Definition \ref{def:mapping-allocation-function}). Lastly, Constraints (\ref{eq:classic_mcf:constr_forbid_small_nodes}), (\ref{eq:classic_mcf:constr_forbid_small_edges}) and (\ref{alg:lp:novel:forbidding-nodes-in-sub-lps}), which forbid specific node and edge mappings, still hold since the algorithm cannot extract a mapping conflicting with these constraints.
 
Lastly, using this invariant, we verify that the resource allocations of $\PotEmbeddings$ are bounded by the values of the allocation variables $\vec{a}$. Assume for contradiction, that for some resource $(x, y) \in \substrateResources$, $\sum_{(\prob, \mappingRequestIteration) \in \PotEmbeddings} \prob \cdot \allocationFunction(\mappingRequestIteration, x, y) > a^{x, y}$ holds. In each iteration of the decomposition algorithm, a mapping $(\prob, \mappingRequestIteration) \in \PotEmbeddings$ is extracted. The value $\prob \cdot\allocationFunction(\mappingRequestIteration, x, y)$ is subtracted from the allocation variable $a^{x, y}$ in Line~\ref{algline:rip-decomp:adapt-load-one} or~\ref{algline:base-decomposition:adapt-load}. Since the total allocations of all mappings exceed the initial value of the allocation variable by assumption, it follows that $a^{x,y} < 0$ after the final iteration of the algorithm, violating the invariant that the LP constraints are satisfied after each iteration of the loop, since the allocation variables are required to be non-negative.

Therefore, the bound $\sum_{(\prob, \mappingRequestIteration) \in \PotEmbeddings} \prob \cdot \allocationFunction(\mappingRequestIteration, x, y) \leq a^{x, y}$ holds for each resource $(x, y) \in \substrateResources$.
\end{proof}

\subsection{Reduction to the Base Algorithm} \label{sec:rip-reduction-to-base-algorithm}
In this section, we examine how the adapted algorithm relates to the base algorithm. The main result is that the base algorithm emerges from the adapted algorithm when using the label set ordering defined by the edge bag's label sets.

We begin by considering the label set ordering obtained by representing each edge label set by the label set of its edge bag. That is, for each edge $e \in \outEdgesExtractionOrder{i}$, the edge label set $\labelsetEdge[e]$ is represented in the extraction label set ordering $\labelsetOrder_i$ by the label set $\bagLabelSet$ of the edge bag $\outEdgeBag[i][b]$ it belongs to, according to Definition~\ref{def:edge-bags}. 

\begin{definition}(Bag Label Set Ordering) \label{def:bag-label-set-ordering}\\ 
Given an extraction order $\reqExtractionOrderDef$ with decomposable edge label assignment $\labelsetsEdges$, we define the bag label set ordering as $\labelsetOrderSet = \{\labelsetOrder_i | i\in \reqNodes\}$, where
\begin{align}
\labelsetOrder_i := (\labelsetIncoming, \bagLabelSet[i][1], ..., \bagLabelSet[i][n])
\label{eq:equivalence-bags-rip-ordered-labelsets}
\end{align}
where $\bagLabelSet[i][b]$ refers to the label set of $\VEbfsBag[b] \in \BagIPlus$, the $b$-th edge bag of node $i$, as introduced in Definition~\ref{def:bag-label-set}.
\end{definition}

\begin{lemma}(Validity of the Bag Label Set Ordering) \label{lemma:validity-of-bag-label-set-ordering}\\ 
The bag label set ordering introduced in Definition~\ref{def:bag-label-set-ordering} is an extraction label set ordering.
\end{lemma}
\begin{proof}
We show that the bag label set ordering satisfies all defining properties of an extraction label set ordering as given in Definition~\ref{def:extractionLabelSetOrdering}.

Property~\ref{def:extractionLabelSetOrdering:item:running-intersection-ordering} requires that the extraction label set ordering has the running intersection property. By the definition of edge bags (see Definition~\ref{def:edge-bags}), the bag label sets $\bagLabelSet$ are mutually disjoint. Therefore any intersection of some bag label set with another label set in the sequence must be with the incoming label set $\labelsetIncoming$, which comes first in the sequence. Therefore, the bag label set ordering has the running intersection property. 

Property~\ref{def:extractionLabelSetOrdering:item:first-nonlocal} requires that the incoming label set is the first element of the sequence. By Definition~\ref{def:bag-label-set-ordering}, $\labelsetIncoming$ is the first element in the sequence.

Property~\ref{def:extractionLabelSetOrdering:item:representative-exists} states that for each edge label set, a representative label set is contained in the extraction label set ordering. Since the edge bags partition the outgoing edges, i.e. every edge is contained in some edge bag, this property is satisfied as well.
\end{proof}

We will now show that Linear Program~\ref{LP:RunningIntersectionProperty} reduces to Linear Program~\ref{IP:novel-AC} from the base algorithm, when using the edge bags as the basis of the label set ordering.

\begin{lemma}(Equivalence to the Base LP Formulation) \label{lemma:bag-label-set-ordering-yields-same-LP}\\
Let $\reqExtractionOrderDef$ be an extraction order with a decomposable edge label assignment $\labelsetsEdges$. 

When using the bag label set ordering, Linear Program~\ref{LP:RunningIntersectionProperty} is equivalent to Linear Program~\ref{IP:novel-AC} in that any valid variable assignment for LP~\ref{LP:RunningIntersectionProperty} corresponds to a valid variable assignment for Linear Program~\ref{IP:novel-AC}.
\end{lemma}
\begin{proof}
We first consider the set of variables defined by both formulations. When using the same edge label assignment, the subformulation variables $(\vec{y}, \vec{z}, \vec{a})$ are the same for both LP formulations. Similarly, the global node mapping and allocation variables, which are independent of the edge bags/label set ordering, are defined the same way in both formulations. 

The $\gamma$-variables of Linear Program~\ref{LP:RunningIntersectionProperty} differ from those in Linear Program~\ref{IP:novel-AC}, since separate $\gamma$ variables are added for the incoming label set which are not present in Linear Program~\ref{IP:novel-AC}. However, note that by Constraint (\ref{eq:lp:novel-rip:in-edges-to-gamma}), these variables are always equal to the corresponding node mapping variables in the in-edges' LP subformulations.

We next show that under the bag label set ordering, the constraints of Linear Program~\ref{LP:RunningIntersectionProperty} enforce the constraints from Linear Program~\ref{IP:novel-AC}. 
Constraint (\ref{eq:lp:novel-rip:gamma-to-out-edges}) generalizes Constraint (\ref{alg:lp:novel:gamma-to-outgoing-edges}) from the original formulation: Since the edge bags partition the outgoing edges, and since the bag label sets are disjoint, the edge label set $\labelsetEdge$ is represented by one of the bag label sets, $\bagLabelSet[i][b]$. Stating $\labelsetRepresentative[e] \in \labelsetOrder_i$ is therefore equivalent to selecting $e \subseteq \VEbfsBag[b]$ and the constraints are otherwise equal.

We next consider Constraints (\ref{eq:lp:novel-rip:in-edges-to-gamma}) and (\ref{eq:lp:novel-rip:gamma-continuity}). Since the bag label sets $\bagLabelSet[i][b]$ are disjoint, any given label set's predecessor index given by $\labelsetPredecessorFunction(b)$ can only be $b$ or $1$, or in other words, any bag label set of outgoing edges may only intersect the incoming label set $\labelsetIncoming$. If the predecessor is $b$ itself, then no constraint of the form (\ref{eq:lp:novel-rip:gamma-continuity}) is generated, as $\labelsetPredecessorFunction(b) \neq b$ is one of the required conditions. 
If on the other hand, $\labelsetPredecessorFunction(b) = 1$, then the following constraints are generated by (\ref{eq:lp:novel-rip:gamma-continuity}):
\begin{align}
\sum_{\begin{subarray}{c}
    \mappingChar_\labelsetmappingIndex \in \MappingSpace[\labelsetIndexed{\labelSetIndex}]: \\
    \restrict[\mappingChar_\labelsetmappingIndex][\labelsetIndexed{\labelSetIndex} \cap \labelsetIncoming] = \mappingPredecessors
\end{subarray}} \hspace{-10pt} \gamma^u_{i, \labelSetIndex, \labelsetmappingIndex}
 ~=~ 
\sum_{\begin{subarray}{c}
    \mappingChar_\labelsetmappingIndexTwo \in \MappingSpace[\labelsetIncoming]: \\
    \restrict[\mappingChar_\labelsetmappingIndexTwo][ \labelsetIndexed{\labelSetIndex} \cap \labelsetIncoming] = \mappingPredecessors
\end{subarray}} \hspace{-10pt} \gamma^u_{i, 1, \labelsetmappingIndexTwo}
&&
\begin{array}{r}
\forall i \in \reqNodes, u \in \substrateNodes, \labelsetOrder_i \in \labelsetOrderSet,\\ 
\labelsetIndexed{\labelSetIndex} \in \labelsetOrder_i: \labelsetPredecessorFunction(\labelsetIndexed{\labelSetIndex}) = 1, \\
\mappingPredecessors \in \MappingSpace[\labelsetIndexed{\labelSetIndex} \cap \labelsetIncoming]
\end{array}
\end{align}

Due to Constraint (\ref{eq:lp:novel-rip:in-edges-to-gamma}), we can substitute $\subLP[y^u_{i}][\EOEdgeToOriginal(e),\mappingChar_\labelsetmappingIndex]$ for $\gamma^u_{i,1,\labelsetmappingIndex}$ for each $e \in \inEdgesExtractionOrder{i}$. Also, $\labelsetPredecessorFunction(\labelsetIndexed{\labelSetIndex}) = 1$ holds if and only if $\labelsetIncoming \cap \labelsetIndexed{\labelSetIndex} \neq \emptyset$. We can substitute ``$\bagLabelSet[i][b]: \VEbfsBag \in \BagIPlus$'' in place of ``$\labelsetOrder_i \in \labelsetOrderSet, \labelsetIndexed{\labelSetIndex} \in \labelsetOrder_i$'' due to (\ref{eq:equivalence-bags-rip-ordered-labelsets}) in Definition~\ref{def:bag-label-set-ordering}. Finally, since $\mappingChar_\labelsetmappingIndex \in \MappingSpace[\labelsetIncoming]$, and since $\labelsetIncoming$ is also the label set associated with the in-edges $e$, we rename $\mappingChar_\labelsetmappingIndex$ to $\mappingChar_e$, and $\labelsetIncoming$ to $\labelsetEdge[e]$. We thus obtain
\begin{align}
\sum_{\begin{subarray}{c}
    \mappingChar_\labelsetmappingIndex \in \MappingSpace[\bagLabelSet[i][b]]: \\
    \restrict[\mappingChar_\labelsetmappingIndex][\bagLabelSet[i][b] \cap \labelsetEdge[e]] = \mappingPredecessors
    \end{subarray}} \hspace{-10pt} \gamma^u_{i, b+1, \labelsetmappingIndex}
~=~ 
\sum_{\begin{subarray}{c}
    \mappingChar_e \in \MappingSpace[\bagLabelSet[i][b] \cap \labelsetEdge[e]]: \\
    \restrict[\mappingChar_e][\bagLabelSet[i][b] \cap \labelsetEdge[e]] = \mappingPredecessors
    \end{subarray}} \hspace{-10pt} \subLP[y^u_{i}][\EOEdgeToOriginal(e),\mappingChar_e]
&&
\begin{array}{r}
\forall i \in \reqNodes, u \in \substrateNodes, e \in \inEdgesExtractionOrder{i}, \\ 
\VEbfsBag \in \BagIPlus: \labelsetEdge[e] \cap \bagLabelSet[i][b] \neq \emptyset, \\
\mappingPredecessors \in \MappingSpace[\bagLabelSet[i][b] \cap \labelsetEdge[e]]
\end{array}
\label{eq:equivalence-bags-rip-ordered-labelsets:in-edges-to-gamma-transformed}
\end{align}
Note that the first label set in the label set ordering is always $\labelsetIncoming \equiv \labelsetEdge[e]$. Therefore, all label set indices of the $\gamma$-variables are offset by one compared to the edge bag's index, and the variable $\gamma^u_{i, b+1, \labelsetmappingIndex}$ is associated with the edge bag $\VEbfsBag[b]$.
Therefore, (\ref{eq:equivalence-bags-rip-ordered-labelsets:in-edges-to-gamma-transformed}) is equivalent to Constraint (\ref{alg:lp:novel:incoming-edges-to-gamma-variables}) from the base formulation. 

Since all other constraints from Linear Program~\ref{IP:novel-AC} are included in Linear Program~\ref{LP:RunningIntersectionProperty} without modification, we have shown that the LP formulations are equivalent when using the bag label set ordering.
\end{proof}
We have therefore shown that under the bag label set ordering, each variable assignment for Linear Program~\ref{LP:RunningIntersectionProperty} corresponds to a variable assignment for Linear Program~\ref{IP:novel-AC}.

We next consider the decomposition algorithms, i.e. we show that when given such equivalent variable assignment, any convex combination of mappings obtained by Algorithm~\ref{alg:decomposition:Algorithm-Novel-AC} can also be obtained by  Algorithm~\ref{alg:decompositionAlgorithm-RIP} using the bag label set ordering as defined in Lemma~\ref{lemma:validity-of-bag-label-set-ordering}.  
\begin{lemma}(Equivalence to the Base Decomposition Algorithm) \label{lemma:bag-label-set-ordering-decomp-equivalence} \\ 
Let $\VNEPInstance$ be an instance of VNEP, and let $\reqExtractionOrderDef$ an extraction order with a decomposable edge label assignment $\labelsetsEdges$. Further, let $\lpvarstilde$ be a solution to Linear Program~\ref{IP:novel-AC} and let $\lpvars$ be the corresponding variable assignment to Linear Program~\ref{LP:RunningIntersectionProperty} using the bag label set ordering $\labelsetOrderSet$.

For any convex combination of mappings $\PotEmbeddings$ obtained by Algorithm~\ref{alg:decomposition:Algorithm-Novel-AC} with input \newline $(\VNEPInstance, \reqExtractionOrder, \lpvarstilde)$, there is an execution of Algorithm~\ref{alg:decompositionAlgorithm-RIP} with input \newline $(\VNEPInstance, \reqExtractionOrder, \labelsetOrderSet, \lpvars)$ yielding the same result, $\PotEmbeddings$.
\end{lemma}
\begin{proof}
The decomposition algorithm only differ in their iteration over the label sets: While the base decomposition algorithm iterates over the bag label sets in arbitrary order in Line~\ref{algline:base-decomposition:foreach-bag-label-set}, the adapted algorithm iterates over the label sets according to the label set ordering in Line~\ref{algline:rip-decomp:edge-in-label-set-loop}. 

Since the bag label sets are disjoint by definition, the order of iteration over the extraction label set ordering in Algorithm~\ref{alg:decomposition:Algorithm-Novel-AC} is also arbitrary. Additionally, it holds that each bag label set occurs somewhere in the label set ordering. Therefore, each iteration of the for-loop over the label sets in Algorithm~\ref{alg:decompositionAlgorithm-RIP} matches an iteration in the corresponding loop in Algorithm~\ref{alg:decomposition:Algorithm-Novel-AC}, in the sense that in both iterations, a mapping is extended by the same set of edges and nodes as in the execution of Algorithm~\ref{alg:decomposition:Algorithm-Novel-AC}.

The same flows are available within each LP subformulation in $\lpvars$ as in $\lpvarstilde$. Algorithm~\ref{alg:decompositionAlgorithm-RIP} can therefore choose the same edge and node mappings in each iteration of the outer loop starting in Line~\ref{algline:rip-decomp:while-flow} that are contained in the mapping for the current iteration, $\mappingRequestIteration$.
It follows that in each iteration of this outer loop, the same mapping $\mappingRequestIteration$ with the same flow value $\prob$ can be extracted, and therefore the resulting convex combination is the same as $\PotEmbeddings$.
\end{proof}
Note that obtaining the same result from either algorithm is not guaranteed. For example, for some edge $(i, j) \in \reqEdges$, multiple non-zero flows might exist, each which could be chosen in some iteration of the algorithm. 

By Lemma~\ref{lemma:bag-label-set-ordering-yields-same-LP}, we have shown that when using the bag label set ordering, the two Linear Programs are equivalent, and by Lemma~\ref{lemma:bag-label-set-ordering-decomp-equivalence}, we show that the decomposition algorithms can in principle yield the same result, given equivalent LP solutions. Therefore, the extended algorithm given by Linear Program~\ref{LP:RunningIntersectionProperty} and Algorithm~\ref{alg:decompositionAlgorithm-RIP} truly generalizes the approach of the base algorithm.

\subsection{Complexity of the Extended Algorithm} \label{sec:size-rip-lp-formulation}

As discussed in Section~\ref{sec:solving-the-vnep:extraction-width}, the size of the base LP formulation and the runtime of the base decomposition algorithm grow exponentially with the extraction width parameter. The size of the adapted formulation in Linear Program~\ref{LP:RunningIntersectionProperty} can similarly be parameterized by the size of the largest label set in the extraction label set ordering $\labelsetOrderSet$. 

In this section, we first define the \emph{extraction label width} parameter. We then discuss LP size and runtime of the adapted algorithm parameterized by the extraction label width, and explore how several findings by Rost and Schmid related to the extraction width can be transferred to the extraction label width. Finally, we consider the hardness of determining the extraction label width of request topologies. We will show that given an extraction order, computing the optimal label set ordering is $\complexityNP$-hard.

\subsubsection{The Extraction Label Width Parameter}

To parameterize the LP size and runtime, we introduce a parameter analogous to the extraction width, which we call the \emph{extraction label width}. 
\begin{definition} (Extraction Label Width) \\\
Let $\reqTopologyDef$ be a request topology and let $\reqExtractionOrder$ be an extraction order for $\reqTopology$. Let $\labelsetOrderSet$ be an extraction label set ordering for $\reqExtractionOrder$. We define the width $\labelWidthEQ(\reqExtractionOrder, \labelsetOrderSet)$ as the size of the largest label set in $\labelsetOrderSet$:
\begin{align}
\labelWidthEQ(\reqExtractionOrder, \labelsetOrderSet) := 1 + \max_{i \in \reqNodes} \max_{~L \in \labelsetOrder_i~} ~|L| 
\end{align}

We define the extraction label width of a request topology  $\reqTopology$ as the minimal width of all possible extraction orders and extraction label set orderings, where $\extractionOrderCharacter(\reqTopology)$ and $\labelsetOrderSetSet(\reqExtractionOrder)$ denote the sets of all extraction orders of $\reqTopology$, and the set of all extraction label set orderings of $\reqExtractionOrder$, respectively:
\begin{align}
\labelWidth(\reqTopology) := \min_{ \reqExtractionOrder \in \extractionOrderCharacter (\reqTopology) }
   \min_{ ~\labelsetOrderSet \in \labelsetOrderSetSet(\reqExtractionOrder)~ } \labelWidthEQ(\reqExtractionOrder, \labelsetOrderSet)
\end{align}
\end{definition}

Analogously to the size of the base LP formulation and the runtime of the base decomposition algorithm, we can show that Linear Program~\ref{LP:RunningIntersectionProperty} and Algorithm~\ref{alg:decompositionAlgorithm-RIP} have polynomial runtime, when using extraction orders and extraction label set orderings with bounded width.
\begin{theorem} (Size of Linear Program~\ref{LP:RunningIntersectionProperty}) \label{thm:adapted-alg-lp-size}\\
Let $\VNEPInstance$ be an instance of VNEP. Given an extraction order $\reqExtractionOrderDef$, a decomposable edge label assignment $\labelsetsEdges$, and an extraction label set ordering $\labelsetOrderSet$, the size of Linear Program $\ref{IP:novel-AC}$ is bounded by $\mathcal{O}\left(|\substrateTopology|^{\labelWidthEQ(\reqExtractionOrder, \labelsetOrderSet)} \cdot |\reqTopology| \right)$.
\end{theorem}
\begin{proof}
The proof is very similar to the corresponding proof for the base algorithm (cf. \cite{rostSchmidFPTApproximations}, Theorem~17). The most significant change is that the extraction label width is substituted for the extraction width parameter.

Analogously to the bag label sets in the base algorithm, the label sets in the extraction label set ordering are by definition supersets of the label sets in the edge label assignment. It follows that the number of LP subformulations for a single edge is bounded by the size of the largest label set's mapping space, i.e. by $|\substrateNodes|^{\labelWidthEQ(\reqExtractionOrder, \labelsetOrderSet) - 1}$. Each such subformulation has $\bigO(|\substrateEdges|)$ many variables. 

The variables $\gamma^u_{i, \labelSetIndex, \labelsetmappingIndex}$ are defined for each request and substrate node, and for each mapping $m_\labelsetmappingIndex \in \MappingSpace[\labelsetIndexed{\labelSetIndex}]$. Therefore, there are $\bigO(|\substrateNodes|^{\labelWidthEQ(\reqExtractionOrder, \labelsetOrderSet)} \cdot |\reqNodes|)$ many such variables.
\end{proof}

\begin{theorem} (Runtime of Algorithm~\ref{alg:decompositionAlgorithm-RIP}) \label{thm:adapted-alg-runtime}\\
Let $\VNEPInstance$ be an instance of VNEP. Let $\reqExtractionOrderDef$ be an extraction order for $\reqTopology$. Given a solution $\lpvars$ for Linear Program~\ref{LP:RunningIntersectionProperty} with the confluence edge label assignment $\labelsetsEdges$, and an extraction label set ordering $\labelsetOrderSet$, Algorithm~\ref{alg:decomposition:Algorithm-Novel-AC} executes in time $\mathcal{O}\left(|\substrateTopology|^{2 \cdot \labelWidthEQ(\reqExtractionOrder, \labelsetOrderSet) + 1} \cdot |\reqTopology|^2 \right)$.
\end{theorem}
\begin{proof}
The proof is very similar to the corresponding proof for the base algorithm (cf. Theorem 18 in \cite{rostSchmidFPTApproximations}). The only significant change is that the extraction label width is substituted for the extraction width parameter.

In each iteration of the outer while-loop starting in Line~\ref{algline:rip-decomp:begin-while-flow}, a single mapping is extracted. At least one node or edge mapping variable is set to zero in Line~\ref{algline:rip-decomp:adapt-variables-one}. Once all node mapping variables are zero, the loop in Line~\ref{algline:rip-decomp:begin-while-flow} condition fails and the algorithm terminates. 

By Theorem~\ref{thm:adapted-alg-lp-size} the number of variables is $\mathcal{O} \left(|\substrateTopology|^{\labelWidthEQ(\reqExtractionOrder, \labelsetOrderSet)}  \cdot |\reqTopology| \right)$. Accordingly, at most $\mathcal{O} \left(|\substrateTopology|^{\labelWidthEQ(\reqExtractionOrder, \labelsetOrderSet)}  \cdot |\reqTopology| \right)$ many mappings are extracted. 

We now consider the time for each mapping extraction. Each iteration of the outer while loop performs at most $|\reqNodes|$ many \textbf{choose} operations in Line~\ref{algline:rip-decomp:choose-compatible-labelset-mapping}. The graph search to decide the edge mapping in Lines~\ref{algline:rip-decomp:compute-path} and~\ref{algline:rip-decomp:compute-path-rev-edge} is performed $\bigO(|\reqEdges|)$ many times and can be implemented in time $\bigO(|\substrateEdges|)$.
\end{proof}

\subsubsection{Extraction Label Width and Extraction Width}

In this subsection, we discuss how several of the results related to the extraction width parameter presented in Section~\ref{sec:solving-the-vnep:small-extraction-width-graphs} extend to the extraction label width.

Using the results from Section~\ref{sec:rip-reduction-to-base-algorithm}, we can first establish an upper bound for the extraction label width: Since the bag label set ordering is always an extraction label set ordering, the extraction label width is at most the extraction width.
\begin{corollary} (Upper Bound for the Extraction Label Width) \\
The extraction label width of a request $\reqTopologyDef$ is bounded above by the extraction width, i.e. $\labelWidth(\reqTopology) \leq \extractionWidthReq(\reqExtractionOrder)$.
\end{corollary}
\begin{proof}
Let $\reqExtractionOrderDef$ be an extraction order minimizing the extraction width and let $\bagLabelsetOrderSet$ denote the bag label set ordering of $\reqExtractionOrder$. The extraction width of the request is equal to the size of the largest label set in $\bagLabelsetOrderSet$. By Lemma~\ref{lemma:validity-of-bag-label-set-ordering}, the bag label set ordering is an extraction label set ordering for $\reqExtractionOrder$. It follows that $\labelWidthEQ(\reqExtractionOrder, \bagLabelsetOrderSet) = \extractionWidthReq(\reqTopology)$, which implies that $\labelWidth(\reqTopology) \leq \extractionWidthReq(\reqTopology)$.
\end{proof}
Note that given an extraction order and decomposable edge label assignment, the corresponding bag label set ordering can be derived in polynomial time. We therefore conclude that the complexity of the extended algorithm should never exceed that of the base algorithm.

We now return to Lemma~\ref{thm:extraction-width-adding-parallel-paths}, which states that the extraction width of a request topology increases by at most the maximal degree, when extending the topology by adding paths parallel to existing edges. We will show an analogous result for the extraction label width, with a tighter bound: Instead of increasing by the maximal degree of the topology, the extraction  label width increases by at most one. 
\begin{lemma} (Increase of Extraction Label Width when Adding Parallel Edges) \\
Given an arbitrary request $\reqTopology$, adding any number of parallel paths for an existing edge increases the extraction label width of $\reqTopology$ by at most one.
\end{lemma}
\begin{proof}
Consider an extraction order $\reqExtractionOrder$ and an extraction label set ordering $\labelsetOrderSet$ for $\reqTopology$ with minimal width, i.e. select $\reqExtractionOrder$, $\labelsetOrderSet$ such that $\labelWidthEQ(\reqExtractionOrder, \labelsetOrderSet) = \labelWidth (\reqTopology)$.

Let $e =(i,j) \in \reqEdges$ be some edge of the graph and assume without loss of generality that $e \in \reqExtractionOrderEdges$. When inserting a parallel path $P$  to the original topology, we include this path in the extraction order in the same orientation as $(i, j)$, i.e. we orient the path in the extraction order such that it starts in $i$ and ends in $j$ in the extraction order. We further denote by $(i,k) \in P$ the first edge in $P$.

The edge $e$ is itself a path from $i$ to $j$, and therefore, $P$ and $(i, j)$ form a confluence ending in $j$. Since $e$ and $P$ both start and end in the same node without branches, the label sets must be equal under the confluence label assignment, i.e. $\labelsetEdge[e] = \labelsetEdge[(i,k)]$. 

There are now two possibilities: 
The node $j$ may be the end node of some confluence in the original extraction order. In this case, $j$ must be in the edge label set $\labelsetEdge[e]$, and accordingly be contained in $\labelsetRepresentative[e] \in \labelsetOrder_i$, the representative label set for $\labelsetEdge[e]$. In this case, the $\labelsetEdge[(i,k)] = \labelsetEdge[e] \subseteq \labelsetRepresentative[e]$ holds and the extraction label set ordering $\labelsetOrderSet$ applies without modification.

Alternatively, $j$ may not be a confluence end node in the original extraction order, and the addition of $P$ may have induced a new confluence. In this case, consider  $\labelsetOrder_i \in \labelsetOrderSet$, the original label set ordering for node $i$. The representative label set of $e$, $\labelsetRepresentative[e] \in \labelsetOrder_i$, does not contain $j$, since $j$ was not originally a confluence end node. However, since $P$ is a path parallel to an existing edge, no other \emph{new} confluences can be formed by its inclusion in the request graph, and since no confluence ending in $j$ exists in the unmodified request, $j$ does not occur in any other label set. Therefore, extending $\labelsetRepresentative[e]$ with the single node $j$ yields an extraction label set ordering for the modified extraction order. 

Only in the second case, i.e. $j \not\in \labelsetEdge[e]$, and if $|\labelsetRepresentative[e]| = \labelWidth (\reqTopology)$, i.e. the representative label set for $e$ was the largest label set in $\labelsetOrderSet$, does the extraction label width of $\reqTopology$ increase by one through the inclusion of $j$. 
\end{proof}

We next consider the extraction label width of the half-wheel graph (see Definition~\ref{def:half-wheel-graphs}), in particular focusing on the non-optimal root placement in the center node.
\begin{lemma} (Extraction Label Width of Half Wheel Graphs) \label{lemma:label-width-of-half-wheel-graphs} \\
Let $\halfwheelGraphDef$  be a half wheel graph with $n$ outer nodes. There exists an extraction order $\reqExtractionOrder$ that is rooted in $w_c$ and a label set ordering $\labelsetOrderSet$, such that $\labelWidthEQ(\reqExtractionOrder, \labelsetOrderSet) = 3$
\end{lemma}
\begin{proof}
We give a proof for the case where $n$ is odd. The argument can easily be extended to even $n$. We first define the extraction order $\reqExtractionOrderDef$ with root placement $\reqExtractionOrderRoot = w_c$, and the edge orientation
\begin{align}
\reqExtractionOrderEdges = \underbrace{\{ (w_c, w_1), \ldots (w_c, w_n) \}}_{\text{radial edges}} \cup \underbrace{\{ (w_{2}, w_{1}), (w_{2}, w_{3}), \ldots (w_{n-1}, w_{n-2}), (w_{n-1}, w_{n}) \}}_{\text{alternating outer edges}}.
\end{align}

The edges of the extraction order are oriented in such a way that all odd numbered outer nodes are confluence end nodes. The outgoing edges $(w_c, w_i) \in \outEdgesExtractionOrder{w_c}$ from the central node $w_c$ then have the following edge label assignment:
\begin{align}
\labelsetEdge[(w_c, w_i)] = \begin{cases}
\{w_i\} \text{ if $i$ is odd.} \\
\{w_{i-1}, w_{i+1} \} \text{ if $i$ is even.}
\end{cases}
\end{align}

We will now define an extraction label set ordering $\labelsetOrderSet$, such that  $\labelWidthEQ(\reqExtractionOrder, \labelsetOrderSet) = 3$ holds. For $w_c$, we define the ordering by directly using the edge label sets of edges from $w_c$ to even-numbered outer nodes, i.e. the ordering $\labelsetOrder_{\reqExtractionOrderRoot} := (\{w_1, w_3\}, \{w_3, w_5\}, \ldots \{w_{n-2}, w_{n} \})$. This ordering defines an extraction label set ordering for $\reqExtractionOrderRoot$. For each outer node $w_i$, we use the label set ordering containing only the incoming label set, i.e. we set $\labelsetOrder_{w_i} = (\labelsetIncoming[w_i])$. The label sets of the out-edges of outer nodes are subsets of $\labelsetIncoming[w_i]$: If $i$ is odd, then no out-edge exists, and if $i$ is even, it has two out edges labeled with $\{i-1\}$ and $\{i+1\}$, respectively, both of which are in $\labelsetIncoming[w_i]$. The largest label set in $\labelsetOrderSet$ therefore has size two, and it follows that $\labelWidthEQ(\reqExtractionOrder, \labelsetOrderSet) = 3$.
\end{proof}
We thus find that while placing the root in the central node of a half wheel graph is still suboptimal, the difference is no longer dependent on the size of the half wheel graph, and therefore the impact of a non-optimal root placement on the extraction label width is far less drastic compared to the impact on the extraction width, as derived by Rost and Schmid in \cite{rostSchmidFPTApproximations}.

Note that the example shown in Figure~\ref{fig:impact-of-adding-one-edge}, highlighting the impact on the extraction width of adding a single edge to an extraction order, is also a half-wheel graph. Due to Lemma~\ref{lemma:label-width-of-half-wheel-graphs}, it follows that the extraction label width of the extraction order remains unaffected by the inclusion of the additional edge.

\subsection{Optimal Extraction Label Set Orderings} \label{sec:extraction-orders-tree-decomp}

We will now consider the problem of finding a small extraction label set ordering $\labelsetOrderSet$ for a specific extraction order, when using the confluence edge label assignment. In this section, we will show that finding extraction label set orderings with minimal width for a fixed extraction order is $\complexityNP$-hard: By calculating the width of a specific extraction order, the treewidth of an arbitrary undirected graph can be computed. Additionally, we show that an optimal extraction label set ordering can be obtained from a minimal-width tree decomposition of an undirected graph related to the extraction order.

To discuss the extraction label set ordering, we first introduce the \emph{incident label sets} to describe the edge label sets of edges incident in a specific request node.
\begin{definition} (Incident Label Sets) \label{def:incident-labelsets}\\
Let $\reqExtractionOrderDef$ be an extraction order with a decomposable edge label assignment $\labelsetsEdges$. Given a node $i \in \reqNodes$, we define the incident label sets of $i$ as the set family containing the edge label sets of all edges incident in $i$, i.e. we define
\begin{align}
\labelsetsIncident[i] := \{\labelsetIncoming[i]\} \cup \{\labelsetEdge[e] ~|~ e \in \outEdgesExtractionOrder{\reqExtractionOrderRoot}, \labelsetEdge[e]\in \labelsetsEdges \}.
\end{align}
\end{definition}

We will now first show that arbitrary hypergraphs can be represented by the label sets in an extraction order, when the confluence edge label assignment is used. The following theorem states that for an arbitrary collection of label sets $\labelsets$, an extraction order can be constructed such that the root node's incident label sets contain all of the label sets in $\labelsets$.
\begin{theorem} (Existence of Extraction Orders Representing Arbitrary Label Sets) \label{lem:existence-EO-arbitrary-labelsets} \\
Given an arbitrary set of label sets $\labelsets$, an extraction order $\reqExtractionOrderDef$ exists, such that the incident label sets of $\reqExtractionOrderRoot$ under the confluence edge label assignment, $\labelsetsIncident[\reqExtractionOrderRoot]$, contain each label set in $\labelsets$. Further, the inclusion-wise maximal subset of $\labelsetsIncident[i]$ is equal to $\labelsets$, i.e. \newline 
$\labelsets = \{ L ~|~ L \in \labelsetsIncident[i]: (\nexists L' \in \labelsetsIncident[i] \setminus \{L\}:  L \subset L' ) \}$.
\end{theorem}
\begin{proof}
We prove the lemma constructively, by describing how to obtain such an extraction order. 
The node set consists of a root node, and two node sets, $V_l$ and $V_\labelsets$. The first node set $V_l$ is the set containing all label nodes, i.e. $V_l = \bigcup_{\labelsetIndexed{\labelSetIndex} \in \labelsets} \labelsetIndexed{\labelSetIndex}$. The second node set $V_\labelsets$ contains one representative node for each label set, i.e. $V_\labelsets = \{1, \ldots |\labelsets| \}$. Given one of these representative nodes $\labelSetIndex \in V_\labelsets$, we will simply refer to the corresponding label set from $\labelsets$ as $\labelsetIndexed{\labelSetIndex}$.
We thus define the node set of the extraction order as 
\begin{align}
\reqNodes = \{\reqExtractionOrderRoot\} \cup V_l \cup V_\labelsets.
\end{align}

We define the edge set as $E = E_1 \cup E_2 \cup E_3$ in terms of the following three components:
\begin{align}
E_1 := & \{ (\reqExtractionOrderRoot, \labelSetIndex) ~|~ \labelSetIndex \in V_\labelsets \} \\
E_2 := & \{ (\labelSetIndex, l) ~|~ \labelSetIndex \in V_\labelsets, l \in \labelsetIndexed{\labelSetIndex} \} \\
E_3 := & \{ (\reqExtractionOrderRoot, l) ~|~ l \in V_l \text{ with } \left| \{ L ~|~ L \in \labelsets, l \in L \} \right| = 1 \}
\end{align}
The first edge set $E_1$ contains edges connecting the root node to each of the nodes representing label sets. The edges in $E_2$ connect each label set node $\labelSetIndex \in V_\labelsets$ with each label node $l \in \labelsetIndexed{\labelSetIndex}$. Finally, the edges $E_3$ connect the root node to label nodes that are only contained in a single label set.

We now show that each label set in $\labelsets$ is the label set of one of the root node's outgoing edges. Let $\labelSetIndex \in V_\labelsets$ be the representative node of some label set $\labelsetIndexed{\labelSetIndex} \in \labelsets$. Since $(\reqExtractionOrderRoot, \labelSetIndex) \in E_1$, and by definition of $E_2$, each label in $\labelsetIndexed{\labelSetIndex}$ is reachable from $\labelSetIndex$. 
A second path from $\reqExtractionOrderRoot$ to $l$ exists: If $l$ is contained in some other label set $\labelsetIndexed{\labelSetIndexTwo}$, then the path via $\labelSetIndexTwo \in V_\labelsets$ can be taken. If $l$ is only contained in $\labelsetIndexed{\labelSetIndex}$, then the edge $(\reqExtractionOrderRoot, l)$ is included in $E_3$.
Therefore, for each label node $l \in \labelsetIndexed{\labelSetIndex}$, the edge $(\reqExtractionOrderRoot, \labelSetIndex)$ lies on a confluence starting in $\reqExtractionOrderRoot$ and ending in $l$. It follows that under the confluence edge label assignment, $\labelsetEdge[(\reqExtractionOrderRoot, \labelSetIndex)] = \labelsetIndexed{\labelSetIndex}$ holds.

For each $e = (\reqExtractionOrderRoot, l) \in E_3$, the edge label set is $\labelsetEdge = \{l\}$, which by construction is not contained in $\labelsets$. However, the label $l$ is contained in some label set in $\labelsets$. Therefore, the label set of any edge in $E_3$ is a subset of some other edge label set and does not affect the inclusion-wise maximal subset of $\labelsetsIncident[\reqExtractionOrderRoot]$.

It follows that the inclusion-wise maximal subset of $\labelsetsIncident[\reqExtractionOrderRoot]$ is the same as $\labelsets$.
\end{proof}
Theorem~\ref{lem:existence-EO-arbitrary-labelsets} not only proves that an extraction order representing any possible edge label set configuration exists, but also shows constructively how this extraction order can be obtained in polynomial time and space. The resulting extraction order is quite similar to the representation of the hypergraph as an incidence graph, as introduced in Definition~\ref{def:hypergraph-incidence-graph}.

Two examples of this construction are shown in Figures~\ref{fig:arbitrary-label-graph-example} and~\ref{fig:arbitrary-label-hypergraph-example}. Figure~\ref{fig:arbitrary-label-graph-example} shows an example where the label sets each contain two labels, and can therefore directly be identified with the edge set of the label graph. Figure~\ref{fig:arbitrary-label-hypergraph-example} shows the more general case of arbitrarily sized label sets. In this case, the label sets must be represented as a hypergraph instead.

\begin{figure}[htbp]
	\centering
	\includegraphics[width=0.75\textwidth]{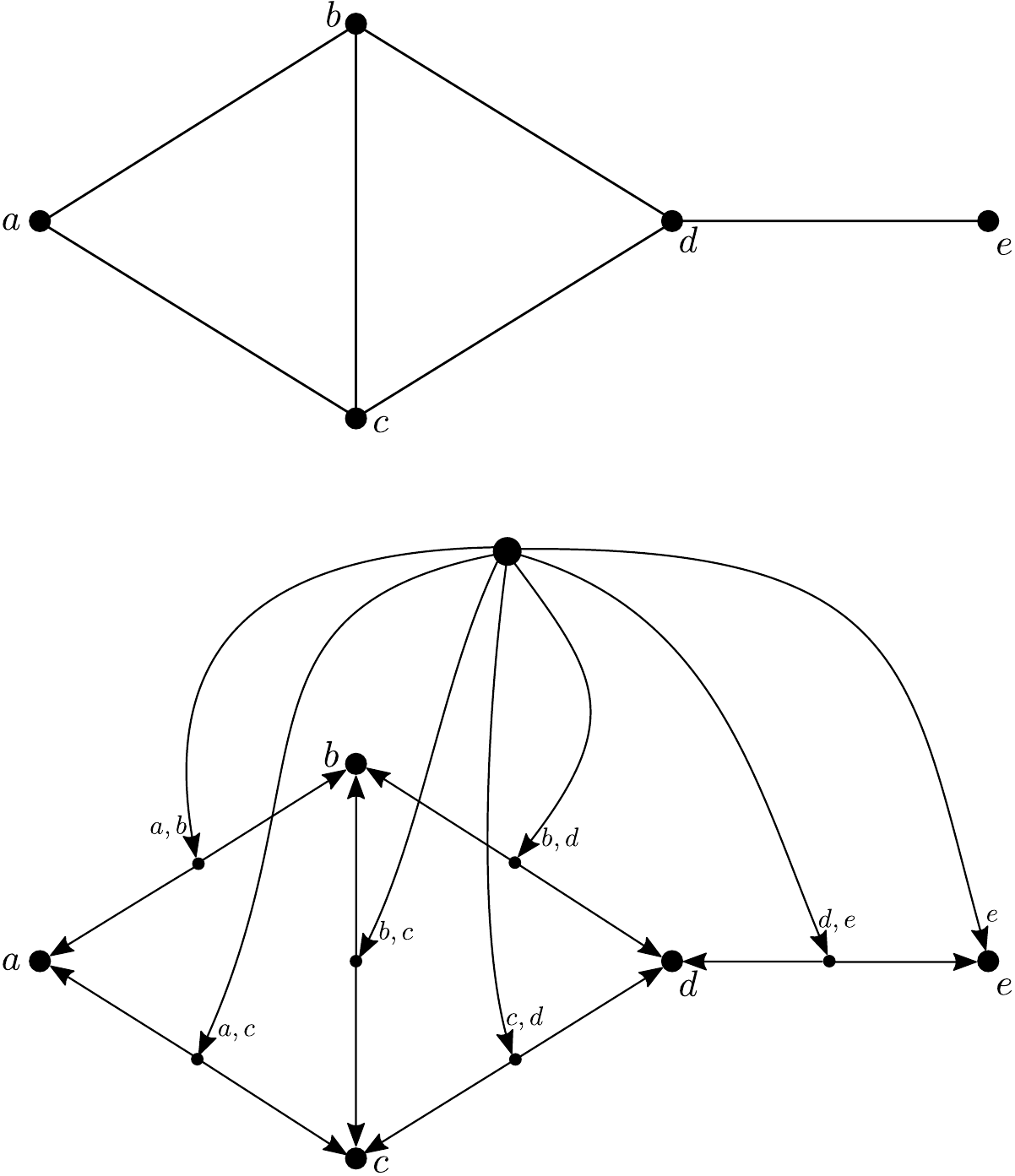}
	\caption{
		An example for the construction described in Lemma~\ref{lem:existence-EO-arbitrary-labelsets} for the label sets \mbox{$\labelsets = \{ \{ a, b \},\{ a, c \},\{ b, c  \}, \{ b, d \},\{ c, d \},\{ d, e \} \}$}. Since the label sets each contain two nodes, they can be represented as the edge set of an undirected graph $G$. An extraction order is created, where the label sets of the root's out-edges correspond to the edge set of $G$.\newline
		Top: The undirected graph $G$, representing the label set graph $\labelsets$. \newline
		Bottom: The extraction order with label sets annotated to the root node's out-edges. Note that the label sets of these edges correspond directly to the edges of $G$. The dashed edge $(r, e)$ belongs to the edge set $E_3$ in Lemma~\ref{lem:existence-EO-arbitrary-labelsets}, and is added to ensure the existence of a confluence ending in $e$, since $e$ only occurs in a single label set.
	}
	\label{fig:arbitrary-label-graph-example}
\end{figure}

\begin{figure}[htbp]
	\centering	\includegraphics[width=0.6\textwidth]{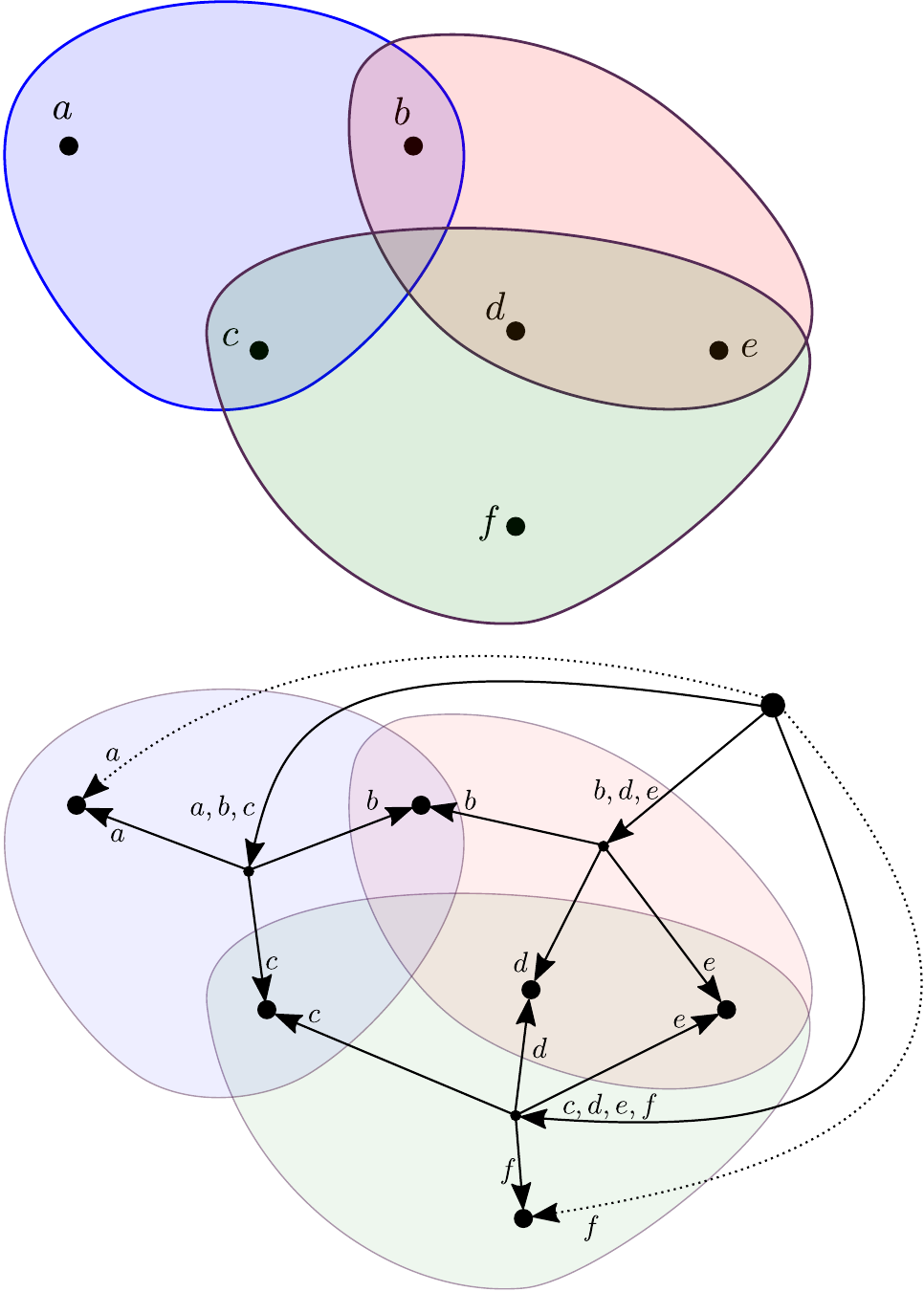}
	\caption{
		An example for the construction described in Lemma~\ref{lem:existence-EO-arbitrary-labelsets} for the label sets \mbox{$\labelsets = \{ \{a, b, c\}, \{ b, d, e \}, \{ c, d, e, f\} \}$}. An extraction order is created, where the label sets of the root's out-edges correspond to the edge set of $G$.\newline
		Top: The hypergraph corresponding to the label set $\labelsets$. The three Hyperedges are colored blue, red, and green, respectively. \newline
		Bottom: The extraction order annotated with the confluence edge label assignment. The label sets of the root node's out-edges directly correspond to the label sets in $\labelsets$. The dashed edges belong to the edge set $E_3$ in Lemma~\ref{lem:existence-EO-arbitrary-labelsets} and are only added for label nodes which would otherwise not be confluence end nodes.
	}
	\label{fig:arbitrary-label-hypergraph-example}
\end{figure}

We now define a representation of the label sets  as an undirected graph. The definition identifies the label sets with the corresponding  hypergraph (see Remark~\ref{remark:label-set-hypergraph}), and defines the \emph{label set graph} as its primal graph.
\begin{definition} (Label Set Graph) \label{def:label-set-graph} \\
Let $\labelsets$ be some set of label sets. We define the label set graph $\labelSetGraphDef[\labelsets]$ as the primal graph of the hypergraph induced by $\labelsets$ (see Definition~\ref{def:hypergraph-primal-graph}).
That is, $\labelSetGraph$ is an undirected graph, whose nodes are the label nodes, i.e. $\labelSetGraphNodes[\labelsets] = \bigcup \labelsets$. Two nodes $a, b \in \labelSetGraphNodes[\labelsets]$ are connected if there is some label set $L \in \labelsets$  with $\{a, b\} \subseteq L$.
\end{definition}
Note that by Lemma~\ref{lem:existence-EO-arbitrary-labelsets}, we have established that arbitrary graphs can occur as the label set graph of $\labelsetsIncident[i]$, i.e. the set of label sets of edges incident in some node $i\in \reqNodes$. That is, given an arbitrary graph $\genericGraph$, an extraction order exists such that $\genericGraph$ is the same as the label set graph $\labelSetGraph[\labelsetsIncident[i]]$ for some $i\in \reqNodes$.

The next lemma states that each extraction label set ordering defines a tree decomposition of the corresponding label set graph.
\begin{lemma} (Tree Decompositions and Extraction Label Set Orderings) \label{lemma:extraction-LSO-is-tree-decomp} \\
Let $\reqExtractionOrderDef$ be some extraction order, and let $i \in \reqNodes$ be some node. Consider the set of incident label sets of $i$, $\labelsetsIncident[i]$ (see Definition~\ref{def:incident-labelsets}). 

Given any extraction label set ordering $\labelsetOrder_i$ for $i$, a tree decomposition $\treeDecompDef$ of the label set graph $\labelSetGraph[\labelsetsIncident[i]]$ with width $\decompwidth(\treeDecomp) = \max_{L \in \labelsetOrder_i} |L|$ can be constructed in polynomial time, and vice versa.
\end{lemma}
\begin{proof}
We first show that given an extraction label set ordering $\labelsetOrder_i$ for a node $i$, a tree decomposition for the label set graph $\labelSetGraph[\labelsetsIncident[i]]$ with width $\max_{L \in \labelsetOrder_i} |L|$ can be computed. 

By Definition~\ref{def:extractionLabelSetOrdering}, the extraction label set ordering $\labelsetOrder_i$ has the running intersection property. Therefore, by Lemma~\ref{lemma:join-tree-from-rip-order}, a join tree can be constructed in polynomial time, whose nodes are the label sets in $\labelsetOrder_i$. By Lemma~\ref{lemma:join-trees-and-tree-decomp-primal-graphs}, this join tree defines a tree decomposition for the primal graph of the hypergraph defined by the label sets in $\labelsetOrder_i$. However, by Definition~\ref{def:label-set-graph}, this graph is the label set graph $\labelSetGraph[\labelsetOrder_i]$. By construction, the label sets in $\labelsetOrder_i$ are the same as the sets $\treeDecompSets$ associated with the tree decomposition. It follows that $\decompwidth(\treeDecomp) = \max_{L \in \labelsetOrder_i} |L|$.

We next show that given a tree decomposition $\treeDecompDef$ for the label set graph $\labelSetGraph[\labelsetsIncident[i]]$, we show that an extraction label set ordering $\labelsetOrder_i$ with width $\decompwidth(\treeDecomp)$ can be obtained.
Firstly, by Lemma~\ref{lemma:tree-decomp-and-join-trees}, the tree decomposition can be transformed to a join tree over the tree decomposition's set family $\treeDecompSets$. By Lemma~\ref{lemma:rip-order-from-join-tree}, we can obtain a running intersection ordering from this join tree, satisfying Property~\ref{def:extractionLabelSetOrdering:item:running-intersection-ordering} of the extraction label set ordering.

Secondly, by Corollary~\ref{corollary:rip-arbitrary-start-edge}, the running intersection ordering described above can be constructed in such a way that it starts with $\labelsetIncoming[i]$. Therefore, Property~\ref{def:extractionLabelSetOrdering:item:first-nonlocal} holds.

Finally, note that by Definition~\ref{def:label-set-graph}, each label set in $\labelsetsIncident[i]$ induces a clique subgraph in the label set graph $\labelSetGraph[\labelsetsIncident[i]]$. By Lemma~\ref{lemma:cliques-in-tree-decomp}, each such clique subgraph must be contained in one of the tree decomposition's sets $\treeDecompSet \in \treeDecompTreeNodes$. Therefore, each label set is a subset of at least one set in the join tree used to define the running intersection ordering, thus satisfying Property~\ref{def:extractionLabelSetOrdering:item:representative-exists}. The size of the largest set in the join label set ordering is equal to the size of the largest set in the tree decomposition. Hence, $\decompwidth(\treeDecomp) = \max_{L \in \labelsetOrder_i} |L|$ holds.
\end{proof}

The following theorem states that given an extraction order, finding optimal extraction label set orderings is $\complexityNP$-hard. 
\begin{theorem} (Hardness of Finding Optimal Extraction Label Set Orderings) \label{thm:hardness-extraction-label-ordering}\\
Given an extraction order $\reqExtractionOrder$, the problem of finding an extraction label set ordering $\labelsetOrderSet$, minimizing the width $\labelWidthEQ(\reqExtractionOrder, \labelsetOrderSet)$ is $\complexityNP$-hard.
\end{theorem}
\begin{proof} 
We show the hardness by reduction of the Treewidth Decision Problem, which is $\complexityNP$-complete by Theorem~\ref{thm:treewidth-hardness}. 

Given an instance $(\genericGraph, k)$ of the treewidth decision problem, use the construction from Theorem~\ref{lem:existence-EO-arbitrary-labelsets} to generate an extraction order $\reqExtractionOrderDef$, such that the label set graph $\labelSetGraph[\labelsetsIncident[\reqExtractionOrderRoot]]$ for the root node $\reqExtractionOrderRoot$ is the same as $\genericGraph$. Determine an extraction label set ordering $\labelsetOrderSet$ for $\reqExtractionOrder$, minimizing the width $\labelWidthEQ(\reqExtractionOrder, \labelsetOrderSet)$. Then, by Lemma~\ref{lemma:extraction-LSO-is-tree-decomp}, the size of the largest label set in $\labelsetOrder_{\reqExtractionOrderRoot}$ is the treewidth of $\genericGraph$.

\end{proof}

We have thus shown that determining an optimal extraction label set ordering for a \emph{specific} extraction order is $\complexityNP$-hard. However, due to the extensive work related to treewidth and tree decomposition, we expect that extraction label set orderings of a reasonable width can be found in many cases.

In \cite{rostSchmidFPTApproximations}, Rost and Schmid prove that computing extraction orders which are optimal for the base algorithm's extraction width parameter is $\complexityNP$-hard. We therefore expect that the problem of computing an optimal extraction order, while additionally optimizing the extraction label set ordering is hard as well. However, this hardness result does not necessarily apply to the underlying request topology. In particular, extraction orders containing label set graphs with large treewidth might be easily identifiable as sub-optimal choices. It is also conceivable that some connection between extraction label width and the treewidth of the request topology might exist. On an intuitive level, one might conjecture that the representation of a label set graph with a certain treewidth might itself require an underlying request topology of a certain treewidth.

\cleardoublepage
\section{Extending the Base Algorithm to Multi-Root Orientations}\label{sec:multiroot}

The second extension of the base algorithm aims to relax the requirement that the extraction order of the request must be a rooted graph. Rost and Schmid have shown in \cite{rostSchmidFPTApproximations} that for the class of half-wheel graphs (see Definition~\ref{def:half-wheel-graphs}), the placement of the root node has drastic implications for the extraction width of the extraction order (see Lemma~\ref{lem:extraction-width-half-wheel}). In Lemma~\ref{lemma:label-width-of-half-wheel-graphs}, we show a similar result for the extraction label width. It follows that the placement of the root node directly affects the size of the LP formulation and the runtime of the decomposition algorithm. 

Now, consider a request topology containing multiple subgraphs, such that the extraction label width of the subgraphs is affected by the root placement. If extraction orders are required to be rooted graphs, it follows that selecting a root optimally for one subgraph may force a sub-optimal root placement in the other subgraphs. An example of such a request is shown in Figure~\ref{fig:03:half_wheels_double_root}. This request consists of two half-wheel graphs, connected by a single edge. By Lemma~\ref{lemma:label-width-of-half-wheel-graphs}, the optimal root placement has an extraction order $\reqExtractionOrder$ with extraction label width $\labelWidthEQ(\reqExtractionOrder, \labelsetOrderSet) = 2$, where $\labelsetOrderSet$ is a minimal extraction label set ordering for $\reqExtractionOrder$. However, under the non-optimal root placement, only a width of $\labelWidthEQ(\reqExtractionOrder, \labelsetOrderSet) = 3$ can be achieved.

\begin{figure}[tbh]
    \centering
    \includegraphics[width=1.0\textwidth]{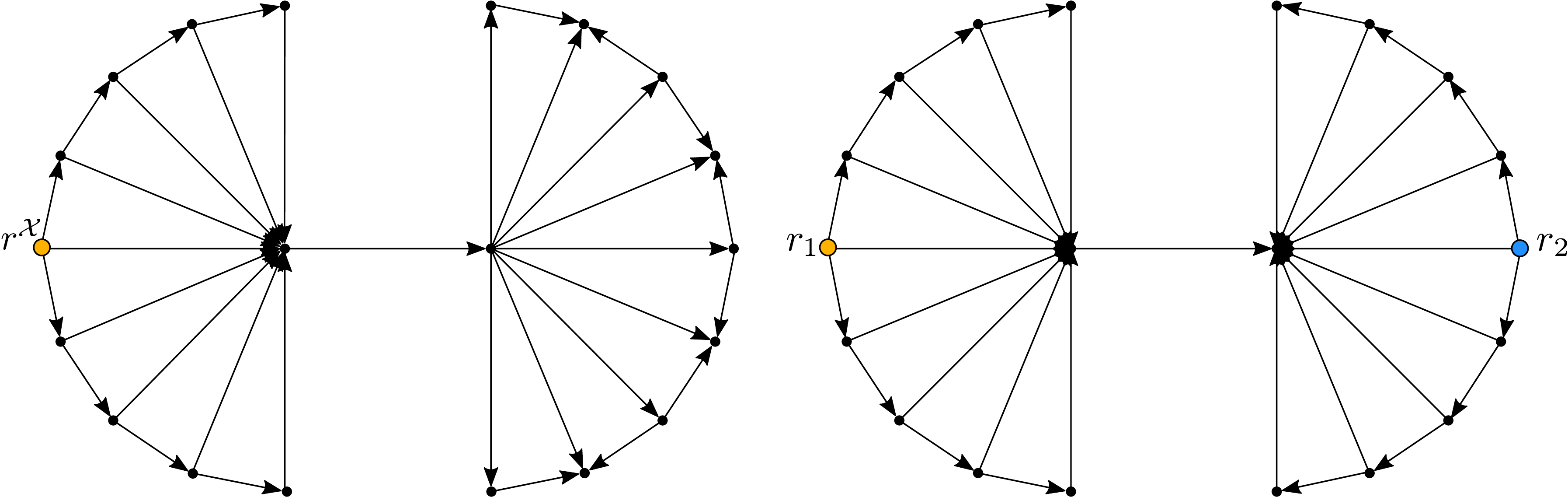}
    \caption{
A request graph motivating the use of multiple root nodes, consisting of two half-wheel graphs (see Definition~\ref{def:half-wheel-graphs}) joined by a single edge.
\newline
Left: A rooted extraction order with the root node $\reqExtractionOrderRoot$ placed optimally in either of the two half-wheel subgraphs forces a suboptimal root placement for the other half-wheel subgraph. The extraction label width using a single root is therefore $3$ according to Lemma~\ref{lemma:label-width-of-half-wheel-graphs}. Note that the extraction width of the base algorithm scales linearly with the size of the half-wheel subgraphs (cf. \cite{rostSchmidFPTApproximations}, Lemma~\ref{lem:extraction-width-half-wheel}).
\newline
Right: Placing an second root node in the right subgraph allows optimal root placement in \emph{both} half-wheel subgraphs. The extraction label width in each half-wheel subgraph is $2$, as each edge now lies on at most one confluence. The improvement for the base algorithm is more drastic, as the extraction width of the base algorithm is also $2$.
	}
    \label{fig:03:half_wheels_double_root}
\end{figure}

Since the size of the LP formulation and the runtime of the decomposition algorithm depend exponentially on the extraction label width (see Theorems~\ref{thm:adapted-alg-lp-size} and~\ref{thm:adapted-alg-runtime}), even reducing the extraction label width parameter by one may have a significant impact on the runtime. It therefore seems useful to extend the algorithm in such a way that a second root node could be placed, allowing the optimal root placement to be chosen in both half-wheel subgraphs. In this section, we introduce such an extension, allowing the decomposition of generalized extraction orders. 

We begin by introducing the basic terminology related to generalized extraction orders in Section~\ref{sec:multiroot:definitions}. In Section~\ref{sec:multiroot:induced-EO}, we describe a simple method of converting a generalized extraction order to a regular extraction order.

In the remainder of Section~\ref{sec:multiroot}, we discuss a more sophisticated algorithm, which results in a lower overall width. This algorithm is based on deriving convex combinations of mappings for rooted subgraphs of the generalized extraction order, and stitching these subgraph mappings together to a solution for the entire request. In Section~\ref{sec:multiroot:tree-like-EOs}, we define the required high-level structure of generalized extraction orders to which the refined algorithm can be applied. Section~\ref{sec:multiroot:rooted-subgraphs} validates the approach of processing subgraphs individually by showing that Algorithm~\ref{alg:decompositionAlgorithm-RIP} can compute convex combinations of mappings for rooted subgraphs. Section~\ref{sec:multiroot:local-subgraph-extensibility} describes how consistent mappings between such subgraphs can be enforced by introducing the notion of \emph{local subgraph extensibility}. In Section~\ref{sec:multiroot:decomposition-alg}, we discuss the decomposition algorithm for generalized extraction orders, and show its correctness. Lastly, in Section~\ref{sec:multiroot:complexity}, we discuss the complexity of the algorithm.

\begin{remark} (Half-wheel Graphs in Diagrams) \\
Many diagrams in this section show request topologies containing multiple half-wheel subgraphs. The use of these half-wheel subgraphs is intended to motivate a specific placement of root nodes: As Rost and Schmid show in \cite{rostSchmidFPTApproximations}, suboptimal root placement in half-wheel graphs leads to an increase in the extraction width which scales linearly in the size of the half-wheel subgraph (see Lemma~\ref{lem:extraction-width-half-wheel}).  Given the results of Section~\ref{sec:hierarchical-bags}, in particular Lemma~\ref{lemma:label-width-of-half-wheel-graphs}, the impact of a bad root placement on the extraction label width is not as drastic and merely increases the extraction label width from $2$ to $3$. Therefore, some figures may show a partitioning of the request which may not be optimal for the extraction label width.
\end{remark}

\subsection{Preliminary Definitions} \label{sec:multiroot:definitions}
We will now introduce the terminology used in the remainder of this section. 
We first define the notion of a \emph{generalized extraction order} as a generalization of the definition of an extraction order given in Definition~\ref{def:extraction-order}.
\begin{definition} (Generalized Extraction Order) \label{def:gen-extraction-order} \\
Given a request topology $\reqTopologyDef$, we call an acyclic graph $\reqDAGOrientationDef$ a (generalized) extraction order of $\reqTopology$, if an edge $e$ is contained in $\reqDAGOrientationEdges$ if either itself or its reversed orientation is contained in $\reqEdges$, i.e. if $\undirected[\reqDAGOrientation] = \undirected[\reqTopology]$.
\end{definition}

\begin{remark}
Throughout Section~\ref{sec:multiroot}, we will primarily consider generalized extraction orders. For simplicity, we will occasionally refer to them simply as extraction orders. When referring to extraction orders as introduced in Definition~\ref{def:extraction-order}, we will explicitly call them \emph{rooted extraction orders}. We also occasionally refer to generalized extraction orders as \emph{multi-root extraction orders} to emphasize this distinction. 
\end{remark}

We next introduce the \emph{root set} containing the \emph{local root nodes}, i.e. the nodes which have no incoming edges in the generalized extraction order.
\begin{definition} (Root Set, Local Root Nodes) \\
Given a generalized extraction order $\reqDAGOrientation$ of a request graph, we define the root set as
\begin{align}
\rootSet := \{r | r \in \reqNodes \text{ with } |\inEdgesExtractionOrder{r}| = 0 \}.
\end{align}
We refer to the nodes in $\rootSet$ as \emph{local root nodes}.
\end{definition}

We next define \emph{root regions} as a partitioning of the multi-root extraction order into a set of subgraphs, such that these subgraphs' edge sets are disjoint. Each root region is rooted in one of the extraction order's local root nodes.
\begin{definition} (Root Regions)\\
Let $\reqDAGOrientation$ be a generalized extraction order with root set $\rootSet$. We define the root region set $\rootRegionSet := \{ \rootRegion[r] ~~|~~ r \in \rootSet\}$ as a partition of the edge set $\reqDAGOrientationEdges$.

We require that the root regions satisfy the following properties:
\begin{enumerate}
	\item The root regions cover the entire extraction order, i.e. $\bigcup_r \rootRegion[r] = \reqDAGOrientationEdges$ 
	\item Root regions are pairwise disjoint, i.e. for $r, r' \in \rootSet$ with $r \neq r'$, it holds that $\rootRegion[r] \cap \rootRegion[r'] = \emptyset$.  
	\item For each local root $r \in \rootSet$ and its corresponding root region $\rootRegion[r] \in \rootRegionSet$, any edge $e \in \rootRegion[r]$ is reachable from $r$.
	\item For any request node, all out-going edges lie in the same root region, i.e. for each $i \in \reqNodes$ there exists a local root node $r \in \rootSet$, such that $\outEdgesDAGOrder{i} \subseteq \rootRegion[r]$.
\end{enumerate}
Lastly, for each root $r \in \rootSet$, we denote the root region's node set as \mbox{$\rootRegionNodes[r] :=  \bigcup_{(i, j) \in \rootRegion[r]} \{i, j\}$} and define the root region subgraph as the rooted graph
\begin{align}
\rootRegionSubgraph[r] := (\rootRegionNodes[r], \rootRegion[r], r).
\end{align}
\end{definition}


Since only the edge sets of different root region subgraphs are defined to be disjoint, their node sets may overlap. We therefore now introduce the notions of \emph{root region boundaries} and \emph{pairwise root region boundaries}.
\begin{definition} (Root Region Boundary)\\
Given an extraction order $\reqDAGOrientation$ and some root node $r \in \rootSet$, we define the root region boundary $\rootRegionBoundary[r]$ as its intersection with all other root regions:
\begin{align}
\rootRegionBoundary[r] := \rootRegionNodes[r] \cap \left(\bigcup_{r' \in \rootSet \setminus \{r\}} \rootRegionNodes[r'] \right)
\end{align}
We additionally define \emph{pairwise} boundary regions between two root nodes $r_1 \in \rootSet$ and $r_2 \in \rootSet$ as
\begin{align}
\rootRegionBoundaryPair[r_1][r_2] := \rootRegionBoundary[r_1] \cap \rootRegionBoundary[r_2]
\end{align}
\end{definition}

To describe the large-scale structure of a generalized extraction order, we define the \emph{root region extraction order}. A root region extraction order is a rooted, directed graph containing a node representing each root region, in which root regions are connected by an edge if they share a non-empty boundary region.
\begin{definition}(Root Region Extraction Order) \label{def:root-region-extraction-order} \\
Let $\reqDAGOrientationDef$ be a generalized extraction order with root set $\rootSet$. A root region extraction order $\RRExtractionOrderDef$ is a directed acyclic graph rooted in $\RRExtractionOrderRoot \in \rootSet$, whose nodes are the roots of the extraction order $\reqDAGOrientation$, and whose edges $\RRExtractionOrderEdges$ are defined as follows:
For all $r, r' \in \rootSet$ with $\rootRegionBoundaryPair[r][r'] \neq \emptyset$, either $(r, r') \in \RRExtractionOrderEdges$ or $(r', r) \in \RRExtractionOrderEdges$. Given a root node $r$, we denote the in- and out-edges of $r$ in the root region extraction order as $\inEdgesRR{r}$ and $\outEdgesRR{r}$, respectively.
\end{definition}
The root region extraction order can also be thought of as a rooted directed acyclic orientation of the \emph{intersection graph} (see Definition~\ref{def:intersection-graph}) of the root regions' node sets. Examples of root region extraction orders are shown in Figures~\ref{fig:03:many-root-regions-tree-example} and~\ref{fig:03:many-root-regions-cycle-example}. Tree-like root region extraction orders will be of particular interest for the main algorithm presented in Section~\ref{sec:multiroot:decomposition-alg}.

\begin{figure}[htb]
	\centering
	\includegraphics[width=0.5\textwidth]{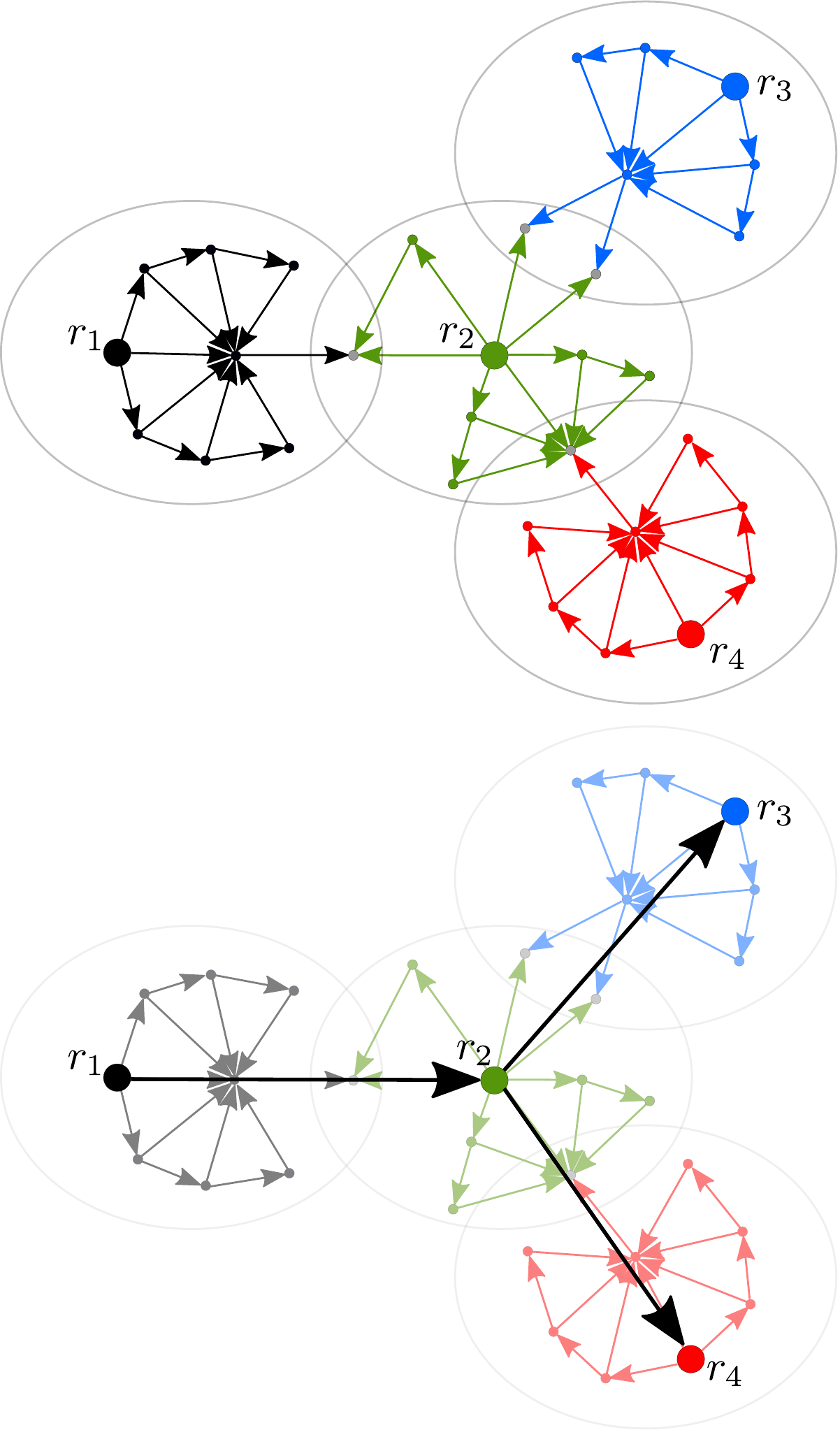}
	\caption{
		A generalized extraction order with several root regions. \newline
		Top: The generalized extraction order. Root regions are enclosed in gray ellipses. \newline
		Bottom: A root region extraction order rooted in $r_1$. Note that the root region extraction order is tree-like.
	}
	\label{fig:03:many-root-regions-tree-example}
\end{figure}

\begin{figure}[p]
	\centering
	\includegraphics[width=0.6\textwidth]{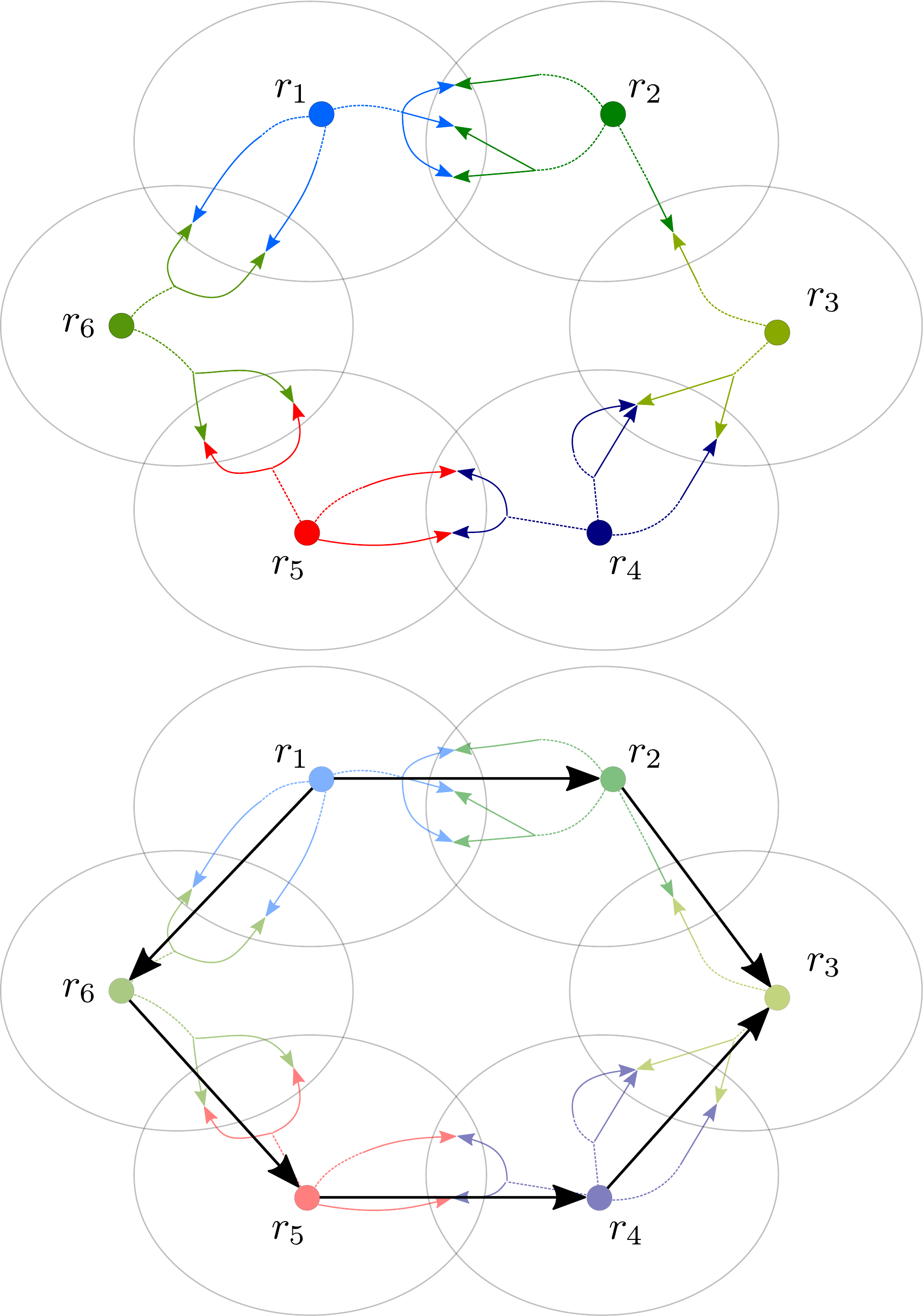}
	\caption{
		A generalized extraction order with several root regions, whose root region extraction order is not tree-like. Dashed edges indicate the existence of a path, i.e. the root regions may also contain half-wheel subgraphs, necessitating the introduction of the local roots. \newline
		Top: The generalized extraction order. Root regions are enclosed in gray ellipses. \newline
		Bottom: A root region extraction order for the request containing a confluence.
	}
	\label{fig:03:many-root-regions-cycle-example}
\end{figure}
\clearpage

\subsection{Induced Extraction Orders}\label{sec:multiroot:induced-EO}

In this section, we propose a simple and universally applicable approach for multi-root extraction orders. This approach relies on a transformation to a rooted extraction order, which can then be handled within the framework of the existing approximation algorithm.

The extraction order is extended by introducing a super-root node $\superroot$ and edges connecting this super-root to each of the local root nodes. We ensure that these additional \emph{virtual} node and edges do not affect the space of possible mappings for the request: The request's demand function is extended by setting $\reqDemand (\superroot) = 0$ for the super-root node and $\reqDemand(e) = 0$ for every virtual edge. By setting their demands to zero, we ensure that no additional substrate resources are allocated when embedding this extended extraction order.

We define the transformation of a generalized extraction order described above as the \emph{induced extraction order}.
\begin{definition} (Induced Extraction Order) \label{def:multiroot:induced-extraction-order}\\
Given a generalized extraction order $\reqDAGOrientation$ with root set $\rootSet$, we define the induced extraction order of $\reqDAGOrientation$ as 
\begin{align}
\reqInducedEO[\reqDAGOrientation] := \left(\reqNodes \cup \{\superroot\}, \reqDAGOrientationEdges \cup E', \superroot \right) \quad \text{where} \quad
E' = \left\{(\superroot, r) ~|~ r \in \rootSet \right\}
\end{align}
We also extend the resource demand function such that $\reqDemand(\superroot) = 0$ and  $\reqDemand(e') = 0$ for all edges in $E'$.
\end{definition} 

Note that the induced extraction order is always a rooted extraction order whose root node is $\superroot$. Further, it requires the same resource allocations as embedding the corresponding request, since no resources are allocated to the superroot node  and the virtual edges. It can therefore directly be used as a rooted extraction order, without further modification to the overall algorithm. Figure~\ref{fig:03:half_wheels_connected_good_root_simple_superroot} shows an example of an induced extraction order, based on the example request of two half-wheel graphs introduced in Figure~\ref{fig:03:half_wheels_double_root}.

\begin{figure}[h]
	\centering
	\includegraphics[width=0.7\textwidth]{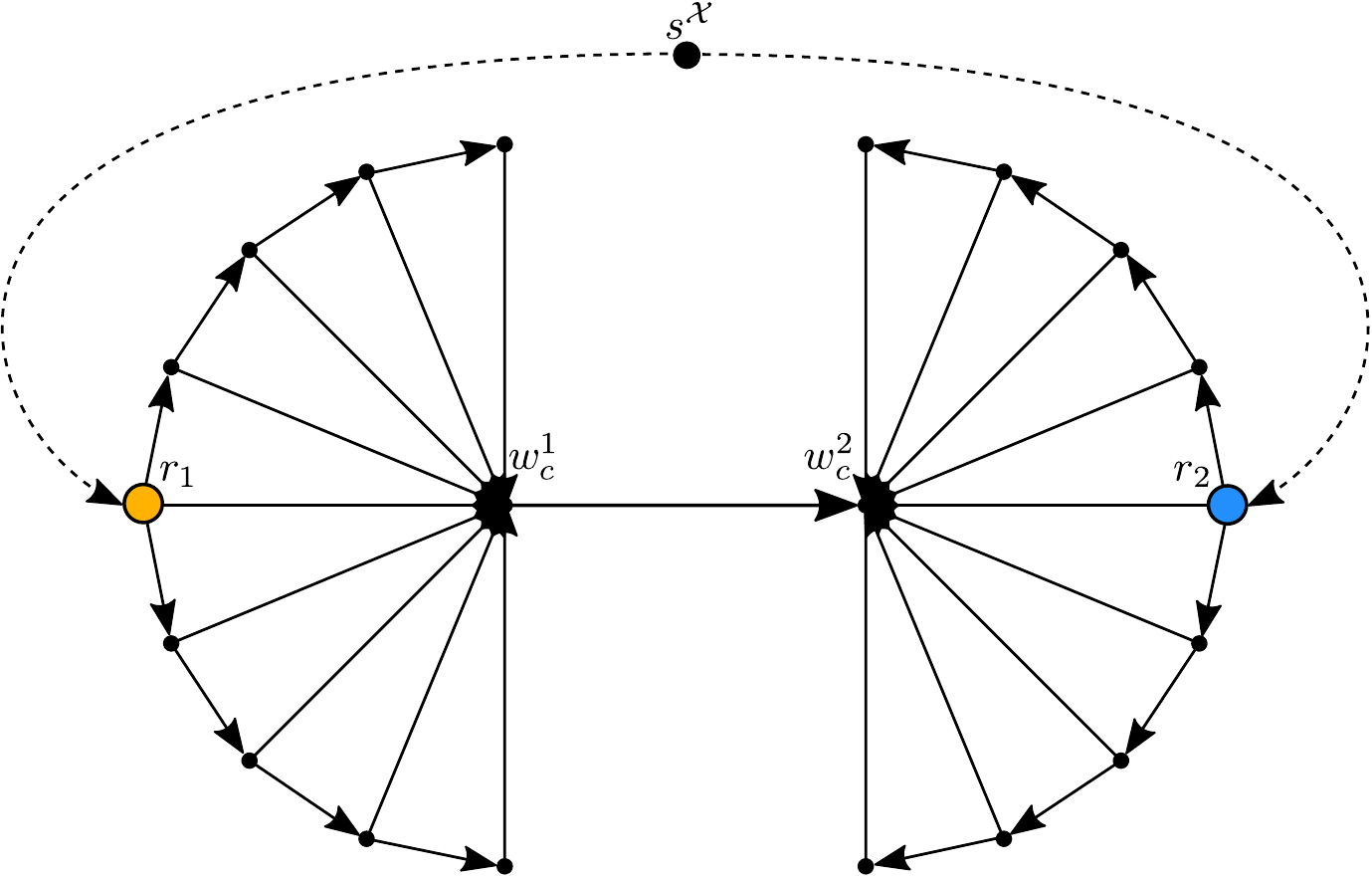}
	\caption{
An illustration of how to transform a generalized extraction order to a rooted extraction order using the example introduced in Figure~\ref{fig:03:half_wheels_double_root}. The transformation is achieved through the introduction of a virtual super-root node connected to each local root node.  Virtual edges connecting $\superroot$ to the local root nodes are dashed. A mapping for the induced extraction order would not allocate any substrate resources to the super-root node or any of the virtual edges, as their demand is set to zero. In this example, the induced extraction order still has extraction label width $3$, as any edge in the left half-wheel graph is labeled with $\{w_c^1, w_c^2\}$. Note that the \emph{classical} extraction width of the base algorithm is also $3$ and therefore drastically improved.
	}
	\label{fig:03:half_wheels_connected_good_root_simple_superroot}
\end{figure}

We now show the validity of using the induced extraction order to determine a mapping for the original multi-root extraction order. First, we show that any valid mapping for the original extraction order can be extended to a valid mapping of its induced extraction order.
\begin{lemma}(Existence of Mappings for Induced Extraction Orders) \label{lemma:multiroot:existence-mapping-induced-order}\\
Given a generalized extraction order $\reqDAGOrientationDef$ and its induced extraction order $\reqInducedEO[\reqDAGOrientation]$, and given a valid mapping $\mappingRequest = (\mappingNodes, \mappingEdges)$ for $\reqDAGOrientation$ with resource allocations bounded by substrate capacities, then there exists a valid mapping $\mappingRequestTilde := (\mappingNodesTilde, \mappingEdgesTilde)$ for $\reqInducedEO[\reqDAGOrientation]$, such that $\restrict[\mappingNodesTilde][\reqNodes] = \mappingNodes$ and $\restrict[\mappingEdgesTilde][\reqDAGOrientationEdges] = \mappingEdges$, and such that the resource allocations of $\mappingNodesTilde$ are bounded by substrate capacities.
\end{lemma}
\begin{proof}
Consider the following extension of the mapping $\mappingRequest$, where $E'$ is the set of virtual edges in the induced extraction order, $u \in \substrateNodes$ is an arbitrary substrate node and $\path[u][v] \subseteq \substrateEdges$ is a path connecting substrate nodes $u$ and $v$:
\begin{align}
\mappingNodesTilde &:= \mappingNodes \cup \{ \superroot \mapsto u \} \\
\mappingEdgesTilde &:= \mappingEdges \cup \{ e \mapsto \path[u][\mappingNodes(r)] | e := (\superroot, r) \in E' \} 
\end{align}
The properties $\restrict[\mappingNodesTilde][\reqNodes] = \mappingNodes$ and $\restrict[\mappingEdgesTilde][\reqDAGOrientationEdges] = \mappingEdges$ are satisfied by this definition of $\mappingRequestTilde$. The validity of the mapping for each node and edge in the general extraction order $\reqDAGOrientation$ follows from the validity of the original mapping $\mappingRequest$. Further, the virtual super-root node $\superroot$ is mapped to a single node, and each virtual edge is mapped to a path connecting the mapped locations of both of its end points.
Lastly, by definition of the induced extraction order, the virtual resources have zero resource demand and therefore, since the original mapping respected the resource capacities, all substrate capacities are respected by $\mappingRequestTilde$ as well.
\end{proof}
From Lemma~\ref{lemma:multiroot:existence-mapping-induced-order}, it follows that  whenever a mapping for the original request exists, the VNEP approximation algorithms discussed so far can find a mapping for the induced extraction order.
Conversely, we next show that any mapping of the induced extraction order can be converted to a valid mapping of the original multi-root extraction order. 
\begin{lemma}(Validity of Mappings of Induced Extraction Orders) \\
Given a generalized extraction order $\reqDAGOrientationDef$ and its induced extraction order $\reqInducedEO[\reqDAGOrientation]$, and given a valid mapping $\mappingRequestTilde = (\mappingNodesTilde, \mappingEdgesTilde)$ for the  induced extraction order $\reqInducedEO[\reqDAGOrientation]$, then the mapping $\mappingRequest$ defined by $\mappingRequest := (\restrict[\mappingNodesTilde][\reqNodes], \restrict[\mappingEdgesTilde][\reqDAGOrientationEdges])$ is a valid mapping for $\reqDAGOrientation$.
\end{lemma}
\begin{proof}
By definition, the induced extraction order's node and edge sets are supersets of the generalized extraction order's node and edge sets, respectively. Therefore, every request node and edge is mapped by $\mappingRequestTilde$, and it follows by the definition of $\mappingRequest$ that these mappings are also included in $\mappingRequest$. Therefore, each node and edge in the original is mapped in $\mappingRequest$. 

The validity of $\mappingRequest$ then follows from the assumption that $\mappingRequestTilde$ is valid: Validity of the node and edge mappings follows directly, and the validity of the resource allocations follows, since no allocations are added in $\mappingRequest$ relative to $\mappingRequestTilde$.
\end{proof}

We now consider the size of the resulting LP formulation when using the induced extraction order. The root regions of the generalized extraction order intersect with other root regions in their boundary node sets $\rootRegionBoundary$. Consider a pairwise root region boundary $\rootRegionBoundaryPair[r][r']$ between two roots $r$ and $r'$. By definition of $\rootRegionBoundaryPair[r][r']$, each node $b \in \rootRegionBoundaryPair[r][r']$ is reachable from both $r$ and $r'$. Further, in the induced extraction order, both $r$ and $r'$ are reachable from the super-root node $\superroot$. Therefore, the root region boundary nodes now form the end nodes of confluences whose start node is $\superroot$. The virtual edges connecting $\superroot$ to the local root nodes lie on these confluences. The size of their label sets is therefore given by the size of the corresponding local root node's boundary region:
\begin{align}
|\labelsetEdge[(\superroot, r)]| = |\rootRegionBoundary[r]| && \forall r \in \rootSet
\end{align}

\subsection{Tree-Like Root Region Extraction Orders}\label{sec:multiroot:tree-like-EOs}

In the remainder of Section~\ref{sec:multiroot}, we introduce a more sophisticated approach which results in a smaller width than the use of the induced extraction order discussed in Section~\ref{sec:multiroot:induced-EO}. We propose an alternative decomposition algorithm, which processes each root region separately and assembles a mapping for the entire request by incrementally extending a partial mapping to each root region.

However, this approach is not applicable for \emph{all} generalized extraction orders. In Definition~\ref{def:root-region-extraction-order}, we introduced the notion of a root-region extraction order as a construct to describe the high-level structure of a multi-root extraction order. The more refined approach requires \emph{tree-like root region extraction orders}, i.e. root region extraction orders, where the underlying undirected graph is a tree.
\begin{definition} (Tree-Like Root Region Extraction Order) \label{def:multiroot:tree-like-RR-EO}\\
Given a generalized extraction order $\reqDAGOrientationDef$, we call its root region extraction order $\RRExtractionOrderDef$ tree-like if $\undirected[\RRExtractionOrder]$, i.e. the undirected version of $\RRExtractionOrder$, is a tree.
\end{definition}
Examples of extraction orders with tree-like and non-tree-like root region extraction orders are shown in Figures~\ref{fig:03:many-root-regions-tree-example} and~\ref{fig:03:many-root-regions-cycle-example}, respectively.

We introduce the notion of an \emph{incoming root region boundary}, assigning to each root region the union of the pairwise boundaries with all parent root regions according to some root region extraction order.
\begin{definition} (Incoming Root Region Boundary) \label{def:incoming-RR-boundary} \\
Let $\reqDAGOrientationDef$ be a generalized extraction order with root set $\rootSet$ and a root region extraction order $\RRExtractionOrderDef$. Given a local root node $r \in \rootSet$, we  define the incoming root region boundary as
\begin{align}
\rootRegionBoundaryIncoming[r] := \begin{cases}
\emptyset & \text{ if } r = \RRExtractionOrderRoot \\
\bigcup_{r' \in \inEdgesRR{r}} \rootRegionBoundaryPair[r'][r] & \text{otherwise.} 
\end{cases}
\end{align}
\end{definition}
Note that in a tree-like root region extraction order, each root region has at most a single in-neighbor in the root region extraction order, and only the root region associated with $\RRExtractionOrderRoot$ has none. The incoming root region boundary can therefore directly be identified with the pairwise boundary between a root region and its parent root region according to the root region extraction order.

\begin{remark} (Restriction of Edges between Boundary Nodes) \label{remark:multiroot-boundary-edges-restriction} \\
Beyond the restriction of only allowing tree-like root region extraction orders, we will further assume that there are no edges between any of the boundary nodes, i.e. for any pair of boundary nodes $b_1, b_2 \in \rootRegionBoundaryPair[r][r']$, it holds that $(b_1, b_2) \neq \reqDAGOrientationEdges$. This restriction can be relaxed to allow for any set of edges that can be extended to a chain, i.e. only forbidding orientations containing edges such as $\{(b_1, b_2), (b_3, b_2)\} \subset \rootRegion[r]$. Due to space constraints, developing such an extension in detail is outside of the scope of this thesis.
We therefore only consider extraction orders containing no edges within the root regions' boundary sets. Where applicable, we describe how the approach might be extended to allow for some edges between boundary nodes.
\end{remark}

\begin{remark} (Generalization to Non-Tree-Like Root Region Extraction Orders) \label{remark:multiroot:generalizing-non-treelike}\\
While Algorithm~\ref{alg:decompositionAlgorithm-multiroot} is only applicable to extraction orders with tree-like root region extraction orders, it can be combined with the approach using induced extraction orders described in Section~\ref{sec:multiroot:induced-EO} by repeatedly proceeding as follows:
\begin{enumerate}
\item Identify a confluence in the root region extraction order.
\item Define an induced extraction order for all root regions lying on the confluence.
\item Substitute this induced extraction order for all contained root regions.
\end{enumerate}
In the following, we assume that these steps have been taken and that the resulting root region extraction order is tree-like.
\end{remark}
An example of the generalization approach described in Remark~\ref{remark:multiroot:generalizing-non-treelike} is shown in Figure~\ref{fig:03:many-root-regions-hybrid-example}, where an induced extraction order is substituted for the root regions which lie on a cycle in the root region extraction order.

\begin{figure}[tbph]
	\centering
	\includegraphics[width=0.7\textwidth]{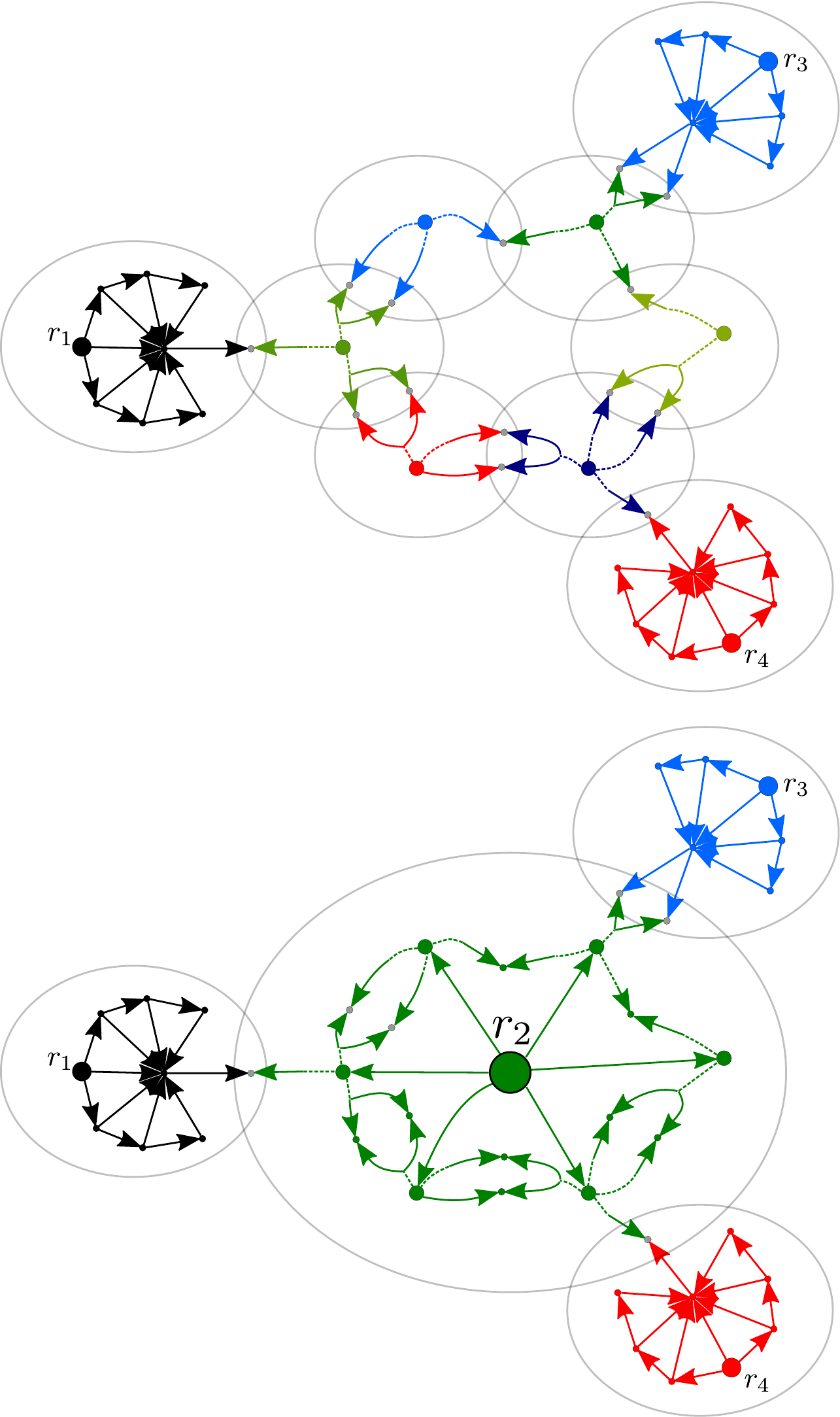}
	\caption{
A generalized extraction order illustrating the procedure described in Remark~\ref{remark:multiroot:generalizing-non-treelike}. Dashed edges indicate the existence of a path, i.e. the root regions may contain additional nodes and edges necessitating the introduction of the local root nodes. \newline
Top: The original extraction containing a root region cycle. \newline
Bottom: Substituting the induced extraction order (see Definition~\ref{def:multiroot:induced-extraction-order}) for the cyclical root regions as described in Remark~\ref{remark:multiroot:generalizing-non-treelike} results in a tree-like root region extraction order. 
	}
	\label{fig:03:many-root-regions-hybrid-example}
\end{figure}

\subsection{Decomposition of Individual Root Regions} \label{sec:multiroot:rooted-subgraphs}
In this subsection, we validate the approach of performing the decomposition for each root region separately. Given some local root node $r\in \rootSet$, the corresponding root region subgraph $\rootRegionSubgraph[r]$ is by definition a rooted, acyclic graph. The root region can therefore be interpreted as a rooted extraction order according to Definition~\ref{def:extraction-order} for the corresponding subgraph of the request. We will first extend the definitions of decomposable edge label assignment (see Definition~\ref{def:decomposable-edge-labels}) and extraction label set orderings (see Definition~\ref{def:extractionLabelSetOrdering}) to the multi-root setting. Secondly, in Lemma~\ref{lemma:decomposition-rooted-subgraphs}, we give the main result of this subsection, namely the decomposability of rooted subgraphs.

We first consider the edge label assignment. At the boundary between two root regions, the current definition of decomposable label assignments from Definition~\ref{def:decomposable-edge-labels}, which is only intended for rooted extraction orders, breaks down. At a boundary node, incoming edges may originate in different root regions, and these edges may have different label assignments. We therefore introduce the notion of \emph{multi-root decomposable edge label assignments}, generalizing the original concept. We require that the usual definition of a decomposable label assignment holds within each root region, but allow for a deviation at the root region boundary nodes.
\begin{definition} (Multi-Root Decomposable Edge Label Assignment) \label{def:03:mr-decomp-label-assignment} \\
Let $\reqDAGOrientationDef$ be a generalized extraction order with root set $\rootSet$.
A label assignment $\labelsetsEdges$ for the edges in $\reqDAGOrientation$ is multi-root decomposable if for each local root node $r \in \rootSet$, the label assignment within the root region subgraph $\rootRegionSubgraph[r]$ is a decomposable label assignment according to Definition~\ref{def:decomposable-edge-labels}.
\end{definition}

Note that the use of a multi-root decomposable edge label assignment by itself does not guarantee that a decomposition of multiple intersecting root regions is possible in a consistent way. Rather, we introduce Definition~\ref{def:03:mr-decomp-label-assignment} as a means to describe a label assignment extending over multiple root regions. Figure~\ref{fig:03:mr-decomposable-label-assignment-example} shows a simple example of a multi-root extraction order highlighting the key difference between a multi-root decomposable label assignment and a decomposable label assignment. Edges belonging to different root regions may be differently labeled, and the incoming label set of a node on the boundary between two root regions is no longer unique. Note that, as required, the multi-root decomposable label assignment still defines a decomposable label assignment according to Definition~\ref{def:decomposable-edge-labels} \emph{within each root region}.

\begin{figure}[tbh]
	\centering
	\includegraphics[width=0.2\textwidth]{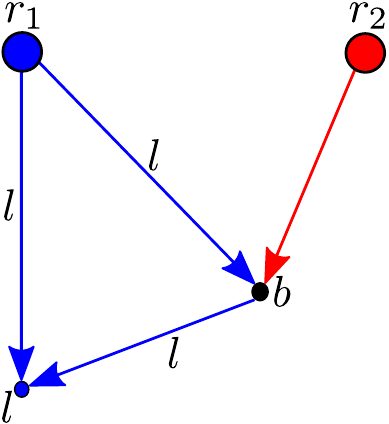}
	\caption{
		An example showing how the multi-root decomposable label assignment introduced in Definition~\ref{def:03:mr-decomp-label-assignment} differs from the original notion of decomposable edge labels from Definition~\ref{def:decomposable-edge-labels}. The extraction order consists of two root regions, rooted in $r_1$ and $r_2$. Note that the incoming edges to node $b$ do not have the same edge label set: The edge $(r_1, b)$ has the label set $\labelsetEdge[(r_1, b)] = \{l\}$, but $(r_2, b)$ is labeled $\labelsetEdge[(r_2, b)] = \emptyset$.
	}
	\label{fig:03:mr-decomposable-label-assignment-example}
\end{figure}

In the next definition, we similarly generalize the notion of an extraction label set ordering (see Definition~\ref{def:extractionLabelSetOrdering}) by introducing the \emph{multi-root extraction label set ordering}, which assigns an extraction label set ordering to each root region subgraph.
\begin{definition} (Multi-root Extraction Label Set Ordering) \label{def:MR-extraction-label-set-ordering}\\
Given a generalized extraction order $\reqDAGOrientationDef$ with a multi-root decomposable edge label assignment $\labelsetsEdges$, we define the multi-root extraction label set ordering $\labelsetOrderSetMR$ by assigning an extraction label set ordering $\labelsetOrderSetRegion[r]$ to each root region of $\reqDAGOrientation$. That is, given the root set $\rootSet = \{r_1, ..., r_n\}$, we define
\begin{align}
\labelsetOrderSetMR := \{ \labelsetOrderSetRegion[r_1], \ldots, \labelsetOrderSetRegion[r_n] \},
\end{align}
where each root node $r_i \in \rootSet$ is assigned an extraction label set ordering $\labelsetOrderSetRegion[r_i] \in \labelsetOrderSetMR$ for the corresponding root region subgraph $\rootRegionSubgraph[r_i]$.
\end{definition}

We now show that Algorithm~\ref{alg:decompositionAlgorithm-RIP} can derive a convex combination of mappings for a single root region of the extraction order, when given a LP solution for the full, generalized extraction order.

\begin{lemma} (Decomposition of Individual Root Regions) \label{lemma:decomposition-rooted-subgraphs}\\ 
Given a VNEP instance $\VNEPInstance$, a generalized extraction order $\reqDAGOrientationDef$, a local root node $r \in \rootSet$ and a solution $\lpvars$ of Linear Program~\ref{LP:RunningIntersectionProperty} constructed over $\reqDAGOrientation$ with a multi-root decomposable edge label assignment $\labelsetsEdges$ and a multi-root extraction label set ordering $\labelsetOrderSetMR$, Algorithm~\ref{alg:decompositionAlgorithm-RIP} may be used to obtain a complete convex combination of valid mappings for the root region subgraph $\rootRegionSubgraphDef[r]$, when given the input  $((\reqTopologyRootRegionSubgraph[r], \substrateTopology), \rootRegionSubgraph[r], \labelsetOrderSetRegion[r], \lpvars)$. Additionally, the resource allocations of the resulting convex combination of mappings are bounded by the substrate capacities.
\end{lemma}
\begin{proof}
By definition, $\rootRegionSubgraph[r]$ is a rooted subgraph of $\reqDAGOrientation$. Since $\labelsetsEdges$ is a multi-root decomposable edge label assignment, the edge label assignment obtained by restricting $\labelsetsEdges$ to the root region's edges $\rootRegion[r]$ is by Definition~\ref{def:03:mr-decomp-label-assignment} a decomposable label assignment according to Definition~\ref{def:decomposable-edge-labels}. Further, $\labelsetOrderSetRegion[r]$ is an extraction label set ordering for $\rootRegionSubgraph[r]$ by Definition~\ref{def:MR-extraction-label-set-ordering}. It follows that the variable set $\lpvars$ for the full extraction order is a superset of the variables which would be present in the LP formulation constructed only for the root region $\rootRegionSubgraph[r]$ with edge label assignment $\labelsetsEdges$ and extraction label set ordering $\labelsetOrderSetRegion[r]$.
Now consider an execution of Algorithm~\ref{alg:decompositionAlgorithm-RIP}, where instead of the full LP solution for the entire multi-root extraction order, only the subset of LP variables which are related to the root region subgraph are passed to the algorithm. From the correctness proof of Algorithm~\ref{alg:decompositionAlgorithm-RIP} in Theorem~\ref{thm:decomp_of_rip_orderings}, it follows that the decomposition succeeds.

We have now shown that a subset of the LP variables can be converted to a convex combination of mappings. We next show that the inclusion of the variables related to other root regions does not impact the execution of the decomposition algorithm.

The execution of Algorithm~\ref{alg:decompositionAlgorithm-RIP} mainly consists of the while-loop starting in Line~\ref{algline:rip-decomp:begin-while-flow}. As discussed in the proof of Theorem~\ref{thm:decomp_of_rip_orderings}, the value of the variable $x$ matches the sum of the node mapping variables $y^u_r$ for the root node, where the sum is taken over each possible mapping location for $r$. It follows that the loop is executed until the sum of root node mapping variables is zero, regardless of any additional variables. 
Consider a single iteration of this outer while-loop. After selecting a mapping for the root node $\superroot$, i.e. the root region's local root node $r$, the while-loop starting in Line~\ref{algline:rip-decomp:begin-while-q} performs a top-down traversal of the rooted extraction order $\reqExtractionOrder$, i.e. the root region subgraph $\rootRegionSubgraph[r]$. Within the body of the loop, only nodes that lie in the root region are added to the queue in Line~\ref{algline:rip-decomp:add-j-to-q}, and out-edges of a node are only included in the iteration starting in Line~\ref{algline:rip-decomp:edge-in-label-set-loop}, if they are contained in the root region's edge set $\rootRegion[r]$. 
Therefore, in its traversal of the extraction order, the decomposition algorithm only accesses LP variables that are directly related to the mapping of nodes and edges contained in the root region.

Finally, the minimal flow value $\prob$ is subtracted from a set of variables in Lines~\ref{algline:rip-decomp:adapt-load-one} and following. These variables are explicitly selected in Line~\ref{algline:rip-decomp:compute-Vk}, such that only variables related to the root region are included. The allocation variables $\vec{a}$ are the only exception, since they also account for resource allocations for the other root regions' nodes and edges. However, the allocation variables do not affect the execution of the decomposition algorithm in any way. Indeed, the resource allocations of the resulting mapping for the root region $\rootRegionSubgraph[r]$ are bounded by the resource allocations of the full extraction order, which are reflected in the allocation variables' initial values. Therefore, the resource allocations from the resulting convex combination of mappings for the single root region do not exceed any substrate capacities.

It follows that the algorithm extracts mappings for all nodes and edges contained in the root region and that the variables related to other root regions are ignored by the algorithm, concluding the proof.
\end{proof}

This result guarantees that for a single root region, a convex combination of mappings can be extracted. However, note that after executing Algorithm~\ref{alg:decompositionAlgorithm-RIP} on one root region, decomposition of any intersecting root region is impossible, since the node mapping variables of the boundary nodes are reduced to zero during the decomposition of the first root region in Line~\ref{algline:rip-decomp:adapt-variables-one} of Algorithm~\ref{alg:decompositionAlgorithm-RIP}. This issue will be resolved in the decomposition algorithm presented in Section~\ref{sec:multiroot:decomposition-alg} by resetting the variables of the LP to their original values for each decomposition of a root region subgraph.

\subsection{Local Subgraph Extensibility} \label{sec:multiroot:local-subgraph-extensibility}

We have now shown that the existing decomposition algorithm can extract a solution for a single root region of a multi-root extraction order. However, in order to extend a valid mapping for one root region with the valid mapping of a second root region, we must ensure that the boundary nodes contained in both root regions are mapped consistently in both mappings. In particular, to combine two mappings of two intersecting root regions, we require that their boundary nodes are mapped to the same substrate nodes in both mappings. That is, for two local root nodes $r_1, r_2 \in \rootSet$, and two mappings $m_1$ and $m_2$ for $\rootRegionSubgraph[r_1]$ and $\rootRegionSubgraph[r_2]$, respectively, we require $\restrict[m_1][\rootRegionBoundaryPair[r_1][r_2]] = \restrict[m_2][\rootRegionBoundaryPair[r_1][r_2]]$ in order to combine $m_1$ and $m_2$. We call this property \emph{local subgraph extensibility}.

In this section, we discuss how such consistent node mappings of root region boundaries can be guaranteed. One mechanism to enforce consistent node mappings is given by the local connectivity property (see Lemma~\ref{lem:local-connectivity-property}). However, as discussed in Section~\ref{sec:02:limitations-of-mcf} in the context of the base algorithm, the local connectivity property is insufficient when consistency of node mappings must be enforced \emph{non-locally}. In Section~\ref{sec:02:the-base-vnep-alg}, the notion of labeling edges and duplicating the LP subformulations was introduced in the framework of the base algorithm in order to enforce consistent mappings of confluence start and end nodes non-locally.

In this section, we will introduce a method of extending the edge label assignment in order to enforce consistent boundary node mappings. We will first briefly describe the construction, discussing root region boundaries with one, two, and $n$ nodes. A proof that the approach guarantees existence of compatible mappings in both root regions is the main result of this subsection, Theorem~\ref{thm:local-subgraph-extensibility}. 
The extension of the edge label assignment for two and $n$ boundary nodes is demonstrated in Figure~\ref{fig:03:refined-construction-2-regions-example}.

\begin{figure}[tbph]
	\centering
	\includegraphics[width=1.0\textwidth]{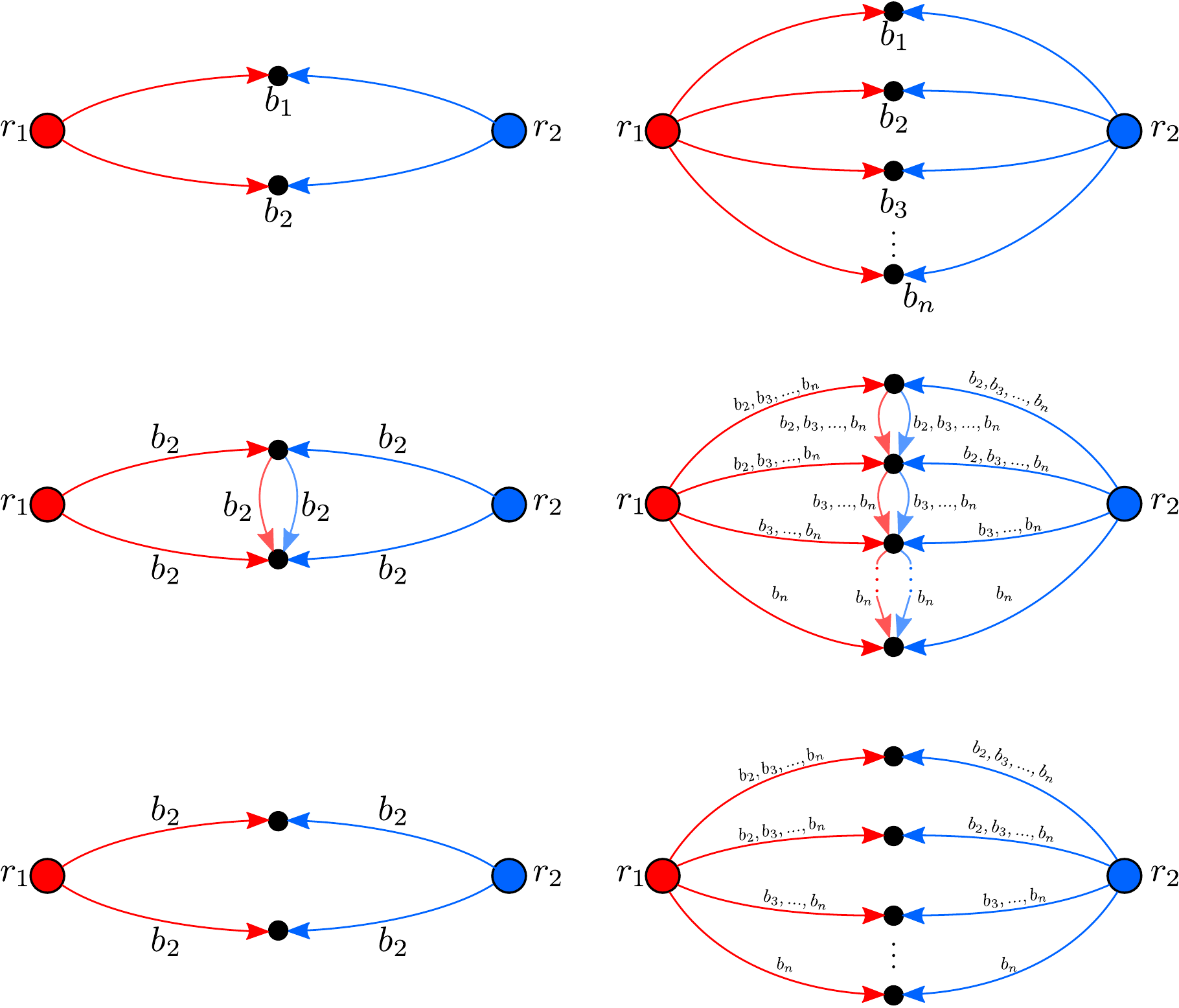}
	\caption{
The multi-root confluence edge label assignment for extraction orders with two root regions. Virtual edges are added in a chain connecting the boundary nodes between the root regions. The label assignment is generated for each extended root region, and then applied to the original extraction order. \newline
Left column: Two boundary nodes. Right column: Arbitrary number of boundary nodes. \newline
Top row: The generalized extraction orders, each consisting of two root regions, shown red and blue. \newline 
Center row: Virtual edges between boundary nodes are added to each root region. The virtual edges introduce new confluences involving the boundary nodes. \newline 
Bottom row: The virtual edges are removed, and we obtain the original extraction order with the multi-root confluence edge label assignment.
	}
	\label{fig:03:refined-construction-2-regions-example}
\end{figure}

\begin{remark} \label{remark:why-only-pairwise-root-regions}
For simplicity of notation, we will only state the results of this subsection for the case where the extraction order contains two root regions, such that $\rootSet = \{r_1, r_2\}$. We will usually assume that a mapping for the first root region $\rootRegionSubgraph[r_1]$ is given and that it has to be extended in a consistent way to a mapping for both root regions. Due to the requirement of tree-like root region extraction orders, this procedure of combining pairs root region mappings can be iteratively extended to the entire extraction order, since each root region intersects with exactly one parent root region according to the root region extraction order.
\end{remark}

\paragraph{One Boundary Node}
We first consider the simple case where the root regions share a single boundary node $b$. The extraction order shown in Figure~\ref{fig:03:half_wheels_connected_good_root_simple_superroot}, with two half-wheel graphs connected by a single edge, is an example of this case. There is only a single node, namely $b$, for which the mappings of the two root region subgraphs must agree.  Since this constraint can be enforced locally, flow conservation ensures that both mappings are consistent and no additional labels are required. Note that using the induced extraction order would have introduced the additional edge label $b$, resulting in a larger LP-formulation. 

\paragraph{Two Boundary Nodes}
Next, consider the case with are two boundary nodes, i.e. let $\rootRegionBoundaryPair = \{b_1, b_2\}$. This case is shown in the left column of Figure~\ref{fig:03:refined-construction-2-regions-example}.
We now extend the extraction order's edge label assignment as follows: We introduce an artificial edge $(b_1, b_2)$ connecting the boundary nodes. This edge must lie on a confluence ending in $b_2$ in both root regions, since $b_1$ and $b_2$ are by definition reachable both from $r_1$ and $r_2$. We now apply the confluence edge label assignment to each of the root region subgraphs $\rootRegionSubgraph[r_1]$ and $\rootRegionSubgraph[r_2]$, where $(b_1, b_2)$ is included in both cases. After calculating the edge label assignment, the virtual edge is removed.

Let $i_1 \in \rootRegionNodes[r_1]$ and $i_2 \in \rootRegionNodes[r_2]$ be the common ancestors of $b_1$ and $b_2$, which are farthest from the corresponding root node. The resulting edge label assignment adds the label $b_2$ to every edge between $i_1$ and $b_2$, and $i_2$ and $b_2$. Note that the induced extraction order would have introduced the label $b_1$ \emph{and} $b_2$, and additionally would have propagated both labels to the root nodes $r_1$ and $r_2$. Therefore, as in the case of one boundary node, the edge label assignment adds one fewer label, and additionally extends the label assignment of fewer edges.

\paragraph{Arbitrary Boundary Size}
This approach can be extended to an arbitrary number $n \geq 2$ of boundary nodes by connecting the boundary nodes to a chain, i.e. for some ordering $(b_1, b_2, ... b_n)$ we introduce artificial edges $E' = \{(b_1, b_2), ... (b_{n-1}, b_n) \}$. This case is shown in the right column of Figure~\ref{fig:03:refined-construction-2-regions-example}.

Again, each boundary node is reachable from both root nodes. Therefore, each of the $n-1$ edges in the chain induces a new confluence, ending in the nodes $b_2, ... b_n$. Therefore, $n-1$ new labels are introduced, and the label $b_i$ with $2 \leq i \leq n$ is in both root regions propagated to the common ancestor of $b_1$ and $b_i$, which is farthest from the root node. Again, the induced extraction order would have introduced $n$ labels, and would have assigned them to \emph{every} edge between the root nodes and the corresponding boundary node.

We now introduce the \emph{extended root regions}, to define the extension of root regions with artificial edges as described above.
\begin{definition}(Extended Root Region) \label{def:extended-extraction-order}\\
Let $\reqDAGOrientation$ be an extraction order with root set $\rootSet$. Let $r \in \rootSet$ be a local root node with root region graph $\rootRegionSubgraphDef[r]$. 

For each local root node $r'$, whose root region graph $\rootRegionSubgraph[r']$ intersects with $\rootRegionSubgraph[r']$ in a set of boundary nodes $\rootRegionBoundaryPair[r][r'] = \{b_1, ..., b_n\}$, we define some ordering $\rootRegionBoundaryPairOrdering[r][r']$ of $\rootRegionBoundaryPair[r][r']$, which we chose w.l.o.g. as $\rootRegionBoundaryPairOrdering[r][r'] = (b_1, ..., b_n)$. 
We then define the extended root region $\rootRegionSubgraphExtended[r]$ as the directed graph obtained by adding a simple path connecting all boundary nodes to the extraction order's edge set according to this ordering, i.e.:
\begin{align}
\rootRegionSubgraphExtended[r] := (\rootRegionNodes[r], \rootRegion[r] \cup \{(b_1, b_2), ... (b_{n-1}, b_n) \}).
\end{align}
Lastly, we require that when constructing the extended root region for $r'$, $\rootRegionSubgraphExtended[r']$, the same ordering of the boundary nodes is used.
\end{definition}

Next, we define the \emph{multi-root confluence edge label assignment}.
\begin{definition}(Multi-Root Confluence Edge Label Assignment) \label{def:multiroot:extended-labels}\\  
Let $\reqDAGOrientationDef$ be an extraction order with local root set $\rootSet$. For each local root node $r \in \rootSet$, the confluence edge label assignment $\labelsetsEdgesRootRegion[r]$ of the extended root region graph $\rootRegionSubgraphExtended[r]$ is determined. The root region is a subgraph of the extended root region graph. Therefore, for each edge $e \in \rootRegion[r]$ in the root region, $\labelsetsEdgesRootRegion[r]$ contains a label set $\labelsetEdge[e]$. We then define the multi-root confluence edge label assignment by assigning to each edge this label set.
\end{definition}
Since the root regions partition all edges of the extraction order, this definition produces a well-defined label assignment for each edge.
We will now show that the multi-root confluence edge label assignment meets the definition of a multi-root decomposable edge label assignment given in Definition~\ref{def:03:mr-decomp-label-assignment}.
\begin{lemma}(Multi-root Confluence Edge Labels are Multi-root Decomposable) \label{corollary:ext-edge-labels-are-mr-decomposable} \\
The multi-root confluence edge label assignment described in Definition~\ref{def:multiroot:extended-labels} is multi-root decomposable.
\end{lemma}
\begin{proof}
We show that the extended edge label assignment is multi-root decomposable by verifying that it satisfies the defining properties of a decomposable edge label assignment within each root region subgraph according to Definition~\ref{def:decomposable-edge-labels}.

The first property requires that the label assignment is a superset of the confluence edge label assignment for the root region. In the confluence edge label assignment, an edge is labeled with some node $k$ if and only if it lies on some confluence ending in $k$. Since the root region is a subgraph of the extended root region, any such confluence in the root region must also be present in the extended root region. Therefore, these confluences are identified and labeled when applying the confluence label assignment to the extended extraction order.

The second property requires that the label-induced subgraph for each label node is rooted. Assume the contrary, i.e. there is some label $k$, such that the label-induced subgraph for $k$ is not rooted. Note that the subgraph is taken over the original root region, and does not contain the virtual edges. 

The third property requires that each node $k$ which occurs as a label must be contained in its label-induced subgraph. This holds for all nodes which are confluence end nodes in the original root region subgraph. Now consider the case where $k$ is a root region boundary node, which only becomes a label node through the addition of virtual edges. By removing the virtual edges after the label assignment, only one of the paths of any confluence ending in $k$ can be interrupted. A path from the label-induced subgraph's root node to $k$ exists within the original root region subgraph. This path is labeled with $k$ according to Lemma~\ref{lemma:decomplabels:all-paths-are-labeled}, and therefore, $k$ is part of its label-induced subgraph.

The fourth property requires that each of a node's incoming edges have the same label set. This is true, since Lemma~\ref{lemma:incomingLabelsUnique} holds for the confluence edge label assignment within each root region, and since the removal of the virtual edges between boundary nodes does not affect the edge label assignment for any other edges.

Finally, the fifth property requires that no node occurs as a label in any of its out-edges. This is satisfied because the root region is a subgraph of the extended root region: If any edge $(i, j) \in \rootRegion[r]$ were labeled with its tail node $i$, the same edge would have been labeled with $i$ in the edge label assignment for the extended extraction order, thus contradicting that the label assignment is for the extended extraction order is decomposable.
\end{proof}

Using the multi-root confluence edge label assignment, we now show that any mapping of one root region subgraph can be consistently extended by a mapping of an adjacent root region subgraph.
\begin{theorem} (Local Subgraph Extensibility) \label{thm:local-subgraph-extensibility}\\
Let $\reqDAGOrientation$ be an extraction order with root set $\rootSet = \{r_1, r_2\}$ and root regions $\rootRegionSet = \{\rootRegion[r_1], \rootRegion[r_2]\}$, such that $\rootRegionBoundaryPair[r_1][r_2] = \{b_1, \ldots b_n\}$, and such that there are no edges between the boundary nodes. 
Further, let $(\vec{y}, \vec{z}, \vec{a})$ be a solution of Linear Program~\ref{LP:RunningIntersectionProperty} for the extraction order $\reqDAGOrientation$, where the edge label assignment is the multi-root confluence edge label assignment from Definition~\ref{def:multiroot:extended-labels}.

Then, the following holds:    
If a valid mapping $m_1$ with non-zero flow for the subgraph $\rootRegionSubgraph[r_1]$ can be extracted by Algorithm~\ref{alg:decompositionAlgorithm-RIP}, then a valid mapping $m_2$  with non-zero flow can be extracted by applying Algorithm~\ref{alg:decompositionAlgorithm-RIP} to the subgraph $\rootRegionSubgraph[r_2]$, such that $\restrict[m_1][\rootRegionBoundaryPair[r_1][r_2]] = \restrict[m_2][\rootRegionBoundaryPair[r_1][r_2]]$. 
\end{theorem}
\begin{proof}
We show the lemma by induction over the number of boundary nodes $n$.

First, assume $n = 1$, i.e. there is a single boundary node $b_1$. According to Lemma~\ref{lemma:decomposition-rooted-subgraphs}, a mapping $m_1$ for $\reqTopologyRootRegionSubgraph[r_1]$ can be extracted. Let $u := m_1(b_1)$ be the substrate node to which $b_1$ was mapped. Since $b_1$ is reachable from $r_1$ in $\rootRegionSubgraph[r_1]$, it has some incoming edge $e_1 \in \rootRegion[r_1]$, that has also been mapped by the decomposition algorithm. The decomposition algorithm can only map this incoming edge in Line~\ref{algline:rip-decomp:map-edge-head-node} or Line~\ref{algline:rip-decomp:map-edge-head-node-reversed}, when $\subLP[y^u_{b_1}][\EOEdgeToOriginal(e_1), \restrict[m_1][\labelsetEdge[e_1]]] > 0$. By Constraint (\ref{alg:lp:novel:node-to-sub-node-mapping}), this implies that the global node-mapping variable $y^u_{b_1}$ is also non-zero. 

Now consider the execution of Algorithm~\ref{alg:decompositionAlgorithm-RIP} on $\rootRegionSubgraph[r_2]$. Again, the decomposition algorithm succeeds due to Lemma~\ref{lemma:decomposition-rooted-subgraphs}, and $b_1$ is reachable from $r_2$ within $\rootRegionSubgraph[r_2]$ by some edge $e_2 \in \rootRegion[r_2]$. As we have established $y^u_{b_1} > 0$, it also holds that there is some $m_{e_2} \in \MappingSpace[\labelsetEdge[e_2]]$ such that $\subLP[y^u_{b_1}][\EOEdgeToOriginal(e_2), m_{e_2}] > 0$, due to Constraint (\ref{alg:lp:novel:node-to-sub-node-mapping}). Therefore, at some point, the decomposition algorithm will generate a mapping $m_2$, such that $m_2(b_1) = u$.

Assume that the lemma holds for some $n \geq 1$. 

Consider a generalized extraction order with two root regions $\rootRegionSubgraph[r_1]$ and $\rootRegionSubgraph[r_2]$, and boundary size $n+1$, such that $\rootRegionBoundaryPair[r_1][r_2] = \{b_1, \ldots b_n, b_{n+1}\}$. Without loss of generality, let the ordering of the boundary nodes, according to which the artificial edges are added, be $(b_1, \ldots b_{n+1})$, such that we refer to the last boundary node in the chain of artificial edges as $b_{n+1}$. By Lemma~\ref{lemma:decomposition-rooted-subgraphs}, we can extract a mapping for $\rootRegionSubgraph[r_1]$, since it is a rooted subgraph of $\reqDAGOrientation$. We introduce the subgraph $\rootRegionSubgraph[r_2]'$, which is induced by the node set $\rootRegionNodes[r_2]' = \rootRegionNodes[r_2] \setminus \{b_{n+1}\}$. The size of the boundary between $\rootRegionSubgraph[r_1]$ and $\rootRegionSubgraph[r_2]'$ is $n$, and $\rootRegionSubgraph[r_2]'$ is also a rooted subgraph of $\reqDAGOrientation$. Therefore, by the induction hypothesis, a mapping for $\rootRegionSubgraph[r_2]'$ can be extracted. 

Next, we include the boundary node $b_{n+1}$ which was previously omitted. Given that $b_{n+1}$ is appended as the last node to the chain of artificial edges, any artificial edge $e_i := (b_{i}, b_{i+1})$ with $i \leq n$ is labeled with $b_{n+1}$. Once again, this holds because $b_{n+1}$ is reachable from $r_1$ and $r_2$ both directly and via $b_i$, inducing a confluence ending in $b_{n+1}$, which contains $e_i$. Any edge labeled with $b_i$ is also labeled with $b_{n+1}$. $b_{i}$ and $b_{n+1}$ share some common ancestor $a \in \reqDAGAncestorSet[\rootRegionBoundary[r_2]] \cap \rootRegionNodes[r_2]$, as they are both reachable from $r_2$. We define the boundary ancestor set $\boundaryAncestorSet[r_1][r_2][n+1]$ as the set of such first common ancestor nodes between $b_{n+1}$ and any other boundary node.

Let $a_1$ be the first ancestor node that is encountered by the decomposition algorithm for $\rootRegionSubgraph[r_2]'$ in some iteration $k$ in Line~\ref{algline:rip-decomp:retrieve-i-from-q}. Since $a_1$ is the first boundary ancestor, none of the other boundary nodes' mappings have been determined. By definition, there is some edge $e \in \outEdgesDAGOrder{a_1}$ that is labeled with $b_{n+1}$. Therefore, the decomposition algorithm applied to the subgraph $\rootRegionSubgraph[r_2]'$ extracts a mapping for $b_{n+1}$ before or at the same time as any other boundary node is mapped. Let $u := \mappingNodesIteration(b_{n+1})$ be the node to which $b_{n+1}$ is mapped. 

The extended edge labels are a decomposable edge label assignment according to Definition~\ref{def:decomposable-edge-labels}. Therefore the node mapping choice $b_{n+1} \mapsto u$ can only occur if $\subLP[y^u_{b_{n+1}}][e, \labelsetIncoming[b_{n+1}]] > 0$ for $e \in \inEdgesExtractionOrder{b_{n+1}}$. By the local connectivity property, there then also exists a non-zero flow in each LP subformulation associated with some edge $e \in \inEdgesDAGOrder{b_{n+1}}$.
These are precisely the edges in $\rootRegion[r_2] \setminus \rootRegion[r_2]'$. Since $b_{n+1}$ is the last boundary node, it holds that $\outEdgesDAGOrder{b_{n+1}} = \emptyset$. It follows that $\labelsetIncoming[b_{n+1}] = \{b_{n+1}\}$. Therefore, for any $i$ with $(i, b_{n+1}) \in \inEdgesDAGOrder{b_{n+1}}$, a mapping can be extracted, since the node mapping for $i$ was chosen consistently with $b_{n+1} \mapsto u$.
\end{proof}

\begin{remark} (Allowing Edges between Boundary Nodes) \\
In Remark~\ref{remark:multiroot-boundary-edges-restriction}, we restricted the approach to extraction orders in which the boundary nodes were not connected by edges. We can somewhat relax this requirement by allowing edges between boundary nodes, as long as the existing edges can be \emph{extended} to a chain by adding virtual edges, as described in Definition~\ref{def:extended-extraction-order}. In this case, the boundary edge is assigned to one of the two root regions, while a virtual edge must still be added to the other root region to obtain a consistent labeling. It is unclear how boundary edges which cannot be extended to a chain might be handled in this framework. We therefore explicitly forbid extraction orders where edges of the form $\{(b_1, b_2), (b_3, b_2)\}$ are contained for some boundary nodes $\{b_1, b_2 \}$.
\end{remark}

\subsection{The Multi-Root Decomposition Algorithm} \label{sec:multiroot:decomposition-alg}


\begin{figure}[p]

\removelatexerror

\begin{algorithm*}[H]

\SetKwInOut{Input}{\small{Input}}\SetKwInOut{Output}{\small{Output}}
\SetKwFunction{ProcessPath}{ProcessPath}{}{}
\SetKwFunction{reverse}{reverse}{}{}
\SetKwFunction{LP}{LP}

\newcommand{\SET}{\textbf{set~}}
\newcommand{\ADD}{\textbf{add~}}
\newcommand{\EACH}{\textbf{each~}}
\newcommand{\DEFINE}{\textbf{define~}}
\newcommand{\AND}{\textbf{and~}}
\newcommand{\LET}{\textbf{let~}}
\newcommand{\WITH}{\textbf{with~}}
\newcommand{\COMPUTE}{\textbf{compute~}}
\newcommand{\FIND}{\textbf{find~}}
\newcommand{\CHOOSE}{\textbf{choose~}}
\newcommand{\DECOMPOSE}{\textbf{decompose~}}
\newcommand{\FORALL}{\textbf{for all~}}
\newcommand{\OBTAIN}{\textbf{obtain~}}
\newcommand{\WITHPROBABILITY}{\textbf{with probability~}}

\caption{Multi-Root decomposition algorithm for solutions of LP \ref{LP:RunningIntersectionProperty}  for generalized extraction orders with tree-like root region extraction orders.}
\label{alg:decompositionAlgorithm-multiroot}

\Input{VNEP-instance $(\substrateTopology, \reqTopology)$, generalized extraction order $\reqDAGOrientationDef$ with tree-like root region extraction order $\RRExtractionOrderDef$, multi-root extraction label set ordering $\labelsetOrderSetMR$, solution $\lpvars$ to Linear Program ~\ref{LP:RunningIntersectionProperty} based on the multi-root confluence edge label assignment}
\Output{Convex combination~$\PotEmbeddings = \{\decomp = (\prob,\mappingRequestIteration)\}_k$ of valid mappings}

  \SET $(\vec{y}_0, \vec{z}_0, \vec{\gamma}_0, \vec{l}_0 ) \gets \lpvars$ \label{algline:mr-decomp:copy-lpvars}\\
  \ForEach{$r \in \rootSet$}{
	 \COMPUTE \text{convex combination of mappings} $\PotEmbeddings_r$ \text{for root region $\rootRegionSubgraphDef[r]$} \text{according to Algorithm \ref{alg:decompositionAlgorithm-RIP} with arguments } $(\VNEPInstance, \rootRegionSubgraph[r], \labelsetOrderSetRegion[r], \lpvars)$ \label{algline:mr-decomp:decomp-root-region} \\
     \SET $\lpvars \gets (\vec{y}_0, \vec{z}_0, \vec{\gamma}_0, \vec{l}_0 )$ \label{algline:mr-decomp:reset-lpvars}\\
  }
  
  \SET $\PotEmbeddings  \gets \emptyset$ \AND $k \gets 1$ \label{algline:mr-decomp:init-D-and-k}\\
  \SET $x \gets 1$ \label{algline:mr-decomp:init-x} \\
  \While{$x > 0$ \label{algline:mr-decomp:while-x-start}}
  { 
  	\SET $\prob \gets 1$ \\
    \SET $\mappingRequestIteration = (\mappingNodesIteration,\mappingEdgesIteration)~\gets (\emptyset,\emptyset)$\\
    
    \SET $\Queue \gets \{ \RRExtractionOrderRoot \}$ \\
    \While{$|\Queue| > 0$ \label{algline:mr-decomp:while-q-start}}{ 
        \CHOOSE $r \in \Queue$ \AND \SET$\Queue \gets \Queue \setminus \{r\}$ \\
        \CHOOSE $(f, \mappingRequestDef) \in \PotEmbeddings_r$ \textbf{such that} $\restrict[\mappingNodes][\rootRegionBoundaryIncoming[r]] = \restrict[\mappingNodesIteration][\rootRegionBoundaryIncoming[r]]$ \label{algline:mr-decomp:choose-mapping} \\
        \SET $\mappingRequestIteration \gets \mappingRequestIteration \cup \mappingRequest$ \label{algline:mr-decomp:extend-mapping} \\
        \If{$f < \prob$}{\SET $\prob \gets f$ }

        \ForEach{$(r, r') \in \outEdgesRR{r}$}{
       		\SET $\Queue \gets \Queue \cup \{ r' \}$\label{algline:mr-decomp:add-neighbor-root-to-Q}
	    }
	}
	
	\ForEach{$r \in \rootSet$}{
		\CHOOSE $(f, \mappingRequest) \in \PotEmbeddings_r$ \textbf{such that} $\mappingRequest = \restrict[\mappingRequestIteration][\rootRegionSubgraph[r]]$ \label{algline:mr-decomp:identify-RR-mapping} \\
		\SET $f \gets f - \prob$ \label{algline:mr-decomp:reduce-mapping-vals}\\
		\If{$f = 0$}{
			\SET$\PotEmbeddings_r \gets \PotEmbeddings_r \setminus (f, \mappingRequest)$ \label{algline:mr-decomp:remove-zero-value-mapping}\\
		}
	}
    \ADD $\decomp = (\prob,\mappingRequestIteration)$ to $\PotEmbeddings$ \\
    \SET $x \gets x - \prob$ \label{algline:mr-decomp:reduce-x}\\
    \SET $k \gets k + 1$ \label{algline:mr-decomp:increment-k}
  }

\KwRet{$\PotEmbeddings$}
\end{algorithm*}

\end{figure}

The decomposition algorithm for generalized extraction orders using tree-like root region extraction orders is given as pseudocode in Algorithm~\ref{alg:decompositionAlgorithm-multiroot}. The execution of the algorithm can be divided into two stages: In the first stage, each root region is processed individually, such that a convex combination of requests is obtained for each root region subgraph. In the second stage, the mappings of the root region subgraphs are stitched together to a mapping for the entire request. A correctness proof is given in Theorem~\ref{thm:multiroot:decomp-correctness}, towards the end of this subsection. We will now describe the execution of the algorithm informally.

In the first stage, the algorithm extracts a convex combination of mappings $\PotEmbeddings_r$ for each root region subgraph. By Lemma~\ref{lemma:decomposition-rooted-subgraphs}, this decomposition succeeds for each individual root region. In Line~\ref{algline:mr-decomp:copy-lpvars}, the variables of the linear program are copied. After performing the decomposition for a root region in Line~\ref{algline:mr-decomp:decomp-root-region}, the original values of the variables are restored in Line~\ref{algline:mr-decomp:reset-lpvars}. This step ensures that the decomposition succeeds in adjacent root regions, because the node mapping variables related to the mapping decision for boundary nodes are reduced to zero in Line~\ref{algline:rip-decomp:adapt-variables-one} of Algorithm~\ref{alg:decompositionAlgorithm-RIP}. Note that in a practical implementation, these steps can be avoided by excluding the variables in question from the variable reduction operation in Algorithm~\ref{alg:decompositionAlgorithm-RIP}.

In each iteration of the while-loop starting in Line~\ref{algline:mr-decomp:while-x-start}, the algorithm combines the root region mappings to a mapping for the entire request. Similarly to the decomposition in Algorithms~\ref{alg:decomposition:Algorithm-Novel-AC} and~\ref{alg:decompositionAlgorithm-RIP}, this process follows a graph traversal of the root region extraction order $\RRExtractionOrder$. 

Initially, the local root node $\RRExtractionOrderRoot$, which is also the root node of the root region extraction order, is added to the queue. In each iteration of the while-loop starting in Line~\ref{algline:mr-decomp:while-q-start}, a single root region is processed: A mapping is selected, which agrees with the partial mapping $\mappingRequestIteration$. The existence of such a mapping is guaranteed by Theorem~\ref{thm:local-subgraph-extensibility}. The theorem applies, because the root region extraction order is by assumption tree-like. Therefore, each root region intersects with at most one previously mapped region. The algorithm tracks the minimal mapping value encountered in this process in the variable $\prob$. Once a mapping for the root region has been selected, all adjacent root regions are added to the queue.

Once the entire root region extraction order is processed, the algorithm iterates over the root regions a second time to reduce the mapping value $f$ for each used mapping in the convex combination of mappings. This is analogous to the step of subtracting the minimal encountered variable value from the used variables in Lines~\ref{algline:rip-decomp:adapt-variables-one} and following of Algorithm~\ref{alg:decompositionAlgorithm-RIP}. The decomposition algorithm must ensure that any given mapping of a root region can only be selected while it has a non-zero value. The notation $\mappingRequest = \restrict[\mappingRequestIteration][\rootRegionSubgraph[r]]$ in Line~\ref{algline:mr-decomp:identify-RR-mapping} should be understood as identifying the mapping which was selected for the root region in Line~\ref{algline:mr-decomp:choose-mapping}, i.e. the mapping for $\rootRegionSubgraph[r]$ which agrees with the global mapping $\mappingRequestIteration$ both in terms of node and edge mappings.

\begin{theorem} (Correctness of Algorithm~\ref{alg:decompositionAlgorithm-multiroot}) \label{thm:multiroot:decomp-correctness}\\
Let $\VNEPInstance$ be a VNEP instance with a generalized extraction order $\reqDAGOrientationDef$ with root set $\rootSet$, a multi-root confluence edge label assignment $\labelsetsEdges$, and a multi-root extraction label set ordering $\labelsetOrderSetMR$. Given a solution $\lpvars$ of Linear Program~\ref{LP:RunningIntersectionProperty}, Algorithm~\ref{alg:decompositionAlgorithm-multiroot} returns a complete convex combination of valid mappings. Additionally, the resource allocations of the convex combination are bounded by the values of the corresponding allocation variables, i.e. $\sum_{(\prob, \mappingRequestIteration) \in \PotEmbeddings} \prob \cdot \allocationFunction(\mappingRequestIteration, x, y) \leq a^{x, y}$ holds for each substrate resource $(x, y) \in \substrateResources$.
\end{theorem}
\begin{proof}
We first show that the first stage of the algorithm, i.e. the decomposition of the root region subgraphs, is guaranteed to yield a complete convex combination of mappings for each root region subgraph. Since the LP variables are reset to their original values in Line~\ref{algline:mr-decomp:reset-lpvars}, the executions of Algorithm~\ref{alg:decompositionAlgorithm-RIP} in Line~\ref{algline:mr-decomp:decomp-root-region} for different root region subgraphs do not interfere with one another. Therefore, due to Lemma~\ref{lemma:decomposition-rooted-subgraphs}, each subgraph decomposition succeeds, and for each local root $r \in \rootSet$, the mapping list $\PotEmbeddings_r$ contains a complete convex combination of mappings for the root region subgraph $\rootRegionSubgraph[r]$.

We next consider the second stage of the algorithm, where the mappings obtained for each root region are combined to mappings for the entire request topology. The following invariant holds in each iteration of the while-loop starting in Line~\ref{algline:mr-decomp:while-x-start} for each root region: The value of the variable $x$ is equal to the total remaining value of the mappings for the root region, i.e. $x = \sum_{(f, \mappingRequest) \in \PotEmbeddings_r} f$ holds for each $r\in \rootSet$. The invariant holds initially, since $x=1$ is initialized in Line~\ref{algline:mr-decomp:init-x}, and $\sum_{(f, \mappingRequest) \in \PotEmbeddings_r} f = 1$ holds for each root region, since $\PotEmbeddings_r$ is initially a convex combination of mappings. The invariant is preserved, since in each iteration the same value $\prob$ is subtracted from each $\PotEmbeddings_r$ and from $x$ in Lines~\ref{algline:mr-decomp:reduce-mapping-vals} and~\ref{algline:mr-decomp:reduce-x}, respectively.

Due to this invariant, the initial choose-operation to select a mapping in Line~\ref{algline:mr-decomp:choose-mapping} is guaranteed to succeed for the root region extraction order's root node $\RRExtractionOrderRoot$. Next, we consider the execution of the while-loop starting in Line~\ref{algline:mr-decomp:while-q-start} for some other local root node $r$. Due to the tree-like root region extraction order, the root region subgraph $\rootRegionSubgraph[r]$ shares a boundary with exactly one root region $\rootRegionSubgraph[r']$ which has been processed in a previous iteration, such that $\rootRegionBoundaryIncoming[r] = \rootRegionBoundaryPair[r][r']$ holds. The local subgraph extensibility from Theorem~\ref{thm:local-subgraph-extensibility} then guarantees that some mapping $\mappingRequestDef$ exists for $\rootRegionSubgraph[r]$ in $\PotEmbeddings_r$, such that $\restrict[\mappingNodes][\rootRegionBoundaryPair[r][r']] = \restrict[\mappingNodesIteration][\rootRegionBoundaryPair[r][r']]$ holds, and the choose operation therefore succeeds. Since the root regions are edge-disjoint, it follows that the mapping $\mappingRequestIteration$ for the entire request can be extended by $\mappingRequest$ in Line~\ref{algline:mr-decomp:extend-mapping}. 

Since the root region extraction order is a rooted tree, every root region is processed in some iteration of the inner while-loop, and it follows that in each iteration, the resulting mapping $\mappingRequestIteration$ is defined for the entire request.

We next show that the algorithm terminates: In each iteration of the outer while loop, at least one mapping is removed from its mapping list $\PotEmbeddings_r$ in Line~\ref{algline:mr-decomp:remove-zero-value-mapping}, since $\prob$ tracks the minimal value of all subgraph mappings that are included in the mapping for the full request. The number of distinct mappings in each root region's mapping list is bounded by the number of LP variables associated with the root region. Therefore, after a finite number of iterations, each list $\PotEmbeddings_r$ is empty. Due to the invariant relating $x$ to the remaining mapping values, the condition of the while-loop in Line~\ref{algline:mr-decomp:while-x-start} eventually fails, and the algorithm terminates.

Lastly, we show that the resulting list of mappings for the entire request, $\PotEmbeddings$, is a complete convex combination of mappings, whose resource allocations are bounded by the initial values of the allocation variables. The variable $x$ is initialized with $1$, and in each iteration of the outer while-loop, one of the values $\prob$ with $(\prob, \mappingRequestIteration) \in \PotEmbeddings$ is subtracted from $x$. After the algorithm exits the while-loop, $x=0$ holds, and it follows that $\sum_{k = 1}^{|\PotEmbeddings|} \prob = 1$ holds.
We now consider the bound on the resource allocations. Due to Theorem~\ref{thm:decomp_of_rip_orderings},  the bound holds for each root region separately, i.e. for each $r \in \rootSet$, $\sum_{(\prob, \mappingRequest) \in \PotEmbeddings_r} \prob \cdot \allocationFunction(\mappingRequest, x, y) \leq a^{x, y}$ holds. Since each root region mapping $(f, \mappingRequest) \in \PotEmbeddings$ is included in global mappings with a total mapping value of $f$, the bound also applies to the convex combination of mappings $\PotEmbeddings$ returned by the algorithm. 
\end{proof}

\subsection{Complexity of the Multi-Root Algorithm} \label{sec:multiroot:complexity}

First, note that Theorem~\ref{thm:adapted-alg-lp-size} regarding the size of Linear Program~\ref{LP:RunningIntersectionProperty} can be generalized, as the LP formulation itself is not modified.
\begin{theorem} (Size of Linear Program~\ref{LP:RunningIntersectionProperty} for Generalized Extraction Orders) \label{thm:multiroot:LP-size} \\
Let $\VNEPInstance$ be a VNEP instance and let $\reqDAGOrientationDef$ be a generalized extraction order with root set $\rootSet$. Given the multi-root confluence edge label assignment $\labelsetsEdges$ and a multi-root extraction label set ordering $\labelsetOrderSetMR$, the size of Linear Program~\ref{LP:RunningIntersectionProperty} is bounded by  $\bigO\left(\sum_{r \in \rootSet}|\substrateTopology|^{\labelWidthEQ(\rootRegionSubgraph[r], \labelsetOrderSetRegion[r])} \cdot |\rootRegionSubgraph[r]| \right)$.
\end{theorem}
\begin{proof}
Due to Theorem~\ref{thm:adapted-alg-lp-size}, the root region subgraph rooted in $r \in \rootSet$ contributes with $\bigO\left(|\substrateTopology|^{\labelWidthEQ(\rootRegionSubgraph[r], \labelsetOrderSetRegion[r])} \cdot |\rootRegionSubgraph[r]| \right)$  to the size of the LP formulation. The result then follows by summation over the root set.
\end{proof}

We next consider the runtime of Algorithm~\ref{alg:decompositionAlgorithm-multiroot}.
\begin{theorem} (Runtime of Algorithm~\ref{alg:decompositionAlgorithm-multiroot} for Generalized Extraction Orders) \label{thm:multiroot:decomp-runtime} \\
Let $\VNEPInstance$ be a VNEP instance and let $\reqDAGOrientationDef$ be a generalized extraction order with root set $\rootSet$. Given the multi-root confluence edge label assignment $\labelsetsEdges$ and a multi-root extraction label set ordering $\labelsetOrderSetMR$, and a solution $\lpvars$ of Linear Program~\ref{LP:RunningIntersectionProperty}, the runtime of Algorithm~\ref{alg:decompositionAlgorithm-multiroot} is bounded by
$\bigO \left(\sum_{r \in \rootSet}  |\substrateTopology|^{2 \cdot \labelWidthEQ(\rootRegionSubgraph[r], \labelsetOrderSetRegion[r]) + 1} \cdot |\rootRegionSubgraph[r]|^3 \right)$.
\end{theorem}
\begin{proof}
We again partition the execution of Algorithm~\ref{alg:decompositionAlgorithm-multiroot} in two stages. In the first stage, the decomposition of each root region is performed in Lines~\ref{algline:mr-decomp:copy-lpvars} to~\ref{algline:mr-decomp:reset-lpvars}. Due to Theorem~\ref{thm:adapted-alg-runtime}, the corresponding runtime of 
\begin{align} \label{eq:mr-decomp-runtime:stage-1}
\bigO\left(\sum_{r \in \rootSet} |\substrateTopology|^{2 \cdot \labelWidthEQ(\rootRegionSubgraph[r], \labelsetOrderSetRegion[r]) + 1} \cdot |\rootRegionSubgraph[r]|^2 \right)
\end{align}
follows directly by summation over the root set.

In the second stage, the mappings of root regions are combined to an overall mapping of the request in Lines~\ref{algline:mr-decomp:init-D-and-k} to~\ref{algline:mr-decomp:increment-k}. In each iteration, (at least) one mapping of one root region subgraph is discarded in Line~\ref{algline:mr-decomp:remove-zero-value-mapping}. The number of iterations is therefore bounded by the total number of mappings contained in the convex combination of each root region. The number of mappings for a single root region subgraph is in turn bounded by the number of LP variables associated with this subgraph. Due to Theorem~\ref{thm:adapted-alg-lp-size}, within a single root region subgraph $\rootRegionSubgraph$, this number is bounded by $\bigO\left(|\substrateTopology|^{\labelWidthEQ(\reqExtractionOrder, \labelsetOrderSet)} \cdot |\reqTopology| \right)$. By summation over the root set, we can bound the total number of mappings in all convex combinations by $\bigO\left( \sum_{r \in \rootSet}|\substrateTopology|^{\labelWidthEQ(\rootRegionSubgraph[r], \labelsetOrderSetRegion[r])} \cdot |\rootRegionSubgraph[r]| \right)$.

Within each iteration of the loop starting in Line~\ref{algline:mr-decomp:while-x-start}, the algorithm performs two traversals of the root region extraction order: First, a mapping is selected for each root region in Lines~\ref{algline:mr-decomp:while-q-start} to~\ref{algline:mr-decomp:add-neighbor-root-to-Q}. In the second traversal of the root region, all values of chosen root region mappings are reduced by $\prob$. For each root region, the choose-operation in Line~\ref{algline:mr-decomp:choose-mapping} can be performed in $\bigO(|\rootRegionBoundary[r]|)$ amortized time if the mappings in $\PotEmbeddings_r$ are stored in a hash table indexed by the mappings of the boundary nodes. The extension of the mapping in Line~\ref{algline:mr-decomp:extend-mapping} takes $\bigO(|\rootRegionSubgraph[r]|)$ time.

Therefore, the second stage of the algorithm has runtime 
\begin{align}\label{eq:mr-decomp-runtime:stage-2}
\bigO\left( \sum_{r \in \rootSet}|\substrateTopology|^{\labelWidthEQ(\rootRegionSubgraph[r], \labelsetOrderSetRegion[r])} \cdot |\rootRegionBoundaryIncoming[r]| \cdot |\rootRegionSubgraph[r]|^2 \right).
\end{align}
Since $|\rootRegionBoundaryIncoming[r]| < |\rootRegionSubgraph[r]|$, and since $|\substrateTopology|^{\labelWidthEQ(\rootRegionSubgraph[r], \labelsetOrderSetRegion[r])}  <  |\substrateTopology|^{2 \cdot \labelWidthEQ(\rootRegionSubgraph[r], \labelsetOrderSetRegion[r]) + 1}$, we can combine (\ref{eq:mr-decomp-runtime:stage-1}) and (\ref{eq:mr-decomp-runtime:stage-2}) to obtain the overall bound.
\end{proof}

We have now shown a bound on the LP size and algorithm runtime for the multi-root algorithm. Note that this is not necessarily a tight bound. The important result is rather that the algorithm is \emph{fixed-parameter tractable} with respect to the extraction label width. We will now discuss the impact of the additional labels introduced by the multi-root confluence edge label assignment. We first show an upper bound for the maximal increase of the extraction label width from using the multi-root confluence edge label assignment.
\begin{lemma} (Increase in Extraction Label Width due to Multi-Root Confluence Edge Labels) \label{lemma:increase-ELW-MR-confluence-edge-labels}\\
Let $\reqDAGOrientationDef$ be a generalized extraction order with root set $\rootSet$. Given a local root node $r \in \rootSet$ with root region subgraph $\rootRegionSubgraphDef[r]$, let $\labelsetsEdges$ be the regular confluence edge label assignment for $\rootRegionSubgraph[r]$ (see Definition~\ref{def:confluence-edge-labels}), and let $\labelsetsEdgesTilde$ be the multi-root confluence label assignment (see Definition~\ref{def:multiroot:extended-labels}) restricted to the root region. Then, it holds that
\begin{align}
\labelWidthEQ(\rootRegionSubgraph[r], \labelsetOrderSetTilde) \leq \labelWidthEQ(\rootRegionSubgraph[r], \labelsetOrderSet) + |\rootRegionBoundary[r]| - 1,
\end{align}
where $\labelsetOrderSet$ and $\labelsetOrderSetTilde$ are the minimal-width extraction label set orderings for $\labelsetsEdges$ and $\labelsetsEdgesTilde$, respectively.
\end{lemma}
\begin{proof}
Recall that the multi-root confluence edge label assignment $\labelsetsEdgesTilde$ is defined using the extended root region (see Definition~\ref{def:extended-extraction-order}).  In particular, for each adjacent root region $r'$, the pairwise boundary nodes $|\rootRegionBoundaryPair[r][r']|$ are connected with edges according to some ordering $\rootRegionBoundaryPairOrdering[r][r']$.  

We now consider the \emph{maximal} number of confluence end nodes that can be \emph{newly} introduced in the extended root region. This number is maximized, if none of the boundary nodes are confluence end nodes in the original root region subgraph $\rootRegionSubgraph[r]$. We therefore assume that the boundary nodes are not confluence end nodes in the original root region subgraph. 

Let $b_1 \in \rootRegionBoundaryPair[r][r']$ be the first node according to the ordering $\rootRegionBoundaryPairOrdering[r][r']$. If $b_1$ was not originally a confluence end node, it follows that $b_1$ is not a confluence end node in the extended root region subgraph, since at most a single outgoing edge $(b_1, b_2)$ is added to $b_1$. Note that if $b_1$ is a confluence end node in the original root region subgraph, the label $b_1$ is not added to any other edges due to Property~\ref{def:decomposability-extended-edge-labels:labels-end-in-label-node} of Definition~\ref{def:decomposable-edge-labels} stating that $k \not\in \labelsetEdge[e]$ for any $e \in \outEdgesExtractionOrder{e}$.

Next, consider any $b_i$ with $i \in \{2, ..., |\rootRegionBoundaryPair[r][r']|\}$, such that $b_i$ occurs in the $i$-th position in $\rootRegionBoundaryPairOrdering[r][r']$. By definition of the root region boundary, every boundary node is reachable from the root node. Due to the added edge $(b_{i-1}, b_i)$ in the extended root region, a second path via $b_{i-1}$ exists. Therefore, any boundary node other than $b_1$ becomes a confluence end node in the extended root region. 

Therefore, the total number of label nodes in the multi-root confluence edge label assignment may increase by at most $|\rootRegionBoundaryPair[r][r']|-1$ due to the pairwise boundary with the root region associated with $r'$. By adding the contributions of all adjacent root regions, we find that the multi-root confluence edge label assignment may  contain at most $|\rootRegionBoundary[r]| - 1$ more label nodes than the regular confluence edge label assignment. Denote the set of newly introduced label nodes as $L^+$.

Now, we consider the worst case, in which each of these additional label nodes is added to the edge label assignment for every edge $e \in \rootRegion[r]$. Consider the extraction label set ordering $\labelsetOrderSet$ for the confluence edge label assignment. We can obtain an extraction label set ordering $\labelsetOrderSetTilde$ for the multi-root confluence edge label assignment by extending every label set in $\labelsetOrderSet$ with $L^+$. That is, for every node $i \in \rootRegionNodes[r]$ with the label set ordering $\labelsetOrder_i = (\labelsetIndexed{1}, ..., \labelsetIndexed{n})$, we set $\labelsetOrderTilde_i := (\labelsetIndexed{1} \cup L^+, ..., \labelsetIndexed{n} \cup L^+)$. 

Then, $\labelsetOrderSetTilde$ is an extraction label set ordering for the multi-root confluence edge label assignment with width 
\begin{align}
\labelWidthEQ(\rootRegionSubgraph[r], \labelsetOrderSetTilde) = \labelWidthEQ(\rootRegionSubgraph[r], \labelsetOrderSet) + |L^+| \leq \labelWidthEQ(\rootRegionSubgraph[r], \labelsetOrderSet) + |\rootRegionBoundary[r]| - 1,
\end{align}
concluding the proof.
\end{proof}
This result shows that the multi-root approach is particularly useful when the generalized extraction order can be partitioned in such a way that the pairwise boundaries between root regions are small. In particular, if root regions are separable with a single node, the multi-root approach does not impact the extraction label width negatively.

We next compare the multi-root confluence edge label assignment to the confluence label assignment of the conceptually simpler induced extraction order approach presented in Section~\ref{sec:multiroot:induced-EO}. Specifically, we show that the approach of the multi-root algorithm results in a smaller edge label assignment.
\begin{lemma} (Comparison of Multi-Root Confluence Labels to Induced Extraction Order) \label{multiroot:comparison-MRCEL-inducedEO} \\
Let $\reqDAGOrientationDef$ be a generalized extraction order with root set $\rootSet$. Further, let $\labelsetsEdges$ be a multi-root confluence edge label assignment (see Definition~\ref{def:multiroot:extended-labels}) and let $\labelsetsEdgesTilde$ be the confluence edge label assignment of the induced extraction order of $\reqDAGOrientation$. 

Given a local root node $r \in \rootSet$ with root region subgraph $\rootRegionSubgraphDef[r]$, for all edges $e \in \rootRegion[r]$ in the root region, $\labelsetEdge[e] \subseteq \labelsetEdgeTilde[e]$ holds. 
\end{lemma}
\begin{proof}
First, note that both $\labelsetsEdges$ and $\labelsetsEdgesTilde$ extend the regular confluence edge label assignment for the root region subgraph $\rootRegionSubgraph[r]$ only by adding boundary nodes to the edge label assignment.

Consider some edge $e =(i,j) \in \rootRegion[r]$. In the confluence edge label assignment of the induced extraction order, every boundary node that can be reached from $j$ occurs in the label set $\labelsetEdgeTilde$: Since each boundary node is also reachable from the super-root node $\superroot$ through a node-disjoint path traversing another root region, $e$ lies on a confluence ending in the boundary node.

On the other hand, in the edge label assignment $\labelsetsEdges$, a boundary node $b\in \rootRegionBoundary$ that is reachable from $j$ does not occur as a label in $\labelsetEdge$ if either of the following hold: Firstly, if the boundary node is the first node in the ordering $\rootRegionBoundaryPairOrdering[r][r']$ used to derive the extended extraction order for some $r' \in \rootSet \setminus \{r\}$, it may not be a confluence end node in the extended root region. Secondly, if $j$ is separable from the boundary node, i.e. if a single node $k \in \rootRegionNodes[r]$ exists such that all paths from $j$ to $b$ pass through $k$, then $b\not\in \labelsetEdge$ holds. 

Since $\labelsetEdgeTilde[e]$ contains all boundary nodes reachable from $j$ and $\labelsetEdge[e]$ only contains some of them, it follows that $\labelsetEdge[e] \subseteq \labelsetEdgeTilde[e]$ holds.
\end{proof}
As a consequence, we conclude that the extraction label width using the multi-root confluence edge label assignment is bounded above by the extraction label width of the induced extraction order.
In addition to this result, the induced extraction order adds virtual edges between the super-root node and each local root node $r$, each of which is assigned $|\rootRegionBoundary[r]|$ many labels.

From the proof of Lemma~\ref{multiroot:comparison-MRCEL-inducedEO}, we can also conclude that the use of the multi-root algorithm is particularly useful if the boundary set of the root region is separable from the rest of the root region by a single node. In this case, the additional labels due to the root region boundary nodes are only propagated to the separating node and the remaining root region remains unaffected.

\begin{remark} (Ordering of Boundary Nodes) \\
One important aspect is not considered in the results presented in this section. The multi-root confluence edge label assignment assumes a choice of the ordering of the root region boundary nodes $\rootRegionBoundaryPairOrdering$ in the construction of the extended root region (see Definition~\ref{def:extended-extraction-order}). This ordering impacts the edge label assignment in non-trivial ways. A detailed study of this problem of selecting the boundary node ordering exceeds the scope of this thesis.
\end{remark}

\cleardoublepage

\section{Conclusion} \label{sec:conclusion}

In this section, we first summarize the important results of the thesis, and subsequently discuss possible avenues for future work.

\subsection{Summary of Results}
In this thesis, two extensions to the VNEP approximation algorithm by Rost and Schmid are presented. 
The first extension in Section~\ref{sec:hierarchical-bags} using extraction label set orderings significantly improves upon the base algorithm in terms of size of the underlying linear program and accordingly can drastically reduce the runtime of the algorithm. We show that the resulting algorithm is fixed-parameter tractable with respect to the new extraction label width parameter (cf. Theorems~\ref{thm:adapted-alg-lp-size} and~\ref{thm:adapted-alg-runtime}). Since the extraction label width is bounded by the extraction width which parameterizes the complexity of the base algorithm, the extended algorithm may significantly improve runtimes compared to the base algorithm.

Several results derived by Rost and Schmid for the extraction width are generalized to the new extraction label width parameter. In particular, we consider the class of half-wheel graphs, whose extraction width grows linearly when a suboptimal root placement is imposed, and show that the extraction label width remains a constant value (cf. Lemma~\ref{lemma:label-width-of-half-wheel-graphs}). Another result by Rost and Schmid stating that the  adding a path parallel to an existing edge increases the extraction width by at most the maximal degree of the request topology is generalized. We show that the extraction label width increases by at most one (cf. Theorem~\ref{thm:extraction-width-adding-parallel-paths}). 

The complexity of deriving optimal extraction label set orderings is investigated, and it is shown that finding optimal extraction label set orderings for a specific extraction order is a $\complexityNP$-hard problem (cf. Theorem~\ref{thm:hardness-extraction-label-ordering}).

In the second extension of the algorithm in Section~\ref{sec:multiroot}, the notion of a rooted extraction order is generalized to allow for more general extraction orders containing multiple root regions. The main result is the multi-root decomposition algorithm described in Section~\ref{sec:multiroot:decomposition-alg}, which is also fixed-parameter tractable with respect to the extraction label width (cf. Theorems~\ref{thm:multiroot:LP-size} and~\ref{thm:multiroot:decomp-runtime}). Lastly, the impact of placing additional root nodes on the extraction label width parameter is considered, with the result that each additional root region increases the extraction label width at most by the number of nodes in the root region boundary.

\subsection{Future Work}

A more detailed investigation of the connection between the extraction label width and the request topology, rather than a specific extraction order, might provide further insight into the hardness of computing optimal extraction label set orderings. In particular, the question of whether some connection between the extraction label width and the treewidth of the request topology exists is of high interest. 

The multi-root algorithm requires a specific high-level structure of the extraction order, namely a tree-like root region extraction order. Although a preprocessing procedure is presented (cf. Remark~\ref{remark:multiroot:generalizing-non-treelike}) which ensures that arbitrary multi-root extraction orders have this property, this procedure relies on the induced extraction order, which is shown to result in larger width (cf. Lemma~\ref{multiroot:comparison-MRCEL-inducedEO}). An algorithm that directly supports extraction orders with non-treelike root region extraction orders is therefore of interest. Additionally, the multi-root algorithm presented in this work places restrictions on the existence of edges between the root regions's boundary node, forbidding such edges in configurations which can not be extended to a chain connecting all such boundary nodes. A multi-root algorithm without this restriction could operate on arbitrary acyclic orientations of the request topology.


\cleardoublepage
\bibliographystyle{plain}
\bibliography{references}

\begin{thebibliography}{10}

\bibitem{arnborg1985efficient}
Stefan Arnborg.
\newblock Efficient algorithms for combinatorial problems on graphs with
  bounded decomposability—a survey.
\newblock {\em BIT Numerical Mathematics}, 25(1):1--23, 1985.

\bibitem{arnborg1987complexity}
Stefan Arnborg, Derek~G Corneil, and Andrzej Proskurowski.
\newblock Complexity of finding embeddings in ak-tree.
\newblock {\em SIAM Journal on Algebraic Discrete Methods}, 8(2):277--284,
  1987.

\bibitem{beeri1983desirability}
Catriel Beeri, Ronald Fagin, David Maier, and Mihalis Yannakakis.
\newblock On the desirability of acyclic database schemes.
\newblock {\em Journal of the ACM (JACM)}, 30(3):479--513, 1983.

\bibitem{bodlaender1986classes}
Hans~L. Bodlaender.
\newblock {\em Classes of graphs with bounded tree-width}, volume~86.
\newblock Unknown Publisher, 1986.

\bibitem{bodlaender1996linear}
Hans~L Bodlaender.
\newblock A linear-time algorithm for finding tree-decompositions of small
  treewidth.
\newblock {\em SIAM Journal on computing}, 25(6):1305--1317, 1996.

\bibitem{bodlaender1996arboretum}
Hans~L Bodlaender.
\newblock {\em A partial k-arboretum of graphs with bounded treewidth}, volume
  1996.
\newblock Utrecht University: Information and Computing Sciences, 1996.

\bibitem{bodlaender1995approximating}
Hans~L Bodlaender, John~R Gilbert, Hj{\'a}lmtyr Hafsteinsson, and Ton Kloks.
\newblock Approximating treewidth, pathwidth, frontsize, and shortest
  elimination tree.
\newblock {\em Journal of Algorithms}, 18(2):238--255, 1995.

\bibitem{bodlaender1993pathwidth}
Hans~L Bodlaender and Rolf~H M{\"o}hring.
\newblock The pathwidth and treewidth of cographs.
\newblock {\em SIAM Journal on Discrete Mathematics}, 6(2):181--188, 1993.

\bibitem{chowdhury2009virtual}
NM~Mosharaf~Kabir Chowdhury, Muntasir~Raihan Rahman, and Raouf Boutaba.
\newblock Virtual network embedding with coordinated node and link mapping.
\newblock In {\em INFOCOM 2009, IEEE}, pages 783--791. IEEE, 2009.

\bibitem{courcelle1990monadic}
Bruno Courcelle.
\newblock The monadic second-order logic of graphs. i. recognizable sets of
  finite graphs.
\newblock {\em Information and computation}, 85(1):12--75, 1990.

\bibitem{cowell2006probabilistic}
Robert~G Cowell, Philip Dawid, Steffen~L Lauritzen, and David~J Spiegelhalter.
\newblock {\em Probabilistic networks and expert systems: Exact computational
  methods for Bayesian networks}.
\newblock Springer Science \& Business Media, 2006.

\bibitem{dantzig1951maximization}
George~B Dantzig.
\newblock Maximization of a linear function of variables subject to linear
  inequalities.
\newblock {\em New York}, 1951.

\bibitem{fagin1983degrees}
Ronald Fagin.
\newblock Degrees of acyclicity for hypergraphs and relational database
  schemes.
\newblock {\em Journal of the ACM (JACM)}, 30(3):514--550, 1983.

\bibitem{fischer2013virtual}
Andreas Fischer, Juan~Felipe Botero, Michael~Till Beck, Hermann De~Meer, and
  Xavier Hesselbach.
\newblock Virtual network embedding: A survey.
\newblock {\em IEEE Communications Surveys \& Tutorials}, 15(4):1888--1906,
  2013.

\bibitem{gill1986projected}
Philip~E Gill, Walter Murray, Michael~A Saunders, John~A Tomlin, and Margaret~H
  Wright.
\newblock On projected newton barrier methods for linear programming and an
  equivalence to karmarkar’s projective method.
\newblock {\em Mathematical programming}, 36(2):183--209, 1986.

\bibitem{jensen1996introduction}
Finn~V Jensen.
\newblock {\em An introduction to Bayesian networks}, volume 210.
\newblock UCL press London, 1996.

\bibitem{karmarkar1984new}
Narendra Karmarkar.
\newblock A new polynomial-time algorithm for linear programming.
\newblock In {\em Proceedings of the sixteenth annual ACM symposium on Theory
  of computing}, pages 302--311. ACM, 1984.

\bibitem{khachiyan1980polynomial}
Leonid~G Khachiyan.
\newblock Polynomial algorithms in linear programming.
\newblock {\em USSR Computational Mathematics and Mathematical Physics},
  20(1):53--72, 1980.

\bibitem{lauritzen1988local}
Steffen~L Lauritzen and David~J Spiegelhalter.
\newblock Local computations with probabilities on graphical structures and
  their application to expert systems.
\newblock {\em Journal of the Royal Statistical Society. Series B
  (Methodological)}, pages 157--224, 1988.

\bibitem{maier1983theory}
David Maier.
\newblock {\em The theory of relational databases}, volume~11.
\newblock Computer science press Rockville, 1983.
\newblock Available at
  \url{http://web.cecs.pdx.edu/~maier/TheoryBook/TRD.html}.

\bibitem{matouvsek1991algorithms}
Ji{\v{r}}{\'\i} Matou{\v{s}}ek and Robin Thomas.
\newblock Algorithms finding tree-decompositions of graphs.
\newblock {\em Journal of Algorithms}, 12(1):1--22, 1991.

\bibitem{mehraghdam2014specifying}
Sevil Mehraghdam, Matthias Keller, and Holger Karl.
\newblock Specifying and placing chains of virtual network functions.
\newblock In {\em Cloud Networking (CloudNet), 2014 IEEE 3rd International
  Conference on}, pages 7--13. IEEE, 2014.

\bibitem{nielsen2009bayesian}
Thomas~Dyhre Nielsen and Finn~Verner Jensen.
\newblock {\em Bayesian networks and decision graphs}.
\newblock Springer Science \& Business Media, 2009.

\bibitem{raghavan1987randomized}
Prabhakar Raghavan and Clark~D Tompson.
\newblock Randomized rounding: a technique for provably good algorithms and
  algorithmic proofs.
\newblock {\em Combinatorica}, 7(4):365--374, 1987.

\bibitem{reed1992computing-treewidth}
Bruce~A Reed.
\newblock Finding approximate separators and computing tree width quickly.
\newblock In {\em Proceedings of the twenty-fourth annual ACM symposium on
  Theory of computing}, pages 221--228. ACM, 1992.

\bibitem{robertson1986treewidth}
Neil Robertson and Paul~D. Seymour.
\newblock Graph minors. ii. algorithmic aspects of tree-width.
\newblock {\em Journal of algorithms}, 7(3):309--322, 1986.

\bibitem{rost:charting-complexity-of-vne-2018}
Matthias Rost and Stefan Schmid.
\newblock {Charting the Complexity Landscape of Virtual Network Embeddings}.
\newblock In {\em Proceedings IFIP Networking}, 2018.
\newblock Extended technical report available at arXiv:1801.03162 [cs.NI].
  [Online]. Available: \url{http://arxiv.org/abs/1801.03162}.

\bibitem{rostSchmidFPTApproximations}
Matthias Rost and Stefan Schmid.
\newblock {(FPT-)Approximation Algorithms for the Virtual Network Embedding
  Problem}.
\newblock Technical Report arXiv:1803.04452 [cs.NI], March 2018.

\bibitem{rost:vne-approx-leveraging-rand-round-2018}
Matthias Rost and Stefan Schmid.
\newblock {Virtual Network Embedding Approximations: Leveraging Randomized
  Rounding}.
\newblock In {\em Proceedings IFIP Networking 2018}, 2018.
\newblock Extended technical report available at arXiv:1803.03622 [cs.NI].
  [Online]. Available: \url{http://arxiv.org/abs/1803.03622}.

\bibitem{rost2014:its_about_time}
Matthias Rost, Stefan Schmid, and Anja Feldmann.
\newblock It's about time: On optimal virtual network embeddings under temporal
  flexibilities.
\newblock In {\em Parallel and Distributed Processing Symposium, 2014 IEEE 28th
  International}, pages 17--26. IEEE, 2014.

\end{thebibliography}

\end{document}